\documentclass[oneside,12pt]{IIScthesisPSnPDF}
\usepackage{amssymb,etoolbox}
\usepackage{amsmath,charter}
\usepackage{latexsym}
\usepackage{multicol}
\usepackage{booktabs}
\usepackage{multirow}
\usepackage{cleveref}
\usepackage{wrapfig}
\usepackage{tabularx}
\usepackage[ruled,vlined,linesnumbered,algochapter]{algorithm2e}
\usepackage{algpseudocode}
\algdef{SE}[PROCEDURE]{Procedure}{EndProcedure}%
   [2]{\algorithmicprocedure\ \textproc{#1}\ifthenelse{\equal{#2}{}}{}{(#2)}}%
   {\algorithmicend\ \algorithmicprocedure}%

\usepackage{enumerate}
\usepackage{cancel,subcaption}
\usepackage{mathrsfs}
\usepackage{color,xcolor}
\usepackage[compact]{titlesec}
\usepackage{mdwlist}
\usepackage{tikz}
\usepackage{varwidth}
\usepackage[vertfit]{breakurl}
\usepackage{datetime}
\usepackage{pdfpages,pifont}
\usepackage{amsthm}
\usepackage{thm-restate,thmtools,xfrac,multirow} 
\usetikzlibrary{shapes,arrows, trees}
\usepackage{charter}
\usepackage{tikz}
\usetikzlibrary{positioning}
\usetikzlibrary{trees}
\usepackage{forest}
\usetikzlibrary{shapes,arrows}
\usetikzlibrary{intersections}
\usepackage{float}
\usepackage{xcolor}
\usepackage{subcaption}
\usepackage{tabularx}

\newtheorem{example}{Example}[chapter]
\newtheorem*{nonexample}{Example}
\newtheorem{theorem}{Theorem}[chapter]
\newtheorem{lemma}{Lemma}[chapter]
\newtheorem{corollary}{Corollary}[chapter]
\newtheorem{definition}{Definition}[chapter]
\newtheorem{proposition}{Proposition}[chapter]

\newtheorem{observation}{Observation}[chapter]

\newcommand{\remove}[1]{}

\RequirePackage{comment}
\usepackage{xifthen}
\usepackage{xspace}
\usepackage{arydshln}
\usepackage{authblk}
\usepackage{dsfont}
\usepackage{bbm}
\usepackage{mathtools}
\usepackage{graphics}
\usepackage{float}
\usepackage{array}
\usepackage{multirow}
\usepackage[shortlabels]{enumitem}
\renewenvironment{proof}{\noindent{\bf Proof:} \hspace*{1mm}}{\hfill $\Box$ }

\newcommand{\CC}{\ensuremath{\mathcal C}\xspace}

\newcommand{\GG}{\ensuremath{\mathcal G}\xspace}

\newcommand{\MM}{\ensuremath{\mathcal M}\xspace}
\newcommand{\NN}{\ensuremath{\mathcal N}\xspace}

\newcommand{\TT}{\ensuremath{\mathcal T}\xspace}

\newcommand{\YY}{\ensuremath{\mathcal Y}\xspace}

\DeclareMathOperator*{\argmaxaa}{arg\,max}
\DeclareMathOperator*{\argminaa}{arg\,min}
\newcommand{\argmi}{\argminaa\limits_}
\newcommand{\argma}{\argmaxaa\limits_}
\newcommand{\instance}{\ensuremath{\mathcal{I}}\xspace}
\newcommand{\suml}{\sum\limits_}
\newcommand{\maxl}{\max\limits_}
\newcommand{\minl}{\min\limits_}
\newcommand{\lb}{\left(}
\newcommand{\rb}{\right)}

\DeclarePairedDelimiter\floor{\lfloor}{\rfloor}
\newcommand{\curly}[1]{\ensuremath{\{#1\}}\xspace}
\newcommand{\nn}{\nonumber}
\newcommand{\NPH}{\ensuremath{\mathsf{NP}}-hard\xspace}

\newcommand{\WWH}{\ensuremath{\mathsf{W[2]}}-hard\xspace}
\newcommand{\Poly}{\ensuremath{\mathsf{P}}\xspace}
\newcommand{\NP}{\ensuremath{\mathsf{NP}}\xspace}
\newcommand{\XP}{\ensuremath{\mathsf{NP}}\xspace}

\newcommand{\subsum}{{\sc Subset Sum}\xspace}

\newcommand{\vercov}{{\sc Vertex Cover}\xspace}
\newcommand{\setcov}{{\sc Set Cover}\xspace}

\newcommand{\yes}{{\sc Yes}\xspace}
\newcommand{\no}{{\sc No}\xspace}
\newcommand{\FPT}{\ensuremath{\mathsf{FPT}}\xspace}

\newcommand{\bud}{\ensuremath{b}\xspace}
\newcommand{\proj}{\ensuremath{P}\xspace}
\newcommand{\voters}{\ensuremath{N}\xspace}
\newcommand{\feasible}{\ensuremath{\mathcal{F}}\xspace}
\newcommand{\aof}[1]{%
  \ifthenelse{\isempty{#1}}%
    {\ensuremath{A_i}\xspace}
    {\ensuremath{A_{#1}}\xspace}
}
\newcommand{\approf}{\ensuremath{\mathcal{A}}\xspace}
\newcommand{\rprof}{\ensuremath{\mathcal{P}}\xspace}
\newcommand{\ii}{\ensuremath{i}\xspace}

\newcommand{\pa}{\ensuremath{a}\xspace}
\newcommand{\distinct}{\ensuremath{\hat{n}}\xspace}
\newcommand{\scalel}{\ensuremath{\delta}\xspace}
\newcommand{\bool}{\ensuremath{\mathds{1}}}

\newcommand{\uof}[2]{%
  \ifthenelse{\isempty{#1}}  {\ensuremath{\operatorname{u_i\!}\left(#2\right)}\xspace}
    {\ensuremath{\operatorname{u_{#1}\!}\left(#2\right)}\xspace}
}
\newcommand{\cof}[1]{\ensuremath{\operatorname{c\!}\left(#1\right)}\xspace}
\newcommand{\domainof}[1]{\ensuremath{\mathcal{X}_{#1}}\xspace}
\newcommand{\winnerfunction}{\ensuremath{\ensuremath{\Phi}}\xspace}
\newcommand{\ruleof}[2]{%
  \ifthenelse{\isempty{#2}}%
    {\ensuremath{\operatorname{#1\!}\left(\instance\right)}\xspace}
    {\ensuremath{\operatorname{#1\!}\left(\instance#2\right)}\xspace}
}
\newcommand{\winners}[2]{%
	\ifthenelse{\isempty{#2}}%
	{\ensuremath{\operatorname{\winnerfunction\!}\left(#1,\instance\right)}\xspace}
	{\ensuremath{\operatorname{\winnerfunction\!}\left(#1,\instance#2\right)}\xspace}
}
\newcommand{\pbrule}{\ensuremath{R}\xspace}
\newcommand{\fullapinstance}{\ensuremath{\langle \voters,\proj,c,\bud,\approf \rangle}\xspace}
\newcommand{\fullrinstance}{\ensuremath{\langle \voters,\proj,c,\bud,\rprof \rangle}\xspace}

\newcommand{\hcbp}{\ensuremath{\mathsf{HCBP}}\xspace}
\newcommand{\mmpb}{\ensuremath{\mathsf{MPB}}\xspace}

\newcommand{\lpalgo}{{\sc Ordered-Relax}\xspace}

\makeatletter
\newcommand{\setword}[2]{%
  \phantomsection
  #1\def\@currentlabel{\unexpanded{#1}}\label{#2}%
}
\makeatother

\newcommand{\opt}{\emph{OPT}}

\newcommand{\lo}{\ensuremath{l_o}\xspace}
\newcommand{\ho}{\ensuremath{h_o}\xspace}
\newcommand{\la}{\ensuremath{l_{\approf}}\xspace}
\newcommand{\ha}{\ensuremath{h_{\approf}}\xspace}
\newcommand{\disuof}[2]{%
  \ifthenelse{\isempty{#1}}%
    {\ensuremath{\operatorname{d_i\!}\left(#2\right)}\xspace}
    {\ensuremath{\operatorname{d_{#1}\!}\left(#2\right)}\xspace}
}
\newcommand{\fdisuof}[2]{%
  \ifthenelse{\isempty{#1}}%
    {\ensuremath{\bud-{\uof{}{#2}}}\xspace}
    {\ensuremath{\bud-{\uof{#1}{#2}}}\xspace}
}
\newcommand{\cxof}[1]{\ensuremath{\operatorname{c\!}\left(#1\right)x_{#1}}\xspace}

\newcommand{\cxsof}[1]{\ensuremath{\operatorname{c\!}\left(#1\right)x_{#1}^*}\xspace}
\newcommand{\br}{\ensuremath{R}\xspace}
\newcommand{\mmpbr}{\ensuremath{M}\xspace}
\newcommand{\firstfrac}[1]{%
  \ifthenelse{\isempty{#1}}%
    {\ensuremath{\frac{|S|-|Y_i|}{|\aof{}|-|Y_i|}}\xspace}
    {\ensuremath{\frac{|S|-|Y_{#1}|}{|\aof{#1}|-|Y_{#1}|}}\xspace}
}
\newcommand{\revfrac}[1]{%
  \ifthenelse{\isempty{#1}}%
    {\ensuremath{\frac{|S|-|\aof{}|}{|S|-|Y_i|}}\xspace}
    {\ensuremath{\frac{|S|-|\aof{#1}|}{|S|-|Y_{#1}|}}\xspace}
}
\newcommand{\frack}[1]{%
  \ifthenelse{\isempty{#1}}%
    {\ensuremath{\frac{|\aof{}|-|Y_i|}{|S|-|Y_i|}}\xspace}
    {\ensuremath{\frac{|\aof{#1}|-|Y_{#1}|}{|S|-|Y_{#1}|}}\xspace}
}
\newcommand{\githublink}{\url{https://github.com/Participatory-Budgeting/maxmin_2022}}

\newcommand{\degreeset}{\ensuremath{\operatorname{D}}\xspace}
\newcommand{\degreeproj}{\ensuremath{\mathcal{D}}\xspace}
\newcommand{\costfunction}{\ensuremath{\operatorname{c}}\xspace}
\newcommand{\valid}{\ensuremath{\mathcal{V}}\xspace}
\newcommand{\lbounds}{\ensuremath{\mathcal{L}}\xspace}
\newcommand{\eachlb}{\ensuremath{l}\xspace}
\newcommand{\ubounds}{\ensuremath{\mathcal{H}}\xspace}
\newcommand{\eachub}{\ensuremath{h}\xspace}
\newcommand{\fulldeginstance}{\ensuremath{\langle \voters,\degreeproj,\costfunction,\bud,\lbounds,\ubounds\rangle}\xspace}
\newcommand{\pdof}[2]{%
  \ifthenelse{\isempty{#1}}%
    {\ifthenelse{\isempty{#2}}
    {\ensuremath{p_{j}^{t}}\xspace}
    {\ensuremath{p_{j}^{#2}}\xspace}
    }
    {\ifthenelse{\isempty{#2}}{\ensuremath{p_{#1}^{t}}\xspace} 
    {\ensuremath{p_{#1}^{#2}}\xspace} 
    }
}
\newcommand{\dcof}[2]{%
  \ifthenelse{\isempty{#1}}%
    {\ifthenelse{\isempty{#2}}
    {\ensuremath{c_{j}^{t}}\xspace}
    {\ensuremath{c_{j}^{#2}}\xspace}
    }
    {\ifthenelse{\isempty{#2}}{\ensuremath{c_{#1}^{t}}\xspace} 
    {\ensuremath{c_{#1}^{#2}}\xspace} 
    }
}
\newcommand{\conrange}{\ensuremath{\tau_p}\xspace}
\newcommand{\maxconrange}{\ensuremath{\overline{\conrange}}\xspace}
\newcommand{\minconrange}{\ensuremath{\underline{\conrange}}\xspace}
\newcommand{\qof}[2]{%
  \ifthenelse{\isempty{#1}}%
    {\ifthenelse{\isempty{#2}}
    {\ensuremath{q_{j}^{t}}\xspace}
    {\ensuremath{q_{j}^{#2}}\xspace}
    }
    {\ifthenelse{\isempty{#2}}{\ensuremath{q_{#1}^{t}}\xspace} 
    {\ensuremath{q_{#1}^{#2}}\xspace} 
    }
}
\newcommand{\qmax}{\ensuremath{q_m}\xspace}
\newcommand{\qsum}{\ensuremath{q_\sigma}\xspace}
\newcommand{\csetof}[1]{
\ifthenelse{\isempty{#1}}
{\ensuremath{\costfunction\!\left(S\right)}\xspace}
{\ensuremath{\costfunction\!\left(#1\right)}\xspace}
}
\newcommand{\mdof}[1]{
\ifthenelse{\isempty{#1}}
{\ensuremath{t_{j}}\xspace}
{\ensuremath{t_{#1}}\xspace}
}
\newcommand{\dsetof}[1]{
\ifthenelse{\isempty{#1}}
{\ensuremath{\degreeset\!\left(p_{j}\right)}\xspace}
{\ensuremath{\degreeset\!\left(p_{#1}\right)}\xspace}
}
\newcommand{\utof}[2]{
\ifthenelse{\isempty{#1}}
{
\ifthenelse{\isempty{#2}}
{\ensuremath{u_i\left(S\right)}\xspace}
{\ensuremath{u_i\left(S,#2\right)}\xspace}
}
{
\ifthenelse{\isempty{#2}}
{\ensuremath{u_{#1}\left(S\right)}\xspace}
{\ensuremath{u_{#1}\left(S,#2\right)}\xspace}
}
}
\newcommand{\utofd}[2]{
\ifthenelse{\isempty{#1}}
{
\ifthenelse{\isempty{#2}}
{\ensuremath{u_i\left(S'\right)}\xspace}
{\ensuremath{u_i\left(S',#2\right)}\xspace}
}
{
\ifthenelse{\isempty{#2}}
{\ensuremath{u_{#1}\left(S'\right)}\xspace}
{\ensuremath{u_{#1}\left(S',#2\right)}\xspace}
}
}
\newcommand{\dutof}[2]{
\ifthenelse{\isempty{#1}}
{
\ifthenelse{\isempty{#2}}
{\ensuremath{d_i\left(S\right)}\xspace}
{\ensuremath{d_i\left(S,#2\right)}\xspace}
}
{
\ifthenelse{\isempty{#2}}
{\ensuremath{d_{#1}\left(S\right)}\xspace}
{\ensuremath{d_{#1}\left(S,#2\right)}\xspace}
}
}
\newcommand{\dutofd}[2]{
\ifthenelse{\isempty{#1}}
{
\ifthenelse{\isempty{#2}}
{\ensuremath{d_i\left(S'\right)}\xspace}
{\ensuremath{d_i\left(S',#2\right)}\xspace}
}
{
\ifthenelse{\isempty{#2}}
{\ensuremath{d_{#1}\left(S'\right)}\xspace}
{\ensuremath{d_{#1}\left(S',#2\right)}\xspace}
}
}
\newcommand{\variance}{variance coefficient\xspace}
\newcommand{\vpara}{\ensuremath{\gamma}\xspace}
\newcommand{\scorefunction}{\ensuremath{\operatorname{s}\!}\xspace}
\newcommand{\dscoreof}[2]{
\ifthenelse{\isempty{#1}}
{
\ifthenelse{\isempty{#2}}
{\ensuremath{\scorefunction\left(\pdof{}{}\right)}\xspace}
{\ensuremath{\scorefunction\left(\pdof{}{#2}\right)}\xspace}
}
{
\ifthenelse{\isempty{#2}}
{\ensuremath{\scorefunction\left(\pdof{#1}{}\right)}\xspace}
{\ensuremath{\scorefunction\left(\pdof{#1}{#2}\right)}\xspace}
}
}
\newcommand{\msof}[1]{
\ifthenelse{\isempty{#1}}
{\ensuremath{\alpha\!\left(\pdof{}{}\right)}\xspace}
{\ensuremath{\alpha\!\left(\pdof{}{#1}\right)}\xspace}
}
\newcommand{\cardof}[1]{
\ifthenelse{\isempty{#1}}
{\ensuremath{|S|}\xspace}
{\ensuremath{|#1|}\xspace}
}
\newcommand{\lowerb}[2]{
\ifthenelse{\isempty{#1}}
{
\ifthenelse{\isempty{#2}}
{\ensuremath{\eachlb_i(j)}\xspace}
{\ensuremath{\eachlb_i(#2)}\xspace}
}
{
\ifthenelse{\isempty{#2}}
{\ensuremath{\eachlb_{#1}(j)}\xspace}
{\ensuremath{\eachlb_{#1}(#2)}\xspace}
}
}
\newcommand{\upperb}[2]{
\ifthenelse{\isempty{#1}}
{
\ifthenelse{\isempty{#2}}
{\ensuremath{\eachub_i(j)}\xspace}
{\ensuremath{\eachub_i(#2)}\xspace}
}
{
\ifthenelse{\isempty{#2}}
{\ensuremath{\eachub_{#1}(j)}\xspace}
{\ensuremath{\eachub_{#1}(#2)}\xspace}
}
}
\newcommand{\appof}[2]{
\ifthenelse{\isempty{#1}}
{
\ifthenelse{\isempty{#2}}
{\ensuremath{A_i(j)}\xspace}
{\ensuremath{A_i(#2)}\xspace}
}
{
\ifthenelse{\isempty{#2}}
{\ensuremath{A_{#1}(j)}\xspace}
{\ensuremath{A_{#1}(#2)}\xspace}
}
}
\newcommand{\chodof}[2]{
\ifthenelse{\isempty{#1}}
{
\ifthenelse{\isempty{#2}}
{\ensuremath{S(j)}\xspace}
{\ensuremath{#2(j)}\xspace}
}
{
\ifthenelse{\isempty{#2}}
{\ensuremath{S(#1)}\xspace}
{\ensuremath{#2(#1)}\xspace}
}
}
\newcommand{\chocof}[2]{
\ifthenelse{\isempty{#1}}
{
\ifthenelse{\isempty{#2}}
{\ensuremath{c^S(j)}\xspace}
{\ensuremath{c^{#2}(j)}\xspace}
}
{
\ifthenelse{\isempty{#2}}
{\ensuremath{c^S(#1)}\xspace}
{\ensuremath{c^{#2}(#1)}\xspace}
}
}
\newcommand{\nrule}{\ensuremath{R_{|S|}}\xspace}
\newcommand{\tstar}{\ensuremath{t^*}\xspace}
\newcommand{\crule}{\ensuremath{R_{\costfunction(S)}}\xspace}
\newcommand{\slimit}{\ensuremath{\delta}\xspace}
\newcommand{\ccaprule}{\ensuremath{R_{\widehat{\costfunction(S)}}}\xspace}
\newcommand{\drule}{\ensuremath{R_{\|\;\|}}\xspace}
\newcommand{\winnersof}[2]{%
  \ifthenelse{\isempty{#2}}%
    {\ensuremath{\pbrule\left(\instance\right)}\xspace}
    {\ensuremath{\pbrule\left(\instance#2\right)}\xspace}
}

\makeatletter
\newcommand\mathcircled[1]{%
  \mathpalette\@mathcircled{#1}%
}
\newcommand\@mathcircled[2]{%
  \tikz[baseline=(math.base)] \node[draw,circle,inner sep=1pt] (math) {$\m@th#1#2$};%
}
\makeatother
\makeatletter
\newcommand{\thickhline}{%
    \noalign {\ifnum 0=`}\fi \hrule height 1.3pt
    \futurelet \reserved@a \@xhline
}
\newcommand{\toothickhline}{%
    \noalign {\ifnum 0=`}\fi \hrule height 2pt
    \futurelet \reserved@a \@xhline
}
\newcolumntype{"}{@{\hskip\tabcolsep\vrule width 1.3pt\hskip\tabcolsep}}
\newcolumntype{x}{@{\hskip\tabcolsep\vrule width 2pt\hskip\tabcolsep}}
\makeatother

\newcommand{\comporder}{\ensuremath{\mathcal{L}(\proj)}\xspace}
\newcommand{\suci}{\ensuremath{\succeq_i}\xspace}

\newcommand{\rof}[1]{\ensuremath{\operatorname{r_i\!}\left(#1\right)}\xspace}
\newcommand{\pat}[1]{\ensuremath{\operatorname{\succeq_i\!}\left(#1\right)}\xspace}
\newcommand{\opat}[2]{
\ifthenelse{\isempty{#1}}
{\ensuremath{\operatorname{\succ\!}\left(#2\right)}\xspace}
{\ensuremath{\operatorname{\succ_{#1}\!}\left(#2\right)}\xspace}
}
\newcommand{\ptill}[1]{\ensuremath{E^i_{[#1]}}\xspace}

\newcommand{\trunkk}[2]{\ensuremath{\operatorname{t_i\!}\left(#1,#2\right)}\xspace}
\newcommand{\trunkkcc}[2]{\ensuremath{\operatorname{y_i\!}\left(#2\right)}\xspace}

\newcommand{\trunks}{\ensuremath{t}\xspace}

\newcommand{\gtr}{\ensuremath{\langle MT,f \rangle}\xspace}

\newcommand{\cwpara}{\ensuremath{\alpha}\xspace}
\newcommand{\cwr}{\ensuremath{\langle \CC_\cwpara,f \rangle}\xspace}
\newcommand{\ebpara}{\ensuremath{\theta}\xspace}

\newcommand{\cwparaof}[1]{\ensuremath{\operatorname{\cwpara\!}\left(#1\right)}\xspace}

\newcommand{\scoreof}[1]{\ensuremath{score(#1)}\xspace}

\makeatletter
\newcommand*{\rom}[1]{\expandafter\@slowromancap\romannumeral #1@}
\makeatother

\newcommand{\repeatcaption}[2]{%
  \renewcommand{\thefigure}{\ref{#1}}%
  \captionsetup{list=no}%
  \caption{#2 (figure repeated from page \pageref{#1}).}%
  \addtocounter{figure}{-1}
}

\makeatletter
\def\thickhline{%
  \noalign{\ifnum0=`}\fi\hrule \@height \thickarrayrulewidth \futurelet
   \reserved@a\@xthickhline}
\def\@xthickhline{\ifx\reserved@a\thickhline
               \vskip\doublerulesep
               \vskip-\thickarrayrulewidth
             \fi
      \ifnum0=`{\fi}}
\makeatother

\newlength{\thickarrayrulewidth}
\newcommand{\functionsymbol}{H}
\newcommand{\shortfu}[2]{%
  \ifthenelse{\isempty{#1}}%
    {\ensuremath{\operatorname{\functionsymbol_{\psi_q}\!}(#2)}\xspace}
    {\ensuremath{\operatorname{\functionsymbol_{\psi_{#1}}\!}(#2)}\xspace}
}

\newcommand{\kmin}{\ensuremath{\kappa_{min}\xspace}}
\newcommand{\emax}{\ensuremath{\eta_{max}}\xspace}

\newcommand{\rscr}{\ensuremath{\varphi:\mathcal{D}^n \to \Delta \proj}\xspace}

\newcommand{\sd}[1]{%
  \ifthenelse{\isempty{#1}}%
    {\ensuremath{\succ^{\text{\tiny{SD}}}}\xspace}
    {\ensuremath{\succ_{#1}^{\text{\tiny{SD}}}}\xspace}
}
\newcolumntype{?}{!{\vrule width 1.5pt}}

\newcommand{\weakf}{$(\kappa_G,\psi_G,\eta_G)$-weak-GEG\xspace}
\newcommand{\strof}{$(\kappa_G,\psi_G,\eta_G)$-strong-GEG\xspace}
\newcommand{\weakfness}{$(\kappa_G,\psi_G,\eta_G)$-weak-GEG\xspace}
\newcommand{\strofness}{$(\kappa_G,\psi_G,\eta_G)$-strong-GEG\xspace}
\newcommand{\spweakf}{$(\kappa_N,\eta_N)$-weak-IEG\xspace}
\newcommand{\spstrof}{$(\kappa_N,\eta_N)$-strong-IEG\xspace}
\newcommand{\spweakfness}{$(\kappa_N,\eta_N)$-weak-IEG\xspace}
\newcommand{\spstrofness}{$(\kappa_N,\eta_N)$-strong-IEG\xspace}

\newcommand{\community}{group\xspace}

\newcommand{\reprange}{representative range\xspace}
\newcommand{\repranges}{representative ranges\xspace}
\newcommand{\repscene}{representation scenario\xspace}
\newcommand{\repscenes}{representation scenarios\xspace}
\newcommand{\resquota}{reservation quota\xspace}

\newcommand{\compliant}{compliant\xspace}
\newcommand{\topi}{top-ranged\xspace}
\newcommand{\topiness}{top-rangedness\xspace}
\newcommand{\correspondencef}{\ensuremath{\psi_q:\mathcal{D}^{|N_q|}\to \small{A \choose \kappa_q}}\xspace}
\newcommand{\correspondence}{\ensuremath{\psi_q}\xspace}

\newdateformat{monthyeardate}{ \monthname[\THEMONTH], \THEYEAR}

\newcommand{\blankpage}{
\newpage
\thispagestyle{empty}
\mbox{}
\newpage
}

\newcommand{\blankpagewithnumber}{
\newpage
\mbox{}
\newpage
}

\crefname{observation}{observation}{observations}
\crefname{algorithm}{algorithm}{algorithms}
\crefname{align}{equation}{equations}
\crefname{eqnarray}{equation}{equations}

\begin{document}

\title{Exploring Welfare Maximization and Fairness in Participatory Budgeting} 

\submitdate{\monthyeardate\today} 
\phd
\dept{Computer Science and Automation}
\faculty{Faculty of Engineering}
\author{Gogulapati Sreedurga}


\maketitle

\begin{center}
\LARGE{\underline{\textbf{Declaration of Originality}}}
\end{center}
\noindent I, \textbf{Gogulapati Sreedurga}, with SR No. \textbf{04-04-00-10-12-18-1-15965} hereby declare that
the material presented in the thesis titled

\begin{center}
\textbf{Exploring Welfare Maximization and Fairness in Participatory Budgeting}
\end{center}

\noindent represents original work carried out by me in the \textbf{Department of Computer Science and Automation} at \textbf{Indian Institute of Science} during the years \textbf{2018-2023}.

\noindent With my signature, I certify that:
\begin{itemize}
	\item I have not manipulated any of the data or results.
	\item I have not committed any plagiarism of intellectual
	property.
	I have clearly indicated and referenced the contributions of
	others.
	\item I have explicitly acknowledged all collaborative research
	and discussions.
	\item I have understood that any false claim will result in severe
	disciplinary action.
	\item I have understood that the work may be screened for any form
	of academic misconduct.
\end{itemize}

\vspace{20mm}

\noindent {\footnotesize{Date:	\hfill	Student Signature}} \qquad

\vspace{20mm}

\noindent In my capacity as supervisor of the above-mentioned work, I certify
that the above statements are true to the best of my knowledge, and 
I have carried out due diligence to ensure the originality of the
report.

\vspace{20mm}

\noindent  {\footnotesize{Advisor Name: \hfill Advisor Signature}} \qquad

\blankpage

\vspace*{\fill}
\begin{center}
\large\bf \textcopyright \ Gogulapati Sreedurga\\
\large\bf \monthyeardate\today\\
\large\bf All rights reserved
\end{center}
\vspace*{\fill}
\thispagestyle{empty}

\blankpage

\vspace*{\fill}
\begin{center}
\Large DEDICATED TO \\[1em]
\Large\it my beloved ideal parents\\
\Large\it who are the reasons behind everything that I am\\[2em]
\begin{figure}[h]
  \centering
  \includegraphics[width=0.67\linewidth]{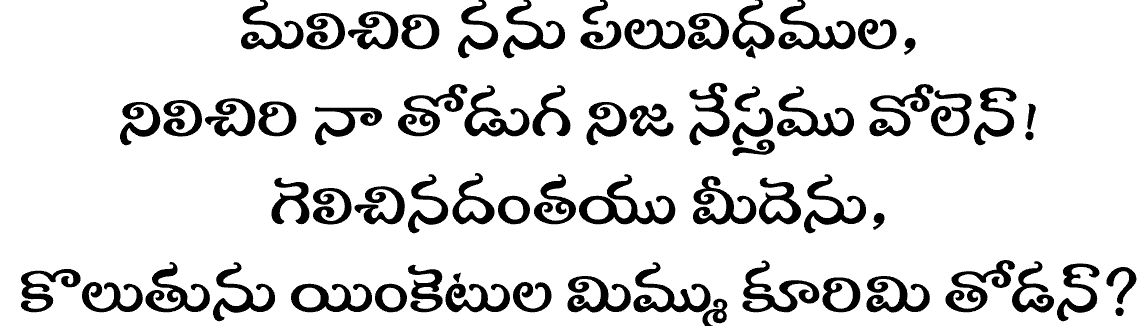}
  \captionsetup{labelformat=empty, width=0.67\linewidth}
  \caption[]{\large\it {English Translation}: You (both) have molded me in countless ways. You always stood by my side like true friends. All my accomplishments indeed belong to you. How else do I worship you and express my boundless \centering{love for you?}}
\end{figure}
\end{center}
\vspace*{\fill}
\thispagestyle{empty}



\blankpage

\setcounter{secnumdepth}{3}
\setcounter{tocdepth}{3}

\frontmatter 
\pagenumbering{roman}

\prefacesection{Acknowledgements}
 First of all, I express my heartfelt gratitude to my esteemed PhD supervisor, Prof. Yadati Narahari, whose unwavering support and unwavering positivity have been invaluable to me. It is truly rare to encounter a supervisor who grants a PhD student such a remarkable degree of freedom and respect, and I consider myself fortunate to be among the privileged few. Prof. Narahari has not only acted as a supervisor but also as a nurturing figure, akin to a parent, offering boundless encouragement and genuine concern for both my personal and professional growth. I also thank his wife, Smt. Padmashri, for showering me and my family with genuine warmth and affection in every interaction we had.

I am incredibly grateful to Assoc. Prof. Siddharth Barman and Assoc. Prof. Neeldhara Misra for their valuable advice and support at several crucial junctures of my PhD. I am truly thankful to Prof. Ulle Endriss for graciously hosting me at the University of Amsterdam, an experience that has greatly enriched my academic journey. The opportunity to engage in both formal and informal discussions with his group has been truly enjoyable. Additionally, I express my gratitude to Dr. Markus Brill for hosting me at the University of Warwick, as this visit has helped me in multiple ways. I am profoundly grateful to Assoc. Prof. Souvik Roy and Dr. Soumyarup Sadhukhan for their support and excellent technical discussions.

I thank the Govt. of India for awarding me the Prime Minister Research Fellowship (PMRF). It is sheer pleasure to have the opportunity to live in the serene campus of Indian Institute of Science. I am very grateful to the staff of Computer Science and Automation (CSA) department - Mrs. Padmavathi, Mrs. Kushael, Mrs. Meenakshi, Mrs. Nishitha. Their support played an indispensable role in managing the practical aspects of my PhD.

I am grateful to my lab mates, Nidhi, Ganesh, and Arpita for their kindness, and valuable help at several junctures of my PhD. I am very thankful to Shraddha, Janaky, and Arpita for proof-reading this thesis. I am supremely grateful to my best friend on campus, Janaky Murthy, for being my twin soul and much more. The hours of chatter, bouncing off technical ideas, nights filled with old music and/or philosophical discussions, the countless joyful moments of lighthearted fun, and the cherished moments of sharing our deepest emotions - she was a much-needed source of solace throughout my demanding PhD journey.

I am indebted to my brother and lab mate, Mayank Ratan Bhardwaj, whose presence has been an immense blessing throughout my PhD journey. He has been a constant source of support and firm belief in my abilities, even when I doubted myself. His guidance and encouragement have played a pivotal role in my personal and professional life. My parents deeply admire him and regard him as a God-given son. Among all the valuable experiences I gained at IISc, I cherish his company as much as, if not more than, my PhD degree.

I extend my warm regards and deep gratitude to my beloved parents-in-law, Smt. Padmaja Yagnamurthy and Sri Dwarakanath Yagnamurthy, for their constant support, love, and a sense of belonging they bestowed upon me.

I am extremely lucky to be blessed with the most encouraging elder sister, Sreeprada, and my brother-in-law (BIL), Dr. Naga Samrat, who is like my big brother in guise. My sister has been my biggest cheerleader ever since I remember. If I ever feel under-confident, a short conversation with her is all I need to feel like I can conquer this world. She transcends the role of a typical sister, serving as my companion in many aspects of my life. Her presence in my life is a cherished treasure. The guidance of my brother-in-law, rooted in his invaluable experience, has been an essential pillar throughout my research journey.

I fall short of words when I need to articulate how lucky I got with my husband, Dr. Sai Saran. Saran entered my life when I was at my lowest and flipped it with great ease. With him by my side, every obstacle became conquerable and every dream became attainable. Being my strongest support, my best friend, my biggest inspiration, and my soul mate in the truest sense, he has forever changed the course of my PhD and, indeed, my entire life.

I am incredibly lucky for receiving the love and unwavering support of my maternal grandmother, Late. Smt. Pingali Rama Syamala. No amount of appreciation or thanks would do justice to the pivotal role she played in my life, and even more importantly, throughout my academic journey. Losing her during my PhD was the biggest low point in my life, and the pain of missing her presence as I submit my PhD thesis is indescribable.

Lastly and most importantly, I have no words to thank my most incredible parents, Sri. Gogulapati Venkata Ramakrishna and Smt. Lakshmi Padma Kumari, who I very often refer to as the best parents in the world. They are the most selfless, most open-minded, most friendly, most loving, most responsible, most kind, most encouraging, and most ideal parents one can imagine and it is no exaggeration if I claim that the space in my thesis is inadequate to describe each of them in detail. My parents are the reasons behind everything that I am personally and professionally, and everything that I will ever be. I humbly dedicate this thesis and every piece of my effort to them.  

\prefacesection{Abstract}
Participatory budgeting (PB) is a voting paradigm for distributing a divisible resource, usually called a budget, among a set of projects by aggregating the preferences of individuals over these projects. It is implemented quite extensively for purposes such as government allocating funds to public projects and funding agencies selecting research proposals to support. This dissertation studies the welfare-related and fairness-related  objectives for different PB models. Our contribution lies in proposing and exploring novel PB rules that maximize welfare and promote fairness, as well as, in introducing and investigating a range of novel utility notions, axiomatic properties, and fairness notions, effectively filling the gaps in the existing literature for each PB model. The thesis is divided into two main parts, the first focusing on dichotomous and the second focusing on ordinal preferences. Each part considers two cases: (i) the cost of each project is \emph{restricted} to a single value and partial funding is not permitted and (ii) the cost of each project is \emph{flexible} and may assume multiple values.

\subsection*{Part I: Dichotomous Preferences}
\subsubsection*{Restricted Costs: Egalitarian Participatory Budgeting}
Egalitarianism holds significance in PB, capturing welfare as well as fairness. Our work introduces and studies a natural egalitarian rule, Maxmin Participatory Budgeting (MPB). The study consists of two parts: computational analysis and axiomatic analysis. In the computational part, we prove that MPB is strongly NP-hard and present several results on its fixed parameter tractability. We also propose an approximation algorithm and further establish an upper bound on the achievable approximation ratio for exhaustive strategy-proof PB algorithms. In the axiomatic part, we investigate MPB by generalizing existing axioms and introducing a new fairness axiom called maximal coverage, which we show MPB satisfies.
\subsubsection*{Flexible Costs: Welfare Maximization when Projects have Multiple Degrees of Sophistication}
We introduce a novel PB model where projects have a discrete set of permissible costs, reflecting different levels of project sophistication. Voters express their preferences by specifying upper and lower bounds on the amount to be allocated for each project. The outcome of a PB rule involves selecting a subset of projects and determining their costs. We propose a variety of utility concepts and welfare-maximizing rules. We prove that all the positive findings from single-cost projects can be extended to this new framework with multiple permissible costs and further analyze the fixed parameter tractability of the problems. We also propose novel and intuitive axioms and evaluate their compatibility with the four PB rules proposed by us.
\subsection*{Part II: Ordinal Preferences}
\subsubsection*{Restricted Costs: Welfare Maximization and Fairness under Incomplete Weakly Ordinal Preferences}
This chapter focuses on incomplete weakly ordinal preferences and has two logical components. The first component concentrates on maximizing welfare, while the second one addresses fairness. In the first component, we introduce a family of rules, \emph{dichotomous translation rules}, and \emph{the PB-Chamberlin-Courant (PB-CC) rule}, which respectively expand on existing welfare-maximizing rules for dichotomous and strictly ordinal preferences. We show that our expansions largely maintain and even enhance the computational and axiomatic properties of these rules. We also propose a new relevant axiom, pro-affordability. The second component introduces the novel class of \textit{average rank-share guarantee rules} to address fairness in participatory budgeting with ordinal preferences, overcoming limitations of existing fairness concepts in the literature.
\subsubsection*{Flexible Costs: Characterization of Group-Fair and Individual-Fair Rules under Single-Peaked Preferences}
We examine a PB model in which the cost of each project is completely unrestricted. Consequently, this model is viewed as random social choice, where the probability associated with a project in the outcome represents the fraction of the budget allocated to it. We investigate fairness in social choice under single-peaked preferences. Existing literature has extensively examined the construction and characterization of social choice rules in the single-peaked domain. We non-trivially extend these findings by incorporating fairness considerations. To address group-fairness, we partition voters into logical groups based on attributes like gender or location. We introduce \textit{group-wise anonymity} to capture fairness within each group and propose the notions of \textit{weak group entitlement guarantee} and \textit{strong group entitlement guarantee} to ensure fairness across groups. We characterize deterministic and random social choice rules that achieve group-fairness. We also explore the case without groups and provide more precise characterizations of rules achieving individual-fairness.

\prefacesection{Publications and Manuscripts based on this Thesis}
\begin{enumerate}
    \item \textit{Gogulapati Sreedurga}, \textit{Mayank Ratan Bhardwaj}, and \textit{Yadati Narahari}. Maxmin Participatory Budgeting. Proceedings of the 31st International Joint Conference on Artificial Intelligence (IJCAI 2022), pages 489-495.
    \item \textit{Gogulapati Sreedurga}. Indivisible Participatory Budgeting with Multiple Degrees of Sophistication of Projects. Proceedings of the 22nd International Conference on Autonomous Agents and Multiagent Systems (AAMAS 2023), pages 2694-2696.
    \item \textit{Gogulapati Sreedurga}. Participatory Budgeting With Multiple Degrees of Projects And Ranged Approval Votes. To appear in the Proceedings of the 32nd International Joint Conference on Artificial Intelligence (IJCAI 2023).
    \item \textit{Gogulapati Sreedurga} and \textit{Yadati Narahari}. Participatory Budgeting under Weakly Ordinal Preferences: Welfare and Fairness. Under review.
    \item \textit{Gogulapati Sreedurga}, \textit{Soumyarup Sadhukhan}, \textit{Souvik Roy}, and \textit{Yadati Narahari}. Individual-Fair and Group-Fair Social Choice Rules under Single-Peaked Preferences. Proceedings of the 22nd International Conference on Autonomous Agents and Multiagent Systems (AAMAS 2023), pages 2872-2874.
    \item \textit{Gogulapati Sreedurga}, \textit{Soumyarup Sadhukhan}, \textit{Souvik Roy}, and \textit{Yadati Narahari}. Characterization of Fair Social Choice Rules under Single-Peaked Preferences. Journal submission (under review).
    \item \textit{Gogulapati Sreedurga}. Participatory Budgeting: Fairness and Welfare Maximization. Proceedings of the 19th European Conference on Autonomous Agents and Multiagent Systems (EUMAS 2022), pages 439-443.
\end{enumerate}
\blankpage
\tableofcontents
\listoffigures
\blankpagewithnumber
\listoftables
\blankpagewithnumber
\newpage

\begin{table}[!ht]
    \centering
    \renewcommand{\arraystretch}{1.3}
    \begin{tabular}[width = \textwidth]{c  l}
    \hline
    Notation & Definition\\
    \hline
    \bud    & Available budget\\
    \voters        & Set of voters\\
    $n$ & Number of voters\\
    \NN & Generic preference profile of voters\\
    \proj       & Set of projects\\
    $m$ & Number of projects\\
    \aof{} & Dichotomous preference of voter $i$\\
    \approf & Dichotomous preference profile\\
    $\succeq_i$ & Weakly ordinal preference of voter $i$\\
    $\succ_i$ & Strictly ordinal preference of voter $i$\\
    $r_i(p)$ & Rank of project $p$ in the ordinal preference of voter $i$\\
    $\succeq_i(t)$, $\succ_i(t)$ & Project(s) ranked exactly $t$ in the ordinal preference of voter $i$\\
    \rprof  & Ordinal preference profile\\
    \domainof{p} & Set of all possible costs of project $p$\\
    \multirow{2}{*}{\cof{S}} & Total cost of a subset $S$ of projects when the costs are restricted\\
    & or partially flexible\\
    \instance & An instance of participatory budgeting\\
    \feasible & Set of all feasible subsets of projects in the restricted costs model\\
    \valid & Set of feasible budget allocations in the flexible costs model\\
    \uof{}{S} & Utility of voter $i$ from an outcome $S$ of the PB rule\\
    \pbrule & A participatory budgeting (PB) rule\\
    \multirow{2}{*}{\ruleof{\pbrule}{}} & Set of all the subsets selected by the PB rule \pbrule in the restricted\\
    & costs model\\
    \multirow{2}{*}{\winners{\pbrule}{}} & Set of all the projects in some subset selected by the PB rule in\\
    & restricted costs model\\
    \hline
    \end{tabular}
    \caption{Common symbols used in multiple chapters}
    \label{tab: common_notations}
\end{table}

\mainmatter 
\setcounter{page}{1}
\chapter{Introduction}\label{chap: intro}
\begin{quote}
    \textit{This chapter serves as an introduction to participatory budgeting, its applications, and the landscape of its variants. It then presents an overview of all the contributions of this thesis.}
\end{quote}

\noindent{In recent years, digital democracy has garnered significant interest, captivating researchers and gaining widespread social support. The growing recognition of the importance of providing equitable representation for every citizen and engaging them directly in decision-making processes has propelled the emergence of citizen-centric governance. Among the various applications of this approach, the one that stands out as particularly noteworthy and currently captures the spotlight is participatory budgeting.}

Participatory Budgeting (PB) is a voting paradigm that deals with the situations where a divisible resource (such as money, time etc.) is to be distributed among a set of projects by taking the preferences of stakeholders over the projects into account. The divisible resource that is to be distributed is typically referred to as the \emph{budget}, whereas the stakeholders are typically referred to as \emph{voters}.

PB finds immediate application in various domains, including government funding the public projects, where the preferences of citizens are aggregated to make funding decisions as displayed in \Cref{fig: publicfunding}. It is also used by conference organizers to manage talk schedules, taking into account the preferences of the audience. Additionally, companies employ PB to select research proposals for funding, by aggregating the preferences of the selection panel as displayed in \Cref{fig: researchfunding}.
\begin{figure}
        \centering
        \includegraphics[width=0.6\linewidth]{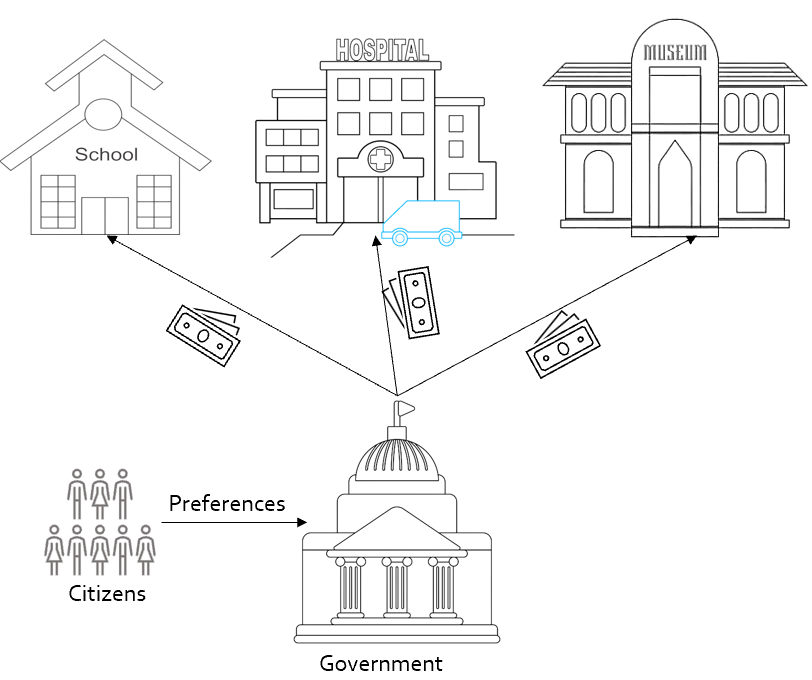}
        \caption{Participatory budgeting for government funding public projects.}
        \label{fig: publicfunding}
\end{figure}

The success of PB is evident through its widespread adoption in numerous countries worldwide, including the United Kingdom, the Netherlands, France, Sweden, Germany, Austria, and the United States \cite{cabannes2004participatory,shah2007participatory,sintomer2008participatory,wampler2010participatory,rocke2014framing}. Several platforms such as Polys platform (\burl{https://polys.vote/participatory-budgeting-platform}) and PB platform maintained by Stanford university (\burl{https://pbstanford.org/}) have also been set up to facilitate participatory budgeting elections. This broad applicability of PB led to its extensive study in computational social choice - an interdisciplinary field at the intersection of computer science, artificial intelligence, and economics, dedicated to investigating preference aggregation. A book chapter by Aziz and Shah \cite{aziz2021participatory} and a recent survey by Rey and Maly \cite{rey2023computational} provide an excellent exposition.
\begin{figure}
        \centering
        \includegraphics[width=0.4\linewidth]{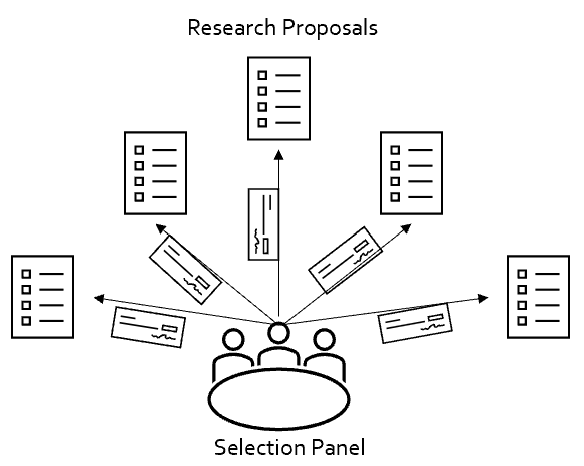}
        \caption{Participatory budgeting for agencies funding research proposals.}
        \label{fig: researchfunding}
\end{figure}

In formal terms, a standard participatory budgeting model comprises several key components: the budget itself, a collection of projects, associated costs for each project, a group of voters, and the preferences expressed by each individual voter. A participatory budgeting rule combines these voter preferences and generates one or more `desirable' budget allocations to the projects as the outcome. The notion of desirability for an outcome is defined based on what we refer to as the `objective' of the rule, which formalizes the criteria guiding the decision-making process.
\section{Landscape of PB Models}\label{sec: intro-variants}
The PB models can be broadly classified with respect to three axes: (i) preferences of the voters (ii) costs of the projects (iii) objective of the PB rule. Let us look at PB through each of these lenses and understand the variants of PB models.
\subsection{Preferences in PB}\label{sec: intro-preferences}
When approaching a voting problem, the initial consideration often revolves around the design of the ballots. The field of social choice offers various approaches for voters to express their preferences regarding the projects under consideration. Furthermore, the participatory budgeting literature proposes additional methods for eliciting preferences that are tailored specifically to PB. In the following discussion, we introduce all the potential ballot formats suitable for PB and illustrate them with a running example.

\begin{example}\label{eg: preferences}
    Suppose there is a total budget of $1M$ units and the following projects need funding: $\{Library, Park, Station, Hospital, School, Museum\}$. There are multiple ways in which voters can express preferences over the projects.
    \begin{figure}[h]
        \centering
        \includegraphics[width=0.63\linewidth]{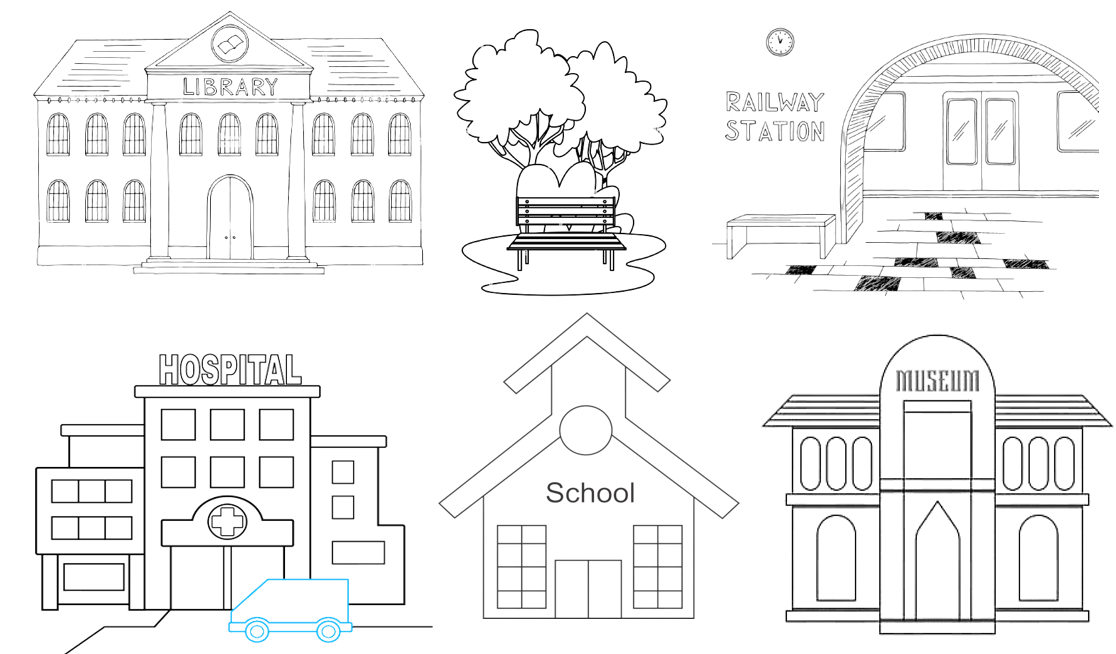}
        \caption{Stylized example to illustrate preferences for \Cref{eg: preferences}.}
        \label{fig: preferences}
    \end{figure}
\end{example}

\begin{itemize}
    \item \textbf{\textit{Dichotomous preferences:}} Every voter reports for each project whether or not she likes the project. In other words, each voter is asked to report a subset of project she likes.
    \begin{itemize}
        \item \textit{Standard dichotomous preferences:} Every voter is free to approve any subset of projects she likes.
        \begin{nonexample}
            A voter may choose to only approve $\{Library, Station, Hospital, School\}$ since they feel park or museum are a waste of money.
        \end{nonexample}
        \item \textit{Knapsack votes:} Every voter approves a subset of projects whose total cost is within the budget limit.
        \begin{nonexample}
            Sometimes, each project may also have some fixed cost associated to them. Every voter may be asked to approve projects such that the total cost is within the budget. For example, suppose library costs $75K$, park costs $50K$, station costs $350K$, hospital costs $500K$, school costs $200K$, and a museum costs $400K$. A voter cannot approve $\{Library, Station, Hospital, School\}$ since the total cost crosses $1M$. She can however approve $\{Library, Station, Hospital\}$.
        \end{nonexample}
        \item \textit{Threshold approval votes:} Every voter approves all and only those projects for which her implicit value crosses a threshold fixed by the PB organizer.
        \begin{nonexample}
            Every voter may be, for instance, asked to approve the projects who she thinks is worth at least $200K$ (this threshold could also be some numerical value of cardinal utility which will be explained shortly).
        \end{nonexample}
    \end{itemize}
    \item \textbf{\textit{Ordinal preferences:}} Every voter gives an ordering or ranking over the projects under consideration.
     \begin{itemize}
        \item \textit{Complete strictly ordinal preferences:} Every voter reports a strict ordering over all the projects.
        \begin{nonexample}
            A voter may report a ranking $Hospital \succ School \succ Museum \succ Library \succ Park \succ Station$.
        \end{nonexample}
        \item \textit{Complete weakly ordinal preferences:} Every voter reports a weak ordering over all the projects. That is, preferences allow for ties between the projects.
        \begin{nonexample}
            A voter may report a ranking $Hospital \succ \{School,Museum\} \succ Library \succ \{Park,Station\}$. That is, the voter has equal preference for $School$ and $Museum$.
        \end{nonexample}
        \item \textit{Incomplete strictly ordinal preferences:}  Every voter reports a strict ordering over only a subset of projects, and ignores the rest.
        \begin{nonexample}
            A voter may report a ranking $Hospital \succ School \succ Museum \succ Library$.
        \end{nonexample}
        \item \textit{Incomplete weakly ordinal preferences:}  Every voter reports a weak ordering over only a subset of projects, and ignores the rest.
        \begin{nonexample}
            A voter may report a ranking $Hospital \succ School \succ \{Museum,Library\}$.
        \end{nonexample}
        \item \textit{Partial ordinal preferences:} Every voter reports only a few pairwise comparisons between the projects, without revealing the entire ordering over those projects.
        \begin{nonexample}
            A voter may report a few comparisons: $School \succ Park$, $School \succ Museum$, $Station \succ Park$, and $Station \succ Museum$. Note that the preference between $Park$ and $Musuem$, or between $School$ and $Station$ cannot be inferred from the above.
        \end{nonexample}
    \end{itemize}
    \item \textbf{\textit{Cardinal preferences:}} Every voter assigns to each project a numerical value that quantifies her liking for the project.
    \begin{nonexample}
            A voter may say that she has a utility of $5$ for $Hospital$ and $School$, a utility of $3$ for $Museum$ and $Station$, a utility of $2$ for $Library$, and $1$ for $Park$. Recall that in a threshold approval vote, the PB organizer can ask a voter to approve all projects valued at least $3$ by her. In such a case, the voter with above stated preferences will approve $\{Hospital, School, Museum, Station\}$.
    \end{nonexample}
\end{itemize}
\subsection{Costs in PB}\label{sec: intro-costs}
In participatory budgeting, the cost of a project signifies the amount that needs to be allocated to the project in the event of its selection for funding. However, the specific context and application of PB can introduce limitations on the funds dedicated to each project. As a consequence, the following scenarios may arise:
\begin{itemize}
    \item \textbf{\textit{Restricted costs:}} The cost of each project is restricted to a single value. A project must either receive an amount exactly equal to this value or must be not funded at all. 
    \begin{nonexample}
        A dam construction project, if funded, may need exactly $\$10B$. Any amount less than this would be inadequate and a higher amount would result in budget wastage.
    \end{nonexample}
    \item \textbf{\textit{Flexible costs:}} The amount allocated to each project is permitted to assume multiple values. This is further split into two possible scenarios:
    \begin{itemize}
        \item \textit{Partially flexible:} A project can be implemented upto different degrees of sophistication and every degree corresponds to a different cost. The amount allocated to the project by the PB rule must belong to this set of permissible costs. 
        \begin{nonexample}
            A building could be constructed with wood, or cement, or stone. Each of these options demands a different cost. The project, if selected, must be allocated an amount equal to one of these permissible costs.
        \end{nonexample}
        \item \textit{Totally flexible:} There is no restriction on the amount allocated to each project. That is, projects do not have any associated cost(s) and any amount can be allocated to each project. 
        \begin{nonexample}
            Any amount, however little or huge, could be used to donate for an environmental cause. It is a welcome contribution and shall be put to use in some way.
        \end{nonexample}
    \end{itemize}
\end{itemize}
\subsection{Objectives of PB}\label{sec: intro-objectives}
As in any social choice setting, PB aims to achieve two primary objectives: maximizing welfare and promoting fairness. Though the precise definitions of welfare and fairness depend on the specific PB model and its application, we informally explain below an overarching idea of what these objectives stand for. In the subsequent chapters, we will provide formal definitions of these concepts, following the introduction of each specific PB model.
\begin{itemize}
    \item \textbf{\textit{Maximizing welfare:}} Utility of a voter is a function that quantifies the benefit a voter derives from an outcome. Various welfare measures are formulated with respect to a given utility notion. We explain three such prominent measures below:
    \begin{itemize}
        \item \textit{Utilitarian welfare:} This is the sum of utilities of all the voters. An outcome that maximizes utilitarian welfare selects a subset of projects that maximizes the sum of utility of all the voters such that the total cost of the set is within the budget.
        \item \textit{Egalitarian welfare:} This is the utility of the worst-off voter, or in other words, the least of the utilities of all the voters. An outcome that maximizes egalitarian welfare selects a subset of projects that maximizes the utility of the worst-off voter such that the total cost of the set is within the budget.
        \item \textit{Nash welfare:} This is the geometric mean of utilities of all the voters, and is often viewed as a trade-off between utilitarian and egalitarian welfare.
    \end{itemize}
    \item \textbf{\textit{Achieving fairness:}} Fairness notions are established with the primary aim of guaranteeing equitable treatment for all voters or ensuring that they receive the treatment they deserve. Various approaches exist to uphold this principle. In the following discussion, we provide an informal explanation of some such ideas.
    \begin{itemize}
        \item \textit{Guarantee-based fairness:}
        Certain notions of individual-fairness guarantee that each voter should find satisfaction with at least a predetermined fraction of the budget allocation. Meanwhile, various group-fairness notions dictate that for a group of voters, an amount proportionate to their group size must be allocated in accordance with their collective preferences.
        \item \textit{Egalitarian fairness:} The goal of maximizing egalitarian welfare is also often viewed as a fairness goal that ensures that a voter from a minority group is treated on an equal footing with a voter from the majority group. This objective aims to promote symmetry in the treatment of individuals, irrespective of their strength.
    \end{itemize}
\end{itemize}
\section{Relating the Thesis to the State-of-the-Art}\label{sec: intro-positioning}
We first contextualize the thesis within the landscape discussed in the preceding section. We particularly confine our work to dichotomous and ordinal preferences, especially due to their cognitive simplicity over cardinal preferences \cite{benade2018efficiency}. 
\vspace*{-0.7\baselineskip}
\begin{center}
\begin{tikzpicture}
    \node (pb){\textbf{Participatory Budgeting}};
    \node (prefs)[below left=of pb,align=center]{
        \begin{tabular}{|ll|}
\hline
\multicolumn{2}{|c|}{\textbf{Preferences}}        \\ \hline
\multicolumn{1}{|l}{{(P1)}} & Dichotomous \\ \hline
\multicolumn{1}{|l}{(P2)} & Ordinal     \\ \hline
\end{tabular}
    };
    \node (costs)[below =of pb,align=center]{
        \begin{tabular}{|ll|}
        \hline
\multicolumn{2}{|c|}{\textbf{Costs}} \\\hline
\multicolumn{1}{|l}{(C1)} & Restricted \\\hline
\multicolumn{1}{|l}{(C2)} & Flexible   \\\hline
\end{tabular}
    };
    \node (objectives)[below right=of pb,align=center]{
        \begin{tabular}{|ll|}
        \hline
\multicolumn{2}{|c|}{\textbf{Objective}} \\\hline
\multicolumn{1}{|l}{(O1)} & Welfare Maximization \\\hline
\multicolumn{1}{|l}{(O2)} & Fairness   \\\hline
\end{tabular}
    };
    \draw[->] (pb)--(prefs);
    \draw[->] (pb)--(costs);
    \draw[->] (pb)--(objectives);
\end{tikzpicture}
\end{center}
\begin{table}[H]
\centering
    \begin{tabular}{cccccl}
(P1) & - & (C1) & - & (O1) & \cite{talmon2019framework,fluschnik2019fair,goel2019knapsack,baumeister2020irresolute,rey2020designing,jain2020participatoryg,freeman2021truthful,benade2021preference,laruelle2021voting}, \Cref{chap: d-re} \\
(P1) & - & (C1) & - & (O2) & \cite{aziz2018proportionally,fluschnik2019fair,laruelle2021voting,fairstein2021proportional,los2022proportional,brill2023proportionality}, \Cref{chap: d-re} \\
(P1) & - & (C2) & - & (O1) &  \cite{bogomolnaia2005collective,aziz2019fair}, \Cref{chap: d-fl} \\
(P1) & - & (C2) & - & (O2) & \cite{bogomolnaia2005collective,duddy2015fair,aziz2019fair}  \\
(P2) & - & (C1) & - & (O1) & \cite{lu2011budgeted, benade2017preference,laruelle2021voting}, \Cref{chap: o-re}  \\
(P2) & - & (C1) & - & (O2) & \cite{aziz2021proportionally,laruelle2021voting,pierczynski2021proportional}, \Cref{chap: o-re} \\
(P2) & - & (C2) & - & (O1) & \cite{anshelevich2017randomized}  \\
(P2) & - & (C2) & - & (O2) & \cite{aziz2014generalization,aziz2018rank,fain2016core,airiau2019portioning}, \Cref{chap: o-fl}
\end{tabular}
\caption{A chart that brings out the positioning of the thesis contributions with respect to the existing literature. The first entry indicates the preference elicitation method, the second indicates whether the costs are restricted or flexible, and the third entry indicates the objective considered. For example, `(P1)-(C2)-(O1)' indicates that the work studies welfare maximization under dichotomous preferences and flexible costs.}
\label{tab: positioning}
\end{table}
\subsection{Relevant Work and Research Gaps in the Literature}\label{sec: intro-gaps}
We now emphasize the existing gaps in the literature and elucidate how the thesis addresses and bridges these gaps.
\subsubsection{Dichotomous Preferences - Restricted Costs}\label{sec: intro-d-re-gap}
In the context of welfare maximization under restricted costs model, Fluschnik et al. \cite{fluschnik2019fair} studied Nash welfare under cardinal preferences, whose results can also be extended to the setting with dichotomous preferences. They prove that Nash welfare maximization is computationally hard even for the most restricted special case with just two voters and unit cost projects. Thus, the literature, with an exception of the work by Laruelle \cite{laruelle2021voting}, only focused on utilitarian welfare \cite{talmon2019framework,goel2019knapsack,baumeister2020irresolute,rey2020designing,jain2020participatoryg,freeman2021truthful,benade2021preference,laruelle2021voting}. Likewise, when it comes to fairness, all the existing works, other than the work by Laruelle \cite{laruelle2021voting}, looked at a guarantee-based fairness notion called proportionality \cite{aziz2018proportionally,fluschnik2019fair,laruelle2021voting,fairstein2021proportional,los2022proportional,brill2023proportionality} which ensures that every group of cohesive voters (i.e., voters having similar preferences) must receive a utility proportional to their group size. Both utilitarian welfare and proportionality suffer a common drawback: they are majority driven. Let us understand this with an example.
\begin{figure}[!ht]
    \centering
    \includegraphics[width=0.8\linewidth]{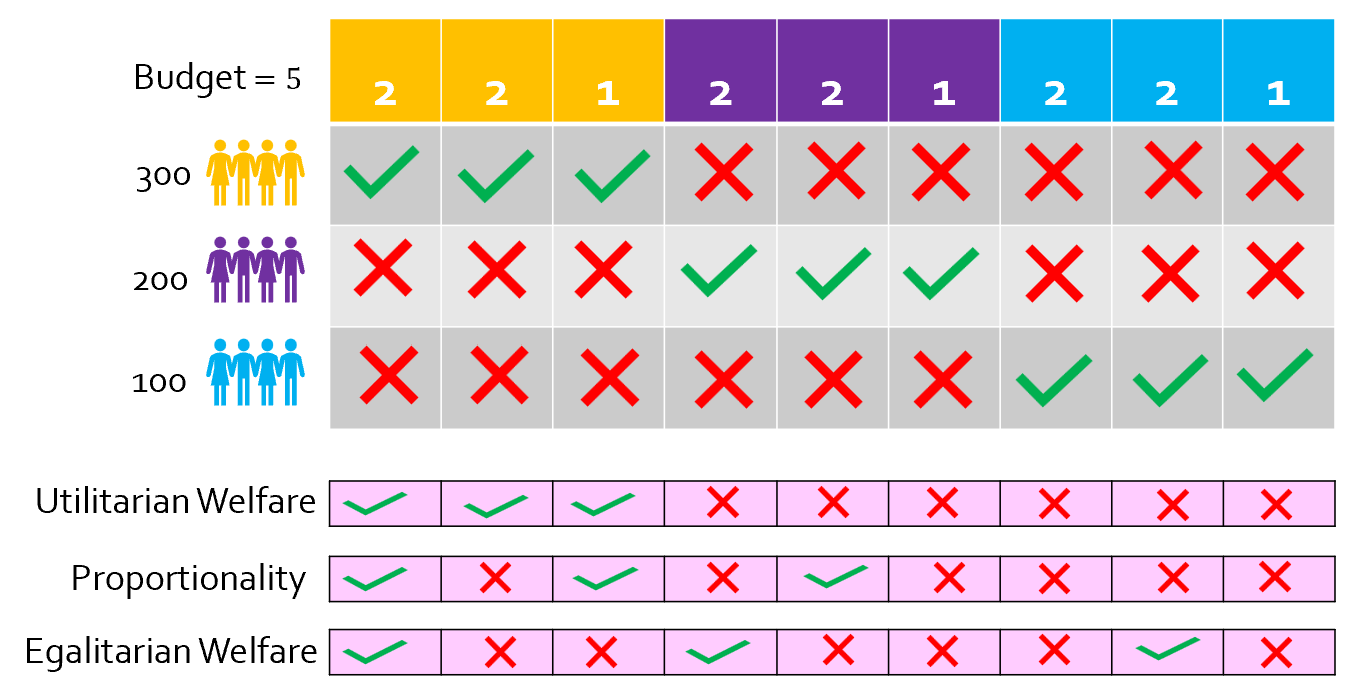}
    \caption{Stylized example to illustrate the need for egalitarian welfare in PB under dichotomous preferences and restricted costs.}
    \label{fig: motivation_egal}
\end{figure}

In \Cref{fig: motivation_egal}, each color denotes a county or district. Each column corresponds to a school construction project proposed in that district and the heading of the column denotes its cost. The total budget is $5$ units. The populations of the three districts are $300,200,$ and $100$ respectively. Every voter approves all and only the projects proposed in her district. The resultant outcomes of various objectives are denoted in pink-colored rows. It can be observed that the outcomes maximizing utilitarian welfare or achieving proportionality ignore the needs of sky-blue district. However, the government may want to construct schools in as many districts as possible and promote universal literacy, instead of constructing multiple schools in one populous district. Therefore, a more desirable goal could be to \emph{cover} all the voters. An outcome that maximizes the egalitarian welfare fulfills this purpose, as displayed in \Cref{fig: motivation_egal}. In fact, this is why the egalitarian welfare maximization is also considered to be a fairness objective as explained in \Cref{sec: intro-objectives}.

Laruelle \cite{laruelle2021voting} experimentally evaluates a \emph{sub-optimal} greedy algorithm for maximizing egalitarian welfare and achieving egalitarian fairness. However, a study of achieving \emph{optimal} egalitarian welfare remained to be explored. We fill this void in our \Cref{chap: d-re} by studying the egalitarian welfare maximization and providing equitable treatment to all the voters.
\subsubsection{Dichotomous Preferences - Flexible Costs}\label{sec: intro-d-fl-gap}
The existing work on PB under dichotomous preferences and flexible costs assumes that the costs are totally flexible and that any amount can be allocated to any project \cite{bogomolnaia2005collective,duddy2015fair,aziz2019fair}. However, there could be many scenarios where the costs are \emph{partially} flexible. One such example is illustrated in \Cref{eg: motivation_degrees}.
\begin{example}\label{eg: motivation_degrees}
    \begin{figure}[h]
    \centering
    \includegraphics[width=0.6\linewidth]{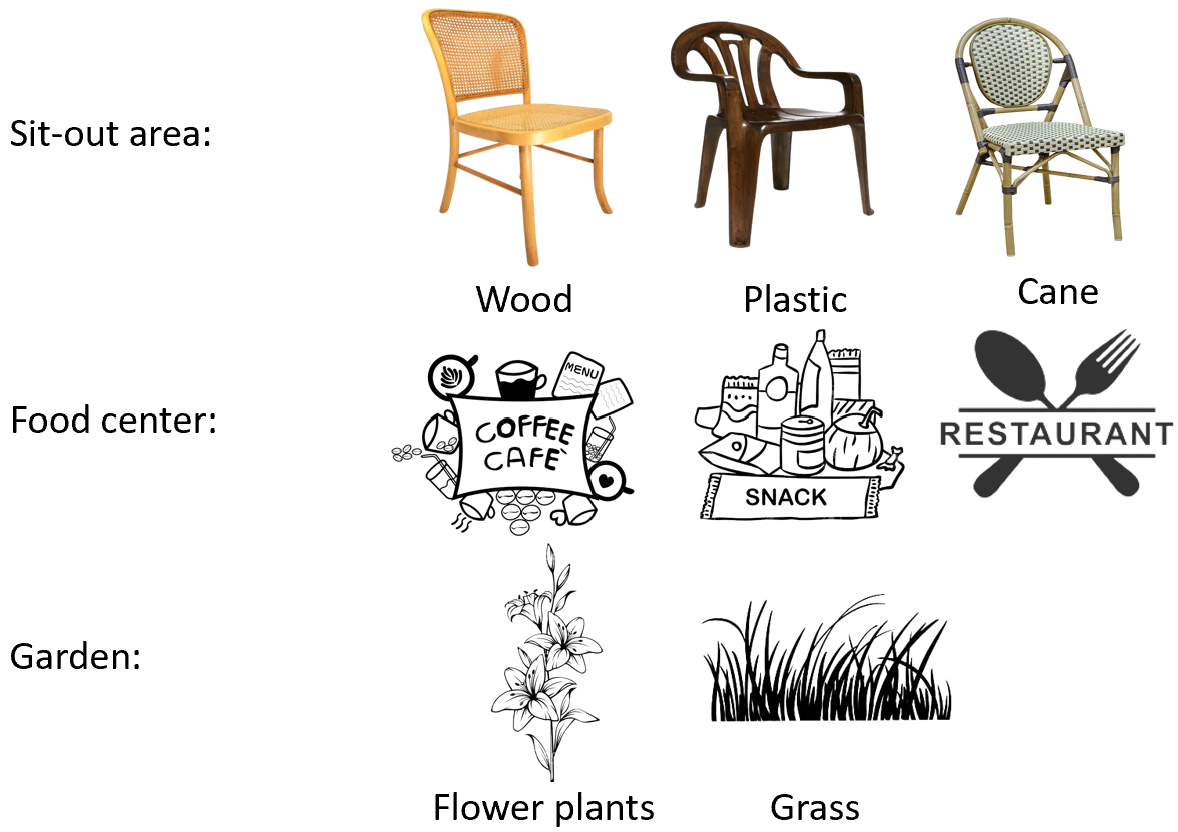}
    \caption{Stylized illustration of \Cref{eg: motivation_degrees} to motivate PB under partially flexible costs.}
    \label{fig: motivation_degrees}
\end{figure}
    Suppose a multi-national company is renovating the facilities at its office to promote mental health of its employees and to motivate them to work from the office. Project proposals include maintaining a small garden, building a sit-out area, constructing a food point etc. and the management decides to aggregate the preferences of employees. Each project can be executed in multiple ways. The food centre could be coffee cafe, or a small snack and fast-food center, or a full fledged restaurant; sit-out area may have chairs made of different materials, and so on. Each of these possibilities correspond to a different cost.
\end{example}

Similarly, a health service construction project could build a primary healthcare center, or a clinic, or a multi-speciality hospital. Assume that the construction of primary healthcare center needs $\$1M$, that of a clinic needs $\$10M$, and that of a multi-speciality hospital needs $\$100M$. We need a constraint which enforces that if the health service project is chosen to be funded, the amount allocated to it cannot be arbitrary and instead must belong to the set $\{1,10,100\}$. In \Cref{chap: d-fl}, we study the model with \emph{partially} flexible costs and maximize the utilitarian welfare. 
\subsubsection{Ordinal Preferences - Restricted Costs}\label{sec: intro-o-re-gap}
\paragraph{Welfare Maximization.}
Welfare maximization in PB under restricted costs and ordinal preferences is studied by Benade et al. \cite{benade2017preference} and Laurelle \cite{laruelle2021voting}. Benade et al. \cite{benade2017preference} assume that all the voters implicitly have cardinal preferences over the projects and that the PB organizer has access only to the ordinal preferences induced by these cardinal preferences. In other words, ordinal preferences are assumed to act as a proxy to the implicit cardinal preferences. Under this assumption, the authors quantified the distortion or loss in the utilitarian welfare caused due to the lack of access to cardinal preferences. However, this work comes with a set of drawbacks: (i) many times, the implicit preferences of the voters could simply be weakly ordinal and not cardinal (ii) they assume that the ordinal preferences are strict, whereas in real world, voters could be indifferent between two projects (iii) they assume that every voter has a preference over every project, whereas in real world, voters could have preferences only over a few projects and have no opinion on the others. Laurelle \cite{laruelle2021voting} experimentally evaluated a sub-optimal greedy algorithm to maximize welfare, but did not study any PB rules that achieve optimal welfare. We bridge all these gaps in \Cref{sec: o-re-welfare} of \Cref{chap: o-re}, wherein we propose PB rules that achieve optimal utilitarian welfare for PB under incomplete weakly ordinal preferences. 

\paragraph{Fairness.}
While the egalitarian fairness for PB under ordinal preferences and restricted costs is studied by Laurelle \cite{laruelle2021voting}, the guarantee-based fairness notions are studied by Aziz and Lee \cite{aziz2021proportionally} and Pierczynski et al. \cite{pierczynski2021proportional}. Aziz and Lee \cite{aziz2021proportionally} proposed two fairness notions, CPSC (Comparative Proportionality for Solid Coalitions) and IPSC (Inclusive Proportionality for Solid Coalitions), and Pierczynski et al. \cite{pierczynski2021proportional} followed the same lines of fairness (though their results are for cardinal preferences, they could be applied to ordinal preferences). Notably, all these notions suffer from major drawbacks: (i) they assume the ordinal preferences to be complete, whereas in real-world, voters could have no opinion on some projects (ii) an outcome satisfying CPSC is not always guaranteed to exist (iii) they consider fairness to be a \emph{hard constraint} instead of \emph{optimizing} it. Many times, this may result in sub-standard choices. Let us look at a toy example that demonstrates this.
\begin{figure}[!ht]
    \centering
    \includegraphics[width=0.3\linewidth]{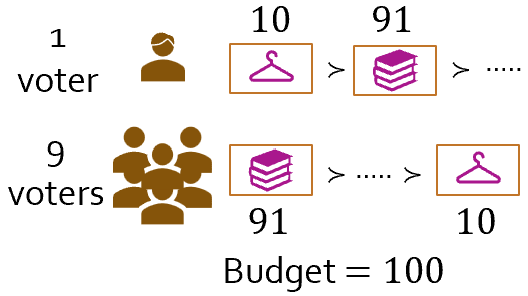}
    \caption{Stylized example to illustrate the drawback of existing fairness notions in PB under ordinal preferences and restricted costs.}
    \label{fig: motivation_args}
\end{figure}

In \Cref{fig: motivation_args}, we have $10$ voters and a budget of $100$ units. A project proposal to construct shopping mall costs $10$ units, whereas all the remaining projects, including a library construction proposal, costs $91$ units each. All the voters except one (say voter $i$) have the same preference: they prefer library the most and shopping mall the least. Whereas voter $i$ being a shopping buff prefers shopping mall over everything else. Existing fairness notions insist that the voter $i$ has a right over $\frac{1}{10}^{\text{th}}$ fraction of budget, i.e., on $10$ units of budget. Thus, any fair outcome ends up selecting only the shopping mall project and chooses no other project, resulting in the wastage of $90$ units of budget. Moreover, mall is the least preferred project for $90\%$ of the voters. Upon closer examination of the example, it appears rather strange that the library project was not chosen instead: library is the \emph{second most} preferred project for $i$ and the \emph{most} preferred project for all the remaining voters. In \Cref{sec: o-re-fair} of \Cref{chap: o-re}, we elaborate on the underlying cause of this drawback in the existing fairness notions and present two families of fair rules that effectively overcome this limitation.
\subsubsection{Ordinal Preferences - Flexible Costs}\label{sec: intro-o-fl-gap}
For the space of PB under ordinal preferences and totally flexible costs, several works studied individual-fairness \cite{aziz2014generalization,aziz2018rank,airiau2019portioning} and a few papers looked at group-fairness \cite{aziz2018rank,fain2016core}. The existing fairness notions have a few limitations. The first limitation is that they assume all the voters to have equal rights over the budget. However, there could be many real-world scenarios where each group is endowed with a different entitlement. One such example is the use of affirmative actions (\burl{https://en.wikipedia.org/wiki/Affirmative_action}) by numerous governments in the world to uplift the minority and historically discriminated groups. The government may want the underprivileged groups to have a say on higher fraction of the budget. The second limitation arises when the existing notions are applied to \emph{strictly} ordinal preferences: the notions simply reduce to dictatorship and yield unfavorable outcomes. In \Cref{chap: o-fl}, we will illustrate this weakness with multiple examples. Finally, in \Cref{sec: o-fl-notions}, we propose fairness notions that overcome these disadvantages. We further characterize group-fair and individually-fair PB rules that also satisfy other desirable properties.
\section{Thesis Contributions and Overview}\label{sec: intro-overview}
In this section, we outline the structure of the thesis and offer a concise overview of the key technical contributions presented in each chapter.
The thesis looks at PB models under each of the four possible combinations (dichotomous preferences-restricted costs; dichotomous preferences-flexible costs; ordinal preferences-restricted costs; ordinal preferences-flexible costs). The results are organized in two parts: (I) dichotomous preferences and (II) ordinal preferences. Each part further contains two chapters: the former chapter considers PB model under restricted costs, whereas the latter chapter considers PB model under flexible costs.

\begin{table}[h]
\begin{tabular}{lcclcl}
\multicolumn{1}{c}{}                   &                                                                                                                             & \multicolumn{4}{c}{\textbf{Costs}}                                                                                                      \\
\multicolumn{1}{c}{}                   &                                                                                                                             & \multicolumn{2}{c}{}                                               & \multicolumn{2}{c}{}                                               \\
\multicolumn{1}{c}{}                   &                                                                                                                             & \multicolumn{2}{c}{\multirow{-2}{*}{Restricted}}                   & \multicolumn{2}{c}{\multirow{-2}{*}{Flexible}}                     \\ \cline{3-6} 
\textbf{}                              & \multicolumn{1}{c|}{{\color[HTML]{FE0000} }}                                                                                & \multicolumn{2}{c|}{{\color[HTML]{FE0000} }}                       & \multicolumn{2}{c|}{{\color[HTML]{FE0000} }}                       \\
\textbf{}                              & \multicolumn{1}{c|}{{\color[HTML]{FE0000} }}                                                                                & \multicolumn{2}{c|}{{\color[HTML]{FE0000} }}                       & \multicolumn{2}{c|}{{\color[HTML]{FE0000} }}                       \\
                                       & \multicolumn{1}{c|}{\multirow{-3}{*}{{\color[HTML]{FE0000} \begin{tabular}[c]{@{}c@{}}Part I:\\ Dichotomous\end{tabular}}}} & \multicolumn{2}{c|}{\multirow{-3}{*}{{\color[HTML]{FE0000} \Cref{chap: d-re}}}} & \multicolumn{2}{c|}{\multirow{-3}{*}{{\color[HTML]{FE0000} \Cref{chap: d-fl}}}} \\ \cline{3-6} 
\multirow{-2}{*}{\textbf{Preferences}} & \multicolumn{1}{c|}{{\color[HTML]{009901} }}                                                                                & \multicolumn{2}{c|}{{\color[HTML]{009901} }}                       & \multicolumn{2}{c|}{{\color[HTML]{009901} }}                       \\
                                       & \multicolumn{1}{c|}{{\color[HTML]{009901} }}                                                                                & \multicolumn{2}{c|}{{\color[HTML]{009901} }}                       & \multicolumn{2}{c|}{{\color[HTML]{009901} }}                       \\
                                       & \multicolumn{1}{c|}{\multirow{-3}{*}{{\color[HTML]{009901} \begin{tabular}[c]{@{}c@{}}Part II:\\ Ordinal\end{tabular}}}}    & \multicolumn{2}{c|}{\multirow{-3}{*}{{\color[HTML]{009901} \Cref{chap: o-re}}}} & \multicolumn{2}{c|}{\multirow{-3}{*}{{\color[HTML]{009901} \Cref{chap: o-fl}}}} \\ \cline{3-6} 
\end{tabular}
\caption{An outline of the organization of the Thesis}
\label{tab: organization}
\end{table}

For each of the four PB models, we propose novel PB rules that maximize welfare or promote fairness. To maximize welfare, we study utility notions existing in the literature and also introduce new notions tailored to each model. We propose PB rules that optimize utilitarian or egalitarian welfare and thoroughly analyze the computational and axiomatic aspects of these rules. In the context of fairness, we critically assess existing notions, emphasizing their limitations and drawbacks. We put forward novel fairness notions which overcome these limitations, and either construct or characterize several families of innovative fair PB rules. Furthermore, we investigate the computational complexity of the newly proposed fair PB rules. Overall, the thesis presents a range of novel utility notions, axiomatic properties, fairness notions, and participatory budgeting rules for the four considered PB models.
\subsection*{\Cref{part: dichotomous}: Dichotomous Preferences}
This part has two chapters, \Cref{chap: d-re} and \Cref{chap: d-fl}, respectively dealing with restricted costs and flexible costs. The first part studies the optimization of egalitarian welfare, whereas the second part proposes four different families of rules each optimizing utilitarian welfare.
\subsubsection*{\Cref{chap: d-re}: Restricted Costs: Egalitarian Participatory Budgeting}
Egalitarianism plays a crucial role in participatory budgeting (PB), functioning as both a welfare and fairness objective, as discussed in \Cref{sec: intro-variants}. We introduce and examine a natural egalitarian rule ,called \emph{Maxmin Participatory Budgeting (MPB)}. The study encompasses two main components: computational analysis and axiomatic analysis.

In the computational analysis, we establish the computational complexity of MPB by demonstrating its strong NP-hardness. We also delve into fixed parameter tractability results and then present an approximation algorithm that guarantees distinct additive approximations for various families of instances. Furthermore, we provide an upper bound on the achievable approximation ratio for exhaustive strategy-proof PB algorithms. In the axiomatic analysis, we explore MPB by extending existing axioms and introducing a novel fairness axiom known as \emph{maximal coverage}. Through rigorous analysis, we explore the compatibility of MPB with each of the axioms and prove that MPB satisfies maximal coverage, making it desirable in terms of fairness.
\subsubsection*{\Cref{chap: d-fl}: Flexible Costs: Welfare Maximization when Projects have Multiple Degrees of Sophistication}
We introduce a novel model that considers projects with a finite set of permissible costs, each representing a distinct level of project sophistication. Voters participate express their preferences for each project, specifying the range of costs they believe the project deserves. The outcome of the participatory budgeting rule involves the selection of a subset of projects and the determination of the amount allocated to each selected project.

We propose four utility notions that capture different aspects of voter preferences and analyze welfare-maximizing rules corresponding to each of these notions. Importantly, we demonstrate that the positive results observed in the context of single-cost projects can be extended to our framework with multiple permissible costs. Additionally, we undertake a thorough analysis of the fixed parameter tractability of the problems, considering both existing parameters and novel parameters to assess computational complexity. We also introduce a set of intuitive axioms and scrutinize their compatibility with the four studied rules, thereby gaining insights into the fairness and desirability properties of the proposed rules.
\subsection*{\Cref{part: ordinal}: Ordinal Preferences}
This part also has two chapters, \ref{chap: o-re} and \ref{chap: o-fl}, respectively dealing with restricted and flexible costs. \Cref{chap: o-re} proposes families of rules that maximize welfare and two families that achieve fairness. \Cref{chap: o-fl} characterizes the fair rules in a special domain of ordinal preferences.
\subsubsection*{\Cref{chap: o-re}: Restricted Costs: Welfare Maximization and Fairness under Incomplete Weakly Ordinal Preferences}
This chapter focuses on incomplete weakly ordinal preferences and has two logical components. The first component presents rules that maximize the utilitarian welfare, whereas the second component presents the rules that achieve fairness.

In the first component, we introduce a family of rules, \emph{dichotomous translation rules}, which expand on existing welfare-maximizing rules for dichotomous ordinal preferences. Likewise, we also introduce the \emph{PB-CC rule} (participatory budgeting chamberlin courant rule) that expands on the well known CC rule for strictly ordinal preferences. We show that all our expansions largely maintain and even enhance the computational and axiomatic properties of the parent rules. We also propose a new relevant axiom, \emph{pro-affordability}. The second component introduces the novel classes of \emph{average rank-share guarantee rules} to address fairness by overcoming limitations of existing fairness concepts in the literature.
\subsubsection*{\Cref{chap: o-fl}: Flexible Costs: Characterization of Group-Fair and Individual-Fair Rules under Single-Peaked Preferences}
We assume that any amount can be allocated to each project, thereby making the model equivalent to random social choice. We study fairness in social choice under single-peaked domain. Literature has extensively examined the characterization of social choice rules in the single-peaked domain. We extend these findings by incorporating fairness considerations.

To address group-fairness, we assume a partition of the voters into logical groups, based on inherent attributes such as gender, race, or location. To capture fairness within each group, we introduce the concept of \emph{Group-Wise Anonymity}. To capture fairness across the groups, we propose a weak and a strong notion of \emph{Group-Entitlement Guarantee (GEG)}. These proposed fairness notions represent generalizations of existing individual-fairness notions. We characterize deterministic and random social choice rules that achieve group-fairness. We provide two separate characterizations of random rules: a direct characterization and an extreme point characterization (as convex combinations of deterministic rules). Additionally, we explore the case where no groups are present and provide more precise characterizations of the rules that achieve individual-fairness.

\blankpage
\chapter{Background} \label{chap:prelims}

\begin{quote} 
 \textit{In this chapter, we provide the background necessary for understanding the technical contributions of this thesis. We begin by introducing fundamental notation and then delve into the essential concepts of participatory budgeting and computational complexity as prerequisites for the ensuing discussions.}
\end{quote}

\section{Participatory Budgeting}\label{sec: prelims-pb}
Participatory Budgeting (PB) is a paradigm used to distribute a divisible resource, called \emph{budget}, among different \emph{projects}. It involves aggregating preferences of stakeholders (referred to as \emph{voters}) to decide how the budget should be allocated to the projects. There are multiple variants of PB models, as explained in \Cref{sec: intro-variants}. We start this section with introducing important notations and mathematically formulating all the variants of PB.
\subsection{Notations}\label{sec: prelims-notations}
Throughout the thesis, we assume that $N = \{1,2,\ldots,n\}$ represents the set of $n$ voters and $\proj = \{p_1,p_2,\ldots,p_m\}$ represents the set of $m$ projects. The available budget that needs to be distributed is represented by \bud. Additionally, there could be different kinds of restrictions on the funds allocated to each project and voters can express their preferences in various ways.
\subsubsection{Costs of Projects}\label{sec: prelims-costs}
The cost of a project signifies the amount that needs to be allocated to the project in the event of its selection for funding. There may be a restriction on the amount each project $p_j \in \proj$ could receive. That is, if $p_j$ is chosen to be allocated a non-zero amount in the final outcome, that amount has to be a value from \domainof{p_j}. We also refer to \domainof{p_j} as the set of possible costs of the project $p_j$. We mathematically define various possibilities of \domainof{p_j}, as described in \Cref{sec: intro-costs}.
\begin{itemize}
    \item \textbf{\textit{Restricted costs:}} The cost of each project $p_j$ is restricted to a single given value $\boldsymbol{\cof{p_j}} \in \mathbb{N}$. That is, $\domainof{p_j} = \{\cof{p_j}\}$. This model is also referred to as \textit{indivisible PB} in the literature.
    \item \textbf{\textit{Flexible costs:}} The allocation amount for each project can take on multiple values, which can further be divided into two cases:
    \begin{itemize}
        \item \textit{Partially flexible:} Each project $p_j \in \proj$ has multiple degrees of sophistication, $\{\pdof{}{1},\ldots,\pdof{}{\mdof{}}\}$, each representing a different way of executing the project. Each degree \pdof{}{} needs a cost of $\dcof{}{} \in \mathbb{N}$. If the project $p_j$ is chosen to be funded, the amount allocated to it must belong to the discrete set $\domainof{p_j} = \boldsymbol{\curly{\dcof{}{1},\ldots,\dcof{}{\mdof{}}}}$.
        \item \textit{Totally flexible:} There are no fixed costs associated with the projects and each project $p_j \in \proj$ may receive any amount. That is, $\domainof{p_j} = (0,\bud]$. This model is also referred to as \textit{divisible PB} in the literature.
    \end{itemize}
\end{itemize}
For a set $S$ of projects, we use $\cof{S}$ to denote the sum of costs of all the projects (or degrees of projects, as applicable) in the set $S$. In the restricted costs model, we say that a set $S$ of projects is \textbf{\textit{feasible}} if its total cost is within the budget limit, i.e., $\cof{S} \leq \bud$. We use $\boldsymbol{\feasible}$ to denote the set of all feasible subsets of projects, i.e., $\feasible = \curly{S \subseteq \proj: \cof{S} \leq \bud}$.
\subsubsection{Preferences of Voters}\label{sec: prelims-preferences}
In this section, we provide mathematical formalizations for the preference elicitation methods discussed in \Cref{sec: intro-preferences}. Our focus is limited to the preference elicitation methods investigated in this thesis.
\begin{itemize}
    \item \textbf{\textit{Dichotomous preferences:}}
    \begin{itemize}
        \item \textit{Standard dichotomous preferences:} Every voter $i \in N$ reports a subset $\boldsymbol{A_i}$ of projects.
        \item \textit{Knapsack votes:} Every voter $i$ reports a subset $\boldsymbol{A_i}$ of projects such that $c(A_i) \leq \bud$.
    \end{itemize}
    Each project in the set $A_i$ is considered to be \textbf{\textit{approved}} by $i$.
    \item \textbf{\textit{Ordinal preferences:}}
    \begin{itemize}
        \item \textit{Complete strictly ordinal preferences:} Every voter $i$ reports a strict ordering $\boldsymbol{\succ_i}$ over all projects in \proj.
        \item \textit{Incomplete weakly ordinal preferences:} Every voter $i$ reports a weak ordering $\boldsymbol{\succeq_i}$ over a subset of projects $S \subseteq \proj$.
    \end{itemize}
    We denote by $\boldsymbol{r_i(p)}$ the \textbf{\textit{rank}} of project $p$ in the preference of $i$. Rank of a project $p$ is defined to be exactly one greater than the number of projects strictly preferred over $p$. Given a value $t \in \{1,\ldots,m\}$, we use $\boldsymbol{\pat{t}}$ to denote all the projects ranked exactly $t$ in $\succeq_i$. Likewise, we use $\boldsymbol{\opat{}{t}}$ to denote the project ranked $t$ in $\succ_i$.
\end{itemize}
Let \NN denote the collection of preferences of all the voters in $N$. Under dichotomous preferences, we use $\boldsymbol{\approf}$ to denote the preference profile, i.e., $\approf = \NN = (A_i)_{i \in N}$. When the preferences are ordinal, we use $\boldsymbol{\rprof}$ to denote the preference profile, i.e., $\rprof = \NN = (\succeq_i)_{i \in N}$ (or $\rprof = (\succ_i)_{i \in N}$ when the preferences are strictly ordinal). In this thesis, henceforth, the term dichotomous preferences implies standard dichotomous preferences unless specified otherwise.

A typical participatory budgeting instance $\boldsymbol{\instance}$ is $\langle \bud,\voters,\proj,(\domainof{p_j})_{p_j \in \proj},\NN \rangle$. In each chapter, the variables $(\domainof{p_j})_{p_j \in \proj}$ and \NN will be adjusted accordingly based on the specific preference elicitation method and cost restriction in each model. Please note that in a model with totally flexible costs, $(\domainof{p_j})_{p_j \in \proj}$ can be omitted from the instance since it is implicitly understood.
\subsubsection{Participatory Budgeting Rule}\label{sec: prelims-pbrule}
In the PB model with restricted costs, a PB rule aggregates the preferences of voters and outputs one or more feasible subsets of projects.
\begin{definition}[\textbf{PB Rule under Restricted Costs}]\label{def: pbrule-restricted}
Given a PB instance \instance, a PB rule \pbrule generates one or more feasible subsets of projects. That is, $\ruleof{\pbrule}{} \subseteq \feasible$.
\end{definition}
A PB rule under restricted costs is said to be \textbf{\textit{irresolute}} if it outputs more than one feasible subsets of projects. In the PB model with flexible costs, a PB rule aggregates the preferences of voters and outputs the corresponding amount allocated to each project. Note that if a project is not chosen to be funded, it gets allocated an amount of $0$.
\begin{definition}[\textbf{PB Rule under Flexible Costs}]\label{def: pbrule-flexible}
Given a PB instance \instance, a PB rule \pbrule outputs the corresponding allocated cost $x_j \in \domainof{p_j} \cup \curly{0}$ for each project $p_j \in \proj$.
\end{definition}
\subsection{Utility Notions}\label{sec: prelims-utilities}
There are various methods to measure the utility that a voter obtains from the outcome of a PB rule. In this section, we present the existing utility notions for both dichotomous preferences and ordinal preferences.
\begin{note}
    Utility notions in this section are only defined for the \textit{restricted costs} model. Utility notions for partially flexible costs model do not exist in the literature (the thesis addresses this gap in subsequent chapters). While utility notions for totally flexible cost exist in the literature, they are beyond the purview of this thesis.
\end{note}
\subsubsection{Utilities under Dichotomous Preferences}\label{sec: prelims-dutilities}
Talmon and Faliszewski \cite{talmon2019framework} introduced three notions of utilities for PB under dichotomous preferences and restricted costs. They are as follows:
\begin{enumerate}
    \item \textbf{Cardinal utility}: Each voter \ii derives a utility of $1$ from an approved project $p$ if the project is chosen in $S$ and $0$ otherwise. That is, $$\utof{}{} = |p: p \in A_i \cap S|.$$
    \item \textbf{Cost utility}: Each voter \ii derives a utility of $\cof{p}$ from an approved project $p$ if the project is chosen in $S$ and $0$ otherwise. That is, $$\utof{}{} = \suml{p \in A_i\cap S}{\cof{p}}.$$
    \item \textbf{Boolean utility}: Each voter \ii derives a utility of $1$ if at least one project in $A_i$ is chosen in $S$ and $0$ otherwise. That is, 
    \begin{align*}
        \utof{}{} \!=\!
        \begin{cases} 
        1 &  A_i \cap S \neq \emptyset\\
        0 & \text{otherwise}
        \end{cases}
    \end{align*}
\end{enumerate}
\subsubsection{Utilities under Ordinal Preferences}\label{sec: prelims-outilities}
While there is no specific utility notion for PB under ordinal preferences defined in the literature, various utility notions established for other social choice problems under ordinal preferences can readily be extended to PB as well. In this thesis, we specifically focus one such utility concept: Chamberlin-Courant (CC) utility \cite{chamberlin1983representative}. The notion intuitively implies that a voter derives utility from the rank of her most preferred funded project.

\vspace*{0.5\baselineskip}
\noindent{\textbf{CC utility}: Each voter $i$ with an ordinal preference derives a utility of $\utof{}{} = m - \minl{p \in S}{r_i(p)}$}.

\begin{nonexample}
    Suppose we have $m = 4$ projects. Let the ordinal preference of voter $i$ be $\{p_1,p_3\} \succ_i p_4 \succ_i p_2$. Consider the set $S = \{p_2,p_4\}$. The most favorite funded project of $i$ is $p_4$ and its rank is $r_i(p_4) = 3$ (since exactly two projects are preferred over it). Therefore $\utof{}{} = 4-3 = 1$.
\end{nonexample}
\subsection{Welfare Measures}\label{sec: prelims-welfare}
Welfare maximization has been one of the most commonly pursued objective in social choice literature. As informally discussed in \Cref{sec: intro-objectives}, the literature formulates several welfare measures with respect to a utility notion. In this thesis, we investigate two such welfare measures: utilitarian welfare and egalitarian welfare.
\begin{enumerate}
    \item \textbf{Utilitarian Welfare}: Also referred to as social welfare in the literature, this measure simply reflects the sum of utilities of all the voters, that is, $UW(S) = \suml{i \in \voters}{\utof{}{}}$.
    \item \textbf{Egalitarian Welfare}: Also referred to as rawlsian welfare in the literature, this measure reflects the utility of the worst-off voter, that is, $EW(S) = \minl{i \in \voters}{\utof{}{}}$.
\end{enumerate}
The typical objective of a welfare-maximizing PB rule is to find feasible set(s) of projects which maximizes the welfare objective. Nash welfare is beyond the purview of this thesis, particularly due to its extreme computational intractability as mentioned in \Cref{sec: intro-d-re-gap}.
\section{Computational Complexity Concepts}\label{sec: prelims-complexity}
Computational complexity of a problem is an evaluation of the resources such as time and space required to solve the problem \cite{garey1979computers}. Computational complexity theory, in general, focuses on classifying problems into specific \emph{complexity classes}, such that all the problems within a class are understood to have similar complexity \cite{papadimitriou2003computational}. This thesis investigates the time complexity of various participatory budgeting problems. We briefly introduce the background on time complexity classes.

The problems in theoretical computer science are broadly classified into two categories: decision problems and optimization problems. A decision problem focuses on verifying whether or not a solution satisfying required properties \emph{exist}. Alternatively, the corresponding optimization problem \emph{computes} such a solution(s). Clearly, optimization version of a problem is at least as hard as its decision version. So, the positive results are usually presented for the optimization version, whereas the negative/impossibility results are presented for the latter.

\subsection{${\Poly}$ and ${\NP}$}\label{sec: prelims-pnp}
The most fundamental classes of time complexity of problems are \Poly and \NP. The class \Poly is a collection of all problems that can be solved in polynomial (in size of the input) time using a deterministic Turing machine, whereas the class \NP is a collection of all those problems that can be solved in polynomial time using a non-deterministic Turing machine. The query on whether or not \Poly and \NP classes are equal remains an open question. However, it is widely believed that \Poly $\neq$ \NP.

For the sake of simplicity, our statements in the thesis are with respect to a deterministic Turing machine. We refer to the problems in \Poly as polynomial-time solvable or computationally tractable problems, and those in \NP or in the classes more complex than \NP as computationally hard or intractable problems. To prove that an optimization problem is in \Poly, we propose algorithms that run in polynomial time in size of the input. To prove that an optimization problem is computationally hard, we generally define its decision version and prove that it is at least as hard as an existing decision problem in \NP class. We achieve this by a process of \emph{reduction} of one problem to the other.

Formally, a decision problem is encoded as a language $D \subseteq \Sigma^*$ over an alphabet $\Sigma$.
\begin{definition}\label{def: reduction}
    A decision problem $D_1$ \textbf{\textit{reduces}} to another decision problem $D_2$ if there exists a polynomial-time computable function $f$ such that for every $x \in \Sigma^*$, $x \in D_1 \iff f(x) \in D_2$.
\end{definition}
\begin{definition}[$\boldsymbol{\NP}$\textbf{-hard} and $\boldsymbol{\NP}$\textbf{-complete}]\label{def: nph}
    A decision problem $D$ is said to be \NP-hard if all problems in \NP reduce to $D$. If a \NP-hard problem $D$ belongs to \NP, $D$ is said to be \NP-complete.
\end{definition}

\subsection{Coping with Intractability}\label{sec: prelims-intractability}
Consider a \NPH decision problem $D$ and its corresponding optimization problem $O_D$. As described in the previous section, the fact that $D$ is \NPH is informally interpreted as follows:
\begin{center}
    Given an \textcolor{red}{\textit{arbitrary}} instance $I$ of $O_D$, it is  \textcolor{red}{\textit{intractable}} to compute an  \textcolor{red}{\textit{optimal}} solution.
\end{center}
To cope up with such a negative result, scientists commonly attack the problem in one or more of the following different directions:
\begin{enumerate}[(i)]
    \item The problem is intractable for \emph{arbitrary} instances. But does it become tractable for \emph{special families of instances} that satisfy a specific property or structure?
    \item The problem is \emph{intractable}, which means that it is not known to be solved in polynomial-time. Could it be \emph{partially tractable}? If yes, what must be the definition of partial tractability?
    \item Finding an \emph{optimal} solution is found to be intractable. Would it be tractable to find a \emph{sub-optimal solution close} to the optimal solution? If yes, how should the closeness be defined?
\end{enumerate}
For the first direction, we look at families of instances satisfying a certain property. For example, in the space of instances of PB under dichotomous preferences, we could confine to instances with only knapsack votes (since knapsack votes are a special case of dichotomous preferences). We could also impose other properties on the instances and shrink the space of instances further.
\subsubsection{Partial Tractability}\label{sec: prelims-partialtractability}
For the second direction, we discuss two notions of partial tractability from the literature: pseudo-polynomial time tractability and fixed parameter tractability.

\paragraph{\textbf{Pseudo-polynomial time tractability.}} A problem $D$ is said to be computed in polynomial time (and thus lie in \Poly) if its running time is a polynomial in the size of the instance $I$, i.e., in the number of bits used to represent the instance. In other words, the running time must be $O(|I|^c)$ for some constant $c$, where the size $|I| = O(\log I)$.

On the flip side, $D$ is said to be computed in \emph{pseudo-polynomial time} if its running time is a polynomial in the size of $I$ when $I$ is expressed in unary. In other words, the running must be $O(|I|^c)$ for some constant $c$, where the size $|I| = I$.
\begin{definition}[\textbf{Weakly} $\boldsymbol{\NP}$\textbf{-hard and Strongly} $\boldsymbol{\NP}$\textbf{-hard}]
    A \NPH problem is said to be weakly \NPH if it is solvable in pseudo-polynomial time and strongly \NPH otherwise.
\end{definition}

\paragraph{\textbf{Fixed parameter tractability.}} An alternative way to define partial tractability is to identify the specific parameter(s) of an instance that is the primary cause of intractability and demonstrate that the running time is polynomial when this particular parameter remains constant \cite{downey2012parameterized}. To formally define this, we first define the problem in terms of parameter.

Formally, a parameterized decision problem is encoded as a language $D \subseteq \Sigma^* \times \Sigma^*$ , where $\Sigma$ is a finite alphabet. The second part of the problem indicates the parameter.
\begin{definition}\label{def: fpt}
    A parameterized decision problem $D$ is said to be \textbf{fixed parameter tractable} with respect to the parameter $k$ if it can be decided whether or not $(I,k) \in D$ by an algorithm in $f(k).|I|^{O(1)}$ time, where $f$ is an arbitrary computable function depending only on $k$.
\end{definition}
\FPT class is a collection of all fixed parameter tractable problems and a problem $D$ as described in \Cref{def: fpt} is said to be \emph{in \FPT w.r.t. parameter $k$}. However, even when the parameters remain constant, some problems could still continue to be intractable. Such problems are naturally harder than the \NPH problems in \FPT. Downey and Fellows \cite{downey2012parameterized} classified parameterized intractable problems into $\mathsf{W}$\textit{-hierarchy} based on boolean circuits. Below is a demonstration of relations between all these classes:
\begin{align*}
    \FPT=\mathsf{W[0]} \subseteq \mathsf{W[1]} \subseteq \mathsf{W[2]} \subseteq \ldots \mathsf{W[\Poly]} \subseteq \XP
\end{align*}
This thesis contains results only on \FPT class and $\mathsf{W[2]}$ class. Parallel to \Cref{def: reduction} and \Cref{def: nph}, we have parameterized reductions and $\mathsf{W[t]}$-hardness as defined below.
\begin{definition}\label{def: parareduction}
    A parameterized decision problem $D_1 \subseteq \Sigma^* \times \Sigma^*$ reduces to another parameterized decision problem $D_2 \subseteq \Sigma^* \times \Sigma^*$ by \textbf{parameterized reduction} if there are two computable functions $h_1$ and $h_2$ depending only on $|k|$ and a function depending on $I$ and $k$ such that
    \begin{enumerate}
        \item $(I,k) \in D_1 \iff f(I,k) \in D_2$
        \item $h_1(k) = k'$ whenever $f(I,k) = (I',k')$
        \item $f$ is computable in $|I|^{O(1)}.h_2(|k|)$
    \end{enumerate}
\end{definition}
\begin{definition}[$\boldsymbol{\mathsf{W[t]}}$\textbf{-hard}]\label{def: wth}
    A parameterized decision problem $D$ is said to be $\mathsf{W[t]}$-hard if all problems in $\mathsf{W[t]}$ reduce to $D$.
\end{definition}

\subsubsection{Approximation Algorithms}\label{sec: prelims-approximation}
Finally, we discuss the third direction of coping up with intractability, in which we find a sub-optimal solution in polynomial time and quantify its closeness to the optimal solution. This is precisely achieved by what are called as \emph{approximation algorithms} in the literature. We recommend the textbooks by Vazirani \cite{vazirani2001approximation} and Williamson and Shmoys \cite{williamson2011design} for a comprehensive overview of the topic.

If the objective in an optimization problem $O_D$ is to maximize a given function, we call the problem a maximization problem. Likewise, if the objective is to minimize, we refer to the problem as a minimization problem. Let \YY denote the space of all the possible input instances of the problem $O_D$. Let $ALG$ denote a polynomial-time approximation algorithm returning a sub-optimal solution for the problem, whereas $OPT$ denotes the optimal algorithm (optimizes the objective exactly). The approximation ratio of $ALG$ denotes the ratio between the optimal value of the objective and the value returned by $ALG$, in the worst case scenario. Formally, we define it below.

\begin{definition}[\textbf{Approximation Ratio}]\label{def: approxratio}
   For a maximization problem, approximation ratio of an algorithm is the maximum value $\zeta \in (0,1)$ such that $ALG(I) \geq \zeta\cdot OPT(I)$ for every $I \in \YY$. For a minimization problem, approximation ratio of an algorithm is the minimum value $\zeta > 1$ such that $ALG(I) \leq \zeta\cdot OPT(I)$ for every $I \in \YY$.
\end{definition}
By definition, approximation ratio can never be $1$. Garey and Johnson \cite{garey1978strong} further introduced the concept of \emph{approximation schemes}, which is to propose a single polynomial-time algorithm with an accuracy parameter $\epsilon$ such that it always guarantees an approximation ratio that is a function of $\epsilon$ and very close to $1$. We formally define them below.
\begin{definition}[\textbf{Polynomial Time Approximation Scheme (PTAS)}]
    An algorithm $ALG$ for the problem $O_D$ is said to be a PTAS, if for a given acuracy parameter $\epsilon>0$, for every instance $I \in \YY$, it holds that $ALG(I) \geq (1-\epsilon)\cdot OPT(I)$ when $O_D$ is a maximization problem and $ALG(I) \leq (1+\epsilon)\cdot OPT(I)$ when it is a minimization problem.
\end{definition}
Note that closer the $\epsilon$ is to $1$, stronger is the approximation ratio. This means that, though getting an approximation ratio of $1$ is intractable, getting a ratio arbitrarily close to $1$ is tractable. An algorithm with a running time $O(|I|^{\frac{1}{\epsilon}})$ is also a PTAS. The next approximation scheme prohibits $\epsilon$ from taking such a position in the exponent.
\begin{definition}[\textbf{Fully Polynomial Time Approximation Scheme (FPTAS)}]
    An algorithm $ALG$ is said to be a FPTAS, if it is a PTAS whose running time is $O(|I|^{O(1)}(\frac{1}{\epsilon})^{O(1)})$, where $|I|$ is the input size and $\epsilon$ is the accuracy parameter.
\end{definition}

\newpage
\part{Dichotomous Preferences}\label{part: dichotomous}
\thispagestyle{empty}

\blankpage

\chapter{Restricted Costs: Egalitarian Participatory Budgeting}\label{chap: d-re}

\begin{quote}
\textit{Egalitarianism, an important objective in participatory budgeting (PB) that maximizes the utility of worst-off voter, has not received much attention in the case where costs of the projects are restricted to a single value. We address this gap in this chapter through a detailed study of a natural egalitarian rule named Maxmin Participatory Budgeting (\mmpb). Our study is in two parts: (1) computational (2) axiomatic. }

\textit{In the computational part, we establish the computational complexity of \mmpb and present pseudo-polynomial time and polynomial-time algorithms, parameterized by well-motivated parameters. We introduce an algorithm that provides distinct additive approximation guarantees for various families of instances and demonstrate through empirical evidence that our algorithm yields exact optimal solutions on real-world participatory budgeting datasets. Additionally, we determine an upper bound on the achievable approximation ratio for \mmpb by examining exhaustive strategy-proof participatory budgeting algorithms.}

\textit{In the second part, we embark on an axiomatic study of the \mmpb rule by extending existing axioms found in the literature. This study culminates in the introduction of a novel fairness axiom named maximal coverage, which effectively captures fairness considerations. We investigate the compatibility of \mmpb with all the axioms.}
\end{quote}

\section{Motivation}\label{sec: d-re-intro}
Among the various preference elicitation methods for participatory budgeting, dichotomous preferences (each voter specifying a subset of projects to be funded) have received much attention in the literature due to their cognitive simplicity \cite{benade2018efficiency}. In participatory budgeting under dichotomous preferences, one way to define the utility of a voter is to base it on the number of approved projects of the voter that are funded \cite{rey2020designing,jain2020participatoryg,pierczynski2021proportional}. This notion, however, fails to capture the role played by the size of a project, as demonstrated below in \Cref{eg: costutil_moti}.
\begin{example}\label{eg: costutil_moti}
    The budget is $\$100M$. A PB rule selects two projects approved by voter $1$,  costing $\$4M$ and $\$6M$, and a project approved by $2$ costing $\$90M$. The utility notion based on cardinality gives higher utility to voter $1$ though $90\%$ of the budget is allocated favorably to voter $2$.
\end{example}
This immediately motivates other well-studied notion of utility, namely, the total amount allocated to the approved projects of the voter \cite{bogomolnaia2005collective,duddy2015fair,talmon2019framework,aziz2019fair,goel2019knapsack,freeman2021truthful}, also explained in \Cref{sec: prelims-dutilities}. Assuming that the utility of a voter is the total amount allocated to its approved projects, the natural question that follows is to decide the appropriate objective of participatory budgeting.

When the objective of participatory budgeting is to maximize welfare, Fluschnik et al. \cite{fluschnik2019fair} studied Nash welfare maximization, which is to maximize the geometric mean of utilities of all the voters. The authors proved that maximizing Nash welfare is computationally hard even for the most restricted special case with just two voters and unit cost projects (their results assume cardinal preferences, but can also be trivially extended to the setting with dichotomous preferences). This resulted in most works focusing on the maximization of utilitarian welfare, which is the sum of utilities of all the voters \cite{talmon2019framework,goel2019knapsack,baumeister2020irresolute,freeman2021truthful,benade2021preference,laruelle2021voting}.

On the other hand, when the objective of participatory budgeting is to achieve fairness, most works looked at a guarantee-based fairness notion called proportionality \cite{aziz2018proportionally,fluschnik2019fair,laruelle2021voting,fairstein2021proportional,los2022proportional,brill2023proportionality} which ensures that every group of cohesive voters (i.e., voters having similar preferences) must receive a utility proportional to their group size. Notably, both the objectives of achieving utilitarian welfare maximization and proportionality suffer a major disadvantage of being majority-driven. To understand it, let us look at two examples: \Cref{eg: egal_moti} and \Cref{fig: motivation_egal}.

\begin{example}\label{eg: egal_moti}
    A budget of \$50M is to be allocated to school construction projects in three villages, X, Y, and Z, with populations of 10000, 6000, and 2000, respectively. Village X proposes schools at localities $\curly{X_1,X_2,X_3,X_4}$ costing  $\curly{\$10M,\$20M,\$20M,\$60M}$ respectively, Y proposes schools at $\curly{Y_1,Y_2,Y_3}$ costing $\curly{\$14M,\$14M,\$16M}$ respectively, and Z proposes a school at $Z_1$ costing \$6M. Suppose each voter in a village approves all and only the projects proposed by her village. 

    Utilitarian welfare maximization would result in the schools being constructed at $\curly{X_1,X_2,X_3}$. The fairness notions based on proportionality also ignore the village $Z$ since the population strength of $2000$ is not proportionate to the amount \$6M needed. However, the government may want to construct schools in as many villages as possible instead of constructing multiple schools at various locations of the same village. In such a scenario, the outcome $\curly{X_1,X_2,Y_1,Z_1}$ represents all the villages including minorities, and promotes universal literacy.
\end{example}
\begin{figure}[!ht]
    \centering
    \includegraphics[width=0.8\linewidth]{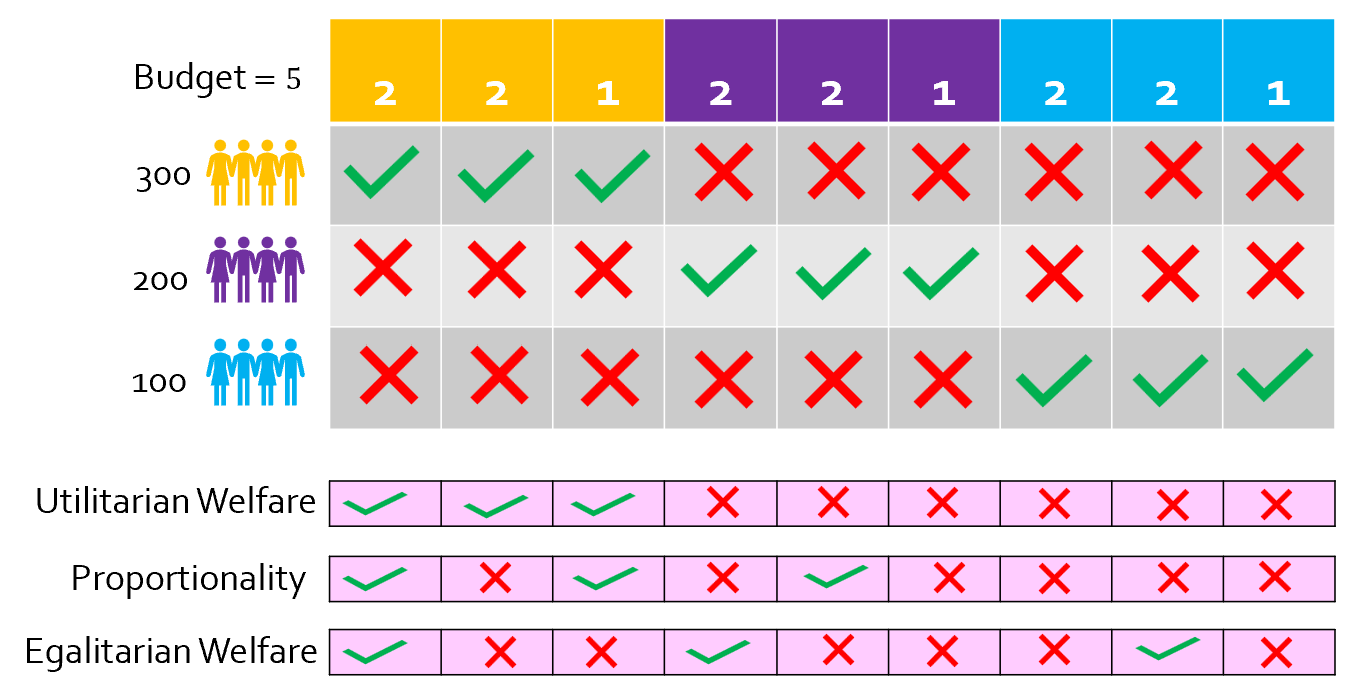}
    \repeatcaption{fig: motivation_egal}{A toy example to demonstrate where egalitarian welfare fares better than utilitarian welfare and proportionality. Colors represent counties or districts, columns represent sewage projects with associated costs. Total budget available is $5$ units. District populations are $300, 200,$ and $100$. Each voter approves only the projects in their district. Outcomes for different objectives are in pink rows. Evidently, utilitarian and proportional outcomes neglect sky-blue district, while egalitarian outcome does not.
}
\end{figure}

The example presented above underscores the importance of prioritizing an egalitarian objective. In simpler terms, when the goal is to maximize the utility of the most disadvantaged voter (i.e., the voter with least utility), the issue of favoring majority preferences, as demonstrated in the example, would not arise. This is precisely why, as explained in \Cref{sec: intro-objectives}, egalitarian welfare maximization is considered to be a welfare objective as well as a fairness objective in the social choice literature.

In the models where the costs of projects are flexible, egalitarian objectives are studied by Aziz and Stursberg \cite{aziz2014generalization}, Aziz et al. \cite{aziz2019fair}, Airiau et al. \cite{airiau2019portioning}, and Tang et al. \cite{tang2020price}. When costs of the projects are restricted to admit only one value, Laruelle \cite{laruelle2021voting} conducted a case study that experimentally evaluated a \emph{sub-optimal} greedy algorithm to optimize an egalitarian objective. However, rather surprisingly, a theoretical or experimental study on achieving \emph{optimal} egalitarian welfare did not exist and this chapter fills the said void. It needs to be mentioned that Brams et al. \cite{brams2007minimax} studied a very restricted special case of our model, named multi-winner voting, where the budget is $k$ units and the projects cost $1$ unit each.

\section{Contributions and Organization of the Chapter}\label{sec: d-re-contri}
We choose the \emph{maxmin} as the egalitarian objective to be optimized. We propose the rule, \emph{Maxmin Participatory Budgeting (\mmpb)}, which maximizes the utility of the voter with least utility. Our study of \mmpb for indivisible PB is in two parts: (1) computational (2) axiomatic. 

In the first part (\Cref{sec: d-re-comp}), we prove that \mmpb is computationally hard and give pseudo-polynomial time and polynomial-time algorithms when parameterized by certain well-motivated parameters. We propose an algorithm that achieves, for \mmpb, additive approximation guarantees for several restricted spaces of instances and empirically show that our approximation algorithm in fact gives \emph{exact} optimal solutions on real-world PB datasets. We also establish an upper bound on the approximation ratio achievable for \mmpb by the family of exhaustive strategy-proof PB algorithms, thereby quantifying the loss in welfare incurred due to strategy-proofness.

In the second part (\Cref{sec: d-re-axioms}), we undertake an axiomatic study of the \mmpb rule by generalizing known axioms in the literature. Our study leads us to propose a new axiom, \emph{maximal coverage}, which captures the fairness notion diversity \cite{faliszewski2017multiwinner}.  We prove that \mmpb satisfies maximal coverage and achieves diversity.
\section{Notations}\label{sec: d-re-prelims}
We recall the necessary notations and terms defined in \Cref{sec: prelims-pb} and also define the MPB rule. Recall that $\voters = \curly{1,\ldots,n}$ is the set of voters and $\proj = \curly{p_1,\ldots,p_m}$ is the set of projects. The function $c: \proj \to \mathbb{N}$ represents the cost and $\bud \in \mathbb{N}$ is the total budget available (note that $\mathbb{N}$ is the set of natural numbers). The dichotomous preference profile of all voters, $(\aof{})_{i \in \voters}$, is represented by $\approf$, where $\aof{} \subseteq \proj$ is the set of projects approved by the voter \ii. An instance \instance of participatory budgeting is \fullapinstance. With a slight abuse of notation, we represent the cost of a set $S$ of projects, $\sum_{p \in S}{\cof{p}}$, by $\cof{S}$. A set $S$ is said to be feasible if $\cof{S} \leq \bud$. A feasible approval vote \aof{} is also called a \emph{knapsack vote}. Let \feasible represent the set of all feasible subsets of projects. Recall the definition of a PB rule as described below.

\newtheorem*{pbrule_repetition}{\Cref{def: pbrule-restricted}}
\begin{pbrule_repetition}
Given a PB instance \instance, a \textbf{PB rule} \pbrule generates one or more feasible subsets of projects. That is, $\ruleof{\pbrule}{} \subseteq \feasible$.
\end{pbrule_repetition}

A \textbf{PB algorithm} is a PB rule such that always outputs only a single feasible subset, or in other words, for all \instance, $|\ruleof{\pbrule}{}| = 1$. Utility of a voter \ii from a set of projects $S$ is defined as the amount of money allocated to the projects approved by \ii, i.e., $\uof{}{S} = \cof{\aof{} \cap S}$.

\begin{definition}[Maxmin Participatory Budgeting (\mmpb)]
Given a PB instance \instance, the $\boldsymbol{\mmpb}$ \textbf{rule} outputs a set of all subsets $S \subseteq \proj$ such that $$S \in \argma{S \in \feasible}{\minl{i \in \voters}{\uof{}{S}}}.$$ 
\end{definition}
Note that our \mmpb rule is irresolute by its definition (\Cref{sec: prelims-pbrule}) since there could be multiple optimal subsets. For ease of presentation, we call the objective optimized by the \mmpb rule the $\boldsymbol{\mmpb}$ \textbf{objective} or simply $\boldsymbol{\mmpb}$. Given a PB instance \instance and a score $s$, the problem of determining if there exists $S \in \feasible$ such that $\min_{i \in \voters}{\uof{}{S}} \geq s$ is called the \textbf{decision version of} $\boldsymbol{\mmpb}$.
\section{Computational Results}\label{sec: d-re-comp}
We first prove the hardness of \mmpb and present some tractable special cases based on certain well-motivated parameters. We then give an approximation algorithm for \mmpb and show empirically that, in fact, it gives exact optimal solutions on real-world PB datasets. We conclude the section by establishing an upper bound on the approximation ratio achieved for \mmpb by any exhaustive strategy-proof PB algorithm. We start by formulating \mmpb as the following integer linear program (\setword{ILP}{word: ILP}) where each variable $x_p$ corresponds to selection of the project $p$.
\begin{align}
    \nn
    \max\; &q\\
    \label{eq: ipl1}
    \text{subject to }&q \leq \suml{p \in \aof{}}{\cxof{p}} \quad \forall i \in \voters\\\nn
    &\suml{p \in \proj}{\cxof{p}} \leq \bud\\
    \label{eq: ilp2}
    &x_p \in \curly{0,1} \quad \forall p \in \proj\\\nn
    &q \geq 0
\end{align}
As we can see, any feasible solution of the above ILP corresponds to a subset $S \subseteq \proj$ that includes all and only the projects whose corresponding variable $x$ is assigned $1$ in the solution. The second constraint ensures feasibility of the subset. \Cref{eq: ipl1} sets the value of $q$ to be equal to $\minl{i \in \voters}{\utof{}{}}$.
\subsection{NP-Hardness}\label{sec: d-re-nph}
Decision version of \mmpb is found to be strongly \NPH based on a reduction from a well known problem \setcov in the literature.
\begin{definition}[\setcov]\label{prob: setcov}
   Given a set $U$ of elements, a collection $S$ of subsets of $U$, and a positive integer $k$, the \setcov problem is to find if there exists $F \subseteq S$ such that $\cup_{P \in F}P = U$ and $|F| = k$. 
\end{definition}
The above problem of \setcov is known to be strongly \NPH \cite{garey1979computers}. We reduce this problem to the decision version of \mmpb to prove the following result.
\begin{theorem}\label{the: strongnph}
The decision version of \mmpb is strongly \NPH.
\end{theorem}
\begin{proof}
    We give a reduction from the \setcov problem. Given an instance $\langle U,S,k \rangle$ of \setcov, we construct an \mmpb instance as follows: For each $C \in S$, create a project $p_C$ with unit cost. For each $i \in U$, create a voter $i$ with $A_i = \curly{p_C: i \in C}$. Set $\bud = k$ and $s = 1$. We claim that both these instances are equivalent.

    To prove the correctness of our claim, first assume that we are given a \yes instance of \setcov. That is, there exists $F \subseteq S$ such that $|F|=k$ and all elements of $U$ are covered in $F$. The corresponding set of projects $\curly{p_C: C \in F}$ is feasible since $|F| = k$. Every voter has at least one approved project selected and hence the minimum utility of a voter is at least $1$. Thus, it is a \yes instance of our \mmpb problem. Likewise, if we assume that the reduced instance is a \yes instance, we have a set of $k$ projects such that at least one approved project of each voter is selected. This implies that the given instance is \yes instance of \setcov and thus completes the proof. 
\end{proof}

\subsection{Tractable Special Cases}\label{sec: d-re-fpt}
A popular way of addressing computational intractability of a problem is to study the special conditions under which the problem becomes tractable (\Cref{sec: prelims-partialtractability}). We discuss some special cases where \mmpb becomes more tractable than the proved hardness in the previous section.
\subsubsection{Constant Number of Projects}\label{sec: smallprojects}
We look at this parameter, since, in many scenarios, there can be an upper bound on the number of projects that can be funded due to logistic reasons. 
\begin{theorem}[Lenstra and Hendrik \cite{lenstra1983integer}]\label{the: ILPvariables}
  Solving an ILP instance is fixed parameter tractable parameterized by the number of variables.
\end{theorem}
\begin{proposition}\label{prop: projects}
\mmpb is fixed parameter tractable parameterized by the number of projects.
\end{proposition}
\Cref{prop: projects} follows from \Cref{the: ILPvariables} since MPB is formulated as an \ref{word: ILP}.
\subsubsection{Constant Number of Distinct Votes}\label{sec: smalln}
In many real-world scenarios, though the number of voters is large, the number of distinct votes is small. That is, the set of voters can be partitioned into a small number of equivalence classes such that all voters in an equivalence class approve the same set of projects. For example, the 2018 PB elections held in Powazki (Warsaw, Poland) had 3482 voters, out of which there were only 16 distinct approval votes. In fact, it has been found that in many real-life PB datasets \cite{stolicki2020pabulib}, the number of distinct votes is less than $20\%$ of the number of voters. Information on some such datasets is included at \githublink. Assuming that the number of distinct votes is small enough, we first prove that \mmpb is not strongly \NPH, but is only weakly \NPH. To prove this, we employ the \subsum problem as defined below.
\begin{definition}[\subsum]\label{prob: subsum}
    Given an integer $Z$ and a set of integers $X= \{x_1,x_2,\ldots,x_n\}$, the \subsum problem is to determine if there exists a subset $X' \subseteq X$ such that $\sum_{x \in X'}{x} = Z$.
\end{definition}
The above \subsum problem is known to be weakly \NPH \cite{garey1979computers}. We are now ready to present our next result.
\begin{theorem}\label{the: weaknph}
The decision version of \mmpb is weakly \NPH when the number of distinct votes is constant.
\end{theorem}
\begin{proof}
We present a reduction from \subsum problem to our problem. Given an instance $\langle Z,X \rangle$ of \subsum, we construct an \mmpb instance as follows: For each $x_i \in X$, create a project $p_i$ with cost $x_i$. Create a single voter who approves all the projects. Set $\bud = s = Z$. We claim that both these instances are equivalent.

To prove the correctness, first assume that we are given a \yes instance of \subsum. That is, there exists $X' \subseteq X$ such that  $\sum_{x \in X'}{x} = Z$. Corresponding set of projects $\curly{p_i: x_i \in X'}$ is feasible since $\bud = Z$ and the utility of voter is the total cost of this set and is exactly $Z$. Thus, this is a \yes instance of our problem. Likewise, if we assume that the reduced instance is a \yes instance of our problem, we have a set $S$ of projects whose total cost is $Z$. The solution $\{x_i: p_i \in S\}$ makes it a \yes instance of \subsum
\end{proof}

\Cref{the: strongnph} implies that it is impossible to have a pseudo-polynomial time algorithm if the number of distinct votes is large. Since the \mmpb is only weakly \NPH when the number of distinct votes is small, we provide a  pseudo-polynomial time algorithm to solve \mmpb for that case.
\begin{theorem}
\mmpb can be solved in pseudo-polynomial time when the number of distinct votes is constant.
\end{theorem}
\begin{proof}
Let the number of distinct votes be \distinct. Let $A_1,\ldots,A_{\distinct}$ represent these distinct votes. We propose a dynamic programming algorithm. Construct a $\distinct+2$ dimensional binary matrix $Q$ such that $Q(i,j,u_1,\ldots,u_{\distinct})$ takes a value $1$ if and only if there exists a subset $S \subseteq \curly{p_1,\ldots,p_i}$ such that $\cof{S} \leq j$ and $\cof{S \cap A_t} = u_t$ for all $t \in [\distinct]$. Here, $i$ takes values from $1$ to $m$, $j$ takes values from $0$ to \bud, and the remaining entries take values from $0$ to $j$. Let the collection of all such $\distinct+2$ sized tuples be $X$. We fill the first row of the matrix as follows:
\begin{align*}
       Q(1,j,u_1,\ldots,u_{\distinct})  = \begin{cases} 
      1 & \text{if }j \geq \cof{p_1}, u_t \in \curly{0,\cof{p_1}}, \\&u_t = 0 \iff p_1 \notin A_t \;\forall t \in [\distinct] \\
      0 & \text{otherwise} 
   \end{cases}
\end{align*}
Now, we fill the matrix recursively as follows:
\begin{align*}
    Q(i,j,u_1,\ldots,u_{\distinct}) = \max\{&Q(i-1,j,u_1,\ldots,u_{\distinct}),\\&Q(i-1,j-\cof{p_i},v_1,\ldots,v_{\distinct})\}
\end{align*}
where for all $t \in [\distinct]$, $v_t$ is $u_t$ if $p_i \notin A_t$ and $u_t - \cof{p_i}$ otherwise. We know that there are $|X|$ entries in our matrix. The solution of \mmpb is as follows:
\begin{align*}
    \max\limits_{(m,\bud,u_1,\ldots,u_{\distinct}) \in X}{\lb Q(m,\bud,u_1,\ldots,u_{\distinct})\cdot\minl{t \in [\distinct]}{u_t}\rb}
\end{align*}
\paragraph{Running Time.} There are at most $m(\bud+1)^{\distinct+1}$ tuples in $X$ and computing each entry of our matrix takes constant time. The computation of \mmpb solution from the matrix takes $O(\distinct \bud^{\distinct}\log\bud)$ time. The total running time is $O(m\distinct \bud^{\distinct}\log\bud)$, which is pseudo-polynomial if \distinct is constant. Correctness follows from the definition of $Q$.
\end{proof}

\subsubsection{Constant Scalable Limit}\label{sec: gcd}
We introduce a new natural parameter, \textit{scalable limit}, that is reasonably small in several real-world PB elections.
\begin{definition}[\textbf{Scalable Limit}]
Given a PB instance \fullapinstance, the ratio $\frac{\max(\cof{p_1},\ldots,\cof{p_m})}{\tiny{GCD}(\cof{p_1},\ldots,\cof{p_m},\bud)}$ is referred to as {scalable limit}, denoted by \scalel.
\end{definition}
Often in many real-world settings, costs of the projects and the budget are expressed as multiples of some large value. For example, suppose a budget of $10$ billion dollars is to be distributed among a set of projects costing hundreds of millions each. That is, the cost of each project is a multiple of $\$100$M. This PB instance could be scaled down by dividing the costs and budget with $\$100$M to derive a new instance with a budget of $100$. If the cost of the most expensive project originally was $\$900$M, it would now cost $9$ in the scaled down instance. This number $9$ is what we call the scalable limit. In other words, the scalable limit of an instance is the cost of the most expensive project after scaling down the costs and budget to values as low as possible. This parameter takes quite low values in many real-world PB election datasets, e.g., Boston, New York District 8, Seattle District 1 (2019) etc. (see \burl{https://pbstanford.org/}).

From \Cref{the: weaknph}, we know that \mmpb is not polynomial time solvable even if the number of distinct votes is small. We now prove that, if the scalable limit is also small in conjunction with  the number of distinct votes being small, then \mmpb is polynomial time solvable. Before we prove this result, we state a theorem that we crucially use in the proof.

\begin{theorem}[Jansen and Rohwedder \cite{jansen2018integer}]\label{the: ILPconstraints}
  Solving an ILP instance is fixed parameter tractable parameterized by sum of the highest value in the coefficient matrix and the number of constraints.
\end{theorem}

\begin{theorem}\label{the: scalablelimit}
\mmpb is fixed parameter tractable parameterized by sum of the scalable limit and the number of distinct votes.
\end{theorem}
\begin{proof}
Let $\distinct$ be the number of distinct votes and let $A_1,\ldots,A_{\distinct}$ represent these distinct votes. Divide the costs and budget of the instance by $\tiny{GCD}(\cof{p_1},\ldots,\cof{p_m},\bud)$ to obtain a new instance $\instance'$ with costs $c'$ and budget $\bud'$. Clearly, $\instance'$ has the same optimal \mmpb solution as that of \instance.

From \Cref{the: ILPconstraints}, it is known that the problem of solving an ILP is in FPT when parameterized by the number of constraints and the highest value in coefficient matrix \cite{ganian2019solving}. We modify the \ref{word: ILP} for \mmpb by replacing \voters with $[\distinct]$ and $c$ with $c'$. Now, the highest value in the coefficient matrix is $\max(c'(p))$, i.e., \scalel, and the number of constraints is \distinct+1. Clearly, the modified ILP is equivalent to the initial ILP and the theorem follows.
\end{proof}
\subsection{An Efficient Approximation Algorithm}\label{sec: algo}
Another popular way of coping with intractability of a problem is resorting to approximation algorithms (\Cref{sec: prelims-approximation}). In this section, we first introduce a family of PB algorithms called as Ordered-Fill algorithms. We then propose an approximation algorithm, called \lpalgo, from this family, that is based on LP-rounding. We prove that it achieves approximation guarantees for \mmpb for some restricted spaces of instances. We show empirically that it provides exact optimal solutions for \mmpb on real-world PB datasets.

\begin{definition}[\textbf{\textit{Ordered-Fill Algorithms}}] Given a participatory budgeting instance \instance, an ordered-fill algorithm with respect to a complete order $\succ$ over \proj selects the projects in the decreasing order of their ranks in $\succ$ until the next ranked project does not fit within the budget.\footnote{Greedy rules \cite{talmon2019framework} are ordered-fill algorithms where $\succ$ is based on utility from each \textit{affordable} project.}
\end{definition}
\begin{example}\label{eg: orderedfillalgo}
Consider an instance where $\proj = \curly{p_1,p_2,p_3}$, $\bud = 4$, $\cof{p_1} = \cof{p_3} = 2$, and $\cof{p_2} = 3$. An ordered-fill algorithm w.r.t. $p_1 \succ p_2 \succ p_3$ outputs $\curly{p_1}$.\footnote{Note that the outcome is not maximal since $\curly{p_1,p_3} \in \feasible$.}
\end{example}

We consider the LP-relaxation of the \ref{word: ILP} for \mmpb objective by relaxing \Cref{eq: ilp2} to $0 \leq x_p \leq 1$.
Consider the following LP-rounding algorithm: Solve the relaxed LP to get $(q^*,x^*)$. Let $S = \phi$ be the initial outcome. Add the project with the highest value of $\cxsof{p}$ to $S$, followed by the one with the second highest value, and so on, till the next project does not fit. Call this algorithm \lpalgo (Algorithm \ref{algo: orderedrelax}).
\begin{algorithm}
	\DontPrintSemicolon
	\KwIn{A PB instance under dichotomous preferences \fullapinstance}
	\KwOut{A feasible subset of projects $S$}
    $(q^*,x^*) \gets$ Solution of the relaxed ILP;\;
    Let $T$ be any array of project indices sorted according to their $\cof{p}x_{p}^*$;\;
    $t \gets 0$;\;
    $S \gets \emptyset$;\;
	\While{$\cof{S} \leq \bud$}{
        $S \gets S \cup T[p_t]$;\;
	}
	\Return{$S$}\;
	\caption{\lpalgo}
	\label{algo: orderedrelax}
\end{algorithm}
\begin{theorem}\label{lem: algobound}
For any PB instance $\instance$, \lpalgo outputs a set $S \in \feasible$ such that $ALG \geq OPT - \frac{|A_j \setminus S|}{|S \setminus A_j|}\cdot(\bud - OPT)$ where $j = \argmi{i \in \voters}{\cof{\aof{} \cap S}}$, $ALG = \cof{\aof{j} \cap S}$, and $OPT$ is the minimum utility in the optimal solution for \mmpb.
\end{theorem}
\begin{proof}
Let $S$ be the outcome of \lpalgo, $j = \argmi{i \in \voters}{\cof{\aof{} \cap S}}$, and $\eta = \frac{|A_j \setminus S|}{|S \setminus A_j|}$.

Let the solution of the relaxed LP be $(q^*,x^*)$. Let $Y_i = S \cap \aof{i}$ for each voter $i$. Since $S \setminus Y_i \subseteq \proj \setminus \aof{i}$, for each $i$,
\begin{align}
    \nn
    \suml{p \in S \setminus Y_i}{\cxsof{p}} &\leq \suml{p \in \proj \setminus \aof{}}{\cxsof{p}}\\
    \label{eq: algo1}
    \suml{p \in S \setminus Y_i}{\cxsof{p}} &\leq \bud - \suml{p \in \aof{}}{\cxsof{p}}
\end{align}
From the design of this algorithm, every project $p$ not in $S$ has $\cxsof{p}$ not more than that of projects in $S$. \begin{align}
    \nn \forall p \in S \setminus Y_i \quad \cxsof{p} &\geq \maxl{p' \in \aof{} \setminus Y_i}{\cxsof{p'}}\\\nn
    &\geq \frac{\suml{p' \in \aof{} \setminus Y_i}{\cxsof{p'}}}{|\aof{}| - |Y_i|}\\\nn
    \suml{p \in S \setminus Y_i}{\cxsof{p}} &\geq \firstfrac{} \suml{p \in \aof{} \setminus Y_i}{\cxsof{p}}
\end{align}
Since $x_p^* \leq 1 \; \forall p$ and $\uof{}{S} = \suml{p \in Y_i}{\cof{p}}$, $\suml{p \in Y_i}{\cxsof{p}} \leq \uof{}{S}$.
\begin{align}
    \label{eq: algo2}
    \suml{ p \in S \setminus Y_i}{\cxsof{p}} \geq \firstfrac{} \lb \suml{p \in \aof{}}{\cxsof{p}} -  \uof{}{S} \rb
\end{align}
From \Cref{eq: algo1} and \Cref{eq: algo2}:
\begin{align}
    \nn
    \firstfrac{} \lb \suml{p \in \aof{}}{\cxsof{p}} -  \uof{}{S} \rb \leq \bud - \suml{p \in \aof{}}{\cxsof{p}}
\end{align}
\begin{align}
    \nn
    \suml{p \in \aof{}}{\cxsof{p}} &\leq \frac{\bud+\firstfrac{}\;\uof{}{S}}{\firstfrac{}+1}\\
    \label{eq: algo3}
    \suml{p \in \aof{}}{\cxsof{p}} &\leq \frac{\eta\bud+ALG}{1+\eta}
\end{align}
Since the optimal solution also belongs to the feasible region of the relaxed LP, we know that $\opt \leq q^*$. From \Cref{eq: ipl1}, we have $q^* \leq \sum_{p \in \aof{j}}{\cxsof{p}}$. Combining these observations with \Cref{eq: algo3}, we get,
\begin{align}
    \nn
    \opt &\leq \frac{\eta\bud+ALG}{1+\eta}\\\nn
    ALG &\geq \opt - \eta(\bud - \opt)
\end{align}
This proves the result.
\end{proof}

Given an instance \instance, let $\boldsymbol{\lo}$ and $\boldsymbol{\ho}$ denote respectively the minimum and maximum cardinalities of the outputs produced by all ordered-fill algorithms on $\instance$ (i.e., consider algorithms with respect to all the possible orderings over \proj).
\begin{lemma}\label{lem: loho}
For any instance \instance, \lo and \ho are polynomial-time computable.
\end{lemma}
\begin{proof}
Let $D$ and $A$ respectively denote the outputs of ordered-fill algorithms when the projects are arranged in the non-increasing order and the non-decreasing order of their costs. 
Let $O$ denote the output of any other ordered-fill algorithm. We claim that $|D| \leq |O| \leq |A|$.\\
~\\
\textbf{Part 1 :} $|D| \leq |O|$\\
For the sake of contradiction, let us assume that $|O| < |D|$.
\begin{figure}[ht]
  \centering
  \includegraphics[width=0.47\linewidth]{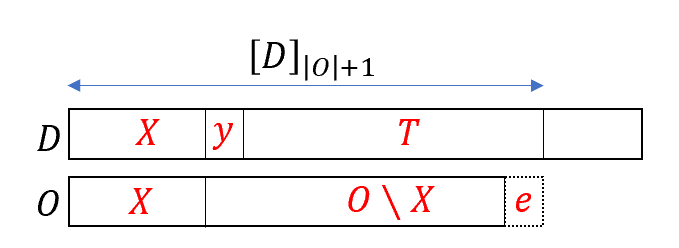}
  \caption{Illustration for $|D| \leq |O|$}
  \label{fig:lleqo}
\end{figure}

Let $X = D \cap O$ be the set of projects selected in both $O$ and $D$. If $D \setminus X = \phi$, our assumption is wrong and the first part of the proof is complete. 

Consider the case where $D \setminus X \neq \phi$. Let $y$ be the project with the highest cost in $\proj \setminus X$. That is,
\begin{align}
    \label{wysx}
    \cof{y} \geq \cof{p} \quad \forall p \in \proj \setminus X
\end{align}
By the definition of $D$,
\begin{align}
    \nn
    y \in D \setminus X
\end{align}
Since $|O| < |D|$, there are at least $|O|+1$ projects in $D$. Hence, there are at least $|O|-|X|$ projects in $D \setminus (X \cup \{y\})$. Let the set of costliest  $|O|-|X|$ projects in $D \setminus (X \cup \{y\})$ be denoted by $T$. By the definition of $D$, $T$ is the set of the costliest $|O| - |X|$ projects in $\proj \setminus (X \cup \{y\})$. Since $y \in D$ and $y \notin X$, $y \notin O$. Thus, $O \setminus X$ is also a set of exactly $|O| - |X|$ projects in $\proj \setminus (X \cup \{y\})$. Therefore,
\begin{align}
    \label{tox}
    \cof{T} \geq \cof{O \setminus X}
\end{align}
As $\{X \cup \{y\} \cup T\} \subseteq D$,
\begin{align}
    \label{xytb}
    \cof{X} + \cof{y} + \cof{T} \leq \bud
\end{align}
Let $e$ be the project where the algorithm stopped adding projects to $O$.
\begin{align}
    \nn
    \cof{O} + \cof{e} &> \bud\\
    \label{xoxeb}
    \cof{X} + \cof{O \setminus X} + \cof{e} &> \bud
\end{align}
Since $e \in \proj \setminus O$, $e \in \proj \setminus X$, from \Cref{wysx},
\begin{align}
    \label{wewy}
    \cof{y} \geq \cof{e}
\end{align}
From \Cref{xytb} and \Cref{xoxeb},
\begin{align}
    \nn
    \cof{X} + \cof{y} + \cof{T} &< \cof{X} + \cof{O \setminus X} + \cof{e}\\\nn
    \cof{T} + \cof{y} &< \cof{O \setminus X} + \cof{e}
\end{align}
From \Cref{wewy},
\begin{align}
    \nn
    \cof{O \setminus X} > \cof{T}
\end{align}
But this contradicts \Cref{tox}, hence contradicting our assumption that $|O| < |D|$. Thus, $|D| \leq |O|$.\\
~\\
\textbf{Part 2 :} $|O| \leq |A|$\\
For the sake of contradiction, let us assume that $|O| > |A|$.
\begin{figure}[ht]
  \centering
  \includegraphics[width=0.47\linewidth]{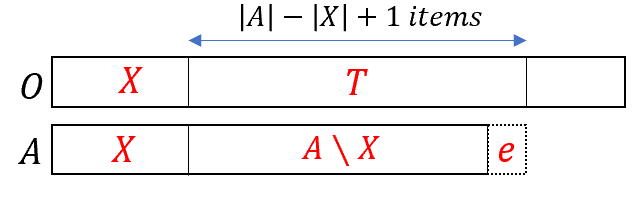}
  \caption{Illustration for $|O| \leq |A|$}
  \label{fig:oleqh}
\end{figure}

Let $X = A \cap O$ be the set of projects selected in both $O$ and $A$. Let $e$ be the first project where the algorithm stopped adding projects to $A$.
\begin{align}
    \nn
    \cof{A} + \cof{e} &> \bud\\
    \label{xaxeb}
    \cof{X} + \cof{A \setminus X} + \cof{e} &> \bud
\end{align}
Since $|O| > |A|$, $|O \setminus X| > |A|-|X|$. Consider a set, $T$, of any $|A|-|X|+1$ projects in $O \setminus X$. From the definition of $A$, $(A \setminus X) \cup \{e\}$ is the set of the least expensive $|A|-|X|+1$ projects in $\proj \setminus X$.
\begin{align}
    \label{tax}
    \cof{A \setminus X} + \cof{e} \leq \cof{T}
\end{align}
Since $T \cup X \subseteq O$,
\begin{align}
    \label{xtb}
    \cof{X} + \cof{T} \leq \bud
\end{align}
From \Cref{xaxeb} and \Cref{xtb},
\begin{align}
    \nn
    \cof{T} &< \cof{A \setminus X} + \cof{e}
\end{align}
This contradicts \Cref{tax}, hence contradicting our assumption that $|O| > |A|$. Thus, $|O| \leq |A|$.

Hence, $|D| \leq |O| \leq |A|$. By definition, $\lo = |D|$ and $\ho = |A|$. $D$ and $A$ can be computed in polynomial time by sorting the projects according to their costs.
\end{proof}

Let $\boldsymbol{\la}$ and $\boldsymbol{\ha}$ respectively denote the lowest and highest cardinalities of the approval votes, i.e., $\la = \minl{i \in \voters}{|\aof{}|}$ and $\ha = \maxl{i \in \voters}{|\aof{}|}$.

\Cref{lem: algobound} can be used to establish theoretical approximation guarantees for some restricted spaces of instances, e.g., let us look at a family of instances that satisfy a property that we call \textit{High Cardinality Budget Property} (\hcbp) which intuitively is a cardinal extension of all votes being knapsack votes. That is, it requires that the budget is high enough to fund more projects than the number of projects in any single approval vote. In other words, an instance satisfies HCBP if and only if $\lo > \ha$. Note that this property is natural in scenarios where budget is high enough to accommodate all projects approved by one voter, e.g., when the budget with the federal government is high enough to fund all proposed projects of one county or when each member in the audience approves less number of talks than the total capacity of the workshop, etc. 
\begin{corollary}\label{cor: hcbp}
    For the instances satisfying \hcbp, \lpalgo ensures that $ALG \geq OPT - \ha(\bud - OPT)$, where $OPT$ denotes the optimal value of \mmpb objective and $ALG$ denotes its value from the output of \lpalgo\footnote{Further, if we define the disutility of a voter $i$ from a set $S$ to be $\bud - \uof{}{S}$ and minimize the maximum disutility, our algorithm achieves $\lb 2 -\frac{1}{\ho} \rb$ approximation for \hcbp instances. Refer Appendix \ref{app: minimax}.}.
\end{corollary}
\begin{proof}
    Consider an instance that satisfies HCBP. From \Cref{lem: algobound}, we know that $ALG \geq OPT - \frac{|A_j \setminus S|}{|S \setminus A_j|}\cdot(\bud - OPT)$ where $j = \argmi{i \in \voters}{\cof{\aof{} \cap S}}$. By the definition of \ha, $|A_j \setminus S| \leq |A_j| \leq \ha$. Likewise, since \lpalgo is an ordered-fill algorithm, $|S| \geq \lo$ and $A_j \leq \ha$. Therefore, $|S-A_j| \geq \lo-\ha$. Since the instance satisfies HCBP, $\lo-\ha \geq 1$. Thus,  $|S-A_j| \geq 1$. Combining this with $|A_j \setminus S| \leq \ha$ proves that $\frac{|A_j \setminus S|}{|S \setminus A_j|} \leq \ha$.
\end{proof}

As another example, we note that the algorithm yields an optimal solution for the instances which guarantee $\aof{j} \subseteq S$. One such family of instances will be the one that includes instances  satisfying the following condition for all $S \in \feasible$: if there exists some $p \in \proj$ such that $\cof{S}+\cof{p}>\bud$, then all approved projects of the worse-off voter from $S$ are included in $S$. Note that this is a sufficient but not a necessary condition for the algorithm to yield an  optimal solution. Although the approximation ratio given by \Cref{lem: algobound} could be trivial for some instances theoretically, our empirical analysis compensates for this.
\begin{table}[h]
\begin{tabular}{c|c|c|c}
\multirow{2}{*}{\textbf{\begin{tabular}[c]{@{}c@{}}PB Election \\ Dataset Location\end{tabular}}} & \multirow{2}{*}{\textbf{Year}} & \multirow{2}{*}{\textbf{\begin{tabular}[c]{@{}c@{}}Minimum Utility in\\ the Optimal Outcome\end{tabular}}} & \multirow{2}{*}{\textbf{\begin{tabular}[c]{@{}c@{}}Minimum Utility in\\ the Outcome of \lpalgo\end{tabular}}} \\
                                                                                                  &                                &                                                                                                            &                                                                                                              \\
                                                                                                  \hline
                                                                                                
Zoliborz                                                                                          & 2020                           & 3531                                                                                                       & 3531                                                                                                         \\
Ursus                                                                                             & 2020                           & 4059                                                                                                       & 4059                                                                                                         \\
Zerzen                                                                                            & 2019                           & 8000                                                                                                       & 8000                                                                                                         \\
Elsnerow                                                                                          & 2019                           & 5000                                                                                                       & 5000                                                                                                         \\
Stare Miasto                                                                                      & 2019                           & 21800                                                                                                      & 21800                                                                                                        \\
Srodmiescie-Polnocne                                                                              & 2019                           & 30000                                                                                                      & 30000                                                                                                        \\
Kamionek                                                                                          & 2018                           & 10000                                                                                                      & 10000                                                                                                        \\
Slodowiec                                                                                         & 2018                           & 21912                                                                                                      & 21912                                                                                                        \\
Boernerowo                                                                                        & 2017                           & 18680                                                                                                      & 18680                                                                                                        \\
Piaski                                                                                            & 2017                           & 13400                                                                                                      & 13400                                                                                                       
\end{tabular}
\caption{Empirical results for arbitrary election datasets from various locations in Poland. The second column denotes the utility of the worst-off voter from the \mmpb rule outcome, while the third column displays the corresponding value from the \lpalgo outcome.}
\label{tab: empirical}
\end{table}
\subsubsection{Empirical Analysis}
Though the theoretical guarantee of our algorithm \lpalgo is limited in terms of approximation ratios and the spaces of instances it covers, empirically it is found to exhibit remarkable performance to give {\em exact optimal solutions\/} on all the PB datasets of real-world PB elections available online \cite{stolicki2020pabulib}. The datasets and corresponding results are included in detail at \githublink. Some of these election datasets are selected arbitrarily and the corresponding results are displayed in \Cref{tab: empirical}. This indicates that, although theoretically the worst cases do not allow the approximation ratio to get better, such worst cases are seldom encountered in real-life, thus explaining the excellent performance of our algorithm.

\subsection{Bound on the Performance of Exhaustive and Strategy-Proof PB Algorithms}\label{sec: sp}
In many cases, real-world voters tend to be strategic and self-interested. Designing rules that are resistant to manipulations by strategic voters is thus a common objective in social choice literature. This property is formally referred to as strategy-proofness. Another property that is of interest in particularly PB settings is exhaustiveness, which ensures that the available budget is utilized without any wastage. In other words, the goal is to design PB rules that allocate the budget without under-utilization. It would be useful to know how the requirements of strategy-proofness and exhaustiveness will degrade the performance of PB algorithms with respect to \mmpb objective.

We establish an upper bound on the approximation ratios of exhaustive strategy-proof PB algorithms by invoking an argument used in the mechanism design literature \cite{procaccia2009approximate,caragiannis2010approximation}. We construct a specific profile whose outcome can be derived using strategy-proofness and establish an upper bound on its approximation ratio.
\begin{definition}[\textbf{Strategy-Proofness}]\label{def: sp}
An algorithm \GG is said to be strategy-proof, if and only if, for any instance $\instance$,
\begin{align}
    \nn
    \forall i \quad \forall S \subseteq \proj \quad \cof{\aof{} \cap \GG(\aof{},\approf_{-i})} \geq \cof{\aof{} \cap \GG(S,\approf_{-i})}
\end{align}
where $\approf_{-i}$ denotes the dichotomous preference profile of all voters except $i$ and $\GG(\approf)$ denotes the outcome of $\GG$ at profile \approf.
\end{definition}
That is, if all voters except $i$ continue to approve the same sets of projects, $i$ should not obtain a better utility by approving any set other than \aof{}. An algorithm is exhaustive if it is impossible to fund an unselected project with the remaining budget.
\begin{definition}[\textbf{Exhaustiveness}]
An algorithm \GG, is said to be exhaustive if and only if for every instance \instance and every $p \notin \GG(\instance)$, it holds that $\cof{p}+ \cof{\GG(\instance)} > \bud$.
\end{definition}
\begin{theorem}\label{the: exsp}
No exhaustive strategy-proof algorithm can achieve an approximation guarantee better than $\frac{2}{3}$ for \mmpb even for instances with only knapsack votes and unit cost projects.
\end{theorem}
\begin{proof}
Consider an exhaustive strategy-proof algorithm \GG. Consider an instance $\instance_1$ with a budget \bud, $2\bud$ projects each costing $1$, two voters, and a dichotomous preference profile $\approf_1 = \curly{\aof{1},\aof{2}}$ such that $\aof{1} \cap \aof{2} = \emptyset$ and $\cof{\aof{1}} = \cof{\aof{2}} = \bud$. Let $\GG(\instance_1) = S$. Without loss of generality, assume that $\cof{\aof{1} \cap S} \leq \cof{\aof{2} \cap S}$. Modify the dichotomous preference profile of $\instance_1$ to get $\instance_2$ as follows: $\approf_2  = \curly{\aof{1},S}$. For $\instance_2$, \GG should output $S$. Else, voter $2$ would have the incentive to misreport $\aof{j}$ instead of the true vote $S$ and get the preferred outcome $S$. Hence, $\GG(\instance_2) = S$. Therefore, the minimum utility of the algorithm for $\instance_2$, $ALG(\instance_2)$, is $\cof{\aof{1} \cap S}$. Let $\opt(\instance_2)$ be the optimal minimum utility for $\instance_2$. That is,
\begin{align}
    \label{eq: ic1}
    \opt(\instance_2) \geq \min{(\cof{\aof{1} \cap X}, \cof{S \cap X})} \quad \forall X \in \feasible
\end{align}
Since \GG is exhaustive, $\cof{S} = \bud$. Consider a set $X \subseteq \proj$ such that $(\aof{1} \cap S) \subseteq X$, $|X \cap (\aof{1} \setminus S)| = \frac{\bud - \cof{\aof{1} \cap S}}{2}$, and $|X \cap (S \setminus \aof{1})| = \frac{\bud - \cof{\aof{1} \cap S}}{2}$. Clearly, $X \in \feasible$.
\begin{align}
    \nn
    \cof{\aof{1} \cap X} = \cof{S \cap X} &= \cof{\aof{1} \cap S} + \frac{\bud - \cof{\aof{1} \cap S}}{2}\\
    \label{eq: ic2}
    &=  \frac{\bud + \cof{\aof{1} \cap S}}{2}
\end{align}
Since $\cof{\aof{1} \cap S} \leq \cof{\aof{2} \cap S}$ and $\cof{S} = \bud$, $\cof{\aof{1} \cap S} \leq \frac{\bud}{2}$. Substituting $\bud \geq 2\cof{\aof{1} \cap S}$ in \Cref{eq: ic2}, we get,
\begin{align}
    \tag{From \Cref{eq: ic1}}
    \opt(\instance_2) &\geq \frac{3\cof{\aof{1} \cap S}}{2}\\\nn
    \frac{ALG(\instance_2)}{\opt(\instance_2)} &\leq \frac{2}{3}
\end{align}
Hence, the approximation guarantee of the algorithm is at most $\frac{2}{3}$. Clearly in both the instances $\instance_1$ and $\instance_2$, all the votes are knapsack votes and the theorem follows.
\end{proof}

\section{Axiomatic Results}\label{sec: d-re-axioms}
An axiomatic study of any voting rule provides valuable insights into the characteristics of the voting rule. We now undertake an axiomatic study of the \mmpb rule by exploring several axiomatic properties available in the literature.
Given an instance \instance and a PB rule \pbrule, we say a project $p$ wins if there exists $S \in \ruleof{\pbrule}{}$ such that $p \in S$. Let \winners{\pbrule}{} represent the set of all projects which win, i.e., $\winners{\pbrule}{} = \curly{p \in \proj: \exists S \in \ruleof{\pbrule}{} p \in S}$. Throughout this section, we represent the \mmpb rule by $\boldsymbol{\mmpbr}$.

First, we examine axioms in the PB literature for rules that output a single feasible subset of projects \cite{talmon2019framework}. We extend these axioms to irresolute rules (like \mmpb) that output multiple feasible subsets. We start by defining all the axioms.

The first axiom requires that, if a winning project is replaced by a set of projects whose total cost is the same, then at least one project of this set should continue to win. In other words, take any winning project $p \in \winners{\pbrule}{}$. Split $p$ into a set $P'$ of smaller projects such that $\cof{p} = \cof{P'}$. In every preference $\aof{}$ such that $p \in \aof{}$, replace $p$ with all the projects in $P'$. Let the new resultant instance be $\instance'$. Splitting monotonicity requires that at least one project from $P'$ must belong to $\winners{\pbrule}{'}$.
\begin{definition}[Splitting Monotonicity \cite{talmon2019framework}]\label{def: splittingmpb}
A PB rule \pbrule is said to satisfy splitting monotonicity iff, for any instance \instance, $\winners{\pbrule}{'} \cap P' \neq \emptyset$ whenever $\instance'$ is obtained from \instance by splitting a project $p \in \winners{\pbrule}{}$ into a set of projects $P'$ with $\cof{P'} = \cof{p}$ and changing every $\aof{}$ having $p$ to $(\aof{} \setminus \curly{p}) \cup P'$.
\end{definition}
The next axiom requires that, if we replace a set of winning projects that are approved by exactly same set of voters by a single project having the same total cost, then the single project wins. In other words, take a set $P'$ such that (i) it consists of only winning projects, i.e., $P' \subseteq \winners{\pbrule}{}$ and (ii) every voter either approves all the projects from $P'$ or approves none of them, i.e., $\aof{} \cap P' \in \{P',\emptyset\}$ for every $i$. Merge all the projects in $P'$ to get a new project $p$ such that $\cof{p} = \cof{P'}$. In every $\aof{}$ such that $P' \subseteq \aof{}$, replace $P'$ with the project $p$. Let the new resultant instance be $\instance'$. Merging monotonicity requires that $p \in \winners{\pbrule}{'}$.
\begin{definition}[Merging Monotonicity \cite{talmon2019framework}]\label{def: mergingmpb}
A PB rule \pbrule is said to satisfy merging monotonicity iff, for any instance \instance, $p \in \winners{\pbrule}{'}$ whenever $\instance'$ is obtained from \instance as follows: for any $S \in \ruleof{\pbrule}{}$ and $P' \subseteq \winners{\pbrule}{}$ such that $\aof{} \cap P' \in \curly{P',\emptyset}$ for every voter \ii, merge the projects in $P'$ into a single project $p$ with cost $\cof{P'}$ and change every $\aof{}$ with $P'$ to $(\aof{} \setminus P') \cup \curly{p}$.
\end{definition}
The next axiom requires that no winning project should be dropped if it becomes less expensive.
\begin{definition}[Discount Monotonicity \cite{talmon2019framework}]\label{def: discountmpb}
A PB rule \pbrule is said to satisfy discount monotonicity iff, for any instance \instance, $p \in \winners{\pbrule}{'}$ whenever $\instance'$ is obtained from \instance by reducing the cost of $p \in \winners{\pbrule}{}$ to $\cof{p}-1$.
\end{definition}
The following axiom requires that no winning project should be dropped if the budget increases. Note that we insist no project to cost exactly $\bud+1$ since that leads to a trivial result of selecting that project, thereby not capturing the essence of axiom.
\begin{definition}[Limit Monotonicity \cite{talmon2019framework}]\label{def: limitmpb}
A PB rule \pbrule is said to satisfy limit monotonicity iff, for any instance \instance such that no project in \proj costs $\bud+1$, $\winners{\pbrule}{} \subseteq \winners{\pbrule}{'}$ whenever $\instance'$ is obtained from \instance by increasing the budget to $\bud+1$.
\end{definition}
The next axiom requires that all sets in \ruleof{\pbrule}{} are maximal.
\begin{definition}[Strong Exhaustiveness \cite{aziz2021participatory}]\label{def: strongempb}
A PB rule \pbrule is said to satisfy strong exhaustiveness iff, for any instance \instance, $$\forall S \in \ruleof{\pbrule}{}\;\forall p \in \proj\setminus S \quad \cof{S} + \cof{p} > \bud$$
\end{definition}
Next, we introduce a new axiom, weak exhaustiveness, which requires that any winning non-maximal set can be made maximal without compromising on the win. In other words, suppose $S$ is a winning set of projects ($S \in \ruleof{\pbrule}{}$). If addition of another project $p$ to $S$ does not make $S$ infeasible, then adding $p$ must not also affect its winning status. That is, $S \cup \{p\}$ must also be a part of $\ruleof{\pbrule}{}$.
\begin{definition}[Weak Exhaustiveness]\label{def: weakempb}
A PB rule \pbrule is said to satisfy weak exhaustiveness iff, for any instance \instance, $$\forall S \in \ruleof{\pbrule}{},\;\forall p \in \proj\setminus S,\; \cof{S} + \cof{p} \leq \bud \implies S\cup\curly{p} \in \ruleof{\pbrule}{}$$
\end{definition}
\begin{theorem}\label{the: pbaxioms_satisfy}
The \mmpb rule satisfies: (a) splitting monotonicity (b) merging monotonicity (c) weak exhaustiveness.
\end{theorem}
\begin{proof}
\begin{enumerate}[label=(\alph*)]
\item {\textbf{Splitting Monotonicity:}}\\
Let $p \in \winners{\mmpbr}{}$ be split into a set of projects $P'$ as above to produce $\instance'$. For the sake of contradiction, let us assume that $\winners{\mmpbr}{'} \cap P' = \emptyset$. Since $p \in \winners{\mmpbr}{}$, there exists $S \in \ruleof{\mmpbr}{}$ such that $p \in S$. Consider a set of projects $K = (S\setminus \{p\}) \cup P'$. Since $\cof{\aof{} \cap K}$ and $\cof{\aof{} \cap S}$ are same for every voter \ii, the minimum utility from $K$ is equal to that from $S$. From our assumption, for any set $T \in \ruleof{\mmpbr}{'}$, $T \cap P' = \emptyset$ and the minimum utility from $T$ is strictly greater than that from $K$ and $S$. This contradicts the fact that $S \in \ruleof{\mmpbr}{}$.
\item {\textbf{Merging Monotonicity:}}\\
Let $S \in \ruleof{\pbrule}{}$ and $P' \subseteq S$ satisfy the condition specified. That is, each voter either approves entire $P'$ or approves no project in $P'$. Let $P'$ be merged into a single project $p$ and the new instance thus produced be $\instance'$. For the sake of contradiction, let us assume that $p \notin \winners{\pbrule}{'}$. Consider the set $K = (S\setminus P') \cup \{p\}$. Since $\cof{\aof{} \cap K}$ and $\cof{\aof{} \cap S}$ are same for every voter \ii, the minimum utility from $K$ is equal to that from $S$. From our assumption, for any set $T \in \ruleof{\mmpbr}{'}$, $p \notin T$ and the minimum utility from $T$ is strictly greater than that from $K$ and $S$. This contradicts $S \in \ruleof{\mmpbr}{}$.
\item {\textbf{Weak Exhaustiveness:}}\\
For the sake of contradiction, let us assume that $\exists S \in \ruleof{\mmpbr}{}$ and $p \in \proj\setminus S$ such that $\cof{S} + \cof{p} \leq \bud$. Consider the feasible set $K = S \cup \{p\}$. Consider any arbitrary voter \ii. If $p \in \aof{}$, $\uof{}{K} = \uof{}{S} + \cof{p}$. Else if $p \notin \aof{}$, $\uof{}{K} = \uof{}{S}$. So, the minimum utility of any voter from $K$ is at least that from $S$. Since $S \in \ruleof{\mmpbr}{}$, $K \in \ruleof{\mmpbr}{}$.
\end{enumerate}
This completes the proof.
\end{proof}
\begin{theorem}\label{the: pbaxioms_nsatisfy}
The \mmpb rule does not satisfy: (a)  discount monotonicity (b) limit monotonicity (c) strong exhaustiveness.
\end{theorem}
\begin{proof}
\begin{enumerate}[label=(\alph*)]
\item {\textbf{Discount Monotonicity:}}\\
Consider an instance $\instance$ with $\proj = \curly{p_1,p_2,p_3,p_4}$ each costing $4$, a budget $12$, and three voters with $\aof{1} = \curly{p_1,p_2}$, $\aof{2} = \curly{p_3}$, and $\aof{3} = \curly{p_4}$. It can be seen that $p_2 \in \winners{\mmpbr}{}$ since $\ruleof{\mmpbr}{}=\{\{p_1,p_3,p_4\}, \{p_2,p_3,p_4\}\}$. Now consider the case where $\cof{p_2}$ is reduced to $3$ to get new instance $\instance'$. Then, $\ruleof{\mmpbr}{'} = \{\{p_1,p_3,p_4\}\}$ and hence $p_2 \notin \winners{\mmpbr}{'}$.
\item {\textbf{Limit Monotonicity:}}\\
Consider an instance $\instance$ with $\proj = \curly{p_1,p_2,p_3,p_4,p_5,p_6}$ costing $\curly{3,1,3,3,3,6}$ respectively, a budget $12$, and four voters with $\aof{1} = \curly{p_1,p_2}$, $\aof{2} = \curly{p_3,p_4}$, $\aof{3} = \curly{p_5}$, and $\aof{4} = \curly{p_6}$. Clearly, $\{p_1,p_3,p_5\} \in \ruleof{\mmpbr}{}$. Therefore, $p_1 \in \winners{\mmpbr}{}$. If \bud is increased to $13$ to get $\instance'$, $\ruleof{\mmpbr}{'} = \{\{p_2,p_4,p_5,p_6\}, \{p_2,p_3,p_5,p_6\}\}$ and $p_1 \notin \winners{\mmpbr}{'}$.
\item {\textbf{Strong Exhaustiveness:}}\\
The above example also illustrates the counter example for strong exhaustiveness. Clearly, $S = \{p_1,p_3,p_5\} \in \ruleof{\mmpbr}{}$ but $\cof{p_2}+\cof{S} < 12$.
\end{enumerate}
This completes the proof.
\end{proof}

We now examine two important axioms from multi-winner voting literature. The first axiom is an analogue of unanimity, which requires that a project approved by all voters needs to win.
\begin{definition}[Narrow-Top Criterion \cite{faliszewski2017multiwinner}]\label{def: narrowtopmpb}
A PB rule \pbrule is said to satisfy narrow-top criterion iff, for any instance \instance, $p \in \winners{\pbrule}{}$ whenever $p \in \aof{}$ for every $i\in \voters$.
\end{definition}
\begin{proposition} \label{prop: narrowtop}
The \mmpb rule  does not satisfy narrow-top criterion.
\end{proposition}
\begin{proof}
Our example uses the fact that the utility from an unanimously approved project could be low and selecting it could make the other projects unaffordable. Consider an instance $\instance$ with $\proj = \curly{p_1,p_2,p_3}$ costing $\curly{1,3,3}$ respectively, a budget $6$, and two voters with $\aof{1} = \curly{p_1,p_2}$ and $\aof{2} = \curly{p_1,p_3}$. Since $\ruleof{\mmpbr}{}=\{\{p_2,p_3\}\}$, we have $p_1 \notin \winners{\mmpbr}{}$ though $p_1 \in \aof{}$ for all \ii.
\end{proof}

The next axiom tries to capture diversity in the voting rules \cite{aziz2018egalitarian,brandt2016handbook}.
\begin{definition}[Clone Independence \cite{brandt2016handbook}]\label{def: clonempb}
A PB rule \pbrule is said to satisfy clone independence iff, for any instance \instance, $\ruleof{\pbrule}{}$ does not change when a group of voters, all having exactly the same approval vote $\aof{v}$, is replaced by a single voter with approval vote $\aof{v}$.
\end{definition}
\begin{proposition}\label{prop: clone}
The \mmpb rule satisfies clone independence.
\end{proposition}
The above proposition follows from the fact that redundant votes do not affect the value of maxmin objective. Clearly, this axiom represents diversity, but it is rather narrow. Faliszewski et al. \cite{faliszewski2017multiwinner} identified that the fairness notion of diversity does not have a clear axiomatic representation in the literature. We address this gap by introducing a new axiom, called \emph{Maximal Coverage}, to capture diversity in PB as well as in multi-winner voting settings.

Let us define {\em covered voters\/}  as the voters with at least one approved project funded and a {\em redundant project\/} as a project whose removal from an outcome does not change the set of covered voters. To achieve diversity, we need to cover as many voters as possible. Our new axiom ensures that a redundant project is funded only when no more voters can be covered by doing otherwise. For example, while allocating time for plenary talks, the organizers might want to explore as many novel ideas as possible, without eliminating any key  area. They might allocate a time slot to a proposed plenary talk on a not-so-popular topic (hoping to prop up a new area), by dropping 1 out of 4 proposed talks in a popular area. This {\em covers\/} or reaches out to the audience from the new area and broadens the reach of conference. 

\begin{definition}[Maximal Coverage]\label{def: maximal}
A PB rule \pbrule is said to satisfy maximal coverage iff, for any instance \instance, $S \in \ruleof{\pbrule}{}$, $p \in S$ such that $\{j_1:p \in \aof{j_1}\} \subseteq \{j_2:(S\setminus\{p\})\cap \aof{j_2} \neq \emptyset\}$, and $i \in \voters$,
\begin{equation}
    \label{eq: maximal1}
    \winners{\pbrule}{} \cap \aof{} = \emptyset \implies \cof{a} > \bud -\cof{S\setminus\{p\}} \; \forall a \in \aof{}
\end{equation}
\end{definition}
We consider an example to understand this better. Suppose a budget of \$2.25B is to be allocated to projects in two counties, X and Y, with populations of 10000 and 6000, resp. County X proposes projects $\curly{X_1, X_2,X_3}$ costing $\curly{\$0.5B,\$1B,\$1B}$, and County Y proposes projects $\curly{Y_1, Y_2,Y_3}$ costing $\curly{\$0.7B,\$0.7B,\$0.8B}$ resp. Assume each voter approves all and only the projects of her county. A rule that maximizes utilitarian welfare selects the set $\{\{X_2, X_3\}\}$. Let $p$ be $X_2$. 
If $i = 2$ and $a = Y_1$, then the first part of \Cref{eq: maximal1} satisfies but $\cof{Y_1} \leq \bud - \cof{X_3}$. Hence it does not satisfy maximal coverage. Now, let us apply the \mmpb rule to this example. That leads to $\ruleof{\mmpbr}{} = \{\{X_1, X_2, Y_1\}, \{X_1, X_2, Y_2\}, \{X_1, X_3, Y_1\}, \{X_1, X_3, Y_2\}\}$ and $\winners{\mmpbr}{} = \{X_1, X_2, X_3, Y_1, Y_2\}$. It is notable that \Cref{eq: maximal1} holds.
\begin{proposition}\label{prop: maximal}
The \mmpb rule satisfies maximal coverage.
\end{proposition}
\begin{proof}
Let $S,p,i,a$ be as stated in \Cref{def: maximal}. Assume $\exists S$  such that the first part of \Cref{eq: maximal1} holds. Since $\winners{\mmpbr}{} \cap \aof{} = \emptyset$, $\uof{}{S} = 0$. Hence the minimum utility from any set in $\ruleof{\mmpbr}{}$ is $0$. But, since $\winners{\mmpbr}{} \cap \aof{} = \emptyset$ and $a \in \aof{}$, $\{a\} \notin \ruleof{\mmpbr}{}$. This is possible only if $\cof{a} > \bud$ and, thus, the second part of \Cref{eq: maximal1} holds.
\end{proof}
\section{Conclusion and Discussion}\label{sec: d-re-conclusion}
This chapter introduces the Maxmin Participatory Budgeting (\mmpb) rule to optimize egalitarian welfare when the costs of projects are restricted to admit only one value. We conduct a comprehensive computational and axiomatic study of this rule.

On the computational front, we demonstrate that \mmpb is strongly NP-hard. We look at two methods of coping up with this intractability: (i) identifying the tractable special cases and (ii) designing approximation algorithms. As a part of studying the first method, we show that \mmpb is only weakly NP-hard when the number of distinct votes is small, and propose a pseudo-polynomial time algorithm for solving it. We also introduce a novel parameter called \textit{scalable limit} and prove that when this parameter is also constant, \mmpb can be solved in polynomial time. As a part of second method, we devise an LP-rounding algorithm \lpalgo that provides distinct additive approximation guarantees for several families of instances, and empirically demonstrate that our algorithm yields optimal outcomes when applied to real-world PB datasets. Furthermore, we establish an upper bound on the achievable approximation ratio for \mmpb within the family of exhaustive and strategy-proof PB algorithms.

On the axiomatic front, we conduct a thorough analysis of \mmpb and introduce a new fairness axiom called \textit{maximal coverage} that aims to capture diversity in PB and multi-winner voting. We examine the compatibility of \mmpb with all other desirable properties and prove that \mmpb results in fair outcomes.

Going forward, an axiomatic characterization of \mmpb is worth exploring. We studied two egalitarian objectives: maxmin and minimax (Appendix \ref{app: minimax}). Both these objectives suffer a common disadvantage: consider a case where the utility of the worst-off voter is $0$ and it is impossible for any feasible subset of projects to give a non-zero utility to all the voters. In such a case, every feasible subset of voters is included in the outcome. To overcome this limitation of maxmin and minmax objectives, studying the egalitarian objective leximin is a promising direction. In the leximin objective, if two feasible subsets of projects give same utility to the worst-off voter, we pick the subset in which the utility of the second worst-off voter is higher, and so on until there is no tie between the sets of projects. This ensures that the ties between several sets of projects is dealt with most efficiency. From the axiomatic point of view, establishing connections between our new axiom maximal coverage and other existing axioms in the literature is an intriguing question to be pursued.

Throughout this chapter, we assumed that the cost of a project is restricted to only one value. In the next chapter, we study the model where the cost of each project is flexible and can take one of multiple possible values.
\blankpage
\chapter{Flexible Costs: Welfare Maximization when Projects have Multiple Degrees of Sophistication}\label{chap: d-fl}

\begin{quote}
\textit{In the realm of participatory budgeting, prior research on dichotomous preferences has largely focused on two scenarios: projects where costs are restricted to admit only one possible value or where costs are totally flexible permitting any amount to be allocated to each project. This chapter introduces a novel perspective by assuming that the costs are partially flexible. That is, we allow the projects to have a range of permissible costs, representing varying degrees of project sophistication. Each voter is allowed to express her preference by specifying upper and lower bounds on the cost she believes each project deserves.}

\textit{Within this framework, the outcome of a PB rule becomes a selection of projects accompanied by their corresponding costs. For such a model, we propose four utility notions and study the corresponding rules maximizing utilitarian welfare. We showcase that the existing positive findings from the model featuring projects with restricted costs can be extended to our enriched framework, where projects can have multiple permissible costs. Moreover, we undertake an analysis of fixed parameter tractability of computationally hard rules in this model.}

\textit{Further, we put forth a set of significant and intuitive axioms, aiming to capture essential aspects of multiple degrees of sophistication of projects. Through our analysis, we evaluate the extent to which the proposed PB rules satisfy these axioms, shedding light on the inherent properties and limitations of each rule.}
\end{quote}

\section{Motivation}\label{sec: d-fl-intro}
The cost of a project is the amount that needs to be allocated to the project in the event of it being selected for funding. As introduced in the \Cref{sec: intro-costs}, there are three possibilities with respect to the costs of the projects:
\begin{itemize}
    \item \textbf{\textit{Restricted costs:}} The cost of each project is restricted to a single value. For example, a dam construction project, if funded, may need a cost of exactly $\$10B$ to be completed.
    \item \textbf{\textit{Flexible costs:}} The amount allocated to each project is permitted to assume multiple values. This is further split into two possible scenarios:
    \begin{itemize}
        \item \textit{Partially flexible:} A project can be implemented upto different degrees of sophistication and every degree corresponds to a different cost. The amount allocated to the project by the PB rule must belong to this set of permissible costs.
        \item \textit{Totally flexible:} There is no restriction on the amount allocated to each project and any amount can be allocated to each project.
    \end{itemize}
\end{itemize}
The existing work on participatory budgeting under dichotomous preferences assume that the costs of projects are totally restricted \cite{talmon2019framework,fluschnik2019fair,goel2019knapsack,baumeister2020irresolute,rey2020designing,jain2020participatoryg,freeman2021truthful,benade2021preference,laruelle2021voting,aziz2018proportionally,fluschnik2019fair,laruelle2021voting,fairstein2021proportional,los2022proportional,brill2023proportionality} or that they are totally flexible \cite{bogomolnaia2005collective,duddy2015fair,aziz2019fair}. This chapter studies the case of costs being \emph{partially flexible}. This is especially applicable in scenarios where there are multiple ways of executing a project. For example, a building can be built with wood or stone, depending on the amount allocated to it. Similarly, a health service project can build a primary healthcare center, a clinic, or a multi-speciality hospital. We call each such option of a project as a possible \emph{degree of sophistication} of the project. Each of these options corresponds to a different cost. \Cref{fig: motivation_degrees} is another illustration of such a model.
\begin{figure}[!ht]
    \centering
    \includegraphics[width=0.6\linewidth]{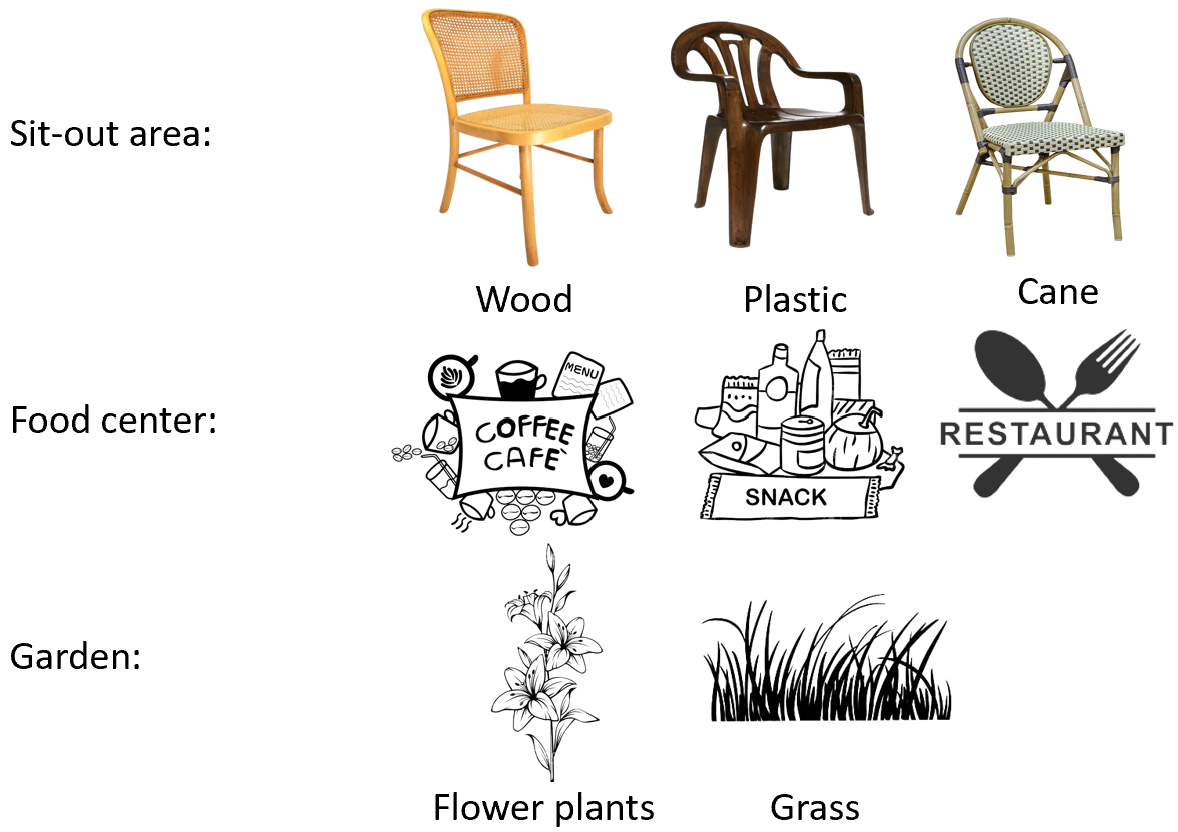}
    \repeatcaption{fig: motivation_degrees}{A toy example of an MNC renovation to demonstrate the applicability of PB model under partially flexible costs. Each project could be executed in multiple ways: a sit-out area could be built with multiple materials; a food center could be a coffee cafe, or a fast-food point, or a restaurant; and so on. Each possible option corresponds to a different cost
}
\end{figure}
The goal in PB under partially flexible costs is to select a subset of projects to be funded and also determine the amount to be allocated to each selected project. 
\section{Prior Relevant Work}\label{sec: d-fl-lit}
From the perspective of flexibility of costs of the projects, several works in the literature study the case where costs are restricted to one value \cite{talmon2019framework,fluschnik2019fair,goel2019knapsack,baumeister2020irresolute,rey2020designing,jain2020participatoryg,freeman2021truthful,benade2021preference,laruelle2021voting,aziz2018proportionally,fluschnik2019fair,laruelle2021voting,fairstein2021proportional,los2022proportional,brill2023proportionality} or are totally flexible \cite{bogomolnaia2005collective,duddy2015fair,aziz2019fair}. Alternately, this chapter studies the case where costs are partially flexible. Notably, a work by Patel et al. \cite{patel2021group} on fairness for groups of projects can be viewed as our model: each degree of sophistication can be considered a separate project, all the degrees of same project are grouped, and a constraint that at most one project from each group can be funded is imposed. As a result, a particular result (\Cref{the: crule-fptas}) in this chapter can be derived from the results of Patel et al. \cite{patel2021group}, while all the other computational and axiomatic results are novel contributions.

From the perspective of preferences of the voters, preference elicitation methods typically studied in any voting framework include dichotomous preferences, ordinal preferences, and cardinal preferences. These methods also continue to be the most studied preference elicitation approaches in PB \cite{shapiro2017participatory,aziz2018proportionally,talmon2019framework,rey2020designing,jain2020participatory,benade2021preference,aziz2021proportionally,fairstein2022welfare,sreedurga2022maxmin}. However, PB is a setting with several attributes like costs associated with the projects. The preferences and utilities of the voters are thus much more complex. This propels the need to devise preference models specific to PB, as pointed out by Aziz and Shah \cite{aziz2021participatory}. One step in this direction is the introduction of a special case of dichotomous preferences, called knapsack votes, by Goel et al. \cite{goel2019knapsack}, where each voter reports her most favorite budget division. This idea however is criticised for its assumption that any project not in this division yields no utility to the voter.

We take another step in the direction of preference modeling for PB. We introduce the approach of \emph{ranged dichotomous preferences}, which strictly generalizes dichotomous preferences. Each voter reports a lower bound and an upper bound on the cost that she thinks each project deserves. All the bounds are initially set to $0$ by default. Voting proceeds in two steps. In the first step, voter starts by approving the projects she likes. For these approved projects, only the upper bounds automatically change to the highest permissible cost. In the second step, voter may \emph{optionally} change bounds for some of these approved projects, if she wishes to have a say on the amount they deserve. Note that this is cognitively not much more demanding than the standard dichotomous preferences since we do not force voters to report the bounds. 

Notably, all computational results in this chapter can also be extended to more general utility functions with minor tweaks. Nevertheless, we present the results for ranged dichotomous preferences due to their cognitive simplicity and natural relevance in the model.

It is worth mentioning that a work by Goel et al. \cite{goel2019knapsack} views divisible PB as a model of indivisible PB where every unit of money in the cost corresponds to one possible degree of sophistication, and proposes a greedy algorithm. Their work however imposes knapsack constraint on each vote and also assumes all lower bounds to be zero. Their model hence forms a very restricted special case of ours. The idea of having multiple permissible costs for each project was first proposed as a future direction in a survey by Aziz and Shah \cite{aziz2021participatory}, in which the authors called each of these costs a \emph{degree of completion} and assumed an existence of a directed graph between the degrees. That is, a degree is assumed to be successor of a less expensive degree and more preferred than it by default. In our model, degrees of projects are independent of each other (a degree could correspond to constructing with wood, while other one corresponds to constructing with cement).

\section{Contributions of the Chapter}\label{sec: d-fl-contri}
This chapter systematically studies the PB model where each project has a set of permissible costs. Such a study has been conducted by Talmon and Faliszewski \cite{talmon2019framework} for approval-based PB, which is a special case of our model with every project having only one permissible cost. We generalize the PB rules and all the positive results in \cite{talmon2019framework} to our model. Namely, we propose polynomial-time, pseudo-polynomial time, and PTAS algorithms. Followed by this, we present results on parameter tractability using the parameter scalable limit introduced in \Cref{chap: d-re} and another novel parameter, variance coefficient, we introduce in the further sections. It needs to be highlighted that, as a part of our study, we also introduce and investigate novel utility notions that are specific to our model and cannot be extended from any existing notions for restricted costs model.

We further propose budgeting axioms for our model and examine their satisfiability with respect to our PB rules. All the axiomatic results are summarized in \Cref{tab: results}. Finally, we discuss the impact of our computational and axiomatic results and make key observations.
\paragraph{\textbf{Organization of the chapter.}} We start by introducing the formal model in \Cref{sec: d-fl-prelims}. We define different utility notions (some of them are extended from the existing ones in the literature on restricted costs model; some are introduced by us explicitly for partially flexible costs model) and the corresponding PB rules in \Cref{sec: utilities}. In \Cref{sec: d-fl-comp}, we analyze the computational complexity of our PB rules and suggest ways to cope up with intractabilities. In \Cref{sec: d-fl-axioms}, we define budgeting axioms for our model and examine their satisfiability.

\section{Notations}\label{sec: d-fl-prelims}
Let us recall the necessary notations from \Cref{sec: prelims-notations} and introduce a few more. A budget \bud is available for allocation. There are $n$ voters $\voters = \curly{1,\ldots,n}$ and $m$ projects in a set \proj. Each project $p_j \in \proj$ has\mdof{}possible degrees of sophistication captured by the set $\dsetof{} = \curly{\pdof{}{0},\pdof{}{1}, \ldots,\pdof{}{\mdof{}}}$. The cost of each degree \pdof{}{} is indicated by \dcof{}{}. We assume that \dcof{}{0} is zero for all $p_j \in \proj$ and it corresponds to not funding the project $p_j$.

Each voter $i \in \voters$ reports for every project $p$, a lower bound\lowerb{}{}and an upper bound\upperb{}{}such that $\lowerb{}{},\upperb{}{} \in \curly{\dcof{}{0},\ldots,\dcof{}{\mdof{}}}$ and $\lowerb{}{} \leq \upperb{}{}$. Let \lbounds and \ubounds respectively denote the collection of all the lower bounds and upper bounds reported by all the voters. 

Let \degreeproj denote the set of all the possible degrees of all projects, or in other words, $\degreeproj = \bigcup_{p_j \in \proj}{\dsetof{}}$. We denote the cost of a set $S \subseteq \degreeproj$, $\sum_{\pdof{}{} \in S}{\dcof{}{}}$, by\csetof{}. Given a subset $S \subseteq \degreeproj$, we use\chodof{}{}to denote the chosen degree(s) of project $p$ in $S$. In other words, $\chodof{}{} = S \cap \dsetof{}$. We use the shorthand notation\chocof{}{}to denote\csetof{\chodof{}{}}. We say a subset $S \subseteq \degreeproj$ is \emph{\textbf{valid}} if $\csetof{} \leq \bud$ and $\cardof{\chodof{}{}} = 1$ for every $p_j \in \proj$. Let \valid denote the collection of all the valid subsets.

A PB instance with multiple degrees of sophistication is denoted by $\instance = \fulldeginstance$. Recall that in the PB model under flexible costs, PB rule outputs an allocation of costs to each project (\Cref{def: pbrule-flexible}). That is, it outputs the corresponding allocated cost $x_j \in \domainof{p_j} \cup \curly{0}$ for each project $p_j \in \proj$, where $\domainof{p_j}$ is the set of permissible costs of $p_j$. In our model with partially flexible costs, $\domainof{p_j} = \curly{\dcof{}{1},\ldots,\dcof{}{\mdof{}}}$ (\Cref{sec: prelims-costs}). Thus, a PB rule is defined as follows:
\begin{definition}[\textbf{PB Rule under Partially Flexible Costs}]
    Given a PB instance \instance with multiple degrees of sophistication, a PB rule \pbrule outputs a valid subset $S \in \valid$, thereby indicating the amount allocated to each project.
\end{definition}

\section{The Utility Measure and PB Rules}\label{sec: utilities}
The PB rules we study are rules maximizing utilitarian welfare. This implies that, given a function $u$ that measures utility a voter derives from a subset of projects, the corresponding PB rule outputs a valid set of projects that maximizes the sum of utilities of all the voters, i.e., it outputs some $S \in \valid$ such that $\sum_{i \in \voters}{\utof{}{}}$ is maximized.
We say a subset $S$ of projects is selected under a PB rule \pbrule if it maximizes the total utility of the voters. 

Now, we only need to define different utility functions. Given the lower and upper bounds reported by the voters, we define the utility of a voter from a valid set $S \in \valid$ in four ways. For each of these utility functions, we also give a shorthand notation for the utilitarian welfare maximizing PB rule associated with it.
\begin{enumerate}
    \item \textbf{Cardinal utility ($\boldsymbol{\nrule}$ rule)}: Each voter \ii derives a utility of $1$ from a project $p$ if the cost of the chosen degree falls within the bounds specified by the voter. Thus, $\utof{}{} = \cardof{\{p_j \in \proj : \lowerb{}{} \leq \chocof{}{} \leq \upperb{}{}, \chocof{}{} \neq 0\}}$.
    \item \textbf{Cost utility ($\boldsymbol{\crule}$ rule)}: A voter \ii derives a utility of\chocof{}{}from a project $p$ if its value falls within the bounds specified by the voter. \begin{center}
        $\utof{}{} = \suml{p_j: \lowerb{}{} \leq \chocof{}{} \leq \upperb{}{}}{\chocof{}{}}$.
    \end{center}
    \item \textbf{Cost capped utility ($\boldsymbol{\ccaprule}$ rule)}: Each voter \ii derives a utility of\chocof{}{}from a project $p$ if the value falls within the bounds reported by her, a utility of\upperb{}{}if $\chocof{}{} > \upperb{}{}$, and a utility of $0$ if $\chocof{}{} < \lowerb{}{}$. That is, with slight abuse of notation, $\utof{}{} = \sum_{p_j \in \proj}{\utof{}{j}}$, where\utof{}{j}is defined as:
    \begin{align*}
        \utof{}{j} =
        \begin{cases} 
        0 &  \chocof{}{} < \lowerb{}{}\\
        \chocof{}{} & \lowerb{}{} \leq \chocof{}{} \leq \upperb{}{}\\
        \upperb{}{} & \text{otherwise}
        \end{cases}
    \end{align*}
    \item \textbf{Distance disutility ($\boldsymbol{\drule}$ rule)}: From every project $p$, each voter derives a disutility of $0$ if the value falls within the bounds reported by her, a disutility of $\chocof{}{}-\upperb{}{}$ if $\chocof{}{} > \upperb{}{}$, and a disutility of $\lowerb{}{}-\chocof{}{}$ if $\chocof{}{} < \lowerb{}{}$. That is, with slight abuse of notation, $\dutof{}{} = \sum_{p_j \in \proj}{\dutof{}{j}}$, where\dutof{}{j}is defined as follows:
    \begin{align*}
        \dutof{}{j} =
        \begin{cases} 
        \lowerb{}{}-\chocof{}{} &  \chocof{}{} < \lowerb{}{}\\
        0 & \lowerb{}{} \leq \chocof{}{} \leq \upperb{}{}\\
        \chocof{}{}-\upperb{}{} & \text{otherwise}
        \end{cases}
    \end{align*}
    The corresponding PB rule $\boldsymbol{\drule}$ minimizes the total disutility.
\end{enumerate}
The first two utility notions are natural extensions of the existing notions \cite{talmon2019framework}, which were also discussed in \Cref{sec: prelims-dutilities} and \Cref{chap: d-re}. The former reflects that the voter is happy as long as the cost allocated to the project is acceptable to her, whereas the latter reflects that as the project gets more money, the voter gets happier if the cost is acceptable to her.

The second utility notion clearly assumes that if the project gets more money than what the voter thinks it deserves, voter derives zero utility. However, this need not always be the case in all applications. For example, suppose a voter will be happy to have an entertainment park in the neighborhood but feels that the park deserves any amount between 1000 and 5000 units. Now, if the park project is allocated 7000 units, the voter could still be happy that there is a park in the neighborhood but derive no more utility than 5000 (since it is the maximum value she thinks the park deserves). The third utility notion tackles such scenarios by modifying the second utility notion to cap the utility at\upperb{}{}instead of dropping it to $0$.

The first three utility notions assume that allocating any cost outside the range reported by the voter yields the same utility to her. However, in many situations, closer the cost of project is to the acceptable range of the voter, higher is the satisfaction voter derives from it. For example, suppose a voter feels that a certain project is worth at least 1000 units. An outcome that allocates 900 units to it is likely to be more preferred by the voter over something that allocates 50 units to it. The fourth rule handles such situations by considering the distance between the cost allocated and the closest acceptable cost as the \emph{disutility}. Farther the cost is from the acceptable range, higher is the amount misspent.
\section{Computational Results}\label{sec: d-fl-comp}
Here, we analyze the computational complexity of our PB rules. We strengthen the existing positive results in the literature \cite{talmon2019framework} and also present a few new results on fixed parameter tractability. All the exact algorithms we present are based on dynamic programming. The approximation schemes we present depend on both dynamic programming as well as a clever rounding scheme.

\subsection{The Rule ${\nrule}$}\label{sec: nrule}
Recall that \nrule outputs a valid set that maximizes the sum of the utilities of voters, where the utility of a voter is defined as $\utof{}{} = \cardof{\{p_j \in \proj : \lowerb{}{} \leq \chocof{}{} \leq \upperb{}{}, \chocof{}{} \neq 0\}}$. Talmon and Faliszewski \cite{talmon2019framework} prove that for the case with restricted cost for each project, a subset maximizing the utility can be computed in polynomial time. We strengthen this result and prove that even when there are multiple permissible costs for each project, a subset maximizing the total utility can be computed in polynomial time. Note that the dynamic programming used in the next result will be referenced multiple times in the chapter.
\begin{theorem}\label{the: nrule-p}
    For any instance \instance, a subset $S \in \valid$ that is selected under \nrule can be computed in polynomial time.
\end{theorem}
\begin{proof}
    We present an algorithm that uses dynamic programming. Let $\dscoreof{}{}$ denote the number of voters $i$ such that $\lowerb{}{} \leq \dcof{}{} \leq \upperb{}{}$. Construct a dynamic programming table such that $A(x,y)$ corresponds to the cost of cheapest valid subset of $\bigcup_{j = 1}^x{\dsetof{}}$ such that the total score of projects in the set is exactly $y$.

    Let $F(y) \subseteq \dsetof{1}$ such that $\dscoreof{1}{} = y$ for every $\pdof{1}{} \in \dsetof{1}$. We compute the first row as: $A(1,y) = \min(\dcof{}{}: \pdof{}{} \in F(y))$. All the remaining rows are computed recursively as follows: $A(x,y) = \min\Big(A\big(x-1,y\big),\;\min_{t = 1}^{\mdof{x}}{\big(A\big(x-1,y-\dscoreof{x}{t}\big)+\dcof{x}{t}\big)}\Big)$.
    In addition to computing this table, we store the sets corresponding to each entry in $A$ in a separate table $B$ as follows: if $A(x,y)$ achieves the minimum value at $A(x-1,y)$, we copy $B(x-1,y)$ to $B(x,y)$ and append $\pdof{}{0}$. If $A(x,y)$ achieves the minimum value at $A(x-1,y-\dscoreof{x}{t})+\dcof{x}{t}$ for some $t \in [1,\mdof{x}]$, we set $B(x,y) = B(x-1,y) \cup \curly{\pdof{x}{t}}$. Finally, we output the set at $B(x,y)$ such that $A(x,y)\leq \bud$ and $y$ is maximized.
    
     \paragraph{Correctness.} Recurrence ensures that for any entry in $B$, we select at most one project from \dsetof{}for every $p_j$. While selecting the outcome from 
    $B$, we ensured that its cost is within the budget. Hence, the output of the algorithm will be a valid subset. Optimality follows from the definition of $A$.
    
    \paragraph{Running time.} Each row in $A$ corresponds to making a decision about one project from \proj, and hence the number of rows is $m$. Each column corresponds to a possible total score. Since we select only one project from each\dsetof{}and maximum score of any degree of a project is $n$ by definition, maximum total score is $mn$. Thus, we have $mn$ columns. Computing each entry in row $x$ takes $O(\mdof{x})$ time. Thus, running time is $\boldsymbol{O(m^2n\tstar)}$, where $\tstar = \max_{p_j \in \proj}{\mdof{}}$.
\end{proof}

\subsection{The Rule ${\crule}$}\label{sec: crule}
Recall that \crule outputs a valid set that maximizes the sum of the utilities of voters, where the utility of a voter is: $$\utof{}{} = \suml{p_j: \lowerb{}{} \leq \chocof{}{} \leq \upperb{}{}}{\chocof{}{}}.$$ Talmon and Faliszewski \cite{talmon2019framework} prove that for the case where cost of each project is restricted to one permissible value, it is \NPH to determine if there exists a feasible subset that guarantees a total utility of at least a given value. Since their model can be modeled as a special case of ours, the hardness directly follows, as displayed in the next proof.
\begin{proposition}\label{the: crule-nph}
    For an instance \instance and a value $s$, it is \NPH to check if \crule outputs a set with at least a total utility $s$.
\end{proposition}
\begin{proof}
We reduce from a known \NPH rule for PB under dichotomous preferences and restricted costs \cite{talmon2019framework}. Given such a PB instance (where each project $p_j$ has some cost $c(p_j)$ and every voter reports a subset $A_i$ of projects she likes), we construct an instance for \crule as follows: For every project $p_j \in \proj$, create exactly two degrees such that $\dcof{}{0} = 0$ and $\dcof{}{1} = c(\pa)$. For each voter $i$, set $\lowerb{}{} = 0$ for every project $p$ and set $\upperb{}{} = \dcof{}{1}$ if and only if $p \in A_i$. Clearly, both the instances are equivalent and we skip the proof of correctness.
\end{proof}

To cope up with the intractability, Talmon and Faliszewski \cite{talmon2019framework} prove that if cost of each project is restricted to one permissible value, the problem has a pseudo-polynomial time algorithm and FPTAS. We extend both these results to our model with multiple permissible costs.
\begin{proposition}\label{the: crule-pseudo}
    For any instance \instance, a subset $S \in \valid$ selected under \crule can be computed in pseudo-polynomial time.
\end{proposition}
\begin{proof}
    We multiply every $\operatorname{s}(P_j^t)$ in the proof of \Cref{the: nrule-p} with \dcof{}{}. DP tables are also constructed as explained in \Cref{the: nrule-p}. The maximum score achievable by each \pdof{}{} is $n\dcof{}{}$. The total score is upper bounded by $n\csetof{S}$ since at most a single project from each \degreeset{} is chosen into each $B(x,y)$. This is bounded by $n\bud$ since $S \in \valid$. Thus, the table size is $O(mn\bud)$ and the time is $\boldsymbol{O(mn\tstar\bud)}$, where $\tstar = \max_{p_j \in \proj}{\mdof{}}$.
\end{proof}
\begin{theorem}\label{the: crule-fptas}
    There is an FPTAS to compute the outcome of \crule for any instance \instance.
\end{theorem}
\begin{proof}
    The idea is inspired from one of the existing FPTAS algorithm techniques of the knapsack problem \cite{ibarra1975fast,lawler1977fast,patel2021group}. We first round the scores of all the projects and then use the DP table explained in \Cref{the: nrule-p} on the modified instance with rounded scores. 

    Given an instance \instance, let $M$ be the maximum score of a degree of project, i.e., $M = \max_{p_j \in \proj, t \in [1,\mdof{}]}{\dscoreof{}{}}$. We can easily ensure that $M \leq OPT$ by eliminating all the project degrees that cannot be part of any set with cost within \bud. Take any $\epsilon \in (0,1)$. Now, create new scores of all the projects as follows: $\msof{} = \floor*{\frac{s(P^t_j)m}{\epsilon M}}$. 
    We construct the DP tables similar to those in \Cref{the: nrule-p}, but while considering new scores. 
    By the definition of\msof{}, for every $p_j \in \proj$ and $t \in [1,\mdof{}]$:
    \begin{align}
        \label{eq: nsltos}
        \msof{} &\leq \frac{\dscoreof{}{}m}{\epsilon M}\\
        \label{eq: nsgtos}
        \dscoreof{}{} &\leq \frac{\epsilon M (\msof{}+1)}{m}
    \end{align}
     Let $S$ be the outcome of this DP algorithm. Suppose $O$ denotes the optimal solution for \instance under \crule. Please note that the DP ensures that $S$ is a valid set. Since $S$ is the optimal solution for the modified scores,
    \begin{align}
        \label{eq: sgto}
        \suml{j=1}^m{\msof{\chodof{}{O}}} \leq \suml{j=1}^m{\msof{\chodof{}{S}}}
    \end{align}
    Now, we prove that the approximation factor of $(1-\epsilon)$ holds.
    \begin{align}
        \tag{From \Cref{eq: nsgtos}}
        OPT = \suml{j=1}^m{\dscoreof{}{\chodof{}{O}}} &\leq \suml{j=1}^m{\frac{\epsilon M (\msof{\chodof{}{O}}+1)}{m}}\\
        \tag{From \Cref{eq: sgto}}
        &\leq \suml{j=1}^m{\frac{\epsilon M \msof{\chodof{}{S}}}{m}}+\epsilon M\\
        \tag{From \Cref{eq: nsltos}}
        &\leq \frac{\epsilon M}{m}\suml{j=1}^m{\frac{\dscoreof{}{\chodof{}{S}}m}{\epsilon M}}+\epsilon M\\
        \tag{$\because M \leq OPT$}
        \suml{j=1}^m{\dscoreof{}{\chodof{}{S}}} &\geq (1-\epsilon)OPT
    \end{align}
    \paragraph{Running Time.} The table has $m$ rows. The modified score of each degree of project is upper bounded by $\frac{s(P^t_j)m}{\epsilon M}$, which is further bounded by $\frac{m}{\epsilon}$ due to the definition of $M$. Since the output is a valid set, it has at most $m$ projects and the maximum possible total score is upper bounded by $\frac{m^2}{\epsilon}$, which is the number of columns. Computing each entry in row $x$ takes $O(\mdof{x})$ time. Thus, the running time is $\boldsymbol{O(\frac{m^3\tstar}{\epsilon})}$, where $\tstar = \max_{p_j \in \proj}{\mdof{}}$.
\end{proof}
\subsubsection{Fixed Parameter Tractability}\label{sec: crule-fpt}
We introduced a new parameter called scalable limit in \Cref{sec: gcd} for PB under restricted costs and observed that this value is often small in PB elections in real life (e.g., datasets at https://pbstanford.org/ where budget and all the costs are multiples of some very high value). Motivated by this, we prove that \crule is in FPT with respect to scalable limit. Before that, we define the scalable limit in flexible costs model as follows:
\begin{definition}[\textbf{Scalable Limit (}$\boldsymbol{\slimit}$\textbf{)}]\label{def: slimit}
    For any instance \instance, we refer to $\frac{\max_{\pdof{}{} \in \degreeproj}{\dcof{}{}}}{\tiny{GCD}(\dcof{1}{1},\ldots,\dcof{m}{\mdof{m}},\bud)}$ as the \textbf{scalable limit} (\slimit) of the instance.
\end{definition}
Intuitively, we scale down all the costs and budget as much as possible, while ensuring that all these values continue to be integers. The scalable limit then refers to the cost of the costliest degree in $\degreeproj$ in this scaled down instance.
\begin{proposition}\label{the: crule-sl}
    For any instance \instance, computing a subset $S \in \valid$ that is selected under \crule is fixed parameter tractable parameterized by the scalable limit \slimit.
\end{proposition}
\begin{proof}
    We again prove this by constructing a DP table. Let $k =\frac{1}{\tiny{GCD}(\dcof{1}{1},\ldots,\dcof{m}{\mdof{m}},\bud)}$. First, we get a new instance $\instance'$ by scaling down all the costs and budget as follows: $\dcof{}{'t} =k\dcof{}{}$ and $\bud' = k\bud$. Then, we may construct the DP table similar to that in the proof of \Cref{the: nrule-p}. 

    \paragraph{Running Time.} The table has $m$ rows. The modified cost of each degree of project is upper bounded by \slimit and hence the maximum score possible from each degree of a project is upperbounded by $n\slimit$. Since the output is a valid set, there can be at most $m$ projects in it and hence the maximum possible total score is upper bounded by $mn\slimit$, which is the number of columns. Computing each entry in row $x$ takes $O(\mdof{x})$ time and the running time is $\boldsymbol{O(m^2n\slimit \tstar)}$.
\end{proof}
\begin{proposition}
    For any instance \instance, computing a subset $S \in \valid$ that is selected under \crule is fixed parameter tractable parameterized by the number of projects $m$.
\end{proposition}
The above proposition follows from the fact the total number of valid subsets is exponential in $m$ (upper bounded by $(\tstar)^m$, where $\tstar = \max_{p_j \in \proj}{\mdof{}}$) and computing the total utility of any set under \crule can be done in polynomial time.

\subsection{The Rule ${\ccaprule}$}\label{sec: ccaprule}
This rule is similar to the previous rule \crule, with the only difference being the way a voter is assumed to view the projects allocated higher cost than her approved limit. Recall that the utility is defined as $\utof{}{} = \sum_{p_j \in \proj}{\utof{}{j}}$, where\utof{}{j}is as follows:
\begin{align*}
        \utof{}{j} =
        \begin{cases} 
        0 &  \chocof{}{} < \lowerb{}{}\\
        \chocof{}{} & \lowerb{}{} \leq \chocof{}{} \leq \upperb{}{}\\
        \upperb{}{} & \text{otherwise}
        \end{cases}
\end{align*}

\begin{proposition}\label{the: ccaprule}
    Given an instance \instance, the following statements hold:
    \begin{enumerate}[(a)]
        \item For any value $s$, it is \NPH to determine if \ccaprule outputs a set that has a total utility of at least $s$.
        \item A subset $S \in \valid$ that is selected under \ccaprule can be computed in pseudo-polynomial time.
        \item There is an FPTAS for \ccaprule.
        \item Computing a subset $S \in \valid$ that is selected under \ccaprule is fixed parameter tractable parameterized by the scalable 
 limit \slimit.
        \item Computing a subset $S \in \valid$ that is selected under \ccaprule is fixed parameter tractable parameterized by the number of projects $m$.
    \end{enumerate}
\end{proposition}
\begin{proof}
    All the above statements follow from the proofs in \Cref{sec: crule} with some minor changes. Statements (a) and (e) follow without any changes in the proof (\ccaprule is the same as \crule when every project has only permissible cost).

    For (b), we modify the definition of\dscoreof{}{}as $\dscoreof{}{} = \suml{i \in \voters}{\utof{}{j}}$, where $\utof{}{j}$ is as defined in the \ccaprule rule. The maximum possible total score of $S$ is again $n\csetof{S}$ (since $\utof{}{j} \leq \chocof{}{S}$ always) and the rest follows. For (c), we replace $\dscoreof{}{}$ in the proof of \Cref{the: crule-fptas} with the above value. Rest of the proof remains the same. The proof for (d) follows from \Cref{the: crule-sl} since $\utof{}{j}$ is upper-bounded by \dcof{}{} for every $i \in \voters$ and $p_j \in \proj$.
\end{proof}

\subsection{The Rule ${\drule}$}\label{sec: drule}
Recall that \drule outputs the feasible set that minimizes the sum of disutilities of all the voters, where the disutility of voter $i$ is defined as $\dutof{}{} = \sum_{p_j \in \proj}{\dutof{}{j}}$, such that\dutof{}{j}is as follows: 
\begin{align*}
        \dutof{}{j} =
        \begin{cases} 
        \lowerb{}{}-\chocof{}{} &  \chocof{}{} < \lowerb{}{}\\
        0 & \lowerb{}{} \leq \chocof{}{} \leq \upperb{}{}\\
        \chocof{}{}-\upperb{}{} & \text{otherwise}
        \end{cases}
    \end{align*}
It is worth bearing in mind that the results from \Cref{sec: crule} and \Cref{sec: ccaprule} cannot be transferred since we have a minimization problem and also because we cannot upper bound the disutilities by \bud. We start the section by proving that this rule too is \NPH, using a slightly different but simple reduction from the same problem (used in \Cref{the: crule-nph}) solved by Talmon and Faliszewski \cite{talmon2019framework}. Then, we give a parameterized approximation algorithm, more specifically a parameterized FPTAS, which guarantees the approximation ratio when the parameter is small \cite{feldmann2020survey}.
\begin{proposition}\label{the: drule-nph}
    For an instance \instance and a value $s$, it is \NPH to determine if \drule outputs a set that has a total disutility of at most $s$.
\end{proposition}
\begin{proof}
For a PB instance under restricted costs (each project $p$ has some cost $c(p)$ and every voter reports a subset $A_i$ of projects that she likes) and a positive value $x$, we construct an instance for \crule as follows: For every project $p_j \in \proj$, create exactly two degrees such that $\dcof{}{0} = 0$ and $\dcof{}{1} = c(\pa)$. For each voter $i$, set $\upperb{}{} = c(\pa)$ for every project $p$ and set $\lowerb{}{} = \dcof{}{1}$ if and only if $p \in A_i$. Let $Z = \sum_{i \in \voters}{\csetof{A_i}}$. Set $s = Z-x$.

\paragraph{Correctness.} To prove the forward direction, suppose the PB instance under restricted costs is a \yes instance. That implies, there exists a set $S$ of projects such that $\sum_{i \in \voters}{\csetof{S \cap A_i}} \geq x$. Now calculate the disutility of $S$ under \drule. Then, $\dutof{}{j} = c(\pa)$ for every $j \in A_i \setminus S$ and $0$ otherwise. Therefore, $\dutof{}{} = -c(A_i)-c(A_i \cap S)$. Then, total utility from $S$ = $Z-\sum_{i \in \voters}{c(A_i \cap S)}$. Since the former is at least $x$, this total disutility is at most $Z-x$. To prove the backward direction, assume that the resultant instance of our problem is a \yes instance. We can construct a set $S$ of projects such that a project $p_j$ is included in $S$ if and only if $p^1_j$ is in the solution. Set $S$ is a solution for PB under restricted costs, making it a \yes instance.
\end{proof}
\subsubsection{Parameterized Approximation Algorithm (FPTAS)}\label{sec: drule-pfptas}
The very high disutilities that are hard to be bound motivate us to propose a \emph{parameterized FPTAS}, which has been of great interest recently \cite{feldmann2020survey}. We consider a parameter that we call \emph{\variance}, \vpara, which, on being given an instance \instance, intuitively shows how divergent the disutilities of different degrees of projects in the instance are. It is captured by measuring the highest disutility a single project can have, relative to sum of the least possible disutilities from all the projects. We explain this formally below.

Given an instance \instance, for each degree of a project $\pdof{}{} \in \degreeproj$, we use $\qof{}{}$ to denote the disutility project \pdof{}{} will contribute to any set that contains it. We call this \emph{disutility contribution} of \pdof{}{}. Suppose for example, $\pdof{}{} \in S$. That implies, $\chocof{}{} = \dcof{}{}$ ($\because S$ is valid). Therefore, $$\qof{}{} = \sum_{i: \dcof{}{}<\lowerb{}{}}{(\lowerb{}{}-\dcof{}{})}+\sum_{i: \upperb{}{}<\dcof{}{}}{(\dcof{}{}-\upperb{}{})}.$$ It can be observed that, $\sum_{i \in \voters}\dutof{}{} = \sum_{\pdof{}{} \in S}{\qof{}{}}$.

Let $\qmax = \max_{\pdof{}{} \in \degreeproj}{\qof{}{}}$. That is, $\qmax$ is the maximum of disutility contributions of all the degrees of all the projects. We use $\qsum$ to denote the sum of minimum disutility contributions from each \dsetof{}, i.e., $\qsum = \sum_{p_j \in \proj}{\min_{t \in \dsetof{}}{\qof{}{}}}$. Our parameter \vpara is the ratio $\qmax/\qsum$. Intuitively, this parameter means that we want the disutility contribution from a single project not to be much higher than the sum of the least disutility contributions from all the $m$ projects.
\begin{theorem}
    For any instance \instance and $\epsilon\in [0,1]$, a subset $S \in \valid$ such that $\sum_{i \in \voters}{\dutof{}{}} \leq (1+\epsilon)OPT$ can be computed in $O(\frac{m^3\vpara\tstar}{\epsilon})$ time, where $OPT$ is the optimal possible total disutility under \drule, \vpara is the \variance, and $\tstar = \max_{p_j \in \proj}{\mdof{}}$. 
\end{theorem}
\begin{proof}
    The idea is similar to that in \Cref{the: crule-fptas}. We define new disutilities for each project by rounding their disutility contributions. We then construct a DP table using which we select a set. For each project \pdof{}{}, we define\msof{}as follows: $$\msof{} = \floor*{{\qof{}{}m}/{\epsilon\qsum}}.$$
    
    We construct the DP tables $A$ and $B$ similar to those in \Cref{the: nrule-p}, with a slight change the each column now represents the total disutility, and\dscoreof{}{}is replaced by\msof{}defined in the above paragraph. 
    We output the set at $B(x,y)$ such that $A(x,y) \leq \bud$ and $y$ is \emph{minimized}. Let $S$ be the resultant outcome.
    Suppose $O$ denotes the optimal solution for \instance under \drule. Please note that the DP ensures that $S$ is a valid set. Since $S$ is the optimal solution with respect to new disutilities (i.e., w.r.t. $\alpha$'s),
    \begin{align}
        \label{eq: sgto'}
        \suml{j=1}^m{\msof{\chodof{}{S}}} \leq \suml{j=1}^m{\msof{\chodof{}{O}}}
    \end{align}
    By the definition of $\alpha$, we have
    \begin{align}
        \label{eq: ndltod}
        \qof{}{} &\geq \frac{\epsilon\qsum\msof{}{}}{m}\\
        \label{eq: ndgtod}
        \msof{}{} &\geq \frac{\qof{}{}m}{\epsilon\qsum}-1
    \end{align}
    Now, we prove that the approximation factor of $(1+\epsilon)$ holds.
    \begin{align}
        \tag{From \Cref{eq: ndltod}}
        OPT = \suml{j=1}^m{\qof{}{\chodof{}{O}}} &\geq \suml{j=1}^m{\frac{\epsilon \qsum \msof{\chodof{}{O}}}{m}}\\
        \tag{From \Cref{eq: sgto'}}
        &\geq \suml{j=1}^m{\frac{\epsilon \qsum \msof{\chodof{}{S}}}{m}}\\
        \tag{From \Cref{eq: ndgtod}}
        &\geq \left(\frac{\epsilon \qsum}{m}\suml{j=1}^m{\frac{m\qof{}{\chodof{}{S}}}{\epsilon \qsum}}\right)-\epsilon \qsum\\
        \label{eq: lastbutone}
        \suml{j=1}^m{\qof{}{\chodof{}{S}}} &\leq OPT+\epsilon \qsum
    \end{align}
Note that \qsum is the sum of minimum possible disutility contribution from each\dsetof{}. Any valid set should have one degree (degree $0$ corresponds to not funding the project) from each\dsetof{}. Since the optimal solution has to be a valid set, the optimal disutility must be at least \qsum, i.e., $OPT \geq \qsum$. Applying this in \Cref{eq: lastbutone}, results in $\sum_{j=1}^m{\qof{}{\chodof{}{S}}} \leq (1+\epsilon)OPT$.

\paragraph{Running Time} Now, we need to check the running time of constructing the DP table. The table has $m$ rows. The disutilities calculated for each degree of project is upper bounded by $\frac{\qof{}{}m}{\epsilon\qsum}$, which is further bounded by $\frac{m\vpara}{\epsilon}$ (by the definition of $\vpara$, $\because \qof{}{}\leq \qmax$). Since the output is a valid set, there can be at most $m$ projects in the output and hence the maximum possible total disutility is upper bounded by $\frac{m^2\vpara}{\epsilon}$, which is the number of columns. Computing each entry in row $x$ takes $O(\mdof{x})$ time. Thus, the running time is $\boldsymbol{O(\frac{m^3\vpara\tstar}{\epsilon})}$, where $\tstar = \max_{p_j \in \proj}{\mdof{}}$.
\end{proof}
\paragraph{Fixed Parameter Tractability} Here, we prove that \drule is also in FPT with respect to the parameters \slimit or $m$.
\begin{proposition}\label{the: drule-sl}
    For any instance \instance, computing a subset $S \in \valid$ that is selected under \drule is fixed parameter tractable parameterized by the scalable limit \slimit.
\end{proposition}
\begin{proof}
    The proof of the above proposition is same as the proof of \Cref{the: crule-sl}. Let $k =\frac{1}{\tiny{GCD}(\dcof{1}{1},\ldots,\dcof{m}{\mdof{m}},\bud)}$. First, we get a new instance $\instance'$ by scaling down all the costs and budget as follows: $\dcof{}{'t} =k\dcof{}{}$ and $\bud' = k\bud$. Then, we may construct the DP table similar to that in the proof of \Cref{the: nrule-p}. Running time is again ${O(m^2n\slimit \tstar)}$.
\end{proof}

\begin{proposition}
    For any instance \instance, computing a subset $S \in \valid$ that is selected under \drule is fixed parameter tractable parameterized by the number of projects $m$.
\end{proposition}
The above result follows from the fact that when $m$ is constant, we can exhaustively search all the valid sets for the one achieving optimal disutility.

\section{Axiomatic Results}\label{sec: d-fl-axioms}
An enormous amount of work has been done on the axiomatic study of PB with approval votes \cite{aziz2018proportionally,aziz2019proportionally,talmon2019framework,rey2020designing,baumeister2020irresolute,sreedurga2022maxmin}. However, our PB model in which each project has a set of permissible costs is unique technically as well as realistically. This uniqueness demands the development of novel axioms applicable explicitly in such a setting. We introduce different axioms and investigate which of the PB rules introduced in \Cref{sec: utilities} satisfy our axioms and which of them do not. We call $[\lowerb{}{},\upperb{}{}]$ to be the interval approved by voter $i$.
\subsection{Monotonicity Properties}\label{sec: monotonicity}
We start by introducing monotonicity properties, which capture how a change in one parameter related to the problem affects the outcome. The first axiom, shrink-resistance, implies that if the voters are narrowing down their interval towards a degree that was winning, then the degree still continues to win.
\begin{definition}[Shrink-resistance]\label{def: shrink}
	A PB rule \pbrule is said to be \emph{shrink-resistant} if for any instance \instance, a voter \ii, a project $p_j \in \proj$, it holds that a set $S$ selected under \pbrule continues to be selected even if\lowerb{}{}and\upperb{}{}are shifted closer to\chocof{}{}.
\end{definition}
\begin{proposition}\label{the: shrink}
	All the four rules \nrule, \crule, \ccaprule, and \drule are shrink-resistant.
\end{proposition}
\begin{proof}
 	First, we look at \nrule. Consider a set $S$ selected under \nrule. If initially $\lowerb{}{} \leq \chocof{}{} \leq \upperb{}{}$, then the condition continues to satisfy even if\lowerb{}{}and\upperb{}{}are shifted closer to\chocof{}{}. Thus,\utof{}{}does not change. If initially $\chocof{}{} < \lowerb{}{}$ or $\upperb{}{} < \chocof{}{}$, then shifting\lowerb{}{}and\upperb{}{}closer to\chocof{}{}may cause\utof{}{}to remain the same or increase by $1$. Utilities of sets without\chodof{}{}remain unchanged thus proving the claim. This logic can also be extended to \crule and \ccaprule. For both these rules, shifting will make\utof{}{}remain the same or increased by\chocof{}{}. Finally, look at \drule. By the definition of\dutof{}{}, shifting\lowerb{}{}and\upperb{}{}closer to\chocof{}{}will cause a strict decrease in\dutof{}{}if it was a non-zero value. Else,\dutof{}{}remains the same.
\end{proof}

Next, we extend discount montonicity \cite{talmon2019framework} to ensure that a winning degree of project should not be omitted if it becomes less expensive.
\begin{definition}[Discount-proofness]\label{def: discount}
    A PB rule \pbrule is said to be \emph{discount-proof} if for any instance \instance, a project $p_j \in \proj$, and a set $S$ that is selected under \pbrule, $S$ continues to be selected if\chocof{}{}is decreased by $1$.
\end{definition}
\begin{proposition}\label{the: discount}
    The rule \nrule is discount-proof, whereas \crule, \ccaprule, and \drule are not.
\end{proposition}
\begin{proof}
    Let the rule \pbrule be \crule, \ccaprule, or \drule. Consider an instance such that: (i) there are two projects with $\dsetof{1} = \curly{\pdof{1}{0},\pdof{1}{1}}$ and $\dsetof{2} = \curly{\pdof{2}{0},\pdof{2}{1}}$ (ii) $\dcof{1}{1} = 2$ and $\dcof{2}{1} = 2$, $\bud=2$ (iii) for every voter $i\in\voters$, $\lowerb{}{1} = \upperb{}{1} = \dcof{1}{1}$ and $\lowerb{}{2} = \upperb{}{2} = \dcof{2}{1}$. Clearly, $S= \curly{\pdof{1}{1}}$ is a set that is selected under \pbrule. If \dcof{1}{1}is changed to $1$, $S$ is not selected since only $\curly{\pdof{2}{1}}$ is the unique set to be selected.

    Now, consider \nrule. This is straight-forward since the utilities are unaffected by costs as long as they are in the acceptable range. Please note that by their definitions, both\lowerb{}{}and\upperb{}{}belong to the set of permissible costs and hence if the cost that is reduced is equal to the lower bound reported by some voter, the lower bound is also automatically considered to be reduced by $1$. Thus, the number of voters finding the chosen degree to be acceptable remains unaffected.
\end{proof}

\subsection{Unanimity Properties}\label{sec: unanimity}
Now, we introduce properties that are based on unanimous approval of all the voters. Given an instance \instance, for each project $p\in \proj$, we define $\conrange = \curly{\dcof{}{}: \pdof{}{} \in \degreeproj,\;\forall i \in \voters\;\; \lowerb{}{} \leq \dcof{}{} \leq \upperb{}{}}$. Let $\maxconrange = \max{(\conrange)}$ and $\minconrange = \min{(\conrange)}$. The first axiom, range-abidingness, requires that if some range of costs is found to be acceptable unanimously by all voters, then the cost allocated to the project must not go beyond this range.
\begin{definition}[Range-abidingness]\label{def: rangeabiding}
	A PB rule \pbrule is said to be \emph{range-abiding} if for any instance \instance, a project $p_j \in \proj$, and a set $S$ selected under \pbrule, it holds that $$\conrange \neq \emptyset \implies \chocof{}{} \leq \maxconrange.$$
\end{definition}
\begin{proposition}\label{the: range-abiding}
	The rules \nrule and \drule are range-abiding, whereas \crule and \ccaprule are not.
\end{proposition}
\begin{proof}
	First, we prove that \crule and \ccaprule are not range-abiding. Recall that \bud denotes the budget. Suppose we have a single project $p$ with exactly two permissible costs, respectively equal to $\floor*{\frac{\bud-1}{n}}$ and \bud. Suppose all the voters set $\lowerb{}{} = \dcof{}{1}$. One voter \ii sets $\upperb{}{} = \dcof{}{2}$, whereas all the remaining voters report $\dcof{}{1}$ as the upper bound. Clearly, in a set $S$ containing \pdof{}{1}, the utility from project $p$ is at most $\bud-1$. Whereas, if $S$ selects \pdof{}{2}, the utility from $p$ is equal to $\bud$ (entirely due to the utility of a single voter). Thus $\pdof{}{2}$ is selected. Note that, in this example, $\conrange = \floor*{\frac{\bud-1}{n}}$ and clearly $\dcof{}{2} > \conrange$.

    Moving on, the rules \nrule and \drule are range-abiding. For the sake of contradiction, assume that they are not. That is, in a set $S$ selected under the rule, $\chocof{}{} > \maxconrange$. Consider the set $S' = S\setminus\curly{\chodof{}{}} \cup \curly{\pdof{}{}:\maxconrange = \dcof{}{}}$. Clearly, $S'$ is valid since $\chocof{}{} > \maxconrange$ and $S$ is valid. Also, $\utofd{}{p} > \utof{}{p}$ since the degree chosen in $S'$ is within the bounds reported by every voter (it gives the optimal utility of $n$ for \nrule and the optimal disutility of $0$ for \drule). Thus, $S'$ is a strictly better set and $S$ must not be selected. This forms a contradiction.
\end{proof}

The next axiom, range-convergingness, requires that increasing the budget should result in the winning degree of some project moving closer to its unanimously approved range of costs. That is, increasing the budget must increase the chance of maximally satisfying \emph{all} the voters, as opposed to dissatisfying some voters.
\begin{definition}[Range-convergingness]\label{def: rangeconverging}
    A PB rule \pbrule is said to be \emph{range-converging} if for any instance \instance, a set $S$ selected for \instance under \pbrule, and a set $S' \neq S$ selected under \pbrule on increasing the budget, it holds that whenever there is at least one project $k$ with $\tau_k \neq \emptyset$, there also exists some project $p_j \in \proj$ such that: $$\chocof{}{} \notin \conrange \implies |\chocof{}{} - \minconrange| > |\chocof{}{S'} - \minconrange|.$$
\end{definition}
\begin{proposition}\label{the: range-converging}
    All the four rules \nrule, \crule, \ccaprule, and \drule are range-converging.
\end{proposition}
\begin{proof}
    First, consider the first three rules, i.e., let \pbrule be \nrule, \crule, or \ccaprule. Note that for these rules, the value $\sum_{i \in \voters}{\utof{}{p}}$ decreases as $\chocof{}{}$ moves farther from all the costs in \conrange (the farther\chocof{}{}moves, more is the number of voters for whom\chocof{}{}falls out of their acceptable range). We prove the claim by contradiction. Assume that for any project $p_j \in \proj$ such that $\chocof{}{} \notin \conrange$ and $\chocof{}{} \neq \chocof{}{S'}$, $\chocof{}{}$ is closer to \conrange than\chocof{}{S'}is. This implies that $\sum_{i \in \voters}{\utofd{}{p}} < \sum_{i \in \voters}{\utof{}{p}}$. By adding these inequalities for all such $p$, we have $\sum_{i \in \voters}{\utofd{}{}} < \sum_{i \in \voters}{\utof{}{}}$. Thus, $S'$ cannot be selected by \pbrule since $S$ continues to be feasible under the increased budget.  
    The proof for rule \drule follows similar idea. This is because, $\sum_{i \in \voters}{\dutof{}{p}}$ increases as $\chocof{}{}$ moves farther from all the costs in \conrange (as\chocof{}{}falls out of the acceptable range for more voters).
\end{proof}

The last monotonicity axiom, range-unanimity, requires that if some range of costs is found to be acceptable unanimously by all the voters, the project must be allocated the maximum amount in this range, or in other words, the highest unanimously approved amount.
\begin{definition}[Range-unanimity]\label{def: rangeunanimous}
    A PB rule \pbrule is said to be \emph{range-unanimous} if for any instance \instance, whenever $\sum_{p_j \in \proj}{\maxconrange} \leq \bud$, the set $\curly{\pdof{}{}: \;$p$ \in \proj,\;\dcof{}{} = \maxconrange}$ is selected under \pbrule.
\end{definition}
Note that the above holds by default if \maxconrange is not defined for some $p$. Range-unanimity and range-abidingness do not imply each other, though they seem closely related. For example, take the rule that picks for each project $p$, degree with minimum cost in \conrange (whenever this set costs lesser than \bud). This is range-abiding but not range-unanimous. In an instance where the maximum costs in \conrange together cost more than the budget, range-unanimity is satisfied by default by any rule. However, a rule that selects a degree for project $p$ whose cost is greater than maximum cost in \conrange, does not satisfy range-abidingness. Thus, range-unanimity does not imply range-abidingness.
\begin{proposition}\label{the: range-unanimous}
    The rules \nrule and \drule are \emph{range-unanimous}, whereas \crule and \ccaprule are not.
\end{proposition}
\begin{proof}
    The proof for \crule and \ccaprule follow from the example in the proof of \Cref{the: range-abiding}. Now, consider the rules \nrule and \drule. Note that the set $\curly{\pdof{}{}: p \in \proj,\;\dcof{}{} = \maxconrange}$ achieves the optimal total utility $mn$ and optimal total disutility of $0$ respectively for both the rules. Hence, if $S'$ not being selected implies that $\csetof{S'} > \bud$. This contradicts $\sum_{p_j \in \proj}{\maxconrange} \leq \bud$ and completes the argument. 
\end{proof}

\subsection{Efficiency Properties}\label{sec: efficiency}
We finally introduce efficiency properties, which capture how a slight betterment of the outcome in one dimension is not possible (similar to the idea behind classic pareto-efficiency property). The first axiom, degree-efficiency, essentially implies that if two valid sets differ only on the degree of one project, then the set with higher degree needs to be preferred.

\begin{definition}[Degree-efficiency]\label{def: degreeefficient}
    A PB rule \pbrule is said to be \emph{degree-efficient} if for any instance \instance, any project $p_j \in \proj$, any set $S$ selected under \pbrule, and the degree $x \in \chodof{}{}$, it holds that $$k>x\implies\csetof{}-\chocof{}{}+\dcof{}{k} > \bud.$$
\end{definition}
\begin{proposition}\label{the: degreeefficient}
    The \ccaprule is degree-efficient, whereas \nrule, \crule, and \drule are not.
\end{proposition}
\begin{proof}
    Let \pbrule be \nrule, \crule, or \drule. Consider the example with single project $p$ such that for every $i$: (i) $\dcof{}{\mdof{}} < \bud$ and $\upperb{}{} < \dcof{}{\mdof{}}$ (ii) there exists $t < \mdof{}$ such that $\lowerb{}{} \leq \dcof{}{} \leq \upperb{}{}$. The set $\curly{\pdof{}{}: \dcof{}{} = \maxconrange}$ is selected under \pbrule since we know that $\conrange \neq \emptyset$ and $\maxconrange < \bud$. Assume that the axiom is satisfied. Then, by the definiton of \mdof{}, $\dcof{}{\mdof{}}>\bud$ contradicting (i).
    
    Now, we look at the rule \ccaprule. Let $S$ be a selected set and $k < \mdof{}$ be such that $k > \chodof{}{}$. Consider the set $S' = S\setminus\curly{\chodof{}{}}\cup \curly{\pdof{}{k}}$. Take any voter $i$. If $\chocof{}{} < \lowerb{}{}$, $\utofd{}{j}$ remains zero like $\utof{}{j}$ or increases by \dcof{}{k}. If $\lowerb{}{}\leq\chocof{}{}\leq\upperb{}{}$, $\utofd{}{j}$ is exactly $\min{(\dcof{}{k},\upperb{}{})}-\chocof{}{}$ more than $\utof{}{j}$. This value is positive since $k > \chodof{}{}$. Finally, if $\chocof{}{} > \upperb{}{}$, $\utofd{}{j}$ is exactly $\dcof{}{k}$ more than $\utof{}{j}$. Thus,\utofd{}{}is clearly greater than\utof{}{}. Since $S$ is selected under \drule, $S'$ must be infeasible. Thus, the given condition holds.
\end{proof}

The next two axioms insist that the valid set closer to the bounds reported by all the voters must be preferred over a valid set farther from them.
\begin{definition}[Lower bound-sensitivity]\label{def: lboundsensitive}
    A PB rule \pbrule is said to be \emph{lower bound-sensitive} if for any instance \instance, any project $p_j \in \proj$, and any two valid set $S,S'$ such that for every voter $i$ we have $\chocof{}{S}<\chocof{}{S'}<\lowerb{}{}$, it holds that $S$ is not selected under \pbrule.
\end{definition}
\begin{proposition}\label{the: lboundsensitive}
    The \drule is lower bound-sensitive, whereas \nrule, \crule, and \ccaprule are not.
\end{proposition}
\begin{proof}
    Let \pbrule be \nrule, \crule, or \ccaprule. For lower-bound sensitivity, consider a counter example as follows: (i) there are two projects with $\dsetof{1} = \curly{\pdof{1}{0},\pdof{1}{1},\pdof{1}{2},\pdof{1}{3}}$ and $\dsetof{2} = \curly{\pdof{2}{0},\pdof{2}{1}}$ (ii) $\dcof{1}{1} = 1$, $\dcof{1}{2} = 2$, $\dcof{1}{3} = \bud-3$, and $\dcof{2}{1} = \bud-2$ (iii) for every voter $i\in\voters$, $\lowerb{}{1} = \upperb{}{1} = \dcof{1}{3}$ and $\lowerb{}{2} = \upperb{}{2} = \dcof{2}{1}$. Clearly, $S= \curly{\pdof{1}{1},\pdof{2}{1}}$ is a set that is selected under \pbrule. Set $S' = \curly{\pdof{1}{2},\pdof{2}{1}}$ and $j = 1$.

    Now, let us consider the rule \drule. The proof is straight-forward. Since the disutility of $i$ from $S$ depends on $\lowerb{}{}-\chocof{}{}$ and $\chocof{}{}-\upperb{}{}$, clearly, $\dutof{}{} < \dutofd{}{}$ for every voter $i$. Hence, $S$ does not get selected by \drule.
\end{proof}
\begin{definition}[Upper bound-sensitivity]\label{def: uboundsensitive}
    A PB rule \pbrule is said to be \emph{upper bound-sensitive} if for any instance \instance, any project $p_j \in \proj$, and any two valid set $S,S'$ such that for every voter $i$ we have $\chocof{}{S}>\chocof{}{S'}>\upperb{}{}$, it holds that $S$ is not selected under \pbrule.
\end{definition}
\begin{proposition}\label{the: uboundsensitive}
    The rules \nrule,\crule, and \drule are upper bound-sensitive, whereas \ccaprule is not.
\end{proposition}
\begin{proof}
    First, let the rule \pbrule be \nrule, \crule, or \drule. Since $\chocof{}{} > \upperb{}{}$, $\utof{}{j} = 0$ and $\dutof{}{} = \chocof{}{}-\upperb{}{}$. Clearly,  utility of $S$ and $S'$ will be same for \nrule and \crule. However, none of them chooses $S$ or $S'$. For \drule, utility from $S$ will be lesser than that from $S'$. Thus, $S$ cannot be selected. Finally, consider \ccaprule. Consider an instance such that: (i) there are two projects with $\dsetof{1} = \curly{\pdof{1}{0},\pdof{1}{1},\pdof{1}{2},\pdof{1}{3}}$ and $\dsetof{2} = \curly{\pdof{2}{0},\pdof{2}{1}}$ (ii) $\dcof{1}{1} = 1$, $\dcof{1}{2} = 2$, $\dcof{1}{3} = 3$, and $\dcof{2}{1} = \bud-3$ (iii) for every voter $i\in\voters$, $\lowerb{}{1} = \upperb{}{1} = \dcof{1}{1}$ and $\lowerb{}{2} = \upperb{}{2} = \dcof{2}{1}$. Clearly, $S= \curly{\pdof{1}{3},\pdof{2}{1}}$ is a set that is selected under \ccaprule. Set $S' = \curly{\pdof{1}{2},\pdof{2}{1}}$ and $j = 1$. The condition is not met.
\end{proof}
\begin{table}
\centering
\begin{tabular}{l"c|c|c|c}
\textbf{PROPERTIES\textbackslash{}RULES} & $\boldsymbol{\nrule}$     & $\boldsymbol{\crule}$ & $\boldsymbol{\ccaprule}$ & $\boldsymbol{\drule}$ \\
\thickhline
{Shrink-resistance}                 & $\checkmark$ & $\checkmark$                             & $\checkmark$                                & $\checkmark$                             \\
\hline
{Discount-proofness}                  & $\checkmark$ & $\times$                                 & $\times$                                    & $\times$                                 \\
\hline
{Range-abidingness}                   & $\checkmark$ & $\times$                                 & $\times$                                    & $\checkmark$                             \\
\hline
{Range-convergingness}                & $\checkmark$ & $\checkmark$                             & $\checkmark$                                & $\checkmark$                             \\
\hline
{Range-unanimity}                 & $\checkmark$ & $\times$                                 & $\times$                                    & $\checkmark$                             \\
\hline
{Degree-efficiency}                & $\times$     & $\times$                                 & $\checkmark$                                & $\times$                                 \\
\hline
{Lower bound-sensitivity}           & $\times$     & $\times$                                 & $\times$                                    & $\checkmark$       \\
\hline
{Upper bound-sensitivity}           & $\checkmark$ & $\checkmark$                             & $\times$                                    & $\checkmark$ 
\end{tabular}
\caption{Results for budgeting axioms in \Cref{sec: d-fl-axioms}}
\label{tab: results}
\end{table}
\section{Conclusion and Discussion}\label{sec: d-fl-conclusion}
Many times, there are multiple ways of executing a public project and hence several, but limited number of, choices for the amount to be allocated to this project. Unfortunately, the existing preference elicitation methods and aggregation rules for participatory budgeting do not take this factor into account and we bridge this gap. We generalized two utility notions defined for PB under restricted costs to our model. We also proposed two other utility notions unique to our model. We analyzed all the corresponding utilitarian welfare maximizing rules computationally and axiomatically.

Our computational part strengthens all the existing positive results, and also introduces several new parameterized tractability results (FPT, parameterized FPTAS) taking into account the parameters recently introduced in the PB literature. It is worth highlighting that all our computational results in \Cref{sec: d-fl-comp} can be generalized by replacing the utilities with cardinal utility for every degree of each project. However, we present all the results for ranged approval votes due to their practical relevance, simplicity, and deep association with axiomatic analysis.

Followed by this, we introduce several axioms for our model with ranged approval votes and investigate which of these are satisfied by our PB rules. Note that, though none of the proposed PB rules satisfies all the axioms, every rule satisfies some axioms. Axiomatic analysis reflects the properties of each rule, using which the PB organizer can pick a rule based on the context. Also, it is worth bearing in mind that the novel disutility notion and PB rule \drule we proposed for our model satisfies as many axioms as any simple approval-based PB rule satisfies. One of the key takeaways of this chapter is hence a conclusion that \drule is a very good choice when each voter approves a range of costs.

In this chapter, we studied welfare maximization for PB under partially flexible costs. An intriguing direction to explore further would be the concept of fairness within the context of partially flexible costs.

Thus far, our emphasis has been on participatory budgeting under dichotomous preferences. In the subsequent part of this thesis, we will delve into the study of participatory budgeting under ordinal preferences, considering both restricted and flexible costs as the cases of interest.

\newpage

\part{Ordinal Preferences}\label{part: ordinal}
\thispagestyle{empty}

\blankpage

\chapter{Restricted Costs: Welfare Maximization and Fairness under Incomplete Weakly Ordinal Preferences}\label{chap: o-re}

\begin{quote}
\textit{We study the participatory budgeting model where the cost of each project is restricted to one value and the voters report incomplete weakly ordinal preferences. The chapter is structured into two distinct parts: one emphasizing the maximization of welfare and the other dedicated to exploring the concept of fairness.}

\textit{In the first part, we extend the existing welfare maximizing rules in the literature on dichotomous and strictly ordinal preferences to propose a class of rules, called \emph{dichotomous translation rules} and another rule named \emph{PB-CC rule}. We prove that our extensions mostly preserve and also enhance the computational and axiomatic properties of the rules. As a part of this, we also introduce a new axiom, \emph{pro-affordability}, explicitly relevant to PB under weakly ordinal preferences. In the latter part of the chapter, we introduce fresh mechanisms to address fairness, called average rank-share guarantee (ARSG) rules. This novel class of rules consists of two distinct families: average rank guarantee rules and share guarantee rules. By employing ARSG rules, we are able to overcome the limitations associated with existing fairness notions in the literature on participatory budgeting under ordinal preferences.}
\end{quote}

\section{Motivation}\label{sec: o-re-intro}
Though dichotomous preferences are well studied both for the cases of restricted costs \cite{bogomolnaia2005collective,duddy2015fair,aziz2019fair} as well as flexible costs \cite{talmon2019framework,goel2019knapsack,jain2020participatory,rey2020designing,sreedurga2022maxmin,fairstein2022welfare}, the dichotomous preferences inherently have limited expressibility since all the approved projects are assumed to be equally preferred by the voter. This motivates eliciting ordinal preferences (rankings) of the projects from the voters. Ordinal preferences are proven to be cognitively easy as well as the most liked ballot by the voters since they get to express more information \cite{benade2018efficiency}.  Various works study strictly ordinal preferences when the costs are totally flexible \cite{aziz2014generalization,aziz2018rank,airiau2019portioning} and also when the costs are restricted to one single value \cite{shapiro2017participatory}. Undeniably, the inherent ranking of projects of a voter is usually weak since ties among the projects are very natural, especially when the number of projects is high. This highlights the need for studying weakly ordinal preferences in participatory budgeting.

Though weakly ordinal preferences are well studied when the costs are flexible \cite{aziz2014generalization,aziz2018rank,ebadian2022optimized}, their study when the costs are restricted has been very limited \cite{aziz2021proportionally}. This chapter skillfully bridges this gap. It is worth emphasizing that we also allow the weakly ordinal preferences to be incomplete, thereby making the elicitation method no more cognitively harder than standard dichotomous preferences (voters can \emph{optionally} rank the approved projects further).

\subsection{Contributions and Organization of the Chapter}\label{sec: o-re-contribution}
The two most desired objectives in social choice literature are \emph{welfare maximization} and \emph{fairness}. While the welfare maximization in participatory budgeting (PB) is studied under dichotomous as well as strictly ordinal preferences (the latter is studied only for multi-winner voting, a special case of PB under restricted costs), it remained to be studied under weakly ordinal preferences. We propose PB rules which extend the existing rules to the incomplete weakly ordinal preferences setting, such that the computational and axiomatic properties of the rules are mostly preserved and even enhanced. We justify the significance of each of our extended rules through an exhaustive axiomatic analysis. \Cref{sec: o-re-welfare} of this chapter studies the PB rules that maximize the utilitarian welfare.

Unlike welfare, fairness for PB under restricted costs and weakly ordinal preferences has been studied in the literature \cite{aziz2021proportionally}. However, these existing fairness notions can be unreasonable in several scenarios, as illustrated in the upcoming sections. We resolve this drawback by proposing a different perspective on fairness and introducing a novel interesting class of fair rules. Also, unlike the work by Aziz et al. \cite{aziz2021proportionally}, we allow for the ordinal preferences to be incomplete (i.e., voters can rank only a few projects and ignore the rest, as explained in \Cref{sec: intro-preferences}). \Cref{sec: o-re-fair} of this chapter studies the fair PB rules. We conclude by making important and compelling observations in \Cref{sec: o-re-conclusion}.
\section{Notations and Preliminaries}\label{sec: o-re-prelims}
Let us recall the necessary notations from \Cref{sec: prelims-notations} and introduce a few more. Recall that $\voters = \{1,\ldots,n\}$ is the set of $n$ voters and $\proj = \{p_1,\ldots,p_m\}$ is the set of $m$ projects. Each voter \ii gives a weakly ordinal preference $\suci$ over a subset $A_i$ of projects.
That is, \suci partitions a set $A_i$ into equivalence classes such that $E^i_1 \succ_i E^i_2 \succ_i \ldots$ (note that every project need not be included in some equivalence class). The rank of a project $p$ is said to be $r$ if exactly $r-1$ projects are strictly preferred over $p$. 
We use $r_i(p)$ to denote the rank of $p$ and $\pat{t}$ to denote all the projects ranked exactly $t$ in $\succeq_i$. We denote by $\ptill{j}$ the set of projects in the first $j$ equivalence classes of $\succeq_i$, i.e., $\ptill{j} = \bigcup_{k \leq j} E^i_k$. A preference profile \rprof is a vector of ordinal preferences of all the voters, i.e., $\rprof = {(\suci)}_{i \in \voters}$. A cost function $c: \proj \to \mathbb{N}$ gives the cost of each project. Cost of a subset $S \subseteq \proj$, $\sum_{p \in S}{\cof{p}}$, is denoted by \cof{S}. \bud denotes the total budget available. A set $S \subseteq \proj$ is said to be \textit{feasible} if $c(S) \leq \bud$. Let \feasible be the set of all feasible subsets of $\proj$.

\begin{example}
	Suppose $\proj = \{p_1,p_2,p_3,p_4\}$ is the set of projects such that $\cof{p_1} = 3, \cof{p_2} = 4, \cof{p_3} = 2$, and $\cof{p_4} = 6$. Let the ranking of a certain voter, suppose $i$, be $\curly{p_2, p_3} \succ_i \curly{p_1}$. Now, $A_i = \curly{p_1,p_2,p_3}$, $r_i(p_1) = 3$, $\pat{2} = \emptyset, \pat{3}=\{p_1\}$, $\ptill{2} = E^i_1 \cup E^i_2 = \curly{p_1,p_2,p_3}$.
\end{example}
A PB instance under incomplete weakly ordinal preferences \fullrinstance is represented by \instance. The rules we present are \emph{irresolute}, i.e., they output multiple feasible subsets. For an instance \instance, a PB rule \pbrule outputs a set of feasible subsets of projects \ruleof{\pbrule}{}.
\section{Welfare Maximizing Rules}\label{sec: o-re-welfare}
In this section, we study PB rules that maximize utilitarian welfare. For this, we first need to define the utility notion. When the preferences are weakly ordinal, utility of a voter from a set of projects must ideally depend on the following factors: (i) the number/costs of the selected projects and (ii) ranks of those projects in her preference. Including both factors (i) and (ii) in the definition of utility may lead to a multi-objective optimization problem. On the other hand, we propose two kinds of PB rules, each considering one of the above two factors. PB rules discussed in \Cref{sec: o-re-dtr} consider the factor (i) while defining the utility, whereas the one in \Cref{sec: o-re-pbcc} considers the factor (ii). Each PB rule outputs all the feasible subsets of projects that maximize the utilitarian welfare (i.e., the sum of utilities of all voters). 

We define the rules by carefully extending the existing rules in the literature to our model of PB under weakly ordinal preferences, without compromising much on their axiomatic properties. We justify the significance of our extended rules by conducting an exhaustive axiomatic analysis towards the end of this section. In fact, taking a close look at our axiomatic results reveals an interesting observation as will be explained in \Cref{sec: o-re-conclusion}.
\subsection{Dichotomous Translation Rules}\label{sec: o-re-dtr}
We start by extending the rules defined for PB under dichotomous preferences \cite{talmon2019framework}. In the model with dichotomous preferences, every voter $i$ reports a subset $A_i$ of projects that she likes and no ranking is involved. Thus, the utility of a voter here depends on factor (i), i.e., on the number of approved projects that got selected or their costs. Talmon and Faliszewski \cite{talmon2019framework} define three such utility notions and study the corresponding utilitarian rules. The utility of $i$ from a set $S$ of projects is defined as $u_i(S) = f(A_i \cap S)$, where $f(S)$ could be $|S|, \cof{S},$ or  $\bool(|S| > 0)$ (takes $1$ if $S$ is non-empty and $0$ otherwise). These utility notions are also explained in \Cref{sec: prelims-dutilities}.

Dichotomous translation rules progress in two phases. In the first phase, a \emph{translation scheme} is used to convert weakly ordinal preferences into dichotomous preferences satisfying certain property. In the second phase, the existing utilitarian rules described in the previous paragraph are applied on the resultant instance of dichotomous preferences. The translation scheme in the first phase must be chosen carefully such that most properties are preserved in the outcome after the second phase. Additionally, our translation also resolves the drawbacks suffered by PB under dichotomous preferences. We discuss two such translation schemes in this section.

\begin{definition}\label{def: lar}
	A \textbf{dichotomous translation rule} $\langle \TT,f \rangle$ is a two-phase protocol with the following phases: (i) translation scheme \TT converts the weak ranking profile \rprof into an approval vote profile $(A_i)_{i \in N}$ (ii) Every $S$ that maximizes $\sum_{i \in \voters}{f(A_i \cap S)}$ is selected in the outcome.   
\end{definition}
\begin{figure}[h]
\centering
	\includegraphics[width=0.7\columnwidth, height = 2.5cm]{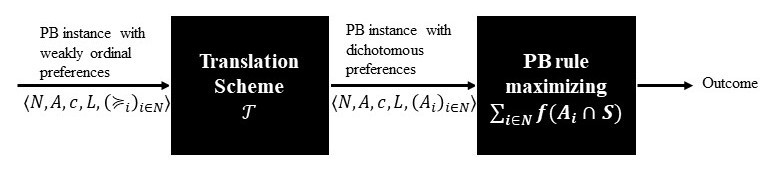}
	\caption{Block diagram to show the working of a dichotomous translation rule $\langle \TT,f \rangle$}
\end{figure}
\subsubsection{Multi-knapsack translation scheme}\label{sec: mt}
In the PB under dichotomous preferences, all the approved projects are interpreted to be equally desirable. Thus, the PB rules fail to capture the feasible subsets of projects most desired by the voter. This motivated Goel et al. \cite{goel2019knapsack} to introduce knapsack voting, where for every voter $i$, $\cof{A_i} \leq \bud$. That is, every voter approves her most desired feasible subset of projects. However, this may not perform well when there are ties among the projects and there are multiple equally desirable feasible subsets for the voter. This is illustrated below.

\begin{example}\label{eg1: dtr}
	Let  $A = \{p_1,p_2,\ldots,p_{8}\}$  and $\bud = 5$, with costs $\{3,3,2,\ldots,2\}$, respectively. Let there be two voters whose inherent preferences for the projects are as follows: $p_1 \succ_1 \curly{p_2,p_3,p_4} \succ_1 p_5$ and $\curly{p_1,p_3} \succ_2 \curly{p_4,p_8}$.
	
	If the voters are asked for standard dichotomous preferences, voter $1$ approves $\{p_1,p_2,p_3,p_4,p_5\}$ and voter $2$ approves $\curly{p_1,p_3,p_4,p_8}$. Here, the information that voter $1$ prefers an outcome $\{p_1,p_4\}$ over $\{p_1,p_5\}$ is lost. Now, suppose the voters are asked to report knapsack votes instead (i.e., each $A_i$ must cost within the budget). Then, voter $1$ is forced to approve $p_1$ and exactly one project from $\{p_3,p_4\}$, thereby implying incorrectly that the remaining project yields no utility to her (note that both $\curly{p_1,p_3}$ and $\curly{p_1,p_4}$ are equally desirable feasible outcomes).
\end{example}

In our model, we allow the voter $i$ to also report a weakly ordinal preference over projects in $A_i$. Our translation scheme uses this weakly ordinal preference to further shrink down $A_i$ such that it captures all and only those projects present in \emph{at least one} of the most desired feasible subsets of $i$. The \emph{\textbf{multi-knapsack translation scheme ($\boldsymbol{MT}$)}} greedily adds the projects following the rank until the budget constraint is respected. At the point where the constraint gets violated, $MT$ includes only those projects whose inclusion will not make the subset infeasible. This translation scheme is demonstrated in Algorithm \ref{algo: mtscheme}. The corresponding dichotomous translation rule is further elucidated below in \Cref{eg2: dtr}.
\begin{algorithm}
	\DontPrintSemicolon
	\KwIn{An ordinal PB instance \fullrinstance}
	\KwOut{A dichotomous preference profile \approf}
        \For{each voter $i$}{
        $A_i \gets \emptyset$;\;
        $j \gets 1$;\;
            \While{$\cof{\aof{}}+\cof{E^i_j} \leq \bud$}{
		      $A_i \gets A_i \cup E^i_j$;\;
                $j \gets j+1$;\;
		  }
            $O_i \gets \emptyset$;\;
            \For{each $p \in E^i_j$}{
                \If{$\cof{p} < \bud-\cof{\aof{}}$}{
                    $O_i \gets O_i \cup \{p\}$;\;
                }
            }
        $A_i \gets A_i \cup O_i$;\;    
        }
	\Return{$(A_i)_{i \in N}$}\;
	\caption{Multi-knapsack translation scheme}
	\label{algo: mtscheme}
\end{algorithm}
\begin{example}\label{eg2: dtr}
	 Consider \Cref{eg1: dtr}. Let us look at the first phase. For the first voter, the multi-knapsack translation scheme starts by adding $E^1_1$ (i.e., $p_1$) to $A_1$. The left over after this is $5-\cof{p_1} = 2$. Cost of $E^1_2$ is $7$, which is greater than left over of $2$ and hence $E^1_2$ does not fit entirely into $A_1$. Thus, we mark only those projects in $E^1_2$ whose costs are within the left over limit, i.e., $\curly{p_3,p_4}$ and add these marked projects into $A_1$ (In Algorithm \ref{algo: mtscheme}, we used $O_i$ to denote the marked projects for $i$). Finally, $A_1 = \curly{p_1,p_3,p_4}$. Similarly, for the second voter, set the left over limit to the budget $5$. Since $E^2_1$ costs $5$, add it to $A_2$. The left over limit is now $0$. Thus, $A_2 = \curly{p_1,p_3}$.
	 
	 In the second phase, if $f(S)$ is defined as $|S|$ or $\cof{S}$, then the rule $\langle MT, f \rangle$ outputs $\curly{\{p_1,p_3\}}$. If $f(S)$ is defined as $\bool(|S| > 0)$, then the $\langle MT, f \rangle$ rule outputs all the sets containing at least one of $p_1$ and $p_3$.
\end{example}

It needs to be mentioned that adding the projects greedily with respect to rank may not always result in an optimal knapsack vote. However, solving the knapsack problem is \NPH and cognitively hard for the voters \cite{benade2021preference}. Thus, especially when costs of the projects are close to each other, the greedy approach provides a close-to-ideal outcome with much less cognitive load. Consequently, the greedy approach is commonly used for knapsack and PB in real-world. It is worth highlighting that $\langle MT ,f \rangle$ generalizes the existing multi-winner voting rules \cite{faliszewski2018multiwinner}. An ordinal multi-winner voting instance can be represented as a PB instance with $\bud = k$ and unit cost projects. If $f(S)$ is $|S|$ or $\cof{S}$, $\langle \MM,f \rangle$ is equivalent to multi-winner bloc rule. If $f(S) = \bool(|S| > 0)$, it is equivalent to $\alpha_k$-CC rule.

We now analyze the computational complexity of the decision problem of the dichotomous translation rule with multi-knapsack translation scheme. That is, the problem is to decide if the sets selected by the rule $\langle MT,f \rangle$ guarantee at least a given utility $s$.
\begin{theorem}\label{the: mtnp}
	Deciding $\langle MT,f \rangle$ is \NPH if any of the following conditions holds:
    \begin{enumerate}[(a)]
        \item $f(S) = c(S)$
        \item $f(S) = \bool(|S| > 0)$.
    \end{enumerate}
\end{theorem}
\begin{proof}
Let $\approf = (A_i)_{i \in N}$ be the dichotomous preference profile obtained from the incomplete weak rankings profile \rprof using the translation scheme $MT$. 
    \begin{enumerate}[(a)]
        \item We reduce \subsum (\Cref{prob: subsum}) to our problem of deciding if there exists $S \in \feasible$ such that $\suml{i \in \voters}{c(A_i \cap S)} \geq s$ for any score $s$. Given an integer $H$ and a set of integers $X= \{x_1,x_2,\ldots,x_n\}$, the \subsum problem is to decide if there exists a subset $X' \subseteq X$ such that $\sum_{x \in X'}{x} = H$. This is known to be \NPH \cite{garey1979computers}. W.l.o.g., we can assume that $X$ is sorted in non-decreasing order.

        Construct a PB instance as follows: set $\bud = s = H$. Create $n$ projects $p_1,p_2,\ldots,p_n$ such that $c(p_i) = x_i$ and another project $p_{n+1}$ with $c(p_{n+1}) = H$. Create $n$ voters such that the preference of voter \ii is $p_i \succ_i p_{n+1} \succ_i others$ (`others' could be ordered arbitrarily). We claim that \subsum is equivalent to our problem. Let us prove the correctness. The scheme $MT$ approves only the top-ranked project for each voter \ii. Hence, for each \ii, $A_i$ is $\{p_i\}$. If the given instance is a \yes instance of \subsum and $X'$ is the required subset, then $\sum_{i \in \voters}{f\big({A_i\;\cap\;S}\big)} = \sum_{x_i \in X'}{x_i} = H$ when $S = \{p_i : x_i \in X'\}$. Hence, this is a \yes instance of the given problem. Now, let us assume that the given instance is a \no instance of \subsum. Therefore, there is no subset $X' \subseteq X$ such that $\sum_{x_i \in X'}{x_i} = H$. Thus, for every $S \in \feasible$, $\sum_{i \in \voters}{f\big({A_i\;\cap\;S}\big)} < s$ (since $\bud = H$), making this a \no instance of the given problem.

        \item We reduce \vercov to our problem of deciding if there exists $S \in \feasible$ such that $\suml{i \in \voters}{\bool(|A_i \cap S|>0)} \geq s$ for any score $s$.
        \begin{definition}[\vercov]\label{prob: vercov}
        Given an undirected graph $G = (V,E)$ and an integer $k$, the \vercov problem is to decide if there exists a $V' \subseteq V$ such that $|V'| \leq k$ and $E = \{(v_1,v_2): (v_1 \in V') \lor (v_2 \in V')\}$. 
        \end{definition}
        The \vercov problem is known to be \NPH \cite{garey1979computers}. We can assume without loss of generality that $k > 2$ (for constant values of $k$, \vercov becomes tractable).

        Construct a PB instance as follows: set $\bud = k$. For each vertex $v$, add a project $p_v$ with cost $1$. Add another dummy project $p_d$ with cost $k-1$. For each edge $e_i = (v^1_i,v^2_i)$, add a voter $i$ with preference $p_{v^1_i} \succ_i p_{v^2_i} \succ_i p_d \succ_i others$ (`others' could be ordered arbitrarily). Set $s = |E|$. We claim that \vercov is equivalent to our problem. Let us prove the correctness. For each voter $i$, the scheme $MT$ gives an output such that $A_i = \{p_{v^1_i},p_{v^2_i}\}$. Suppose the given instance $(G,k)$ is a \yes instance of \vercov. Let $V' \subseteq V$ cover entire $E$ such that $|V'| \leq k$. Consider the feasible subset of projects $S = \{p_v: v \in V'\}$. Since $V'$ is a vertex cover, for any voter \ii, $|A_i \cap S| > 0$. Therefore, $\sum_{i \in \voters}{f\big({A_i\;\cap\;S}\big)} = |E|$ making it a \yes instance. Now, let us assume that $(G,k)$ is a \no instance. Any feasible subset of projects gives a total score less than $|E|$, making it a \no instance.
    \end{enumerate}
This completes the proof.
\end{proof}

To find the computational complexity when $f(S)=|S|$, we rely on a result presented by Talmon and Faliszewski \cite{talmon2019framework}.
\begin{proposition}[Talmon and Faliszewski \cite{talmon2019framework}]\label{cor: talmon}
    Given any PB instance \instance with dichotomous preference profile, a set $S$ maximizing $\suml{i \in \voters}{|A_i \cap S|}$ can be computed in polynomial time.
\end{proposition}
The next result is a direct consequence of \Cref{cor: talmon}.
\begin{corollary}\label{the: mtp}
	If $f(S) = |S|$, then $\langle MT,f \rangle$ can be decided in polynomial time.
\end{corollary}
Next, we discuss another good translation scheme, called cost-worthy translation scheme (CT), which results in an outcome satisfying many good properties. 

\subsubsection{Cost-worthy translation scheme}\label{sec: ct}
The CT scheme represents the  desirability of economical projects. A major disadvantage of existing PB rules under dichotomous preferences is that some of them treat all the approved projects same, while others prefer an expensive project assuming that the cost reflects its prestige and quality  \cite{aziz2018proportionally,talmon2019framework,goel2019knapsack}. While this could be reasonable in some contexts, often there are real-world scenarios in which inexpensive projects are to be preferred and an expensive project is to be funded only if it is \emph{preferred over other projects} by a high enough number of voters. The CT scheme captures the above requirement using the weakly ordinal preferences of the voter and a parameter $\cwpara: [m] \to \mathbb{N}$. The function \cwpara is such that $\cwpara(m) = 0$ and it is monotonically non-increasing function, i.e., $\cwpara(i) \geq \cwpara(j)$ whenever $i < j$. This is called a \emph{worth function} and is used to capture the relation between ranks and the costs of the projects. For any $j \in [m]$, $\cwpara(j)$ denotes the maximum amount a project ranked at $j$ deserves. The scheme approves, for each voter, only those projects that are affordable according to their ranks and the worth function, as shown in Algorithm \ref{algo: ctscheme}. We illustrate this with an example.
\begin{algorithm}
\DontPrintSemicolon
\KwIn{Incomplete weakly ordinal preference profile \rprof, a worth function \cwpara}
\KwOut{Dichotomous preference profile \approf}
\For{each voter $i$}{
    $A_i \gets \emptyset$\\
    \For{$j = 1 ; j\leq m; j++$}{
        \If{$\cof{p_j} \leq \cwpara(r_i(p_j))$}{
            $A_i \gets A_i$ $\cup \{p_j\}$\;
        }
    }
}
\Return{$(A_i)_{i \in N}$}\;
\caption{Cost-worthy translation scheme}
\label{algo: ctscheme}
\end{algorithm}
\begin{example}\label{eg1: cwr}
	Organizers of a 2-hour seminar must select a subset of candidate talks (each of different duration) based on the preferences of audience. They may decide that only one of the talks, the plenary talk, be given 60 min., while the others be allocated at most 50 or 30 or 20 min. Based on this, they could set the worth function such that $\cwparaof{1} = 60, \cwparaof{2} = 50, \cwparaof{3} = 30, \cwparaof{4} = \cwparaof{5} = 20, $ and $\cwparaof{x} = 0$ for $x>5$.
	\begin{table}[H]
		\scalebox{0.85}{
			\begin{tabular}{c|ccccc}
				$\boldsymbol{\alpha}$ & $\mathbf{60}\:\;\;\;$ & $\mathbf{50}\:\;\;\;$ & $\mathbf{30}\:\;\;\;$ & $\mathbf{20}\:\;\;\;$ & $\mathbf{20}\:\;\;\;$\\
				\hline
				$40\%$ & $\textcolor{orange}{A}\:\;\;\succ$ & $\textcolor{orange}{B}\:\;\;\succ$ & $\{\textcolor{orange}{C},\textcolor{orange}{D}\}\:\;\;\succ$ & $\;\:\;\;\;\:\;\;\;\:\;\;\:\:\;\succ$ & $\textcolor{orange}{E}\:\;\;$\\
				{\small{of voters}}& $30\:\;\;\;$ & $50\:\;\;\;$ & $30,20\:\;\;\;$ & $\;\:\;\;\;$ & $20\:\;\;\;$ \\
				\hline
				$60\%$ & $\textcolor{orange}{A}\:\;\;\succ$ & $\textcolor{orange}{C}\:\;\;\succ$ & $\:\;\;\;\:{B}\:\;\;\:\;\;\:\succ$ & $\{\textcolor{orange}{D},\textcolor{orange}{E}\}\:\;\;$ & $\;\:\;\;\;$\\
				{\small{of voters}}& $30\:\;\;\;$ & $30\:\;\;\;$ & $50\:\;\;\:\;\;$ & $\;20,20\:\;\;\;$ & $\;\:\;\;\;$\\
			\end{tabular}
		}
	\end{table}
	All projects ranked adequately high are marked above (in the first ranking, rank of $E$ is $5$ since four projects are preferred over it). $60\%$ of the voters do not rank $B$ high enough to justify its cost. If $f(S)$ is $|S|$ or $\cof{S}$, the dichotomous translation rule $ \langle CT,f \rangle$ outputs $\{A,C,D,E\}$. If $f(S)$ is $\bool(|S| > 0)$, the rule outputs all the sets having at least one of $\curly{A,C,D,E}$.
\end{example}

One may wonder why we use a separate worth function instead of asking the voters to simply approve only the projects they consider worth the cost. However, asking so would assume that a voter deeming a project to be cost-worthy is same as the organizer deeming it to be so. But, a voter decides the worth of a project based on the value derived from it with respect to cost, whereas the organizer may decide the worth also based on how it fares in \emph{comparison} to other projects. A voter may feel that a park worth 900 does bring the proportional value and hence approve or rank it decently high. But if there are several projects she prefers over the park, the organizer may decide spending 900 on it is not worth it. The cost-worthy translation scheme offers the organizer the flexibility to decide whether the degree of preference for an expensive project (i.e., its rank) is good enough to warrant a huge cost.

We describe a few situations where CT scheme provides a desirable choice, by specifying the corresponding worth functions. If a funding agency wants no project costing higher than $x$ to be selected unless it is of the highest quality, setting $\cwparaof{1} = \bud$ and all the others to be $x$ captures the requirement. If the funding agency believes that the projects in the top $p$ positions have stiff competition among themselves and there is a consensus regarding the remaining projects, a function whose top $p$ entries are almost similar while the rest of the entries have larger gaps can be the parameter. If the projects ranked less than a certain threshold should never be approved, we can achieve this by setting to zero, all the entries in the range of worth function to be zero after the threshold. Such a vector can also be used to avoid an inexpensive project being approved by all the voters trivially.

Note that while deciding the worth function parameter can be done intuitively and heuristically in the real-world in several scenarios (organizing a seminar, funding research projects etc.), it may be hard in some other scenarios. This is the primary and only drawback of this rule. Thus, there is an interesting optimization problem underlying their choice. Learning optimal parameters of desirable rules such as median and min-max rules and threshold approval rules \cite{benade2021preference} are considered important open directions in the voting literature. Similarly, the question of determining optimal parameters for dichotomous translation rules with cost-worthy translation scheme is an important problem which we preserve for our future work.

The proofs of the next two results is similar to that of \Cref{the: mtnp} and \Cref{the: mtp}. A similar approach as that of \Cref{the: mtp} can be used to prove \Cref{the: ctp1}.
\begin{theorem}\label{the: ctnp}
	Deciding $\langle CT,f \rangle$ is \NPH if any of the following conditions holds:
    \begin{enumerate}[(a)]
        \item $f(S) = c(S)$ and $\cwparaof{1} = \bud$
        \item $f(S) = \bool(|S| > 0)$ and $\cwparaof{1} \neq \cwparaof{m}$.
    \end{enumerate}
\end{theorem}
\begin{proof}
Let $\approf = (A_i)_{i \in N}$ be the dichotomous preference profile obtained from the incomplete weak rankings profile \rprof using the translation scheme $CT$. 
    \begin{enumerate}[(a)]
        \item We reduce \subsum (\Cref{prob: subsum}) to our problem of deciding if there exists $S \in \feasible$ such that $\suml{i \in \voters}{c(A_i \cap S)} \geq s$ for any score $s$. Given an integer $H$ and a set of integers $X= \{x_1,x_2,\ldots,x_n\}$, the \subsum problem is to decide if there exists a subset $X' \subseteq X$ such that $\sum_{x \in X'}{x} = H$. This is known to be \NPH \cite{garey1979computers}. Without loss of generality, we can assume that $X$ is sorted in non-decreasing order.

        Construct a PB instance as follows: set $\bud = s = H$. For each integer $x_i$, create a project $p_i$ costing $x_i$. Create a single voter who ranks all these $n$ projects in the first place. We claim that \subsum is equivalent to our problem. Let us prove the correctness. The scheme $CT$ approves all the projects and hence is equivalent to \subsum by its definition.

        \item We reduce \vercov (\Cref{prob: vercov}) to our problem of deciding if there exists $S \in \feasible$ such that $\suml{i \in \voters}{\bool(|A_i \cap S|>0)} \geq s$ for any score $s$. Given an undirected graph $G = (V,E)$ and an integer $k$, the \vercov problem is to decide if there exists a $V' \subseteq V$ such that $|V'| \leq k$ and $E = \{(v_1,v_2): (v_1 \in V') \lor (v_2 \in V')\}$. The \vercov problem is known to be \NPH \cite{garey1979computers}. We can assume without loss of generality that $k > 2$ (for constant values of $k$, \vercov becomes tractable).

        Construct a PB instance as follows: set $\bud = k(\cwpara(m)+1)$. For each vertex $v$, add a project $p_v$ with cost $\cwpara(m)+1$. Add $m-2$ dummy projects $d_1,\ldots,d_{m-2}$ such that $\cof{d_i} = \cwpara(i+1) + 1$. For each edge $e_i = (v^1_i,v^2_i)$, add a voter $i$ with preference $\{p_{v^1_i}, p_{v^2_i}\} \succ_i d_1 \succ_i \ldots \succ_i d_{m-2} \succ_i others$. Set $s = |E|$. We claim that \vercov is equivalent to our problem. Let us prove the correctness. Suppose the given instance $(G,k)$ is a \yes instance of \vercov. Let $V' \subseteq V$ cover entire $E$ such that $|V'| \leq k$. Consider the feasible subset of projects $S = \{p_v: v \in V'\}$. Since $V'$ is a vertex cover, for any voter \ii, $|A_i \cap S| > 0$. Therefore, $\sum_{i \in \voters}{f\big({A_i\;\cap\;S}\big)} = |E|$ making it a \yes instance. Now, let us assume that $(G,k)$ is a \no instance. Any feasible subset of projects gives a total score less than $|E|$, making it a \no instance.
    \end{enumerate}
This completes the proof.
\end{proof}
\begin{corollary}\label{the: ctp}
	If $f(S) = |S|$, then $\langle CT,f \rangle$ can be decided in polynomial time.
\end{corollary}
Similar to \Cref{the: mtp}, the aforementioned corollary is also an immediate consequence of \Cref{cor: talmon}.
\begin{theorem}\label{the: ctp1}
	If $f(S) = \cof{S}$, then $\langle CT,f \rangle$ can be decided in the time polynomial in $\cwparaof{1}$.
\end{theorem}
\begin{proof}
    Let $\approf = (A_i)_{i \in N}$ be the approval-vote profile obtained from the incomplete weakly ordinal preference profile \rprof using Algorithm \ref{algo: ctscheme}. We present a dynamic programming based algorithm. Let us define score of a project $p$ as $score(p) = |\{i \in \voters: p \in A_i\}|*\cof{p}$. Construct a table $T$ with $m$ rows and $mn\cwparaof{1}$ columns. The entry $T(i,j)$ corresponds to the cost of cheapest subset of $\{p_1,\ldots,p_i\}$ for which the total score of projects is exactly $j$. We fill the first row as follows:
    \[T(1,j) = \begin{cases} 
          c(p_1) & j = score(p_1) \\
          0 & \text{otherwise} 
       \end{cases} \quad \forall j \in [mn\cwparaof{1}].
    \]
    Now the remaining rows can be filled recursively as follows:
    \begin{equation}
        \label{eqn: btcardinality-recurrence}
        T(i,j) = \min\Big\{T\big(i-1,j\big)\;,\;T\big(i-1,j-score(p_i)\big)+c(p_i)\Big\}
    \end{equation}
    The maximum score possible is $\max\{j \in [mn\cwparaof{1}]: T(m,j) \leq \bud\}$. 
    \paragraph{Running time.} Note that computing each entry of $T$ takes constant time and there are $m^2n\cwparaof{1}$ entries. The running time is $O(m^2n\cwparaof{1})$. Correctness follows from the definition of $T$.
\end{proof}

The next theorem follows trivially since the translation scheme approves all and only those projects whose cost is at most $\cwparaof{1} = \cwparaof{m} = Q$. Hence, any such project, when it exists will be an optimal solution.
\begin{theorem}\label{the: ctp2}
	If $f(S) = \bool(|S|>0)$, then $\langle CT,f \rangle$ can be decided in polynomial time if $\cwparaof{1} = \cwparaof{m}$.
\end{theorem}
\subsection{The PB-CC Rule}\label{sec: o-re-pbcc}
In this section, we generalize the Chamberlin-Courant (CC) rule defined for multi-winner voting under strictly ordinal preferences to our model of PB under weakly ordinal preferences. As explained previously, multi-winner voting is a special case of indivisible PB in which there are no costs involved and exactly $k$ projects are chosen. Clearly, the utilities in this model are determined from the ranks alone. In their groundbreaking work, Chamberlin and Courant introduced one of the most popular multi-winner voting rules, called CC-rule \cite{chamberlin1983representative}.

We generalize the CC-rule to our model where the preferences could be weak, the projects are allowed to have different costs, and a budget constraint is enforced on the outcome. The utility function of our rule however is same as that of standard CC-rule: every voter derives a utility equal to $m-r$ from a subset of projects, where $r$ is the rank of the voter's most favorite project in the subset. This is illustrated in the below example.

\begin{example}
	Suppose the preference of a voter is $\curly{p_1,p_3} \succ \curly{p_2,p_4} \succ p_5$. The utility of this voter from a set $S= \curly{p_2,p_5}$ is $2$ since $p_2$ is the most preferred project of the voter in $S$ and the rank of $p_2$ is $3$.
\end{example}

\begin{definition}\label{def: pbcc}
	Given a PB instance \fullrinstance, the PB-CC rule selects every feasible subset $S \in \feasible$ which maximizes the value $\sum\limits_{i \in N} (m-\min\limits_{p \in S}{r_i(p)})$.
\end{definition}

The following result is a consequence of the fact that deciding CC-rule is \WWH even for multi-winner voting under strictly ordinal preferences \cite{betzler2013computation}, which is a special case of our model where every project costs $1$ and the total budget $\bud = k$.

\begin{proposition}\label{the: pbcc}
	Deciding PB-CC rule is \WWH with respect to the budget \bud.
\end{proposition}

\subsection{Axiomatic Justification}\label{sec: o-re-axioms}
Here, we further justify the PB rules proposed in Sections \ref{sec: o-re-dtr} and \ref{sec: o-re-pbcc} using an axiomatic analysis. We prove that by extending the existing rules in the literature, we do not compromise much on their properties and that, in fact, our rules additionally satisfy a few other interesting properties. As a part of this exercise, we generalize the axioms defined for dichotomous preferences as well as strictly ordinal preferences, both of which are special cases of weakly ordinal preferences. We also introduce a novel axiom, pro-affordability, explicitly for PB under weakly ordinal preferences. The key observations from the axiomatic results are explained in detail in \Cref{tab: o-re-results} and \Cref{sec: o-re-conclusion}.

Throughout the section, for any given PB rule \pbrule and an instance \instance, we use \winners{\pbrule}{} to denote the set of all projects included in the output of \pbrule. Formally, $\winners{\pbrule}{} = \{p \in \proj: \exists S \in \ruleof{\pbrule}{}, p \in S\}$. We say a project $p$ wins if and only if $p \in \winners{\pbrule}{}$. We use $\boldsymbol{\scoreof{p}}$ to denote the number of voters such that $p$ is included in $A_i$ by the translation scheme in dichotomous translation rules. We use $\boldsymbol{\trunkkcc{}{S}}$ to denote $\min\limits_{p \in S}{r_i(p)}$. We defer the axiomatic analysis of dichotomous translation rules with cost-worthy translation scheme to Appendix \ref{app: cwaxioms}. It must be noted that almost all the axioms are satisfied for this family of rules, but the proofs are deferred to the appendix due to their simplicity.
\subsubsection{Generalization of Axioms in the Literature on Dichotomous Preferences}\label{axioms: dichotomous}
Talmon and Faliszewski \cite{talmon2019framework} identify compelling axioms for PB under dichotomous preferences. In \Cref{sec: d-re-axioms}, we extended these axioms to the case where PB rules under dichotomous preferences are irresolute. Now, we further extend these axioms to allow weakly ordinal preferences.

The first axiom, \emph{splitting monotonicity}, asserts that, if a winning project is replaced by a set of multiple new projects that together cost the same, then at least one of the new projects continues to win. In other words, take any winning project $x \in \winners{\pbrule}{}$. Split $x$ into a set $X$ of smaller projects such that $\cof{x} = \cof{X}$. In every preference $\succeq_i$ with $x \in E^i_j$, replace $x$ in $E^i_j$ with all the projects of $X$. Let the new resultant instance be $\instance'$. Splitting monotonicity requires that at least one project from $X$ must belong to $\winners{\pbrule}{'}$.
\begin{definition}[Splitting Monotonicity]\label{def: splittingm}
	A PB rule \pbrule satisfies splitting monotonicity if for any instance \instance and $x \in \winners{\pbrule}{}$, we have $X \cap \winners{\pbrule}{'} \neq \emptyset$ whenever $\instance'$ is obtained from $\instance$ by splitting $x$ into a set $X$ of new projects such that $\cof{x} = \cof{X}$ and replacing $x$ by $X$ in every $\succeq_i$ in \rprof.
\end{definition}
\begin{theorem}\label{the: splittingmp}
	All the following rules satisfy splitting monotonicity:
	\begin{enumerate}[(a)]
		\item A dichotomous translation rule $\langle MT,f \rangle$ such that $f(S)$ is any of the following: (i) $|S|$ (ii) $\cof{S}$ or (iii) $\bool(|S|>0)$. 
		\item The PB-CC rule
	\end{enumerate}
\end{theorem}
\begin{proof}
For an instance \instance and PB rule \pbrule, let $x$ be a project in \winners{\pbrule}{} and let $S_x$ be a set such that $S_x \in \ruleof{\pbrule}{}$ and $x \in S_x$. Let $X$ be a new set of projects such that $\cof{X} = \cof{x}$. Let $\instance'$ be an instance obtained by replacing $x$ with $X$ in every ranking $\suci \in \rprof$. Let $S' = (S_x \setminus \curly{x}) \cup X$. It is enough to prove that $S' \in \ruleof{\pbrule}{'}$.
    \begin{enumerate}[(a)]
        \item First, we prove that for the dichotomous translation rule \gtr, for any $x_1 \in X$, if $x \in A_i$ in \instance, then $x_1 \in A_i$ in $\instance'$. This follows since the point where adding of projects to $A_i$ stops remains the same in both the instances. Since $\cof{x_1} < \cof{x}$ and $x \in A_i$ in \instance, $x_1 \in A_i$ in $\instance'$.
        
        To prove (i) and (ii), suppose $f(S)$ is $|S|$ or $\cof{S}$. For any $x_1 \in X$, we know from the above that $\scoreof{x} = \scoreof{x_1}$. This implies, if $f(S) = |S|$, total utility of $S'$ in $\instance'$ is exactly $|X|-1$ times the total utility of $S_x$ in $\instance$. If $f(S) = \cof{S}$, for any $x_1 \in X$, we know from the above that $\scoreof{x_1} = (\cof{x_1}/\cof{x}) \scoreof{x}$. Therefore, total utility of $S'$ in $\instance'$ is the same as the total utility of $S_x$ in $\instance$.
        
        To prove (iii), suppose $f(S) = \bool(|S|>0)$. From the above, the number of voters who have some project of $S_x$ in $A_i$ is the same as those who have some project of $S'$ before it. In all the above scenarios, since $S_x \in \ruleof{\gtr}{}$, $S' \cap \ruleof{\gtr}{'} \neq \emptyset$.

        \item Finally, let \pbrule be the PB-CC rule. For any voter \ii, $\trunkkcc{\ebpara}{S_x}$ in $\instance$ is the same as $\trunkkcc{\ebpara}{S'}$ in $\instance'$. Therefore, total utility of $S_x$ in \instance is the same as total utility of $S'$ in $\instance'$. Since $S_x \in \ruleof{\pbrule}{}$, $S' \cap \ruleof{\pbrule}{'} \neq \emptyset$.
    \end{enumerate}
    This completes the proof.
\end{proof}

The next axiom, discount monotonicity, asserts that a winning project continues to win if it becomes less expensive. The condition of insisting $\cof{x} \geq 2$ is to ensure that it is possible for the cost of $x$ to decrease.
\begin{definition}[Discount Monotonicity]\label{def: discountm}
	A PB rule \pbrule satisfies discount monotonicity if for any instance \instance and $x \in \winners{\pbrule}{}$ such that $\cof{x} \geq 2$, we have $x \in \winners{\pbrule}{'}$ whenever $\instance'$ is obtained from $\instance$ by reducing the cost of $x$ by 1.
\end{definition}
\begin{theorem}\label{the: discountmp}
	The PB-CC rule satisfies discount monotonicity.
\end{theorem}
\begin{proof}
	For an instance \instance, let $x$ be a project in \winners{\pbrule}{} such that $\cof{x} \geq 2$, where \pbrule is the PB-CC rule. Let $S_x$ be a set such that $S_x \in \ruleof{\pbrule}{}$ and $x \in S_x$. Let $\instance'$ be an instance obtained from \instance by reducing \cof{x} by $1$. It is enough to prove that $S_x \in \ruleof{\pbrule}{'}$. For any set $S$ and voter \ii, we know that $\trunkkcc{\ebpara}{S}$ is the rank of the best ranked project of $S$ in $\suci$. Consider an arbitrary voter \ii. If some project of $S_x$ is ranked before $x$ in $\suci$, $\trunkkcc{\ebpara}{S_x}$ is the same in \instance and $\instance'$. If $x$ is the best ranked project of $S_x$ in $\suci$ in \instance, then $\trunkkcc{\ebpara}{S_x}$ remains the same in $\instance'$ since $\cof{x}$ in $\instance'$ is at least $1$. Therefore, total utility of $S_x$ in $\instance'$ is same as that in \instance. For any arbitrary voter \ii and a set $S$ without $x$ in it, \trunkkcc{\ebpara}{S} remains the same in $\instance'$. Hence, the total utility of any set without $x$ remains the same in $\instance'$. Therefore, any set in $\ruleof{\pbrule}{}$ with $x$ in it continues to be in $\ruleof{\pbrule}{'}$ and hence $x \in \winners{\pbrule}{'}$.
\end{proof}
\begin{theorem}\label{the: discountmn}
	No dichotomous translation rule $\langle MT,f \rangle$ such that $f(S)$ is (i) $|S|$, (ii) $\cof{S}$, or (iii) $\bool(|S|>0)$ satisfies discount monotonicity.
\end{theorem}
\begin{proof}
	Let \pbrule be the rule \gtr. For (i) and (ii), consider an instance \instance with budget $\bud > 3$, projects $\proj = \curly{p_1,p_2,p_3,p_4,p_5}$ costing $\curly{2,\bud-1,2,1,\bud}$ respectively, and four voters with preferences $p_1 \succ_1 p_2 \succ_1 others$, $p_2 \succ_2 p_3 \succ_2 others$, $p_4 \succ_3 p_5 \succ_3 others$, and $p_4 \succ_4 p_5 \succ_4 others$ respectively. Now, $\ruleof{\pbrule}{} = \curly{\curly{p_1,p_4},\curly{p_2,p_4}}$ and $p_1 \in \winners{\pbrule}{}$. If $\cof{p_1}$ is reduced by $1$, $\ruleof{\pbrule}{'} = \curly{\curly{p_2,p_4}}$ and $p_1 \notin \winners{\pbrule}{'}$.
	
	For (iii), consider an instance \instance with budget $\bud > 3$, projects $\proj = \curly{p_1,p_2}$ each costing \bud, and two voters with preferences $p_1 \succ_1 p_2$ and $p_2 \succ_2 p_1$ respectively. Then, $\ruleof{\pbrule}{} = \curly{\curly{p_1},\curly{p_2}}$ and $p_1 \in \winners{\pbrule}{}$. However, if $\cof{p_1}$ is reduced by $1$, $\ruleof{\pbrule}{'} = \curly{\curly{p_2}}$ and $p_1 \notin \winners{\pbrule}{'}$.
\end{proof}

The next axiom, limit monotonicity, requires that a winning project will continue to win if the budget is increased.
\begin{definition}[Limit Monotonicity]\phantomsection\label{def: limitm}
	A PB rule \pbrule satisfies limit monotonicity if for any instance \instance with no project costing exactly $\bud+1$, we have $x \in \winners{\pbrule}{'}$ whenever $x \in \winners{\pbrule}{}$ and $\instance'$ is obtained from $\instance$ by increasing the budget by 1.
\end{definition}
\begin{theorem}\label{the: limitmn}
	The following rules do not satisfy limit monotonicity:
	\begin{enumerate}[(a)]
		\item A dichotomous translation rule $\langle MT,f \rangle$ such that $f(S)$ is any of the following: (i) $|S|$ (ii) $\cof{S}$, or (iii) $\bool(|S|>0)$
		\item The PB-CC rule
	\end{enumerate}
\end{theorem}
\begin{proof}
    \begin{enumerate}[(a)]
        \item For (i), consider an instance \instance with budget $\bud > 3$, projects $\curly{p_1,p_2,p_3,p_4,p_5}$ respectively costing $\curly{2,\bud-1,3,1,\bud}$, and four voters with rankings: $p_1 \succ p_2 \succ \ldots \succ p_5$, $p_2 \succ p_3 \succ p_1 \succ p_4 \succ p_5$, $p_4 \succ p_5 \succ p_3 \succ p_1 \succ p_2$, and $p_4 \succ p_5 \succ p_3 \succ p_2 \succ p_1$. Then, $\ruleof{\gtr}{} = \curly{\curly{p_1,p_4},\curly{p_2,p_4}}$. Now if the budget is increased by $1$ to get $\instance'$, we have, $\ruleof{\gtr}{'} = \curly{\curly{p_2,p_4}}$. Hence, $p_1 \notin \winners{\gtr}{'}$. The same example also works for (iii). 

        For (ii), consider an instance \instance with budget $\bud > 4$, projects $\curly{p_1,p_2,p_3,p_4}$ respectively costing $\curly{2,\bud-1,2,\bud}$, and three voters with rankings: $p_1 \succ p_2 \succ p_3 \succ p_4$, $p_3 \succ p_4 \succ p_1 \succ p_2$, and $p_3 \succ p_4 \succ p_2 \succ p_1$. Then, $\ruleof{\gtr}{} = \curly{\curly{p_1,p_3}}$. Now if the budget is increased by $1$ to get $\instance'$, since $\bud+3 > 7$, we have, $\ruleof{\gtr}{'} = \curly{\curly{p_2,p_3}}$. Hence, $p_1 \notin \winners{\gtr}{'}$.

        \item Let \br be the PB-CC rule. Construct an instance with budget $\bud = 1$, projects $\proj = \curly{p_1,p_2,p_3}$ costing $1$ each, and two voters with preferences $p_1 \succ p_2 \succ p_3$ and $p_3 \succ p_2 \succ p_1$ respectively. Note that no project in $\proj$ costs $\bud+1$. Clearly no non-singleton set is feasible and utility of every singleton set is $2$. Therefore, $\ruleof{\br}{} = \curly{\curly{p_1},\curly{p_2},\curly{p_3}}$. Now, construct an instance $\instance'$ by increasing \bud to $2$. Now though $\curly{p_2}$ is feasible, it has a utility of $2$ whereas the feasible set $\curly{p_1,p_3}$ has an optimal utility of $4$. Therefore, $\ruleof{\br}{'} = \curly{\curly{p_1,p_3}}$ and $p_2 \notin \winners{\br}{'}$.
    \end{enumerate}
This completes the proof.
\end{proof}

Next, we study inclusion maximality \cite{talmon2019framework}, which implies that whenever some budget is remaining, voters would wish to use it.
\begin{definition}[Inclusion Maximality]\label{def: inclusionmax}
	A PB rule \pbrule satisfies inclusion maximality if for any instance \instance and $S, S' \in \feasible$ such that $S \subset S'$ and $S \in \ruleof{\pbrule}{}$, we have $S' \in \ruleof{\pbrule}{}$.
\end{definition}
\begin{proposition}\label{the: inclusionmp}
	All the following rules satisfy inclusion maximality:
	\begin{enumerate}[(a)]
		\item Any dichotomous translation rule $\langle MT,f \rangle$ such that $f(S)$ is any of the following: (i) $|S|$, (ii) $\cof{S}$, or (iii) $\bool(|S|>0)$
		\item The PB-CC rule
	\end{enumerate}
\end{proposition}
\begin{proof}
    \begin{enumerate}[(a)]
        \item It may be observed that, for all the three utility notions studied, utility from a set $S$ is at least as much as the utility from any subset of $S$. Such functions are said to be subset monotone functions. For any subset monotone functions, inclusion maximality is satisfied by default since adding additional projects may only increase the utility.
        \item The function \trunkkcc{}{S} is subset-monotone in $S$. This is because, $\trunkkcc{}{S} \leq \trunkkcc{}{S'}$ for any $S' \subseteq S$. Since the goal is to minimize $\sum_i\trunkkcc{}{S}$, inclusion-maximality follows.
    \end{enumerate}
    This completes the proof.
\end{proof}
\subsubsection{Generalization of Axioms in the Literature on Strictly Ordinal Preferences}\label{axioms: strict}
We extend the axioms defined for strictly ordinal preferences \cite{elkind2017properties} to allow for weakly ordinal preferences. The first axiom, candidate monotonicity, asserts that if a winning project is exchanged with a project ranked in the equivalence class just before it, it continues to win.
\begin{definition}[Candidate Monotonicity]\label{def: candidatem}
	A PB rule \pbrule satisfies candidate monotonicity if for any instance \instance and a project $x \in \winners{\pbrule}{}$, we have $x \in \winners{\pbrule}{'}$ whenever $\instance'$ is obtained from \instance by exchanging $x \in E^i_j$ with a project $x' \in E^i_{j-1}$ in some preference $\succeq_i$.
\end{definition}
\begin{theorem}\label{the: candidatemp}
	All the following rules satisfy candidate monotonicity:
	\begin{enumerate}[(a)]
		\item A dichotomous translation rule $\langle MT,f \rangle$ such that $f(S)$ is any of the following: (i) $|S|$, (ii) $\cof{S}$, or (iii) $\bool(|S|>0)$
		\item The PB-CC rule
	\end{enumerate}
\end{theorem}
\begin{proof}
	For any instance \instance and a PB rule \pbrule, let $x$ be a project such that $x \in \winners{\pbrule}{}$. Consider any $\succeq_i \in \rprof$. Let $x \in E^i_j$ and $x'$ be any project whose rank in $E^i_{j-1}$. Construct $\instance'$ by exchanging $x$ and $x'$ in $\succeq_i$.
    \begin{enumerate}[(a)]
        \item For (i) and (ii), it can be observed that in $\instance'$, \scoreof{x} will increase or remain the same, whereas, \scoreof{x'} will remain the same or decrease. This is because, depending on \cof{x} and \cof{x'}, the point where we stop adding projects to $A_i$ will either stay the same in $\instance'$ or change from $E^i_j$ to $E^i_{j-1}$ or from $E^i_{j-1}$ to $E^i_j$. Now, consider any other project $p$ whose score in $\instance'$ is greater than that in $\instance$. This is possible only when $p$ is ranked after $x$ in the new instance. So, the increase in the score of $p$ cannot be greater than the increase in the score of $x$. Since $x \in \winners{\pbrule}{}$, there exists $S_x \in \ruleof{\pbrule}{}$ such that $x \in S_x$. By definition, the utility of any set is the sum of scores of all projects in that set. Hence, some set containing $x$ will continue to win. Therefore, $x \in \winners{\pbrule}{'}$.

        For (iii), the utility of a set $S$ is the number of voters who have some project of $S \in A_i$. Utility of any $S_x \in \ruleof{\pbrule}{}$ increases by $1$ in $\instance'$ (when $x$ is the only project in $S_x$ in $A_i$) or remains the same, otherwise. Similarly, the utility of any set without $x$ stays the same or decreases $\instance'$. Therefore, $x \in \winners{\pbrule}{'}$.
        \item The total utility of set $S$ can change only because of \trunkkcc{\ebpara}{S} since preferences of all the other voters are unperturbed. Consider any set $S$. Let $x \in S$ and $x' \notin S$. Then, $\trunkkcc{\ebpara}{S}$ decreases from \instance to $\instance'$ or remains the same. Hence, utility of $S$ in $\instance'$ is at least that in $\instance$. Now, suppose $x \notin S$ and $x' \in S$. Then, $\trunkkcc{\ebpara}{S}$ decreases from \instance to $\instance'$ or remains the same. Hence, utility of $S$ in $\instance'$ is at most that in $\instance$. If both $x$ and $x'$ are present or absent together in $S$, $\trunkkcc{\ebpara}{S}$ and the total utility of $S$ are always the same in both \instance and $\instance'$. Therefore, utility of any set with $x$ in $\instance'$ will be at least that in \instance and utility of remaining sets will remain the same or decrease. Since $x \in \winners{\pbrule}{}$, some set with $x$ continues to win and $x \in \winners{\pbrule}{'}$.
    \end{enumerate}
This completes the proof.
\end{proof}

The next axiom, non-crossing monotonicity, asserts that if some project in a selected set $S$ is exchanged with a project in the equivalence class just before it without disturbing any other project in $S$, the set $S$ continues to win.
\begin{definition}[Non-crossing Monotonicity]\label{def: noncrossingm}
	A PB rule \pbrule satisfies non-crossing monotonicity if for any instance \instance, a set $S \in \ruleof{\pbrule}{}$, and a project $x \in S$, we have $S \in \ruleof{\pbrule}{'}$ whenever $\instance'$ is obtained from \instance by exchanging $x \in E^i_j$ with a project $x' \in E^i_{j-1}$ in some preference $\succeq_i$ such that $x' \notin S$.
\end{definition}
\begin{theorem}\label{the: noncrossingmp}
	 A dichotomous translation rule $\langle MT,f \rangle$ such that $f(S)$ is (i) $|S|$ or (ii) $\cof{S}$ satisfies non-crossing monotonicity.
\end{theorem}
\begin{proof}
	Let \pbrule denote the rule $\langle MT,f \rangle$. Let $S_x$ be a set such that $S_x \in \ruleof{\pbrule}{}$ and $x \in S_x$. Let $\succeq_i \in \rprof$ be a preference such that $x \in E^i_j$ for some $j$ and $S_x \setminus E^i_{j-1} \neq \emptyset$. Construct $\instance'$ by exchanging $x$ and $x'$ in $\succeq_i$. For the sake of contradiction, assume $S_x \notin \ruleof{\pbrule}{'}$. Let $f(S)$ be $|S|$ or $c(S)$.
	
	The proof is an extension to the proof of \Cref{the: candidatemp} and we use the fact that for the said rules, the utility of a set can be expressed as the sum of scores of projects in the set. We give an outline of the argument. From \Cref{the: candidatemp}, we know that the utility of $S_x$ increases or stays the same. We also know that there exists $S \in \ruleof{\pbrule}{'}$ such that $x \in S$. This is possible only if \scoreof{x} increases since $S_x \notin \ruleof{\pbrule}{'}$. But if \scoreof{x} increases, since $x' \notin S_x$, utility of $S_x$ also increases. This implies $S_x \in \ruleof{\pbrule}{'}$. 
\end{proof}
\begin{theorem}\label{the: noncrossingmn}
    The following rules do not satisfy non-crossing monotonicity:
    \begin{enumerate}[(a)]
        \item The dichotomous translation rule $\langle MT,f \rangle$ with $f(S) = \bool(|S|>0)$
        \item The PB-CC rule
    \end{enumerate}
\end{theorem}
\begin{proof}
    \begin{enumerate}[(a)]
        \item Consider an instance \instance with budget \bud$(>2)$, a set of projects $\proj = \curly{p_1,p_2,p_3,p_4,p_5,p_6}$ costing $1$, $\bud-1$, $\bud-1$, $\bud$, $1$, and $\bud$ respectively, and three voters whose preferences are: $p_1 \succ_1 p_2 \succ_1 p_3 \succ_1 others$, $p_3 \succ_2 p_6 \succ_2 others$, and $p_5 \succ_3 p_4 \succ_3 others$ respectively. Clearly, $\ruleof{\pbrule}{} = \curly{\curly{p_1,p_3},\curly{p_2,p_5},\{p_3,p_5\},\curly{p_1,p_5}}$. Now let $S = \curly{p_1,p_3}$. See that $p_2 \notin S$. Exchange $p_2$ and $p_3$ in $\succeq_1$ to obtain a new instance $\instance'$. In $\instance'$, $\curly{p_3,p_5}$ has a strictly higher utility than $S$ and hence $S \notin \ruleof{\pbrule}{'}$.
        \item The proof for PB-CC follows from the proof for CC rule by Elkind et al. \cite{elkind2017properties}. Their CC rule is a special case of our PB-CC, where the budget is $k$ and all the projects cost $1$ each. The authors prove that the CC rule does not satisfy non-crossing monotonicity.
    \end{enumerate}
    This completes the proof.
\end{proof}
\subsubsection{Axioms in the Literature on Generic Preferences}\label{axioms: generic}
We examine the axioms in the literature applicable for any voting model \cite{brandt2016handbook,skowron2019axiomatic,lackner2021consistent}. The first two axioms we study ensure an impartial treatment of all the voters and projects by enforcing that the outcome should not depend on the indexing of voters or projects. Let $\Sigma_S$ be the set of all permutations on a set $S$.
\begin{definition}[Anonymity]\label{def: anonymity}
	A PB rule \pbrule is said to be anonymous if for any instance \instance and $\sigma \in \Sigma_\voters$, we have $\pbrule(\instance)=\pbrule(\instance')$ whenever $\instance'$ is obtained from \instance by replacing $\succ_i$ with $\succ_{\sigma(i)}$ for every $i \in \voters$.
\end{definition}
\begin{definition}[Neutrality]\label{def: neutrality}
	A PB rule \pbrule is said to be neutral if for any instance \instance and $\sigma \in \Sigma_\proj$, we have $\pbrule(\instance)=\pbrule(\instance')$ whenever $\instance'$ is obtained from \instance by replacing \cof{j} with \cof{\sigma(j)} and \pat{j} with \pat{\sigma(j)} in every $\succ_i$ for every $j \in [m]$.
\end{definition}
\begin{theorem}\label{the: anonymityneutrality}
	Following rules satisfy both anonymity and neutrality:
	\begin{enumerate}
		\item Any dichotomous translation rule $\langle MT,f \rangle$ such that $f(S)$ is $|S|$, $\cof{S}$, or $\bool(|S|>0)$
		\item The PB-CC rule
	\end{enumerate}
\end{theorem}
Anonymity follows since all the rules are utilitarian and permuting the voters does not affect the summation. Neutrality follows since all rules depend only on the cost of projects and their ranks in preferences. The next popularly studied axiom, consistency, requires that for a voting rule, if two disjoint groups of voters $N_1$ and $N_2$ both choose an outcome $a$, $a$ continues to be chosen if the groups participate together in an election, i.e., for $N_1 \cup N_2$.
\begin{definition}[Consistency]\label{def: consistency}
	A PB rule \pbrule is said to be consistent if for any two ordinal PB instances $\instance_1 = \langle \voters_1,\proj,c,\bud,\rprof_1 \rangle$ and $\instance_2 =\langle \voters_2,\proj,c,\bud,\rprof_2 \rangle$ such that $\voters_1 \cap \voters_2 = \emptyset$, we have, $$\pbrule(\instance_1) \cap \pbrule(\instance_2) \subseteq \pbrule(\langle \voters_1 \cup \voters_2,\proj,c,\bud,\rprof_1 \cup \rprof_2 \rangle).$$
\end{definition}
\begin{theorem}\label{the: consistency}
	All the following rules satisfy consistency:
	\begin{enumerate}[(a)]
		\item A dichotomous translation rule $\langle MT,f \rangle$ such that $f(S)$ is any of the following: (i) $|S|$, (ii) $\cof{S}$, or (iii) $\bool(|S|>0)$
		\item The PB-CC rule
	\end{enumerate}
\end{theorem}
\begin{proof}
    Consider any $\instance_1$ and $\instance_2$ as given in \Cref{def: consistency}. Take any dichotomous translation rule \gtr, a set $S \in \gtr(\instance_1) \cap \gtr(\instance_2)$, and the instance $\instance' = \langle \voters_1 \cup \voters_2,\proj,c,\bud,\rprof_1 \cup \rprof_2 \rangle$. Since $S \in \gtr(\instance_1)$, $\sum_{i \in \voters_1}{f\big({A_i\;\cap\;S}\big)} \geq \sum_{i \in \voters_1}{f\big({A_i\;\cap\;S'}\big)}$ for any $S' \subseteq \proj$. Likewise, $\sum_{i \in \voters_2}{f\big({A_i\;\cap\;S}\big)} \geq \sum_{i \in \voters_2}{f\big({A_i\;\cap\;S'}\big)}$. 
By including both these inequalities, we get $S \in \ruleof{\gtr}{'}$.

The proof for part (b) follows the same idea.
\end{proof}
\subsubsection{Pro-Affordability}\label{sec: proafford}
We propose an axiom called pro-affordability explicitly for indivisible PB under weakly ordinal preferences, to assert that we always prefer a project that is less expensive and ranked higher by everyone.
\begin{definition}[Pro-affordability]\label{def: proafford}
	A PB rule \pbrule satisfies pro-affordability if for any instance \instance, $x \in \winners{\pbrule}{}$, and $x' \in \proj$ such that $\cof{x'} < \cof{x}$ and $x' \suci x$ for every voter \ii, we have $x' \in \winners{\pbrule}{}$.
\end{definition}
\begin{theorem}\label{the: proaffordp}
	All the following rules satisfy pro-affordability:
	\begin{enumerate}[(a)]
		\item A dichotomous translation rule $\langle MT,f \rangle$ such that $f(S)$ is any of the following: (i) $|S|$ or (ii) $\bool(|S|>0)$
		\item The PB-CC rule
	\end{enumerate}
\end{theorem}
\begin{proof}
For any rule \br, an instance \instance, and a project $x \in \winners{\br}{}$, let $x'$ be a project such that $\cof{x'} < \cof{x}$ and $x' \suci x$ for every voter \ii. Let $S_x$ be a set such that $S_x \in \ruleof{\br}{}$ and $x \in S_x$. Let $S' = (S_x \setminus \curly{x}) \cup \curly{x'}$. Note that $S' \in \feasible$ since $S_x \in \feasible$ and $\cof{x'} < \cof{x}$. It is enough to prove that $S' \in \ruleof{\br}{}$.
    \begin{enumerate}[(a)]
        \item Let us prove for (i). For any arbitrary voter \ii, if $x \in A_i$, then $x' \in A_i$ since $x' \suci x$. Utility of the set $S'$ is at least that of $S_x$ since $\scoreof{x'} \geq \scoreof{x}$. Therefore, $S' \in \ruleof{\gtr}{}$. Now, we prove for (ii). For any arbitrary voter \ii, if $S_x \cap A_i \neq \emptyset$, then $S' \cap A_i \neq \emptyset$ since $x' \suci x$. Thus, the utility of $S'$ is at least as much as that of $S_x$ and $S' \in \ruleof{\gtr}{}$.
        \item Let \br be the PB-CC rule. Take any arbitrary voter $i$. If $\rof{x'} < \trunkkcc{\ebpara}{S_x} \leq \rof{x}$, then $\trunkkcc{\ebpara}{S'} < \trunkkcc{\ebpara}{S_x}$. Else, $\trunkkcc{\ebpara}{S'} = \trunkkcc{\ebpara}{S_x}$. Hence, utility of $S'$ is at least that of $S_x$. Since $S_x \in \ruleof{\br}{}$, $S' \in \ruleof{\br}{}$.
    \end{enumerate}
This completes the proof.
\end{proof}
\begin{theorem}\label{the: proaffordn}
	The dichotomous translation rule $\langle MT,f \rangle$ such that $f(S) = \cof{S}$ does not satisfy pro-affordability.
\end{theorem}
\begin{proof}
	Suppose $f(S) = \cof{S}$ and \pbrule is \gtr. Consider an instance \instance with budget $\bud > 2$, four projects $\curly{p_1,p_2,p_3,p_4}$ respectively costing $\curly{1,\bud-1,1,\bud}$. Suppose there are three voters such that one voter has a ranking $p_1 \succ p_2 \succ p_3 \succ p_4$ and two voters have a ranking $p_3 \succ p_4 \succ p_1 \succ p_2$. Clearly $\ruleof{\pbrule}{} = \curly{\curly{p_2,p_3}}$ though $p_1$ is preferred over $p_2$ by all the voters and $\cof{p_1} < \cof{p_2}$.
\end{proof}
\begin{table*}[h!]
        \scalebox{0.8}{
	\begin{tabularx}{\textwidth}{c|cc|cc|cc|cc}
 \centering
		$\boldsymbol{f(S)}$                                           & \multicolumn{2}{c|}{$\boldsymbol{|S|}$}                                                                                             & \multicolumn{2}{c|}{$\boldsymbol{c(S)}$}                                                                                            & \multicolumn{2}{c|}{$\boldsymbol{\bool(|S|>0)}$}                                                                                    & \multicolumn{1}{c}{{\color[HTML]{009901} }}                     & \multicolumn{1}{c}{{\color[HTML]{009901} }}                        \\ \cline{1-7}
		\textbf{Property} $\boldsymbol{\downarrow}\;\;$ \textbf{Rule} $\boldsymbol{\rightarrow}$         & \multicolumn{1}{c}{{\color{orange} ${R_f}$}} & \multicolumn{1}{c|}{{\color{orange} $\boldsymbol{\langle MT,f \rangle}$}} & \multicolumn{1}{c}{{\color{orange} ${R_f}$}} & \multicolumn{1}{c|}{{\color{orange} $\boldsymbol{\langle MT,f \rangle}$}} & \multicolumn{1}{c}{{\color{orange} ${R_f}$}} & \multicolumn{1}{c|}{{\color{orange} $\boldsymbol{\langle MT,f \rangle}$}} & \multicolumn{1}{c}{\multirow{-2}{*}{{\color{blue} {CC}}}} & \multicolumn{1}{c}{\multirow{-2}{*}{{\color{blue} \textbf{PB-CC}}}} \\ \cline{1-9}
		Computational Complexity                         & P                                                 & P                                                                  & \NPH                               & \NPH                                                & \NPH                               & \NPH                                                & \WWH                                              & \WWH                                                 \\ \cline{1-9}
		{\color{orange} Splitting Monotonicity}    & $\checkmark$                                      & $\boldsymbol{\checkmark}$                                                       & $\checkmark$                                      & $\boldsymbol{\checkmark}$                                                       & $\checkmark$                                      & $\boldsymbol{\checkmark}$                                                       &                                                                  & {$\boldsymbol{\checkmark}$}                                               \\ \cline{1-7}\cline{9-9}
		{\color{orange} Discount Monotonicity}     & $\checkmark$                                      & $\boldsymbol{\times}$                                                           & $\times$                                          & $\boldsymbol{\times}$                                                           & $\checkmark$                                      & $\boldsymbol{\times}$                                                           &                                                                  & \textbf{$\boldsymbol{\checkmark}$}                                               \\ \cline{1-7}\cline{9-9}
		{\color{orange} Limit Monotonicity}        & $\times$                                          & $\boldsymbol{\times}$                                                           & $\times$                                          & $\boldsymbol{\times}$                                                           & $\times$                                          & $\boldsymbol{\times}$                                                           &                                                                  & \textbf{$\boldsymbol{\times}$}                                                   \\ \cline{1-7}\cline{9-9}
		{\color{orange} Inclusion Maximality}      & $\checkmark$                                      & $\boldsymbol{\checkmark}$                                                       & $\checkmark$                                      & $\boldsymbol{\checkmark}$                                                       & $\checkmark$                                      & $\boldsymbol{\checkmark}$                                                       & \multirow{-5}{*}{N.A.}                                           & \textbf{$\boldsymbol{\checkmark}$}                                               \\ \cline{1-9}
		{\color{blue} Candidate Monotonicity}    &                                                   & \textbf{$\boldsymbol{\checkmark}$}                                              &                                                   & \textbf{$\boldsymbol{\checkmark}$}                                              &                                                   & \textbf{$\boldsymbol{\checkmark}$}                                              & $\checkmark$                                                     & $\boldsymbol{\checkmark}$                                                        \\ \cline{1-1}\cline{3-3}\cline{5-5}\cline{7-9}
		{\color{blue} Non-crossing Monotonicity} & \multirow{-2}{*}{N.A.}                            & \textbf{$\boldsymbol{\checkmark}$}                                              & \multirow{-2}{*}{N.A.}                            & \textbf{$\boldsymbol{\checkmark}$}                                              & \multirow{-2}{*}{N.A.}                            & \textbf{$\boldsymbol{\times}$}                                                  & $\times$                                                         & $\boldsymbol{\times}$                                                            \\ \cline{1-9}
		Anonymity                                        & $\checkmark$                                      & $\boldsymbol{\checkmark}$                                                       & $\checkmark$                                      & $\boldsymbol{\checkmark}$                                                       & $\checkmark$                                      & $\boldsymbol{\checkmark}$                                                       & $\checkmark$                                                     & $\boldsymbol{\checkmark}$                                                        \\ \cline{1-9}
		Neutrality                                       & $\checkmark$                                      & $\boldsymbol{\checkmark}$                                                       & $\checkmark$                                      & $\boldsymbol{\checkmark}$                                                       & $\checkmark$                                      & $\boldsymbol{\checkmark}$                                                       & $\checkmark$                                                     & $\boldsymbol{\checkmark}$                                                        \\ \cline{1-9}
		Consistency                                      & $\checkmark$                                      & $\boldsymbol{\checkmark}$                                                       & $\checkmark$                                      & $\boldsymbol{\checkmark}$                                                       & $\checkmark$                                      & $\boldsymbol{\checkmark}$                                                       & $\checkmark$                                                     & $\boldsymbol{\checkmark}$                                                        \\ \cline{1-9}
		\textit{Pro-affordability}                       & N.A.                                              & $\boldsymbol{\checkmark}$                                                       & N.A.                                              & $\boldsymbol{\times}$                                                           & N.A.                                              & $\boldsymbol{\checkmark}$                                                       & N.A.                                                             & $\boldsymbol{\checkmark}$                                                        \\ 
	\end{tabularx}}
	\caption{Axiomatic properties of welfare-maximizing PB rules. \textcolor{orange}{Orange} color is used to mark the axioms and rules that are the extensions of those existing in the literature on \textcolor{orange}{dichotomous} preferences. In particular, $R_f$ denotes the PB rule under dichotomous preferences \cite{talmon2019framework}, whereas $\langle MT,f \rangle$ denotes our rule which extends $R_f$ to allow for weakly ordinal preferences (Sec. \ref{sec: mt}). \textcolor{blue}{Blue} color is used to mark the axioms and rules that extend the ones existing in the literature on \textcolor{blue}{strictly ordinal} preferences. In particular, CC denotes the multi-winner voting rule under strictly ordinal preferences \cite{chamberlin1983representative,elkind2017properties} and the PB-CC rule is our extension of the CC rule which allows the projects to have costs and the preferences to be weakly ordinal (Sec. \ref{sec: o-re-pbcc}). Orange-colored axioms are not applicable to the CC rule since the latter does not consider costs or budget. Blue-colored axioms are not applicable to the $R_f$ rules since a $R_f$ does not permit any ranking of projects. All axioms are applicable to our $\langle MT,f \rangle$ and PB-CC rules, since weakly ordinal preferences capture dichotomous and strictly ordinal preferences as special cases.}
	\label{tab: o-re-results}
\end{table*}
\subsubsection{Conclusions from the Axiomatic Analysis}\label{sec: axiomtakeaway}
As seen in \Cref{tab: o-re-results}, our proposed extensions preserve most properties satisfied by their corresponding parent rules. It is important to highlight that our dichotomous translation rules $\langle MT,f \rangle$, despite originating from a PB model with dichotomous preferences, outperform the CC and PB-CC rules in relation to axioms based on models with strictly ordinal preferences. Interestingly, the PB-CC rule, derived from the multi-winner voting model under strictly ordinal preferences, surpasses the existing dichotomous PB rules in relation to the axioms based on PB under dichotomous preferences. In other words, our orange-colored rules excel in relation to the blue-colored axioms, while our blue-colored rule performs exceptionally well with respect to the orange-colored axioms (in fact, even outperforming rules of the same color). This solidifies the significance of each extension: rankings play a vital role in the dichotomous translation rules, and similarly, costs, budget, and indifferent preferences between projects significantly influence the PB-CC rule. Notably, the PB-CC rule emerges as the most favorable in terms of satisfying axioms among the four rules examined.

\section{Fair Rules}\label{sec: o-re-fair}
Two fairness notions exist in the literature for indivisible PB under weakly ordinal preferences: Comparative Proportional Solid Coalitions (CPSC) and Inclusion Proportional Solid Coalitions (IPSC) introduced by Aziz and Lee \cite{aziz2021proportionally}.

\subsection{CPSC and IPSC}\label{sec: cpscipsc}
We briefly introduce the fairness notions of CPSC and IPSC, and discuss their drawbacks. We start by defining some necessary terminology to understand the definitions of CPSC and IPSC.
\begin{definition}[\textbf{Generalized Solid Coalition} \cite{dummett1984voting}]
    Given a PB instance \fullrinstance, we say that a subset of voters $N'$ \textbf{solidly supports} a set $P' \subseteq \proj$ if for any $i \in N'$ and $p' \in P'$, it holds that $p' \succeq_i p$ for every $p \in \proj \setminus P'$. Such a pair $(N',P')$ is said to be a \textbf{generalized solid coalition}.
\end{definition}
In other words, a set of projects $P'$ is solidly supported by a set of voters $N'$ if every voter in $N'$ prefers any project in $P'$ at least as much she prefers a project outside $P'$. We now define a \emph{periphery}, which is the collection all projects ranked equal to some project in $P'$ by some voter in $N'$. An example is illustrated in \Cref{fig: periphery}.
\begin{definition}[\textbf{Periphery} \cite{aziz2021proportionally}]\label{def: periphery}
    Given a PB instance \fullrinstance, \textbf{periphery} of a generalized solid coalition $(N',P')$ is the set of all projects ranked equal to some project in $P'$ by some voter in $V'$. In other words, periphery of $(N',P')$, denoted by $\boldsymbol{per(N',P')}$, is defined as $\curly{p: \exists i \in N', p' \in P' s.t.\; p \succeq_i p'}$.
\end{definition}
\begin{figure}
    \centering
    \includegraphics[width=0.7\linewidth]{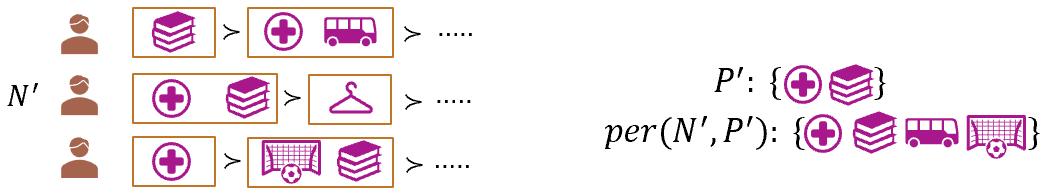}
    \caption{An example to illustrate \Cref{def: periphery}.}
    \label{fig: periphery}
\end{figure}
We are now ready to define the CPSC and IPSC notions. Intuitively, both the notions require that $\cof{S \cap per(N',P')}$ must be at least $\frac{\bud |N'|}{n}$. CPSC notion requires that, at the very least, if $\cof{S \cap per(N',P')} < \frac{\bud |N'|}{n}$, there is no subset $P'' \subseteq P'$ such that $\cof{S \cap per(N',P')} < \cof{P''} \leq \frac{\bud |N'|}{n}$.
\begin{definition}[\textbf{Comparative Proportional Solid Coalition (CPSC)} \cite{aziz2021proportionally}]\label{def: cpsc}
    Given an instance \instance, a set of projects $S$ is said to satisfy CPSC if for every generalized solid coalition $(N',P')$ and every set $P'' \subseteq P'$ such that $\cof{P''} \leq \frac{\bud |N'|}{n}$, it holds that $\cof{S \cap per(N',P')} \geq \cof{P''}$.
\end{definition}
 IPSC notion requires that, at the very least, if $\cof{S \cap per(N',P')} < \frac{\bud |N'|}{n}$, there is no project $p^* \in P'\setminus S$ such that $\cof{p^* \cup (S \cap per(N',P'))} \leq \frac{\bud |N'|}{n}$.
\begin{definition}[\textbf{Inclusive Proportional Solid Coalition (IPSC)} \cite{aziz2021proportionally}]\label{def: ipsc}
    Given an instance \instance, a set of projects $S$ is said satisfy IPSC if for every generalized solid coalition $(N',P')$ and every project $p^* \in P' \setminus S$, it holds that $\cof{p^* \cup (S \cap per(N',P'))} >\frac{\bud |N'|}{n}$.
\end{definition}
Notably, several rules in the literature, such as the Method of Equal Shares by Pierczyński et al. \cite{pierczynski2021proportional} and the Expanding Approval Rules by Aziz and Lee \cite{aziz2021proportionally} satisfy CPSC and/or IPSC.
\subsubsection{Limitations of CPSC and IPSC}\label{sec: limitations}
A drawback of both these notions is that both of them assume the preferences to be complete. Moreover, CPSC is not always guaranteed to exist \cite{aziz2021proportionally}. Another drawback is that, these notions look at coalitions of voters, not paying enough attention to individual fairness. Thus, many a time, CPSC and IPSC can unreasonably upset almost all the voters as illustrated in the below example.
\begin{example}\label{eg1: arsg}
	Let $\bud = 100$ and $n = 10$. \Cref{fig: motivation_args} below demonstrates the preferences of the voters. Let $p_1$ be the project whose cost is $10$, and $p_2$ be the be the library project costing $91$. Let all the other projects also cost $91$ each.
    \begin{figure}[!h]
    \centering
    \includegraphics[width=0.3\linewidth]{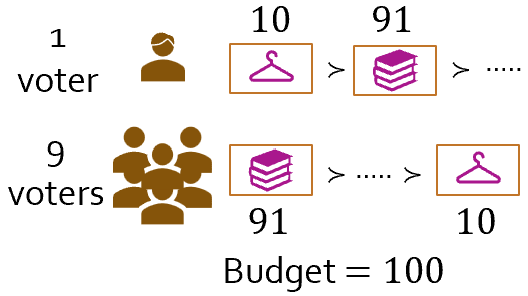}
    \repeatcaption{fig: motivation_args}{Preferences of voters in \Cref{eg1: arsg}}
    \end{figure}
    It is easy to note that $(1,p_1)$, i.e., the first voter and the project costing $10$, form a generalized solid coalition whose periphery is $\{p_1\}$. Any set $S$ that satisfies CPSC or IPSC must have $p_1$ (otherwise, \Cref{def: cpsc} does not hold for $P'' = \{p_1\}$ and  \Cref{def: ipsc} does not hold for $p^* = p_1$). Thus, $\{p_1\}$ is the only feasible set that satisfies CPSC or IPSC. However, $p_1$ is the least preferred project for $90\%$ of the voters. Moreover, $90\%$ of the budget remains unused. Whereas, the outcome $\curly{p_2}$ could be much more reasonable since it is the most preferred project for $90\%$ of the voters and the second best project for the remaining $10\%$ of them. It also makes use of $91\%$ of the budget.
\end{example}
From the above, we observe that CPSC and IPSC notions suffer from the drawback of encoding fairness as a hard constraint, which in turn may lead to sub-optimal choices. However, fairness is often not a hard constraint. For instance, in \Cref{eg1: arsg}, selecting $p_2$ is not really unfair to the group of $10$ similar voters. This is because, when the number of projects is high, it is safe to assume that the second ranked project is also satisfactory to the voters. This idea motivates the formulation of rules that \emph{optimize} fairness.

Our family of fair rules operate on two parameters: $k$ and $\theta$. The core idea of our fairness is to ensure that, on an average, at least an amount of $\theta$ is collectively allocated to the projects ranked at most $k$ by any voter. We introduce a class of rules, called Average Rank-Share Guarantee (ARSG) rules, that encapsulates this idea. It is worth mentioning that such share-based fairness notions are well studied in PB under flexible costs due to their wide relevance and flexibility (e.g., \cite{sreedurga2022characterization}). Here, we study such fairness for PB under restricted costs.

\subsection{Average Rank-Share Guarantee Rules}\label{sec: arsg}
We introduce the class of ARSG rules which offers the PB organizer a high degree of freedom to define fairness. The idea, as explained above, is to ensure that on an average, for each voter, at least a share of $\theta$ is allocated to the projects ranked at most $k$ in her preference. Before defining the rules formally, we explain a notation needed to understand the rules. Given a value $\theta \in \curly{1,2,\ldots,\bud}$ (called \emph{share}) and a set $S \subseteq A$, let us define \trunkk{\ebpara}{S} as the smallest \emph{rank} $k$ such that an amount \ebpara is allocated to projects in $S$ ranked at most $k$ by the voter $i$:
\begin{align*}
	\trunkk{\ebpara}{S} = \begin{cases} 
		\min \Big\{j : c\big(S \cap \bigcup\limits_{q \leq j}\pat{q}\big) \geq \ebpara\Big\} & \text{if }\cof{S} \geq \ebpara \\
		m+1 & \text{otherwise} 
	\end{cases}
\end{align*}
The higher \trunkk{\ebpara}{S} is, the farther is the rank before which at least the share of \ebpara is allocated for the voter $i$. This is further illustrated below with an example.
\begin{example}\label{eg2: arsg}
	Let $\bud = 12$ and $A = \{p_1,p_2,p_3,p_4,p_5\}$ with costs $4,2,5,3,$ and $2$ respectively. Let us suppose $\ebpara = 7$. Suppose there are two voters whose preferences are as below.
	\begin{table}[h!]
        \centering
		\delimitershortfall=0pt
		\setlength{\dashlinegap}{1pt}
		\scalebox{0.9}{
			\begin{tabular}{ccc!{\color{orange}\vrule}cccc}
				${\textcolor{orange}{p_1}}$ & $\succ_1$ & $\{{p_2},\textcolor{orange}{p_4}\}$ & $\succ_1$ & ${\textcolor{orange}{p_3}}$\\
				$4$ & & $2,\;3$ & & $5$\\
		\end{tabular}}\\
		~\\
		\scalebox{0.9}{
			\begin{tabular}{c!{\color{orange}\vrule}cccc}
				$\{{\textcolor{orange}{p_3}},\textcolor{orange}{p_4}\}$ & $\succ_2$ & ${\textcolor{orange}{p_1}}$ & $\succ_2$ & $p_5$\\
				$5,\;3$ & & $4$ & & $2$\\
		\end{tabular}}
	\end{table}
	{A set $S = \curly{p_1,p_3,p_4}$ is represented by orange color. Then for voter $1$, the rank before which the share of \ebpara is guaranteed is $\trunks_1(7,S) = 2$ since $\cof{p_1}+\cof{p_4} \geq 7$ and the rank of $p_4$ is $2$. Ranks before which the shares are allocated are marked for both the voters: $\trunks_1(7,S) = 2$ and $\trunks_2(7,S) = 1$.}
\end{example}
We are now ready to define two families of ARSG rules: Average Rank Guarantee (ARG) rules and Share Guarantee (SG) rules.
\subsubsection{Average Rank Guarantee Rules}\label{sec: o-re-arg}
An average rank guarantee (ARG) rule takes a parameter $k$ (a rank) and selects an outcome which, on average, maximizes the share guaranteed to be allocated to the projects ranked at most $k$ by each voter. The intuition behind the rule is that the PB organizer considers the top $k$ ranks to be satisfactory for a voter (e.g., if there are $100$ projects, an outcome that allocates high amount to the top $3$ ranked projects may be safely assumed to be reasonably fair towards the voter).

The rule is clearly explained with the pseudocode in Algorithm \ref{algo: arg}. Condition in the line 2 of Algorithm \ref{algo: arg} calculates the optimal sum of ranks of all the voters before which a share of \ebpara is allocated and checks if it is at most $kn$. That is, the condition checks if it is possible, on an average, to allocate a share of \ebpara to the top $k$ projects of all the voters. Lines 1 and 3 of the code together find the maximum \ebpara such that this condition in line 2 holds. Line 4 outputs all the optimal subsets for such a \ebpara. This is also illustrated in \Cref{eg3: arsg}.
\begin{algorithm}[h!]
	\DontPrintSemicolon
	\KwIn{An ordinal PB instance \fullrinstance, a rank $k \in \curly{1,2,\ldots,m}$}
	\KwOut{A feasible subset of projects $S$}
	\For{$\ebpara = L$ \texttt{to} $1$}{
		\If{$\min\limits_{S \in \feasible}{\suml{i \in \voters}{\trunkk{\ebpara}{S}}} \leq kn$}{
			break;\;
		}
	}
	$X \gets \{S^*: S^* \in \argmi{S \in \feasible}{\suml{i \in \voters}{\trunkk{\ebpara}{S}}}\}$;\;
	\Return{$X$}\;
	\caption{ARG rule}
	\label{algo: arg}
\end{algorithm}
\begin{example}\label{eg3: arsg}
	Consider the \Cref{eg1: arsg}. Let us consider the average rank guarantee (ARG) rule with the rank parameter given by $k = 2$. That is, the rule considers the top two ranks to be satisfactory ranks and tries to maximize the share guaranteed within the top two ranks to all the voters.
	
	We start with $\ebpara = \bud = 100$, as given on line 1 of Algorithm \ref{algo: arg}. If $S = \curly{p_1}$, $\trunkk{\ebpara}{S} = m+1$ for every voter $i$, where $m$ is the number of projects. Assuming that $m > 2$, $\sum_{i \in \voters}{\trunkk{\ebpara}{\curly{p_1}}} > 3n$. If $S = \curly{p_2}$, $\trunks_1(\ebpara,S) = 2$ and $\trunkk{\ebpara}{S} = 1$ for every other voter. Therefore, $\sum_{i \in \voters}{\trunkk{\ebpara}{\curly{p_2}}} < 2n$. Thus, the condition on the line 2 of Algorithm \ref{algo: arg} is satisfied and we see that the for loop is broken for $\ebpara = L$ itself. As seen above, the set $S=\curly{p_2}$ is that set that minimizes $\sum_{i \in \voters}{\trunkk{100}{S}}$ and hence $X = \curly{\curly{p_2}}$ is the outcome of the ARG rule with rank parameter $k = 2$. Notably, it can be checked that $X$ continues to be the outcome for the ARG rule with any parameter $k$ (When $k = 1$, no \ebpara satisfies condition on line 2 in Algorithm \ref{algo: arg} and hence the loop ends with $\ebpara = 1$. Even then, $\curly{p_2}$ continues to be the unique set that minimizes $\sum_{i \in \voters}{\trunkk{1}{S}}$).
\end{example}
\subsubsection{Share Guarantee Rules}\label{sec: o-re-sg}
A share guarantee (SG) rule takes a parameter $\ebpara$ (called share) and selects an outcome which minimizes the average rank before which the projects are guaranteed to be allocated an amount of \ebpara for every voter. The intuition behind it is that the PB organizer considers the amount $\ebpara$ to be the share of budget required to make a voter happy and tries to minimize the rank within which this share is allocated to each voter.

Clearly, given a share parameter \ebpara, the optimal sum of ranks before which \ebpara gets allocated for every voter is given by $\min\limits_{S \in \feasible}\sum\limits_{i \in \voters}\trunkk{\ebpara}{S}$. Also, the minimum average rank before which \ebpara is guaranteed to get allocated for every voter must be $\large\lceil{\frac{\min\limits_{S \in \feasible}\sum\limits_{i \in \voters}\trunkk{\ebpara}{S}}{n}}\large\rceil$ and the corresponding set $S$ achieving this value must be included in the outcome. We formally define the SG rule below and illustrate it in \Cref{eg4: arsg}.
\begin{definition}\label{def: sg}
	A share guarantee (SG) rule with a parameter $\ebpara \in \curly{1,2,\ldots,\bud}$ selects every feasible subset $S \in \feasible$ that minimizes $\sum_{i \in \voters}{\trunkk{\ebpara}{S}}$.
\end{definition}

\begin{example}\label{eg4: arsg}
	Consider the \Cref{eg1: arsg}. For any share parameter $\ebpara \geq 1$, it can be observed that $\trunks_1(\ebpara,\curly{p_2}) = 2$ and $\trunkk{\ebpara}{\curly{p_2}} = 1$ for every other voter $i$. Thus, the SG rule outputs $\curly{\curly{p_2}}$.
	
	Similarly, consider the \Cref{eg3: arsg}. The set $S$ has the optimum value of $\sum_{i \in \voters}{\trunkk{7}{S}}$, equal to $3$ (Note that $\trunks_1(7,S) > 1$ for any set since $\cof{p_1} = 4$).
\end{example}
\subsubsection{Computational Complexity}\label{sec: faircc}
We study the computational complexity of average rank-share guarantee rules. We start by making a small yet important observation. The next theorem follows from this observation and \Cref{the: pbcc}.
\begin{observation}\label{obs: faircc}
	The SG rule with share parameter $\ebpara = 1$ is equivalent to the PB-CC rule.
\end{observation}
\begin{lemma}\label{the: sgcc}
	Deciding the SG rule is \WWH with respect to its share parameter \ebpara.
\end{lemma}
Now, let us again look at Algorithm \ref{algo: arg}. SG rule needs to be solved at line 2 and 4 of Algorithm \ref{algo: arg}. However, line 1 could further be optimized significantly. For example, we could use binary search instead of linear search since the function $\sum_{i \in \voters}{\trunkk{\ebpara}{S}}$ is monotonic in \ebpara. Thus, executing the ARG rules requires executing at least one, and at most $O(\log \bud)$ SG rules as subroutines. The next result follows from \Cref{the: sgcc}.
\begin{lemma}\label{the: argcc}
	Deciding the ARG rule is \WWH with respect to its rank parameter $k$.
\end{lemma}
\section{Conclusion and Discussion}\label{sec: o-re-conclusion}
We bridge the gap in the indivisible PB literature by studying weakly ordinal preferences (we also allow the preferences to be incomplete). We propose PB rules that maximize utilitarian welfare and those that achieve fairness.

\paragraph*{Welfare maximization.} For the welfare maximization, we extend the existing rules in the literature. The rules maximizing utilitarian welfare are only defined in the literature for some special cases of our model: PB under dichotomous preferences and multi-winner voting under strictly ordinal preferences. We extend the former rules to allow for incomplete weakly ordinal preferences (\Cref{sec: o-re-dtr}) and the latter to allow for weakly ordinal preferences, the projects having costs, and imposition of a budget constraint (\Cref{sec: o-re-pbcc}). We justify the significance and novelty of our extensions using an axiomatic analysis by proving that our extensions enhance the axiomatic properties. 

As depicted in \Cref{tab: o-re-results}, our proposed extensions preserve the properties exhibited by their respective parent rules. Moreover, it is worth highlighting that our dichotomous translation rules $\langle MT,f \rangle$ surpass the CC and PB-CC rules in terms of the axioms derived from models with strictly ordinal preferences, despite originating from a model with dichotomous preferences. Conversely, the PB-CC rule, derived from the multi-winner voting model under strictly ordinal preferences, outperforms the existing dichotomous PB rules in relation to the axioms based on PB under dichotomous preferences. Essentially, our orange-colored rules exhibit exceptional performance in relation to blue-colored axioms, and our blue-colored rule excels with respect to orange-colored axioms (in fact, their performance even surpasses rules of the same color). This underscores the significance of each extension: rankings play a crucial role in the dichotomous translation rules, while costs, budget, and indifference preferences between projects greatly influence the PB-CC rule. It is worth noting that among the four rules studied, the PB-CC rule emerges as the most favorable axiomatically.

\paragraph{Fairness.} Unlike welfare maximization, fairness for PB under weakly ordinal preferences is already studied in the literature \cite{aziz2021proportionally}. However, we identify some major drawbacks suffered by the existing fairness notions (\Cref{sec: limitations}) and propose two families of rules that fill this gap (Examples \ref{eg3: arsg} and \ref{eg4: arsg}). Rank-share based fairness (guaranteeing a share of \ebpara within the first $k$ ranks) is extensively studied in PB under totally flexible costs \cite{aziz2014generalization,aziz2018rank,airiau2019portioning,sreedurga2022characterization}. Here, we study rank-share based fairness for PB under restricted costs and these results are in fact, the first to focus exclusively on guaranteeing fairness at individual level for each voter (much like several fairness axioms in the fair division literature). Our perspective on fairness also offers PB organizer the flexibility to choose the parameters $k$ and \ebpara.

\paragraph{Future work.} This the first step to systematically study indivisible PB under weakly ordinal preferences, which leaves a lot of room for future directions. Welfare maximization is studied for more sophisticated models of PB with dichotomous preferences (projects grouped into categories \cite{jain2020participatory,jain2020participatoryg} etc). Also, there are rules other than the CC rule (e.g., positional scoring rules) for multi-winner voting under ordinal preferences that maximize utilitarian welfare. An interesting direction would be to extend these rules to our model and check if the new rules perform satisfactorily w.r.t. all the axioms. From the fairness point of view, since the proposed ARSG rules are proved to be computationally intractable, it will be interesting to find tractable special cases with structured preferences.

Throughout this chapter, we assumed that the cost of a project is restricted to only one value. In the next chapter, we study the model where the cost of each project is totally flexible and can take any value.
\blankpage
\chapter{Flexible Costs: Characterization of Group-Fair and Individual-Fair Rules under Single-Peaked Preferences}\label{chap: o-fl}

\begin{quote}
\textit{We delve into the concept of fairness in participatory budgeting when project costs are totally flexible. In this scenario, participatory budgeting can be seen as a random social choice problem, where a voting rule outputs a probability distribution over the set of projects indicating the fraction of budget allocated to each project. We assume that the preferences of voters belong to a special domain of ordinal preferences, referred to as single-peaked domain in the literature. In the single-peaked domain, construction and characterization of random social choice rules that satisfy properties such as unanimity and strategy-proofness have been extensively studied in prior works.}

\textit{Expanding on the current research, we include fairness considerations in the study of random social choice. This chapter specifically addresses fairness among groups of voters. We consider an existing partition of the voters into logical groups, based on natural attributes such as race and location.  To capture intra-group fairness, we introduce the concept of group-wise anonymity. To ensure fairness across groups, we propose two novel notions: weak group entitlement guarantee (weak-GEG) and strong group entitlement guarantee (strong-GEG). We provide two separate characterizations of  random social choice rules that satisfy group-fairness: (i) direct characterization (ii) extreme point characterization (expressed as probability mixtures of rules that allocate probability $1$ to a single project).}
\end{quote}

\section{Introduction}\label{sec: o-fl-intro}
Participatory budgeting (PB) involves allocating an available budget to a set of projects by aggregating the preferences of voters over the projects. Dichotomous preferences allow the voters to report a set of projects they like. Alternatively, ordinal preferences allow a voter to rank the projects according to the value she has for each of them. Ordinal preferences turn out to be the most-liked preference elicitation method by the voters due to their expressibility and cognitive simplicity \cite{benade2018efficiency}. 

In this chapter, we focus on ordinal preferences and the case where the project costs are totally flexible. In other words, projects are not associated with fixed costs and a PB rule may allocate \emph{any amount} to each project (\Cref{sec: prelims-costs}). For example, the projects could be to award scholarships to different researchers or to fund different environment-friendly causes. Any amount, however little or large, can be allocated to each of these causes. This model is also referred to as \emph{bounded divisible PB, portioning,} and \emph{fair mixing} in the literature \cite{aziz2021participatory}.

The outcome of a PB rule when the costs are totally flexible shall be a $m$-sized vector (recall that $m$ is the number of projects), such that all entries in the vector sum up to the value of available budget. The budget can be normalized to be $1$. After such a normalization, the outcome of a PB rule shall be a vector whose entries are fractions between $0$ and $1$, such that sum of all the entries is $1$. Clearly, such an output can be viewed as a probability distribution over the set of projects. This makes our model equivalent to \emph{random social choice} model, which is a probabilistic version of \emph{deterministic social choice} model. In deterministic social choice, preferences of the voters are aggregated to output one of the available alternatives. Contrarily, random social choice outputs a probability distribution over the available alternatives, and is thus equivalent to our model where projects are considered to be the alternatives.

In the conventional social choice literature, there are two fundamental properties that are considered essential and non-negotiable: unanimity and strategy-proofness. Unanimity entails that if all voters rank a particular project as their top choice, that project must be selected with a probability of 1. Failing to uphold this property can undermine the persuasiveness and credibility of the rule. Alternately, strategy-proofness requires that no voter can manipulate the outcome by misrepresenting their preferences to achieve a more favorable result. This property is regarded as indispensable, especially in settings where strategic behavior is prevalent. Unfortunately, unanimity and strategy-proofness have been proven to be incompatible in both deterministic and random social choice, unless the rules are dictatorships \cite{gibbard1973manipulation,satterthwaite1975strategy,gibbard1977manipulation}.

The incompatibility of unanimity and strategy-proofness prompted economists to identify different structures on the ordinal preferences for which the two properties become compatible. During this pursuit, seminal work by Black \cite{black1948rationale} introduced \emph{single-peaked domain}. Single-peakedness is an inherent structure on ordinal preferences that is often naturally exhibited in scenarios where alternatives can be ordered based on their intrinsic characteristics, such that every voter prefers an alternative closer to her top-ranked alternative over an alternative farther from it. For example, political parties can be ordered based on their ideology \cite{black1948rationale}, products based on their utility \cite{barbera1997strategy,bade2019matching}, and public facilities based on their target audience and locations \cite{bochet2012priorities}. Tideman \cite{tideman2017collective} accessed numerous instances of  real-life ordinal preference ballots and found out that most of them were single-peaked.

Wide applicability of the single-peaked domain led to its extensive study in social choice \cite{fotakis2016conference,brandt2012computational,sprumont1991division,barbera1997strategy}.  In this domain, Moulin \cite{moulin1980strategy} established that all deterministic social choice rules that are unanimous and strategy-proof can be classified as min-max rules, while those that additionally satisfy anonymity are known as median rules. Building upon this foundation, Ehlers et al. \cite{ehlers2002strategy}, Peters et al. \cite{peters2014probabilistic}, and Pycia and Ünver \cite{pycia2015decomposing} provided characterizations for all random social choice rules adhering to these properties. This chapter views random social choice as participatory budgeting model, introduces novel fairness notions based on this perspective, and characterizes all the fair rules that satisfy unanimity and strategy-proofness.
\begin{note}
    Throughout the chapter, for the sake of uniformity in the thesis, we present the existing results on deterministic and random social choice also in terms of participatory budgeting. Particularly, we refer to the alternatives as \emph{projects}.
\end{note}

\section{Prior Relevant Work}\label{sec: o-fl-lit}
In this section, we summarize the existing results on individual-fairness and group-fairness in random social choice and emphasize the gaps in the literature.

Under dichotomous preferences, each voter $i$ reports a set $A_i$ of projects she likes. For such a model, Bogomolnaia et al. \cite{bogomolnaia2005collective} introduced two notions of individual-fairness: fair outcome share and fair welfare share. Fair outcome share ensures that each set $A_i$ is assigned a probability of at least $|A_i|/m$, where $m$ is the number of projects. Fair welfare share, on the other hand, guarantees that $A_i$ receives a probability of at least $1/n$, where $n$ is the number of voters. The concept of fair welfare share is also known as individual fair share \cite{aziz2019fair}. Fair outcome share and fair welfare share have been extended to weakly ordinal preferences by Aziz and Stursberg \cite{aziz2014generalization} and Aziz et al. \cite{aziz2018rank}, respectively. Additionally, Airiau et al. \cite{airiau2019portioning} extended fair welfare share to strictly ordinal preferences. All these studies primarily focus on fairness at the level of individual voters.

The study of fairness for groups of voters has been rather limited in random social choice. The first kind of group-fairness considered is proportional sharing, studied by Duddy \cite{duddy2015fair} for dichotomous preferences and Aziz et al. \cite{aziz2018rank} for weakly ordinal preferences. This notion requires that for any subset $S$ of voters, the union of the most-liked projects of each of them receives a probability of at least $|S|/n$. The second kind of group-fairness considered is the core, which is primarily studied for cardinal preferences but can be applied both for dichotomous and ordinal preferences \cite{aziz2019fair,fain2016core}. The core guarantees that for any subset $S$ of voters, there will not exist a partial distribution of $|S|/n$ (i.e., probabilities sum to $|S|/n$ instead of to $1$) which gives strictly higher utility to the voters in $S$. Both core and proportional sharing, when applied to strictly ordinal preferences, simply reduce to random dictatorship rule: a rule that selects a dictator from the voters uniformly at random and assigns probability of $1$ to her most preferred project.

\subsection{Limitations of Random Dictatorship}\label{sec: rdlimitations}
To understand one of the drawbacks of random dictatorship rule, consider the following simple example.
\begin{example}\label{eg: rdrule}
    Suppose there are $5$ voters and $10$ projects ordered as $p_1 \triangleleft \ldots \triangleleft p_{10}$. Let us suppose voters $\{1,2\}$ have the same preference $p_3
    \succ p_2 \succ p_4\ldots \succ p_{10} \succ p_1$, voter $3$ has the preference $p_2\succ p_3\succ p_4\ldots \succ p_{10} \succ p_1$, whereas the voters $\{4,5\}$ have the same preference $p_1 \succ p_2\ldots \succ p_{10}$. The random dictatorial rule allocates $\frac{2}{5}$ probability to $p_1$ and $p_3$ each, and $\frac{1}{5}$ probability to $p_2$. As can be seen, $p_1$ is the least preferred project by two of the voters whereas $p_2$ is in the top two ranks of all the five preferences. Clearly, it is more desirable to  allocate higher probability to $p_2$ compared to $p_1$ from a societal viewpoint.
\end{example}
The above example motivates the idea of considering an existing partition of the voters into groups and having a set of \emph{representatives} for each group, instead of distributing the entire probability among only the top ranked projects of the voters. For instance, suppose the voters in \Cref{eg: rdrule} are naturally partitioned into two groups $\{1,2,3\}$ and $\{4,5\}$ and the voters in the same group have closely related preferences. Suppose $\{p_2,p_3\}$ are the representatives of the first group and $\{p_1,p_2,p_3\}$ are the representatives of the second group. Then, an outcome that distributes the probability between $p_2$ and $p_3$ will be collectively fairer towards all the voters since they will likely be happy with a probability close to, if not equal to, $1$.

As clear from the above discussion, imposing a fairness constraint on every subset of voters is a strong requirement that could lead to trivial and undesirable outcomes for strictly ordinal preferences. Often in real-world, we can find a natural partition of voters into groups based on factors such as gender, race, economic status, and location. It will be sensible and adequate to guarantee fairness, both within and across these existing groups.

Another serious drawback of the random dictatorial rule is its assumption of entitlement of exactly $\frac{1}{n}$ for each voter and exactly $\frac{|S|}{n}$ for every subset of voters $S$. However, in many real-world scenarios, the entitlement could be different for different groups. For example, the government may want underprivileged sections of society to have a suppose on higher fraction of budget. Such measures are often referred to with names such as affirmative actions, reservations etc. (\burl{https://en.wikipedia.org/wiki/Affirmative_action}).

We assume that a partition of the voters exists and introduce novel fairness notions overcoming both these drawbacks. We then characterize the fair rules satisfying other desiderata.
\section{Contributions of the Chapter}\label{sec: o-fl-contri}
We study the model where voters are naturally partitioned into groups. Our first key contribution is to propose three notions that capture fairness for groups of voters. To ensure fairness within each group or intra-group fairness, we propose \emph{group-wise anonymity}, which implies that the voters within any given group are treated symmetrically. To ensure fairness across groups, we first propose a weak notion of fairness, weak group entitlement guarantee (\emph{weak-GEG}) followed by a stronger variant (\emph{strong-GEG}) of the same. The PB organizer gets to choose three parameters: (1) the number of representatives to be selected for each group ($\kappa$); (2) a method of selecting them ($\psi$); and (3) the \resquota of each group ($\eta$). For every group $q$, the function $\psi_q$ selects some $\kappa_q$ projects to represent the preferences of voters in $q$. Our weak-GEG notion ensures that all these representatives collectively receive a probability of at least $\eta_q$, whereas strong-GEG ensures that at least one project from the $\kappa_q$ representatives receives a probability of at least $\eta_q$ (demonstrated in \Cref{fig: geg}).
\begin{figure}[!ht]
    \centering
    \includegraphics[width=0.12\linewidth, height=4.3cm]{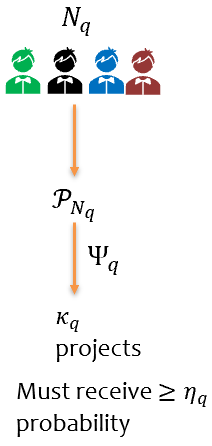}
    \caption{Illustration of Group-Entitlement Guarantee. $\rprof_{N_q}$ denotes the preference profile of the voters in the group $N_q$. The function $\psi_q$ selects $\kappa_q$ representative projects which together have an entitlement over $\eta_q$ probability.}
    \label{fig: geg}
    \end{figure}
    
Our second key contribution is to characterize the space of unanimous and strategy-proof social choice rules, both deterministic and random, that satisfy the proposed fairness notions. Random Social Choice Rules (RSCRs) are typically expressed in two ways in the literature. One way is to express them directly \cite{chatterji2022probabilistic,morimoto2022group} and the other is to express them as probability mixtures of deterministic social choice rules (DSCRs) which allocate the entire probability to a single project \cite{brand2023strategyproof}. It is crucial to emphasize that both these representations are incomparable, and the choice between them is purely a matter of personal preference and convenience. We present characterizations of fair RSCRs in both these ways: \emph{direct characterizations}  and \emph{extreme point characterizations} in which RSCRs are expressed as convex combinations of DSCRs.  To the best of our knowledge, we are the first ones to provide a complete characterization of group-fair rules on single-peaked domain. The characterization also serves as a foundation for identifying families of instances with algorithmically tractable fair rules. A few such families, along with the corresponding fair rules, are mentioned in \Cref{sec: o-fl-examples}.
\begin{figure}[h]
    \begin{subfigure}{.45\textwidth}
      \begin{tikzpicture}
        \draw [black] (-2.8,1.25) rectangle (3.75,-1.25) 
        node at (2.5,0.9) {\textbf{Strong-GEG}};
        \node at (0.5,0.2) {\textcolor{red}{\Cref{sec: dc_sfgwa}}};
        \node at (0.5,-0.3) {\textcolor{brown}{\Cref{sec: ep_sfgwa}}};

        \draw [black] (-3.0,3.75) rectangle (4.2,-1.6) 
        node at (3.1,3.45) {\textbf{Weak-GEG}};
        \node at (0.5,2.7) {\textcolor{red}{\Cref{sec: dc_wfgwa}}};
        \node at (0.5,2.2) {\textcolor{brown}{\Cref{sec: ep_wfgwa}}};

        \draw [black] (-3.2,7.75) rectangle (4.5,-1.95) 
        node at (1.38,7.40) {\textbf{Unanimity, Strategy-Proofness,}};
        \node at (2.08,6.85) {\textbf{Group-Wise Anonymity}};
        \node at (0.5,5.8) {\textcolor{red}{\Cref{sec: dc_gwa}}};
        \node at (0.5,5.3) {\textcolor{brown}{\Cref{sec: ep_gwa}}};
    \end{tikzpicture}
    \caption{Organization of results on the characterization of rules satisfying \textbf{\textit{group-fairness}}}
    \label{fig: gforganization}
    \end{subfigure}
    \hfill
    \begin{subfigure}{.45\textwidth}
      \begin{tikzpicture}
        \draw [black] (-2.8,1.25) rectangle (3.75,-1.25) 
        node at (2.57,0.9) {\textbf{Strong-IEG}};
        \node at (2.3,-0.3) {\textcolor{red}{\Cref{sec: dc_sieg}}};
        \node at (2.3,-0.8) {\textcolor{brown}{\Cref{sec: epc_sieg}}};

        \draw [black] (-3.0,3.75) rectangle (4.2,-1.6) 
        node at (3.17,3.45) {\textbf{Weak-IEG}};
        \node at (2.8,2.7) {\textcolor{red}{\Cref{sec: dc_wieg}}};
        \node at (2.8,2.2) {\textcolor{brown}{\Cref{sec: epc_wieg}}};

        \draw [black] (-3.2,7.75) rectangle (4.5,-1.95) 
        node at (1.40,7.40) {\textbf{Unanimity, Strategy-Proofness}};
        \node at (3.0,5.8) {\textcolor{red}{\Cref{sec: un_sp_dc}}};
        \node at (3.0,5.3) {\textcolor{brown}{\Cref{sec: un_sp_ep}}};

        \draw [black] (-2.5,5.75) rectangle (1.25,0) 
        node at (0.1,5.45) {\textbf{Anonymity}};
        \node at (-0.5,4.5) {\textcolor{brown}{\Cref{sec: median}}};
        \node at (-0.5,2.5) {\textcolor{brown}{\Cref{sec: wiegta}}};
        \node at (-0.5,0.72) {\textcolor{brown}{\Cref{sec: siegta}}};
    \end{tikzpicture}
    \caption{Organization of results on the characterization of rules satisfying \textbf{\textit{individual-fairness}}}
    \label{fig: iforganization}
    \end{subfigure}
\caption{Demonstration of organization of the chapter. Note that when each group is a singleton group and has only one voter, the rectangles in \Cref{fig: gforganization} coincide with the corresponding rectangles in \Cref{fig: iforganization}. The sections colored in \textcolor{red}{red} represent direct characterizations and those colored in \textcolor{brown}{brown} represent extreme point characterizations.}
\label{fig: organization}
\end{figure}
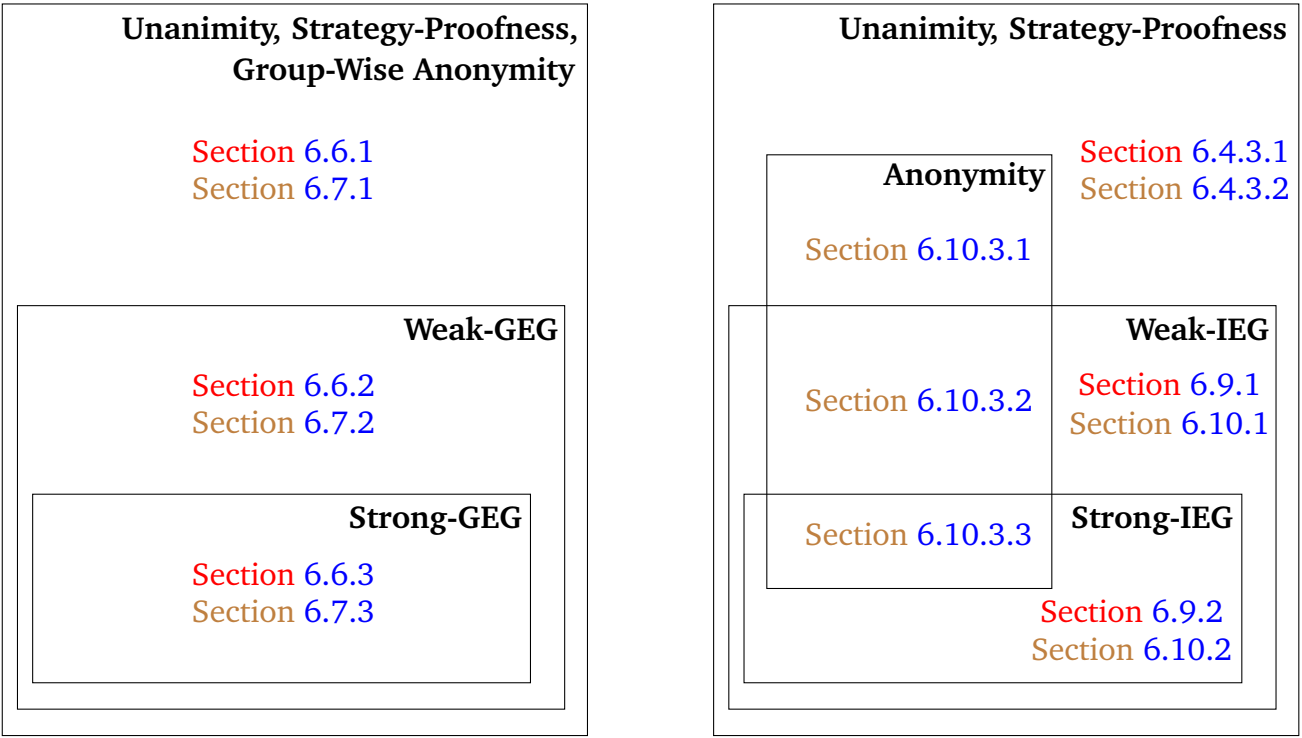
\paragraph{\textbf{Organization of the chapter.}} We commence by introducing essential preliminaries in \Cref{sec: o-fl-prelims} and describing the existing characterizations of unanimous and strategy-proof RSCRs. Following this, in \Cref{sec: o-fl-notions}, we introduce our model with the voters partitioned into groups, explain the parameters related to each group, and introduce our novel fairness notions. 
Subsequently, we provide direct (\Cref{sec: o-fl-direct}) and extreme point (\Cref{sec: o-fl-extreme}) characterizations of unanimous and strategy-proof RSCRs that satisfy these fairness notions.

Following this, in \Cref{sec: o-fl-indnotions}, we proceed to a special case of our model where there are no groups, or in other words, each group has exactly one voter. We propose the fairness notions \textit{weak-IEG} and \textit{strong-IEG} to guarantee individual-fairness and provide the characterizations of the unanimous and strategy-proof rules that satisfy these notions. We also discuss the case where, in addition to groups having only voter, anonymity holds across all the voters. Finally, we provide illustrative examples of some algorithmically tractable fair RSCRs for some families of instances (\Cref{sec: o-fl-examples}).
\section{Notations and Preliminaries}\label{sec: o-fl-prelims}
Let $N=\{1,\ldots,n\}$ be a finite set of voters. We assume that $n \geq 2$. Let $\proj =\{p_1,\ldots,p_m\}$ be a finite set of projects with a prior ordering $\triangleleft$ given by $p_1\triangleleft  \ldots \triangleleft p_m$. Such a prior ordering exists with respect to an intrinsic attribute of the projects such as location. Whenever we write minimum or maximum of a subset of \proj, we mean it with respect to the ordering $\triangleleft$ over \proj. By $p\trianglelefteq p'$, we mean $p=p'$ or $p \triangleleft p'$. For $p, p' \in \proj$, we define $[p,p']=\{\tilde{p} \mid \mbox{ either } p \trianglelefteq \tilde{p} \trianglelefteq p' \mbox{ or } p' \trianglelefteq \tilde{p} \trianglelefteq p\}$. By $(p,p')$, we define $[p,p'] \setminus \{p,p'\}$. For notational convenience, whenever it is clear from the context, we do not use braces for singleton sets, i.e., we denote $\{i\}$ by $i$.

The budget \bud is normalized to $1$. There are no costs associated with projects, that is, $\domainof{p} = (0,1]$. A project $p$ may receive any amount from $\domainof{p} \cup \{0\}$, i.e., from $[0,1]$.

\subsection{Preferences}\label{sec: pref}
A complete, reflexive, anti-symmetric, and transitive binary relation (also called a complete linear order) on a(ny) set $S$ is called a strictly ordinal preference on $S$. We denote by \comporder the set of all strictly ordinal preferences on \proj. For $\succ \in \comporder$ and $p,p' \in \proj$, $p \succ p'$ is interpreted as $p$ being strictly preferred over $p'$ according to $\succ$. For $\succ \in \comporder$ and $k \leq |\proj|$, by $\opat{}{k}$ we refer to  the $k$-th ranked project according to $\succ$, i.e., $\opat{}{k} = p$ if and only if $|\{p' \in \proj \mid p'\succ p \}|=k-1$. For $\succ \in \comporder$ and $p \in \proj$, the \textit{upper contour set} of $p$ at $\succ$, denoted by $U(p,\succ)$, is defined as the set of projects that are as good as $p$ in $\succ$, i.e., $U(p,\succ) = \{p\} \cup \{p' \in \proj \mid p'\succ p\}$.\footnote{Observe that $p\in U(p,\succ)$ by reflexivity.}
\begin{definition}\label{def: sp}
	A preference $\succ \in \comporder$ is called \textbf{single-peaked} if for all $p,p'\in \proj$, $[ \opat{}{1} \triangleleft p \triangleleft p' \mbox{ or } p' \triangleleft p \triangleleft \opat{}{1}]$ implies 
	$p\succ p'$.
	A set of preferences is called \textbf{single-peaked} if each preference in it is single-peaked.
\end{definition}
Intuitively, it means that the preference $\succ$ is such that a project closer to \opat{}{1} (according to the order $\triangleleft$) is strictly preferred over a project farther from it. Let $\mathcal{D}$ be the set of all single-peaked preferences on \proj. Any upper contour set at a single-peaked preference $\succ \in \mathcal{D}$ forms an interval w.r.t. the prior ordering $\triangleleft$ of projects. Collection of preferences of all voters in $N$ is denoted by $\rprof_N$ (we use subscript in this chapter since we need profiles of subsets of voters very often). That is, $\rprof_N \in \mathcal{D}^n$. For $\rprof_N\in \mathcal{D}^n$ and a \community $N_q$, we denote a preference profile $(\succ_i)_{i\in N_q}$ of the members of the \community $N_q$ by $\rprof_{N_q}$.
\subsection{Social Choice Rules and Their Properties}\label{sec: sc}
The set of all probability distributions over the set of projects \proj is denoted by $\Delta \proj$. For a probability distribution $\delta \in \Delta \proj$ and a project $p \in \proj$, $\delta(p)$ denotes the probability of $p$ in the distribution $\delta$.  A \textit{\textbf{Deterministic Social Choice Rule (DSCR)}} on the single-peaked domain $\mathcal{D}^n$ is a function $f:\mathcal{D}^n \to \proj$, and 
a \textit{\textbf{Random Social Choice Rule (RSCR)}} on $\mathcal{D}^n$ is a function $\varphi:\mathcal{D}^n \to \Delta \proj$. For a RSCR $\varphi:\mathcal{D}^n \to \Delta \proj$ and a project $p \in \proj$, $\varphi_p(\rprof_N)$ denotes the probability allocated to $p$ at the preference profile $\rprof_N \in \mathcal{D}^n$. In other words, $\varphi_p(\rprof_N)$ is a short-hand notation to denote $\varphi(\rprof_N)\large(p\large)$. Similarly, for a set $A \subseteq \proj$, we define $\varphi_A(\rprof_N)=\sum_{p \in A} \varphi_p(\rprof_N)$.

\begin{note}\label{note: rscr-pb}
    A RSCR $\varphi$ is equivalent to a PB rule under flexible costs (\Cref{def: pbrule-flexible}), where $\varphi_p(\rprof_N)$ denotes the fraction $x_p$ of budget allocated to project $p$ at the preference profile $\rprof_N$. Henceforth in this chapter, we refer to the PB rules as random social choice rules and the fraction of budget allocated to a project as the probability allocated to it.
\end{note}

In order to choose between various random social choice rules, we need a way for the voter to compare the outcomes of different rules. One possible way to do this is to define a utility notion which quantifies the benefit derived by each voter from a given probability distribution over the projects \cite{airiau2019portioning}. Another popular way of comparing two probability distributions is stochastic dominance \cite{bogomolnaia2001new}.

\begin{definition}
	A probability distribution $\delta \in \Delta \proj$ is said to \textbf{stochastically dominate} another probability distribution $\delta' \in \Delta \proj$ with respect to an ordinal preference $\succ$, denoted by $\delta \sd{} \delta'$, if and only if $$\delta(U(p,\succ)) \geq \delta'(U(p,\succ)) \qquad \forall p \in \proj.$$
 \end{definition}

Following are the two properties considered to be non-negotiable in the literature on random social choice. The first of them, unanimity, implies that a project unanimously ranked at the top by all the voters must be allocated a probability of $1$, i.e., the entire budget.
\begin{definition}
	A RSCR $\varphi:\mathcal{D}^n \to \Delta \proj$ is called \textbf{unanimous} if for all projects $p \in \proj$ and all $\rprof_N \in \mathcal{D}^n$, $$[\opat{i}{1}=p \mbox{ for all }i \in N ]  \Rightarrow [\varphi_p(\rprof_N)=1].$$ 
\end{definition}
The second property, strategy-proofness, implies that it is impossible for any voter $i$ to report a lie $\succ'_i$ instead of her true preference $\succ_i$ and get a stochastically dominant outcome.
\begin{definition}
    A RSCR $\varphi:\mathcal{D}^n \to \Delta \proj$ is called \textbf{strategy-proof} if for all $i\in N$, all $\rprof_N \in \mathcal{D}^n$, and all $\succ'_i \in \mathcal{D}$, we have $$\varphi(\succ_i,\succ_{-i}) \; \sd{i} \;\varphi(\succ'_i,\succ_{-i}),$$ where $\succ_{-i}$ denotes $(\succ_j)_{j \in N \setminus \{i\}}$.
\end{definition}

\subsection{Unanimous and Strategy-Proof RSCRs}\label{sec: un_sp}
The RSCRs satisfying unanimity and strategy-proofness are characterized in the literature using two approaches. The first approach gives a direct definition and characterization of RSCRs, whereas the second approach expresses RSCRs as a convex combination of DSCRs.
\subsubsection{Direct Characterization: Unanimity and Strategy-Proofness}\label{sec: un_sp_dc}
All the unanimous and strategy-proof RSCRs are characterized to be probabilistic fixed ballot rules \cite{ehlers2002strategy}. We present their definition below. Let $S(t; \rprof_N)$ denote the set of voters whose most preferred project lies at or before $p_t$, i.e., $\{i \in N: \opat{i}{1} \trianglelefteq p_t\}$.
\begin{example}\label{eg: dc_basic}
	Assume that there are four voters $\{1,2,3,4\}$ and there are three projects $\{p_1,p_2,p_3\}$. Consider a preference profile $\rprof_N = (\succ_1,\succ_2,\succ_3,\succ_4)$ such that $p_1\succ_1p_2\succ_1p_3$, $p_2\succ_3p_3\succ_3p_1$, and $p_3\succ p_2 \succ p_1$ for $P \in \{\succ_2,\succ_4\}$. The top-ranked projects are $(p_1,p_3,p_2,p_3)$. For this profile, $S(1; \rprof_N)= \{1\}$, $S(2; \rprof_N)=\{1,3\}$, $S(3; \rprof_N)=\{1,2,3,4\}$.
\end{example}

\begin{definition}\label{def: pfbr}
    A RSCR $\varphi$ on $\mathcal{D}^n$ is said to be a \textbf{probabilistic fixed ballot rule (PFBR)} if there is a collection $\{\delta_{S}\}_{S \subseteq N}$ of probability distributions satisfying the following two properties:
	\begin{enumerate}[(i)]
		\item  \textbf{Ballot Unanimity:} $\delta_{\emptyset}(p_m)=1$ and $\delta_{N}(p_1)=1$, and
		\item \textbf{Monotonicity}: for all $p_t \in \proj$, $S \subset T \subseteq N \implies \delta_S([p_1,p_t]) \leq \delta_T([p_1,p_t])$
	\end{enumerate}
	such that for all $\rprof_N\in \mathcal{D}^n$ and $p_t\in \proj$, we have
    $$\varphi_{p_t}(\rprof_N)= \delta_{S(t; \rprof_N)}([p_1, p_t])-\delta_{S(t-1; \rprof_N)}([p_1, p_{t-1}]);$$
    where $\delta_{S(0; \rprof_N)}([p_1, p_{0}])= 0$.
\end{definition}

\begin{example}\label{eg: pfbr}
	Consider \Cref{eg: dc_basic}. Consider a PFBR corresponding to the probabilistic ballots $\{\delta_{S}\}_{S \subseteq N}$ listed in Table \ref{tab: pfbr}. Clearly, the ballots satisfy ballot unanimity and monotonicity. The PFBR with this collection of probabilistic ballots works as follows: in \Cref{eg: dc_basic}, $S(1; \rprof_N)= \{1\}$, $S(2; \rprof_N)=\{1,3\}$, and $S(3; \rprof_N)=\{1,2,3,4\}$. We know that, the probability allocated to $p_2$ at this profile is $\delta_{S(2; \rprof_N)}([p_1, p_2])-\delta_{S(1; \rprof_N)}([p_1,p_1])$. From Table \ref{tab: pfbr}, $\delta_{\{1,3\}}([p_1,p_2])=0.8$ and $\delta_{\{1\}}([p_1,p_1])=0.3$. Thus, the probability of $p_2$ at $(\succ_1,\succ_2,\succ_3,\succ_4)$ is $0.5$. Similarly, we can compute other probabilities.
 \begin{table}
  \centering
	\begin{tabular}{ c | c | c | c ? c | c | c | c }
			 & $p_1$ & $p_2$ & $p_3$ && $p_1$ & $p_2$ & $p_3$ \\ 
			 \hline
			 $\delta_{\emptyset}$ & $0$ & $0$ & $1$ & $\delta_{\{1\}}$ & $0.3$ & $0.2$ & $0.5$ \\ 
			 \hline
			 $\delta_{\{2\}}$ & $0.1$ & $0.5$ & $0.4$ & $\delta_{\{3\}}$ & $0.2$ & $0.4$ & $0.4$\\ 
			 \hline
			  $\delta_{\{4\}}$ & $0.2$ & $0.4$ & $0.4$ & $\delta_{\{1,2\}}$ & $0.4$ & $0.3$ & $0.3$ \\ 
			  \hline
			   $\delta_{\{1,3\}}$ & $0.5$ & $0.3$ & $0.2$ & $\delta_{\{1,4\}}$ & $0.3$ & $0.4$ & $0.3$\\ 
			   \hline
			    $\delta_{\{2,3\}}$ & $0.4$ & $0.3$ & $0.3$ & $\delta_{\{2,4\}}$ & $0.5$ & $0.3$ & $0.2$\\ 
			    \hline
			     $\delta_{\{3,4\}}$ & $0.3$ & $0.4$ & $0.3$ & $\delta_{\{1,2,3\}}$ & $0.8$ & $0.2$ & $0$\\ 
			     \hline
			      $\delta_{\{1,2,4\}}$ & $0.8$ & $0.2$ & $0$ & $\delta_{\{1,3,4\}}$ & $0.9$ & $0.1$ & $0$\\
			      \hline
			       $\delta_{\{2,3,4\}}$ & $0.9$ & $0.1$ & $0$ & $\delta_{\{1,2,3,4\}}$ & $1$ & $0$ & $0$\\ 
			     \end{tabular}
		 		\caption{The probabilistic ballots $\{\delta_{S}\}_{S \subseteq N}$ for the PFBR in \Cref{eg: pfbr}}\label{tab: pfbr}
	 \end{table}
\end{example}

\begin{lemma}\label{the: un_sp_dc}
    A RSCR on $\mathcal{D}^n$ is unanimous and strategy-proof if and only if it is a probabilistic fixed ballot rule \cite{ehlers2002strategy}.
\end{lemma}
\subsubsection{Extreme Point Characterization: Unanimity and Strategy-Proofness}\label{sec: un_sp_ep}
All the unanimous and strategy-proof DSCRs are characterized to be min-max rules and the RSCRs that satisfy these two properties are characterized to be random min-max rules \cite{peters2014probabilistic,pycia2015decomposing}. We first define min-max rules.

\begin{definition}\label{def: mmr}
	A DSCR $f$ on $\mathcal{D}^n$ is a \textbf{min-max} rule if  for all $S \subseteq N$, there exists $\beta_S \in \proj$ satisfying $$\beta_{\emptyset}= p_m, \beta_N=p_1,  \mbox{ and  } \beta_T  \trianglelefteq \beta_{S} \mbox{ for all }S \subseteq T$$ such that $$f(\rprof_N)=\min_{S \subseteq N}\left [\max_{i \in S}\{\opat{i}{1}, \beta_S\}\right].$$ 	
\end{definition}

\begin{example}\label{eg: mmr}
	Consider the instance specified in \Cref{eg: dc_basic}. Consider a min-max rule corresponding to the probabilistic ballots $\{\beta_{S}\}_{S \subseteq N}$ listed in Table \ref{tab: mmr}. It can be seen that they satisfy required properties. For the profile in \Cref{eg: dc_basic}, $\opat{1}{1} = p_1$, $\opat{2}{1} = p_3$, $\opat{3}{1} = p_2$, and $\opat{4}{1} = p_3$. For any set $S$ with voters $2$ or $4$, $\max_{i \in S}\{\opat{i}{1}, \beta_S\}$ is $p_3$. For the sets $\{3\}$ and $\emptyset$, the value continues to be $p_3$ since their corresponding parameters are $p_3$. For the sets $\{1,3\}$ and $\{1\}$, the value is $p_2$ since the parameters are $p_2$. Therefore, the outcome of the rule is $\min\{p_2,p_3\}$, that is $p_2$.
    \begin{table}
  \centering
	\begin{tabular}{ c | c ? c | c }
			 $\beta_{\emptyset}$ & $p_3$ & $\beta_{\{1\}}$ & $p_2$\\ 
			 \hline
			 $\beta_{\{2\}}$ & $p_2$ & $\beta_{\{3\}}$ & $p_3$\\ 
			 \hline
			  $\beta_{\{4\}}$ & $p_3$ & $\beta_{\{1,2\}}$ & $p_1$\\ 
			  \hline
			   $\beta_{\{1,3\}}$ & $p_2$ & $\beta_{\{1,4\}}$ & $p_2$\\ 
			   \hline
			    $\beta_{\{2,3\}}$ & $p_2$ & $\beta_{\{2,4\}}$ & $p_2$\\ 
			    \hline
			     $\beta_{\{3,4\}}$ & $p_3$ & $\beta_{\{1,2,3\}}$ & $p_1$\\ 
			     \hline
			      $\beta_{\{1,2,4\}}$ & $p_1$ & $\beta_{\{1,3,4\}}$ & $p_2$\\
			      \hline
			       $\beta_{\{2,3,4\}}$ & $p_2$ & $\beta_{\{1,2,3,4\}}$ & $p_1$\\ 
			     \end{tabular}
		 		\caption{The parameters $\{\beta_{S}\}_{S \subseteq N}$ of the min-max rule in \Cref{eg: mmr}}\label{tab: mmr}
	 \end{table}
\end{example}

A \textbf{random min-max rule} is a convex combination of min-max rules. That is, it can be expressed in the form of $\varphi=\sum_{w\in W} \lambda_w \varphi_w$ where $\sum_{w\in W} \lambda_w = 1$, and for every $j \in W$, $\varphi_j$ is a min-max rule and $0\leq\lambda_j\leq 1$.

\begin{lemma}\label{the: un_sp_ep}
    A RSCR on $\mathcal{D}^n$ is unanimous and strategy-proof if and only if it is a random min-max rule \cite{pycia2015decomposing}.
\end{lemma}
\section{Group-Fairness Notions}\label{sec: o-fl-notions}
There are many real-life scenarios where there is a natural partition of voters based on factors such as gender, region, race, and economic status. For example, faculty members in a university voting to select projects to be granted funds can be naturally grouped based on their departments or areas of expertise. Similarly, the citizens of a state can be grouped based on the districts or counties they belong to. We model such natural settings and ensure fairness within and across these groups. Let $G = \{1,\ldots, g\}$ and let $\mathcal{N}$ be a partition of the set $N$, that is, $\mathcal{N}=(N_1,\ldots,N_g)$ where $\cup_{q \in G }N_q=N$ and $N_p\cap N_q=\emptyset$ for all distinct $p, q\in G$. Each $N_q$ is referred to as a \community.

\subsection{Intra-Group Fairness}\label{sec: intrafair}
The primary fairness requirement is to ensure intra-group fairness, that is fairness \emph{within} each \community. We ensure that all the voters in the same \community are treated symmetrically. This property sounds familiar to a social choice theorist - it is same as anonymity, except that it is now required only within each \community. This is applied in many countries in USA, Europe, and Asia under the label of affirmative actions, which treat people in the same group symmetrically but favour weaker groups over the others. We call this property \emph{group-wise anonymity} and explain it formally below.

A permutation $\pi$ of $N$ is \textit{\community preserving} if for any $q  \in G$ and any $i \in N_q$, it holds that $\pi(i) \in N_q$. The property of \community-wise anonymity requires that permuting the preferences of voters within the \community does not change the outcome.

\begin{definition}
    A RSCR $\varphi: \mathcal{D}^n \to \Delta \proj$ is \textbf{\community-wise anonymous} if for all \community preserving permutations $\pi$ of $N$ and all $\rprof_N\in \mathcal{D}^n$, we have $\varphi(\rprof_N)=\varphi(\rprof_{\pi(N)})$ where $\rprof_{\pi(N)}=(\succ_{\pi(1)},\ldots,\succ_{\pi(n)})$.
\end{definition}

\subsection{Inter-Group Fairness}\label{sec: interfair}
We now discuss fairness \emph{across} the groups. While there are situations in which all the groups receive equal weightage, in many real-life situations, the quota for each group could differ. For example, the federal government may choose to give a higher probability to the projects preferred by underprivileged groups or a department might set higher threshold of funds for experimental subjects. We propose two novel fairness notions, weak and strong group entitlement guarantees (weak-GEG and strong-GEG), to capture such requirements. 

Each of our fairness notions takes three parameters, $\kappa_G$, $\psi_G$, and $\eta_G$. Here $\kappa_G = (\kappa_q)_{q \in G}$, where every \community $N_q$ is associated with a value $\kappa_q \in [1,m]$, which we call \emph{\reprange} of $N_q$. Similarly, $\eta_G = (\eta_q)_{q \in G}$, where every \community $N_q$ is assumed to be entitled to a certain probability $\eta_{q} \in [0,1]$, which we call \emph{\resquota} of the \community. Both these parameters are decided by the PB organizer based on the context and various attributes of groups such as size and status. The parameter $\psi_G = (\psi_q)_{q \in G}$, called \emph{\repscene}, is a collection of representative functions of every group. Our fairness notions assume that all \repscenes satisfy some properties. We discuss these properties of \repscenes in the following section.

\subsubsection{Selecting Representative Projects of Each Group}\label{sec: repfun}
Each \community $N_q$ is associated with a function $\psi_q$, called representative function, which on receiving the preferences of voters in $N_q$, selects a set of projects as the \emph{representatives} of all those preferences. We impose the following property on the representative functions of all the groups.

\begin{definition}\label{def: topi}
    A representative function $\psi_q: \mathcal{D}^{|N_q|} \to {2^\proj}$ for a \community $N_q$ is said to be \textbf{\topi} if for every profile $\rprof_{N_q} \in \mathcal{D}^{|N_q|}$, there exist a project $p \in \psi_q(\rprof_{N_q})$ and voters $i,j \in N_q$ (could be the same voter) such that $\opat{i}{1} \trianglelefteq p \trianglelefteq \opat{j}{1}$.
\end{definition}
In other words, \topiness implies that at least one representative project of each group lies between the minimum and maximum top-ranked projects of the voters in that group. A collection of representative functions of all the groups, $(\psi_q)_{q \in G}$, is denoted by $\psi_G$ and is called a \emph{\repscene}. A \repscene in which every representative function $\psi_q$ is \topi and selects exactly $\kappa_q$ consecutive projects is said to be \compliant with \repranges $\kappa_G$.

\begin{definition}\label{def: compliance}
    A \repscene $(\psi_q)_{q \in G}$ is said to be \textbf{\compliant} with $\kappa_G$ if for every $q \in G$, $\psi_q$ is \topi and for every $\rprof_{N_q} \in \mathcal{D}^{N_q}$, $\psi_q(\rprof_{N_q})$ is a $\kappa_q$ sized interval of $\triangleleft$.
\end{definition}

The fairness notions we introduce are defined only for \compliant \repscenes. Imposing that the representative functions must be intervals is to ensure that the functions preserve single-peakedness, thereby ensuring that the representatives chosen are also in accordance with the ordering $\triangleleft$. Restricting the output size guarantees that the PB organizer gets to decide the number of representatives each \community deserves, based on factors such as size and diversity of the \community. Functions being \topi will make sure that at least one representative of each group is either a top-ranked project or is sandwiched between two top-ranked projects.

Examples of several \compliant \repscenes that also satisfy a few other desirable properties are given in \Cref{sec: o-fl-compliant}. In all the following sections, we assume that \repscenes are \compliant.
\subsubsection{Inter-Group Fairness: Group Entitlement Guarantee Notions}\label{sec: gegs}
We define two fairness notions called \textit{\textbf{group entitlement guarantees (GEGs)}} for a given $\kappa_G$, $\eta_G$, and a \repscene $\psi_G$ \compliant with $\kappa_G$. The idea of these fairness notions is to ensure that for every \community $N_q$, at least $\eta_q$ probability is assigned to its $\kappa_q$ representatives. The first notion, \weakfness, guarantees that the $\kappa_q$ representatives of $N_q$ together receive a probability of at least $\eta_q$.

\begin{definition}\label{def: weakfair}
	A RSCR $\varphi$ satisfies $\boldsymbol{(\kappa_G,\psi_G,\eta_G)}$-\textbf{weak-GEG} if for all $\rprof_N\in \mathcal{D}^n$ and all $q \in G$,
	$$\varphi_{\psi_q(\rprof_{N_q})}(\rprof_N)\geq \eta_q.$$
\end{definition}

In weak-GEG, the probability $\eta_q$ can be distributed among all the $\kappa_q$ representatives, which could result in an insignificant probability for each of them, especially if $\kappa_q$ is large. The second notion, \strofness, ensures that at least one of the $\kappa_q$ representatives of $N_q$ receives a probability of at least $\eta_q$. That is, the whole of $\eta_q$ is concentrated on one representative of the group.

\begin{definition}\label{def: strongfair}
	A RSCR $\varphi$ satisfies $\boldsymbol{(\kappa_G,\psi_G,\eta_G)}$-\textbf{strong-GEG} if for all $\rprof_N\in \mathcal{D}^n$ and all $q \in G$, there exists $ p \in \correspondence(N_q) $ such that 
		$$\varphi_{p}(\rprof_N)\geq \eta_q.$$
\end{definition}

It is trivial to see that any RSCR that satisfies strong-GEG also satisfies weak-GEG. Also, note that we are guaranteeing fairness at group level. Guaranteeing fairness at individual level can be easily modeled as a special case of this setting where every group is a singleton set (\Cref{sec: o-fl-indnotions}). In the latter case, fair rules take a simpler form and we discuss them in detail in \Cref{sec: o-fl-inddirect} and \Cref{sec: o-fl-indextreme}.
\subsubsection{Some Shorthand Notations for Characterizations}\label{sec: shorthand}
Before we move to the characterizations of fair rules, we introduce some shorthand notations related to the partition $\mathcal{N}$ of voters. We crucially use these notations in our characterizations.

\noindent{\textbf{(i) The value $\boldsymbol{\Gamma}$:}} Let $\Gamma$ be the set of all $g$ dimensional vectors $\gamma:= (\gamma_1,\ldots, \gamma_g)$ such that $\gamma_q \in \{0,\ldots,|N_q|\}$ for all $q \in G$. For $\gamma,\gamma' \in \Gamma$, we suppose $\gamma \gg \gamma'$ (in other words, $\gamma$ dominates $\gamma'$) if $\gamma_q \geq \gamma'_q$ for all $q \in G$. We denote by $\underline{\gamma}\in \Gamma$, the vector of all zeros and by $\overline{\gamma}\in \Gamma$, the vector whose components take the maximum value. More formally, for all $q \in G$, $\underline{\gamma}_q=0$ and $\overline{\gamma}_q=|N_q|$.  For a preference profile $\rprof_N$ and  $1\leq t\leq m$, let us denote by  $\alpha(t;\rprof_N)$ the element $(\alpha_1,\ldots,\alpha_g)$ of $\Gamma$ such that $\alpha_q$ is the number of voters in $N_q$ whose peaks are not to the right of $p_t$, i.e., $\alpha_q = |\{i\in N_q\mid \opat{i}{1}\trianglelefteq p_t\}|$.

\begin{example}\label{eg: dc_alpha}
    Consider the scenario explained in \Cref{eg: dc_basic}. Suppose the voters $\{1,2,3,4\}$ are partitioned into two groups $\{N_1,N_2\}$ such that $N_1=\{1\}$ and $N_2=\{2,3,4\}$. Consider the same preference profile $\rprof_N = (\succ_1,\succ_2,\succ_3,\succ_4)$ where top-ranked projects of the voters are $(p_1,p_3,p_2,p_3)$. For this profile, $\alpha(1;\rprof_N)=(1,0)$, $\alpha(2;\rprof_N)=(1,1)$, $\alpha(3;\rprof_N)=(1,3)$.
\end{example}

\noindent{\textbf{(ii) The function $\boldsymbol{\tau_i}$:}} For a preference profile $\rprof_{N_q}$ and $t \leq |N_q|$, we denote by $\tau_t(\rprof_{N_q})$ the project at the $t^{th}$ position when the top-ranked projects in $\rprof_{N_q}$ are arranged in increasing order (with repetition), i.e., $\tau_t(\rprof_{N_q})$ is such that  $|\{i\in N_q \mid \opat{i}{1} \triangleleft \tau_t(\rprof_{N_q})\}|< t$ and $|\{i\in N_q \mid \opat{i}{1}\trianglelefteq \tau_t(\rprof_{N_q})\}| \geq t$. Note that by definition, $\tau_0(\rprof_{N_q})=p_1$ for all $q \in G$.
\begin{example}\label{eg: ep_tau}
    Consider the scenario in \Cref{eg: dc_alpha}. By sorting these top-ranked projects (with repetition), we get $(p_1)$ for $N_1$ and $(p_2,p_3,p_3)$ for $N_2$. Therefore, $\tau_1(\rprof_{N_1})=p_1$, $\tau_1(\rprof_{N_2})=p_2$, and $\tau_2(\rprof_{N_2})=\tau_3(\rprof_{N_2})=p_3$.
\end{example}

\noindent{\textbf{(iii) Notion of feasibility:}} Consider a representative function $\psi_q$ and two positive integers $z_1$ and $z_2$. A project $p \in \proj$ is said to be \textbf{feasible at }$\boldsymbol{(z_1,z_2;\psi_q)}$ if there exists a profile $\rprof_{N_q} \in \mathcal{D}^{|N_q|}$ such that:
\begin{enumerate}
    \item $p = \min{\{\psi_q(\rprof_{N_q})\}}$
    \item the top ranked projects of exactly $z_1$ voters lie before $\min{\{\psi_q(\rprof_{N_q})\}}$, and
    \item the top ranked projects of exactly $|N_q|-z_2$ voters lie after $\max{\{\psi_q(\rprof_{N_q})\}}$.
\end{enumerate}
The idea is to capture the possible outcomes of the representative function given the top ranked projects of the voters. For example, suppose we have a \compliant \repscene. Then, any representative function $\psi_q$ is \topi and hence any $p \in \proj$ is not feasible at $(0,0; \psi_q)$ as well as $(|N_q|,|N_q|; \psi_q)$. That is, the interval starting at $p$ cannot be the outcome if all the top-ranked projects lie on the same side of the interval.
\section{Direct Characterization (DC) of Group-Fair Rules}\label{sec: o-fl-direct}
As seen in \Cref{the: un_sp_dc}, all the unanimous and strategy-proof RSCRs are characterized to be probabilistic fixed ballot rules. But, these rules do not take into account the partition of voters into groups. Motivated by this, in this section, we characterize unanimous and strategy-proof rules that account for the partition of voters and ensure fairness for all the groups of voters.
\subsection{DC: Group-Wise Anonymity}\label{sec: dc_gwa}
We start by characterizing unanimous, strategy-proof rules that are group-wise anonymous (i.e., fair within each group). We generalize the idea of probabilistic fixed ballot rules  (PFBRs) and introduce \emph{probabilistic fixed group ballot} rules (PFGBRs) that satisfy \community-wise anonymity along with unanimity and strategy-proofness.

The family of RSCRs \emph{probabilistic fixed group ballot rules} is a generalization of probabilistic fixed ballot rules. The key difference between them is that, in a PFBR, probabilistic ballot is defined for each subset of projects whereas in a PFGBR, it is defined for each element in $\Gamma$. That is, a PFGBR is based on a collection of pre-specified parameters $\{\delta_{\gamma}\}_{\gamma\in \Gamma}$, where for every $\gamma\in \Gamma$, a probabilistic ballot $\delta_\gamma\in \Delta \proj$ is specified.

\begin{definition}\label{def: pfgbr}
	A RSCR $\varphi$ on $\mathcal{D}^n$ is said to be a \textbf{probabilistic fixed group ballot rule} if there is a collection of probabilistic ballots $\{\delta_{\gamma}\}_{\gamma\in \Gamma}$ which satisfies
	\begin{enumerate}[(i)]
		\item \textbf{Ballot Unanimity:} $\delta_{\underline{\gamma}}(p_m)=1$ and $\delta_{\overline{\gamma}}(p_1)=1$, and
		\item \textbf{Monotonicity}: for all $\gamma,\gamma' \in \Gamma$,   $\gamma \gg \gamma'$ implies $\delta_\gamma([p_1,p_t]) \geq \delta_{\gamma'}([p_1,p_t]) $ for all $t\in [1,m]$,
	\end{enumerate}
	such that for all $\rprof_N\in \mathcal{D}^n$ and all $p_t\in \proj$,
	$$\varphi_{p_t}(\rprof_N)= \delta_{\alpha(t,\rprof_N)}([p_1, p_t])-\delta_{\alpha(t-1,\rprof_N)}([p_1, p_{t-1}]);$$
	where $\delta_{\alpha(0,\rprof_N)}([p_1, p_{0}])= 0$.
\end{definition}

\begin{example}\label{eg: pfgbr}
	Consider the instance specified in Example \ref{eg: dc_alpha}. Note that $\Gamma$ has eight elements. In Table \ref{tab: pfgbr}, we list down a collection of probabilistic ballots $\{\delta_{\gamma}\}_{\gamma\in \Gamma}$ satisfying ballot unanimity and monotonicity. The PFGBR with these parameters works as follows: 
	In \Cref{eg: dc_alpha}, $\alpha(1; \rprof_N)=(1,0)$, $\alpha(2; \rprof_N)=(1,1)$, and $\alpha(3; \rprof_N)=(1,3)$. From Table \ref{tab: pfgbr}, $\delta_{(1,3)}([p_1,p_3])=1$ and $\delta_{(1,1)}([p_1,p_2])=0.5$. Thus, the probability of $p_3$ at $(\succ_1,\succ_2,\succ_3,\succ_4)$ is $0.5$. Similarly, we can compute other probabilities.
		\begin{table}
  \centering
	\begin{tabular}{ c | c | c | c }
			 & $p_1$ & $p_2$ & $p_3$ \\ 
			 \hline
			 $\delta_{(0,0)}$ & $0$ & $0$ & $1$ \\ 
			 \hline
			 $\delta_{(0,1)}$ & $0$ & $0.1$ & $0.9$ \\ 
			 \hline
			  $\delta_{(0,2)}$ & $0.1$ & $0.1$ & $0.8$ \\ 
			  \hline
			   $\delta_{(0,3)}$ & $0.2$ & $0$ & $0.8$ \\ 
			   \hline
			    $\delta_{(1,0)}$ & $0.4$ & $0.1$ & $0.5$ \\ 
			    \hline
			     $\delta_{(1,1)}$ & $0.5$ & $0$ & $0.5$ \\ 
			     \hline
			      $\delta_{(1,2)}$ & $0.7$ & $0.2$ & $0.1$ \\ 
			      \hline
			       $\delta_{(1,3)}$ & $1$ & $0$ & $0$ \\ 
			     \end{tabular}
		 		\caption{The probabilistic ballots $\{\delta_{\gamma}\}_{\gamma\in \Gamma}$ for the PFGBR in \Cref{eg: pfgbr}}\label{tab: pfgbr}
	 \end{table}
\end{example}
It needs to be highlighted that when all the groups are singleton, every element in $\Gamma$ is a $n$-dimensional binary vector which can be directly interpreted as a subset of voters (a voter is present in the subset if and only if her corresponding entry in the vector is $1$). This interpretation results in the equivalence of PFGBRs with PFBRs (\Cref{def: pfbr}).
\begin{observation}\label{obs: pfgbr}
    If each group has only one voter, then a probabilistic fixed group ballot rule coincides with a probabilistic fixed ballot rule (\Cref{def: pfbr}). Conversely, if a single group contains all the voters, then a probabilistic fixed group ballot rule coincides with a random median rule, which is simply a convex combination of the median rules \cite{pycia2015decomposing}.
\end{observation}

\begin{theorem}\label{the: dc_gwa}
	A RSCR \rscr is unanimous, strategy-proof, and \community-wise anonymous if and only if it is a probabilistic fixed group ballot rule.
\end{theorem}

\begin{proof}
\textbf{(Only if:)} First, we prove that for a RSCR to satisfy unanimity, strategy-proofness, and group-wise anonymity, it needs to be a PFGBR. Assume that $\varphi:\mathcal{D}^n \to \Delta \proj$ is unanimous, strategy-proof, and group-wise anonymous. We must show that $\varphi$ is a PFGBR. Note that since $\varphi$ is unanimous and strategy-proof, by \Cref{the: un_sp_dc}, $\varphi$ is a PFBR. Assume that $\{\delta_S\}_{S\subseteq N}$ are the probabilistic ballots of $\varphi$.

For any set $S\subseteq N$, let $c(S)=(\alpha^S_1,\ldots,\alpha^S_g)$ where $\alpha^S_q=|\{i\in S\cap N_q\}|$. We claim that for any two sets $S$ and $S'$ such that $c(S)=c(S')$, we have $\delta_S=\delta_{S'}$. Take the preference profile $\rprof_N$ where $\opat{i}{1}=p_1$ for all $i\in S$ and $\opat{i}{1}=p_m$ for all $i\in N\setminus S$, and $\rprof_N'$ where $\opat{i}{1}=p_1$ for all $i\in S'$ and $\opat{i}{1}=p_m$ for all $i\in N\setminus S'$. By the definition of PFBR, this means $\varphi(\rprof_N)=\delta_S$ and $\varphi(\rprof_N')=\delta_{S'}$. Note that as $c(S)=c(S')$, $\rprof_N$ and $\rprof_N'$ are group-wise equivalent. Combining these observations with group-wise anonymity of $\varphi$, we have  $\delta_S=\delta_{S'}$. This completes the proof of the claim.
	Since for any $S\subseteq N$, $c(S)\in \Gamma$, in view of our above claim, we can write $\{\delta_S\}_{S\subseteq N}$ as $\{\delta_{\gamma}\}_{\gamma\in\Gamma}$, which in turn implies that $\varphi$ is PFGBR.

	\textbf{(If:)} Now, we show that every PFGBR satisfies unanimity, strategy-proofness, and group-wise anonymity. Let $\varphi:\mathcal{D}^n \to \Delta \proj$ be a PFGBR. We will show that $\varphi$ satisfies unanimity, strategy-proofness, and group-wise anonymity. Note that since $\varphi$ is a PFGBR, it is a PFBR (see \Cref{def: pfbr}). Thus, by \Cref{the: un_sp_dc}, $\varphi$ is unanimous and strategy-proof. To show group-wise anonymity, consider two profiles $\rprof_N$ and $\rprof_N'$ which are group-wise equivalent, that is, there exists a group-preserving permutation $\pi_N$ such that $\rprof'_N=\rprof_{\pi(N)}$. Take $p_k\in \proj$. Since $\rprof_N$ and $\rprof_N'$ are group-wise equivalent, we have $\alpha(k,\rprof_N)=\alpha(k,\rprof_N')$ and $\alpha(k-1,\rprof_N)=\alpha(k-1,\rprof_N')$. Therefore, by the definition of PFGBR, $\varphi_{p_k}(\rprof_N) = \delta_{\alpha(k,\rprof_N)}([p_1,p_t])-\delta_{\alpha(k-1,\rprof_N)}([p_1,p_{t-1}]) = \delta_{\alpha(k,\rprof'_N)}([p_1,p_t])-\delta_{\alpha(k-1,\rprof'_N)}([p_1,p_{t-1}]) = \varphi_{p_k}(\rprof'_N)$. Hence, $\varphi$ is group-wise anonymous.
\end{proof}

\subsection{DC: Weak-GEG and Group-Wise Anonymity}\label{sec: dc_wfgwa}
It is evident from \Cref{the: dc_gwa} that the unanimous, strategy-proof, and \community-wise anonymous RSCRs are characterized to be probabilistic fixed group ballot rules. In this section, we further impose weak-GEG and characterize the probabilistic fixed group ballot rules that are \weakf (\Cref{def: weakfair}). The characterization follows from the definition, and intuitively requires that for any $\kappa_q$-sized interval of $\triangleleft$, the difference of probabilities allotted by two probabilistic ballots is at least $\eta_q$ if $\gamma$ corresponding to one of them dominates the other and the interval could be the outcome of $\psi_q$ under some conditions. 

\begin{theorem}\label{the: dc_weak}
	A RSCR \rscr is unanimous, strategy-proof, \community-wise anonymous, and satisfies \weakf if and only if it is a probabilistic fixed group ballot rule such that for all $q \in G$, for all $\gamma,\gamma' \in \Gamma$ such that $\gamma \gg \gamma'$, and for all $p_x \in \proj$ feasible at $(\gamma'_q,\gamma_q;\psi_q)$, we have
	$$\delta_{\gamma}([p_1, p_{x+\kappa_q-1}])-\delta_{\gamma'}([p_1, p_{x-1}]) \geq \eta_q.$$
\end{theorem}
\begin{proof}
    \textbf{(Only if:)} We prove that for a PFGBR to satisfy \weakf, the given condition needs to hold. Consider any $q, \gamma, \gamma',$ and $p_x$ as given in the theorem. Since $p_x$ is feasible at $(\gamma'_q,\gamma_q; \psi_q)$, we know that there exists a profile $\rprof_{N_q}$ such that exactly $\gamma'_q$ voters have their top-ranked projects before $p_x$, exactly $\gamma_q - \gamma'_q$ voters have their top-ranked projects in the interval $[p_x,p_{x+\kappa_q-1}]$, and $\psi_q(\rprof_{N_q}) = [p_x,p_{x+\kappa_q-1}]$. For every other group $h \in G \setminus \{q\}$, construct an arbitrary profile $\rprof_{N_h}$ such that exactly $\gamma'_h$ voters have their top-ranked projects before $p_x$ and exactly $\gamma_h - \gamma'_h$ voters have their top-ranked projects in the interval $[p_x,p_{x+\kappa_q-1}]$. Let the combined preference profile be $\rprof_N$.

Since the fairness requirement of group $q$ is met at $\rprof_N$, the probability allocated to $[p_x,p_{x+\kappa_q-1}]$ is at least $\eta_q$. That is, $\delta_{\alpha(x+\kappa_q-1,\rprof_N)}([p_1,p_{x+\kappa_q-1}]) - \delta_{\alpha(x-1,\rprof_N)}([p_1,p_{x-1}]) \geq \eta_q$. By the definition of $\alpha$ and construction of $\rprof_N$, we know that $\alpha(x+\kappa_q-1,\rprof_N) = \gamma$ and $\alpha(x-1,\rprof_N) = \gamma'$. The condition follows.

\textbf{(If:)} Now we prove the converse. That is, we prove that if the stated condition holds for a PFGBR, then it satisfies \weakf. Consider any group $q \in G$ and an arbitrary preference profile $\rprof_N$. Let $p_x = \min{\{\psi_q(\rprof_{N_q})\}}$. Set $\gamma$ and $\gamma'$ such that for every $h \in G$, $\gamma'_h = |i \in N_h: \opat{i}{1} \triangleleft p_x|$ and $\gamma_h = |i \in N_h: \opat{i}{1} \trianglelefteq p_{x+\kappa_q-1}|$. By construction, $\gamma \gg \gamma'$. Also, since $\psi_q(\rprof_{N_q}) = [p_x,p_{x+\kappa_q-1}]$, $p_x$ is feasible at $(\gamma'_q,\gamma_q;\psi_q)$. Thus $q,\gamma,\gamma',$ and $p_x$ satisfy all the given conditions. Therefore, $\delta_{\gamma}([p_1, p_{x+\kappa_q-1}])-\delta_{\gamma'}([p_1, p_{x-1}]) \geq \eta_q$. From the construction of $\gamma$ and $\gamma'$, it can be seen that $\alpha(x+\kappa_q-1,\rprof_N) = \gamma$ and $\alpha(x-1,\rprof_N) = \gamma'$. Therefore, $\delta_{\alpha(x+\kappa_q-1,\rprof_N)}([p_1, p_{x+\kappa_q-1}])-\delta_{\alpha(x-1,\rprof_N)}([p_1, p_{x-1}]) \geq \eta_q$. That is, the probability allocated to $[p_x,p_{x+\kappa_q-1}]$ is at least $\eta_q$. Hence, the fairness requirement of the group $q$ is met.
\end{proof}

\subsection{DC: Strong-GEG and Group-Wise Anonymity}\label{sec: dc_sfgwa}
Here, we characterize the probabilistic fixed group ballot rules that are \strof (\Cref{def: strongfair}).
For this, we extend the concept of feasibility of a project to the feasibility of a set of projects as follows: a set of projects $\boldsymbol{\{p^1,p^2,\ldots,p^t\}}$ \textbf{is feasible at} $\boldsymbol{(z_0,z_1,z_2,\ldots,z_{t}; \psi_q)}$ if there exists a profile $\rprof_{N_q}$ such that $p^1 = \min\{\psi_q(\rprof_{N_q})\}$, $|\{i \in N_q: \opat{i}{1} \triangleleft p^1\}| = z_0$, and $|\{i \in N_q: \opat{i}{1} \trianglelefteq p^j\}| = z_j$ for every $j \in {1,\ldots,t}$. That is, $p_x$ being feasible at $(z_1,z_2; \psi_q)$ is another way of saying that there exists $z$ such that $\{p_x,p_{x+\kappa_q-1}\}$ is feasible at $(z_1,z,z_2; \psi_q)$.

This characterization also follows from the definition, and intuitively requires that there cannot be $\kappa_q$ consecutive projects which could be selected by $\psi_q$ under some conditions and are allotted a probability less than $\eta_q$ each.

\begin{theorem}\label{the: dc_strong}
	A RSCR \rscr is unanimous, strategy-proof, \community-wise anonymous, and satisfies \strof if and only if it is a probabilistic fixed group ballot rule such that for all $q \in G$, for all $\gamma^0,\gamma^1,\ldots,\gamma^{\kappa_q} \in \Gamma$ such that $\gamma^{\kappa_q} \gg \ldots \gg \gamma^1 \gg \gamma^0$, and for all $p_x \in \proj$ such that $\{p_x,p_{x+1},\ldots,p_{x+\kappa_q-1}\}$ is feasible at $(\gamma^0,\gamma^1,\ldots,\gamma^{\kappa_q};\psi_q)$, there exists $t \in [0,{\kappa_q-1}]$ such that
	$$\delta_{\gamma^{t+1}}([p_1, p_{x+t}])-\delta_{\gamma^{t}}([p_{1}, p_{x+t-1}]) \geq \eta_q.$$
\end{theorem}
\begin{proof}
    \textbf{(Only if:)} First, we prove that for a PFGBR to satisfy \strof, it is necessary that it satisfies the given condition. Consider any $q, \gamma^0,\gamma^1,\ldots,\gamma^{\kappa_q},$ and $p_x$ as given in the theorem. Construct a profile $\rprof_{N_q}$ such that: (i) For every group $N_h$, set the top-ranked projects of exactly $\gamma^0_h$ voters before $p_x$ and those of exactly $\gamma^{i+1}_h - \gamma^{i}_h$ voters at $p_{x+i}$ for every $i \in \{0,\ldots,\kappa_q-1\}$ (ii) $\psi_q(\rprof_{N_q}) = [p_x,p_{x+\kappa_q-1}]$. Note that such construction is possible since $\{p_x,p_{x+1},\ldots,p_{x+\kappa_q-1}\}$ is feasible at $(\gamma^0,\gamma^1,\ldots,\gamma^{\kappa_q};\psi_q)$. 

Since the fairness requirement of group $q$ is met at $\rprof_N$, there exists $p_t \in [p_x,p_{x+\kappa_q-1}]$ such that probability allocated to $p_t$ is at least $\eta_q$. This implies that there exists $t \in [0,{\kappa_q-1}]$ such that $ \delta_{\alpha(x+t,\rprof_N)}([p_1, p_{x+t}])-\delta_{\alpha(x+t-1,\rprof_N)}([p_1, p_{x+t-1}]) \geq \eta_q$. By the construction of $\rprof_N$, $\alpha(x+t,\rprof_N) = \gamma^{t+1}$ and $\alpha(x+t-1,\rprof_N) = \gamma^{t}$. The condition follows.

\textbf{(If:)} Now, we prove that any PFGBR satisfying the given condition also satisfies \strof. Consider any group $q \in G$ and an arbitrary preference profile $\rprof_N$. Let $p_x = \min{\{\psi_q(\rprof_{N_q})\}}$. For every $h \in G$, set $\gamma^0_h = |\{i \in N_h: \opat{i}{1} \triangleleft p_x\}|$. For every $h \in G$ and $j \in \{1,\ldots,\kappa_q\}$, set $\gamma^j_h = |\{i \in N_h: \opat{i}{1} \trianglelefteq p_{x+j-1}\}|$. By the construction, $\gamma^{\kappa_q} \gg \ldots \gg \gamma^1 \gg \gamma^0$. Also, $\{p_x,p_{x+1},\ldots,p_{x+\kappa_q-1}\}$ is feasible at $(\gamma^0,\gamma^1,\ldots,\gamma^{\kappa_q};\psi_q)$. Thus, the condition in the theorem holds. This implies that there exists $t \in [0,\kappa_q-1]$ such that $\delta_{\gamma^{t+1}}([p_1, p_{x+t}])-\delta_{\gamma^{t}}([p_{1}, p_{x+t-1}]) \geq \eta_q$. By the construction of $\gamma^{t+1}$ and $\gamma^{t}$, $\gamma^{t+1} = \alpha(x+t,\rprof_N)$ and $\gamma^{t} = \alpha(x+t-1,\rprof_N)$. Therefore, there exists $p_t \in \psi_q(\rprof_{N_q})$ such that $\varphi_{p_t}(\rprof_N) \geq \eta_q$. 
\end{proof}
\begin{example}\label{eg: dc_fair}
    Consider \Cref{eg: pfgbr}. You may recall the preferences of voters from \Cref{eg: dc_basic}. We know $N_1 = \{1\}$ and $N_2 = \{2,3,4\}$. Suppose $\kappa_1 = 1$, $\kappa_2 = 2$, $\eta_1 = \frac{1}{3}$, and $\eta_2 = \frac{2}{5}$. Suppose $\psi_1(\rprof_{N_1}) = \{p_1\}$ and $\psi_2(\rprof_{N_2}) = \{p_2,p_3\}$. The PFGBR gives the outcome $(0.4,0.1,0.5)$ and satisfies \strof.
\end{example}

\section{Extreme Point Characterization (EPC) of Group-Fair Rules}\label{sec: o-fl-extreme}
Often, it is cognitively hard to express RSCRs as probability distributions over the projects, especially when the number of projects is large. Besides, the probabilistic fixed ballot rules need many such probability distributions to be decided. This motivates the direction of defining RSCRs in terms of extreme points, or in other words, as the convex combinations of DSCRs. Such a RSCR is easily expressible, owing to the lucidness and simplicity of DSCRs.

Inspired by this, many works in the literature express RSCRs as convex combinations of DSCRs \cite{picot2012extreme,peters2014probabilistic,pycia2015decomposing,peters2017extreme,roy2021unified}. The \Cref{the: un_sp_ep} provides such a characterization of all the unanimous and strategy-proof RSCRs by proving that they are equivalent to random min-max rules. However, these rules do not take into account the partition of voters into groups. Motivated by this, in this section, we characterize, in terms of extreme points or DSCRs, all the unanimous and strategy-proof rules that consider the partition of voters and ensure fairness for all the groups of voters.

\subsection{EPC: Group-Wise Anonymity}\label{sec: ep_gwa}
We start by characterizing unanimous, strategy-proof, and group-wise anonymous RSCRs in terms of DSCRs. For this, we introduce the family of DSCRs, \emph{group min-max rules}, by generalizing the min-max rules (\Cref{def: mmr}). The key difference between them again is that, in a min-max rule, parameters are defined for every subset of projects whereas in a group min-max rule, they are defined for every element in $\Gamma$. That is, a GMMR is based on a collection of pre-specified parameters $\{\beta_{\gamma}\}_{\gamma\in \Gamma}$, where for every $\gamma\in \Gamma$, a parameter $\beta_\gamma\in \proj$ is specified.

\begin{definition}\label{def: gmmr}
	A DSCR  $f:\mathcal{D}^n \to \proj$ is said to be a \textbf{group min-max rule (GMMR)} if for every $\gamma\in \Gamma$, there exists $\beta_\gamma\in \proj$ satisfying 
	$$\beta_{\underline{\gamma}}= p_m, \beta_{\overline{\gamma}}=p_1,  \mbox{ and  } \beta_\gamma  \trianglelefteq \beta_{\gamma'} \mbox{ for all }\gamma \gg \gamma'$$ such that $$f(\rprof_N)=\min_{\gamma \in \Gamma}\left [\max\{\tau_{\gamma_1}(\rprof_{N_1}),\ldots,\tau_{\gamma_g}(\rprof_{N_g}), \beta_\gamma\}\right].$$ 	
\end{definition}

\begin{example}\label{eg: gmmr}
	Consider the framework specified in \Cref{eg: dc_alpha}. In Table \ref{tab: gmmr}, we specify the parameter values of a group min-max rule.
	\begin{table}[H]
		\centering
		\begin{tabular} {c | c c c c c c c c}
			$\beta$ & $\beta_{(0,0)}$ & $\beta_{(0,1)}$ & $\beta_{(0,2)}$ & $\beta_{(0,3)}$ & $\beta_{(1,0)}$ & $\beta_{(1,1)}$& $\beta_{(1,2)}$ & $\beta_{(1,3)}$\\
			\hline
			& $p_3$ & $p_2$ & $p_2$ & $p_1$ & $p_3$ & $p_3$ & $p_3$ & $p_1$\\
		\end{tabular}
		\caption{Parameters of the group min-max rule  in \Cref{eg: gmmr}} \label{tab: gmmr}
	\end{table}
	\noindent{Recall that, in the profile, top ranked projects of the voters $1$, $2$, $3$, and $4$ are $p_1$, $p_3$, $p_2$, and $p_3$ respectively. Also recall from \Cref{eg: ep_tau} that $\tau_1(\rprof_{N_1})=p_1$, $\tau_1(\rprof_{N_2})=p_2$, and $\tau_2(\rprof_{N_2})=\tau_3(\rprof_{N_2})=p_3$. The outcome of $f$ at this profile is determined as follows.}
	\begin{align}
		f(\rprof_N) & =\min_{\gamma \in \Gamma}\left[\max\{\tau_{\gamma_1}(\rprof_{N_1}),\ldots,\tau_{\gamma_g}(\rprof_{N_g}), \beta_\gamma\}\right] \nonumber \\
		& =\min\big[\max\{\tau_0(\rprof_{N_1}),\tau_0(\rprof_{N_2}),\beta_{(0,0)}\}, \max\{\tau_0(\rprof_{N_1}),\tau_1(\rprof_{N_2}),\beta_{(0,1)}\},\nonumber\\
		&\hspace{16mm} \max\{\tau_0(\rprof_{N_1}),\tau_2(\rprof_{N_2}),\beta_{(0,2)}\},\max\{\tau_0(\rprof_{N_1}),\tau_3(\rprof_{N_2}),\beta_{(0,3)}\},\nonumber\\
		& \hspace{16mm} \max\{\tau_1(\rprof_{N_1}),\tau_0(\rprof_{N_2}),\beta_{(1,0)}\}, \max\{\tau_1(\rprof_{N_1}),\tau_1(\rprof_{N_2}),\beta_{(1,1)}\}, \nonumber \\ 
		& \hspace{16mm}\max\{\tau_1(\rprof_{N_1}),\tau_2(\rprof_{N_2}),\beta_{(1,2)}\},  \max\{\tau_1(\rprof_{N_1}),\tau_3(\rprof_{N_2}),\beta_{(1,3)}\}\big] \nonumber \\
		& =\min\big[\max\{p_1,p_1,p_3\}, \max\{p_1,p_2,p_2\},\max\{p_1,p_3,p_2\},\max\{p_1,p_3,p_1\},\nonumber\\
		& \hspace{16mm} \max\{p_1,p_1,p_3\}, \max\{p_1,p_2,p_3\},\max\{p_1,p_3,p_3\},  \max\{p_1,p_3,p_1\}\big] \nonumber \\
		&=p_2. \nonumber
		\hspace{15cm} 	\square  
	\end{align}
\end{example}

\begin{theorem}\label{the: ep_dscr_gwa}
	A DSCR on $\mathcal{D}^n$ is unanimous, strategy-proof, and \community-wise anonymous if and only if it is a group min-max rule.
\end{theorem}
\begin{proof}
    \textbf{(Only if:)} Assume that $f$ is a unanimous, strategy-proof, and group-wise anonymous. We show that it is a GMMR. Note that since $f$ is unanimous and strategy-proof, by Lemma 2.10, $f$ is a MMR. Let $\{\beta_{S}\}_{S\subseteq N}$ be the parameters of $f$. 

For any set $S\subseteq N$, let $c(S)=(\alpha^S_1,\ldots,\alpha^S_g)$ where $\alpha^S_q=|\{i\in S\cap N_q\}|$. We first show that $\beta_S=\beta_{S'}$ for all $S$ and $S'$ with $c(S)=c(S')$. Take the preference profile $\rprof_N$ where $\opat{i}{1}=p_1$ for all $i\in S$ and $\opat{i}{1}=p_m$ for all $i\in N\setminus S$, and $\rprof_N'$ where $\opat{i}{1}=p_1$ for all $i\in S'$ and $\opat{i}{1}=p_m$ for all $i\in N\setminus S'$. Consider the profile $\rprof_N$. Since $\opat{i}{1}=p_1$ if $i\in S$ and $\opat{i}{1}=p_m$ if $i\in N\setminus S$, we have for all $T$ with $T\subsetneq S$, $\max_{i\in T}\{\opat{i}{1},\beta_T\}=p_m$ and for all $T\subseteq S$, $\max_{i\in T}\{\opat{i}{1},\beta_T\}=\beta_T$. This together with $\beta_{T'}\trianglelefteq \beta_{T''}$ for $T''\subseteq T'$ implies $f(\rprof_N)=\beta_S$. Similarly, $f(\rprof_N')=\beta_{S'}$. However, as $c(S)=c(S')$, $\rprof_N$ and $\rprof_N'$ are group-wise equivalent. Combining these observations with group-wise anonymity of $f$, we have  $\beta_S=\beta_{S'}$. 

Since for any $S\subseteq N$, $c(S)\in \Gamma$, the parameter set of $f$ can be written as $\{\beta_\gamma\}_{\gamma\in \Gamma}$. Consider a set $S\subseteq N$ and a profile $\rprof_N$. The value $\max_{i\in S}\{\opat{i}{1}\}$ is same as $\max_{q\in G}\{\tau_{\gamma_q}(\rprof_{N_q})\}$. On combining this together with the fact that $\beta_S=\beta_{S'}$ for all $S$ and $S'$ with $c(S)=c(S')$, we have for all $\rprof_N\in \mathcal{D}^n$ $$f(\rprof_N)=\min_{\gamma \in \Gamma}\left [\max\{\tau_{\gamma_1}(\rprof_{N_1}),\ldots,\tau_{\gamma_g}(\rprof_{N_g}), \beta_\gamma\}\right].$$

\textbf{(If:)} Let $f:\mathcal{D}^n \to \proj$ be a GMMR.  We have to show that $f$ is unanimous, strategy-proof, and group-wise anonymous. Since $f$ is GMMR, it is a MMR (see Definition 2.8). Hence, by Lemma 2.10, $f$ is strategy-proof and unanimous.  To show that $f$ is group-wise anonymous, take two group-wise equivalent profiles $\rprof_N$ and $\rprof'_N$, that is,  there exists a group-preserving permutation $\pi_N$ such that $\rprof_N'=\rprof_{\pi(N)}$. Let $\{\beta_\gamma\}_{\gamma\in \Gamma}$ be the parameters of $f$. Since $\rprof_N$ and $\rprof_N'$ are group-wise equivalent, for any $\gamma\in \Gamma$ and any group $q\in G$, we have $\tau_{\gamma_q}(\rprof_{N_q})=\tau_{\gamma_q}(\rprof_{N'_q})$. Therefore, by the definition of GMMR, $f(\rprof_N)=f(\rprof_N')$.
\end{proof}

We express a RSCR as a convex combination of DSCRs $\phi_1,\ldots,\phi_q$. That is, $\varphi=\sum_{w\in W} \lambda_w \varphi_w$ where $W = \{1,2,\ldots,q\}$, $\sum_{w\in W} \lambda_w = 1$, and for every $j \in W$, $0\leq\lambda_j\leq 1$ and $\varphi_j$ is a DSCR. If every $\varphi_j$ such that $\lambda_j > 0$ is a group min-max rule, such a RSCR is called \textbf{random group min-max rule}. Next, we prove that random group min-max rules are equivalent to probabilistic fixed group ballot rules.

\begin{theorem}\label{the: ep_gwa}
	Every probabilistic fixed group ballot rule on $\mathcal{D}^n$ is also a random group min-max rule on $\mathcal{D}^n$. 
\end{theorem}

\begin{proof}
	For every $\gamma\in \Gamma$, let $\rprof^\gamma$ denote the profile where from group $q$, $\gamma_q$ number of voters have top-ranked projects at $p_1$ and $|N_q|-\gamma_q$ number of voters have top-ranked projects at $p_m$. We call these profiles \textit{boundary profiles}. Since the outcome of a probabilistic fixed group ballot rule at any profile is linearly dependent on the outcomes at the boundary profiles (see \cite{ehlers2002strategy}), it is sufficient to show that any unanimous and strategy-proof RSCR on the boundary profiles can be written as a convex combination of unanimous and strategy-proof DSCRs on the boundary profiles.

	 Let $\varphi$ be a unanimous and strategy-proof RSCR defined on these boundary profiles. We will show that there are unanimous and strategy-proof DSCRs such that $\varphi$ can be written as a convex combination of those deterministic rules. More formally, we show that  there exist $f_1,\ldots,f_r$ unanimous and strategy-proof DSCRs and non-negative numbers $\lambda_s$ where $s\leq r$ with $\sum_{s=1}^r=1$ such that for all $\gamma\in \Gamma$, $$\varphi_{p_k}(\rprof^\gamma)=\sum_{s=1}^{r}\lambda_sf_{sp_k}(\rprof^\gamma),$$
	for all $p_k\in \proj$. Here $f_{sp_k}(\rprof^\gamma)=1$ if $f_s(\rprof^\gamma)=p_k$ otherwise $f_{sp_k}(\rprof^\gamma)=0$. Let $z=\sum_{q=1}^{g}2^{|N_q|}$. These system of equations can be represented in matrix form  as $Z\lambda=d$ where $Z$ is a $z  m\times r$ size matrix of $0-1$, $\lambda$ is column vector of length $r$ with $\lambda_s$ in row $s$, and $d$ is a column vector of length $zm$ with $\varphi_{p_k}(\rprof^\gamma)$ in the row corresponding to $(\gamma,k)$. By Farka's Lemma, having a solution to this system of equations is equivalent to show that $d'y\geq 0$ for any $y\in \mathbb{R}^{z m}$ with $Z'y\geq 0\in \mathbb{R}^r$.
	
	We will prove this by using a network flow formulation of the problem and using the max-flow min-cut theorem.  Consider an arbitrary numbering of  $\gamma\in \Gamma$, $\gamma_1,\ldots,\gamma_{z}$ with $\gamma_1=\underline{\gamma}$ and $\gamma_{z}=\overline{\gamma}$. The set of vertices is $V=\{x,y\}\cup \{(S_i,j)\mid i=1,\ldots,z,j=1,\ldots,m\}$ where $x$ is the source and $y$ is the sink. The edges are defined as follows:
	
	\begin{itemize}
        \itemsep0em
		\item For every $s\in \{1,\ldots, r\}$, let $E_s=\{((\gamma_i,j),(\gamma_{i+1},k))\mid i=1,\ldots, z-1,f_{sp_j}(P^{\gamma_i})=f_{sp_k}(P^{\gamma_{i+1}})=1\}$.
		
		\item There is an edge $(x,(\gamma_1,j))$ for all $j=1,\ldots,m$.
		\item There is an edge $((\gamma_{z,j}),y)$ for all $j=1,\ldots,m$.
	\end{itemize} 

Now the set of edges $E$ is the union of sets $E_s$. We can define a path for every $s$, it is defined as $E_s\cup \{(x,(\gamma_1,j)),((\gamma_{z},j'),y)\mid f_{sj}(\rprof^\gamma_1)=f_{sj'}(P^{z})=1\}$.  Hence, every strategy-proof deterministic rule has a path from source $x$ to sink $y$. The capacities of the vertices are defined as $c(\gamma_i,j)=\varphi_{p_j}(P^{\gamma_i})$ and $c(x)=c(y)=1$.

A cut is a set of vertices such that every deterministic rule intersects the cut at least once.

\begin{lemma}\label{le_1}
	The minimum capacity of a cut is equal to 1.
\end{lemma}

We prove the above lemma in Appendix \ref{app4: le_1}. By Lemma \ref{le_1} and the max-flow min-cut theorem, it follows that the maximal flow through the network is  1. Since the total capacity of the nodes corresponding to any given profile is $1$ and every path will intersect one such node, it must be that the flow through each node in the network will be exactly the capacity of the node. 

We are now ready to complete the proof of the theorem. Consider a
maximum flow through the network. It follows from the  definition of the network that each path is determined by a deterministic  strategy-proof rule. Consider such a rule $f_s$. Suppose that  $Fl(s)$ denotes the flow through the path induced by the rule. 
Since it is a flow, we have  $Fl(s)\geq0$, and hence
\begin{equation}\label{e_1}
	\sum_{s=1}^{r}\sum_{\gamma\in \Gamma}y(\gamma,r(\Gamma))Fl(s)\geq 0.
\end{equation}
Consider the coefficient of an arbitrary term $y(\gamma, j)$ at the left-hand side of \Cref{e_1}. Note that the total flow through an edge of the network is the sum of the flows through all paths containing the edge. Hence, the
total flow at the vertex $(\gamma, j)$ in the network is the sum of the flows through the vertex through all paths containing $(\gamma, j)$. We have already shown  $\varphi_{p_j}(\rprof^\gamma)$, which yields 
$$\sum_{\gamma\in \Gamma}\sum_{j=1}^{m}y(\gamma,j)\varphi_{p_j}(\rprof^\gamma)\geq 0.$$
\end{proof}	

From \Cref{the: dc_gwa} and \Cref{the: ep_gwa}, it can be concluded that a RSCR is unanimous, strategy-proof, and \community-wise anonymous if and only if it can be expressed as a convex combination of group min-max rules.

\subsection{EPC: Weak-GEG and Group-Wise Anonymity}\label{sec: ep_wfgwa}
It is evident from \Cref{sec: ep_gwa} that the unanimous, strategy-proof, and group-wise anonymous RSCRs are characterized to be convex combinations of group min-max rules. In this section, we further impose weak-GEG and characterize the random group min-max rules that are \weakf (\Cref{def: weakfair}). The characterization intuitively requires that for any $\kappa_q$-sized interval of $\triangleleft$ that can be the outcome of $\psi_q$ under some conditions, the weightage allotted to DSCRs having both the parameters in the interval is at least $\eta_q$ whenever $\gamma$ corresponding to one of them dominates the other.

\begin{theorem}\label{the: ep_weak}
	A RSCR \rscr is unanimous, strategy-proof, \community-wise anonymous, and \weakf if and only if it is a random group min-max rule $\varphi=\sum_{w\in W} \lambda_w \varphi_w$ such that for all $q \in G$, for all $\gamma,\gamma' \in \Gamma$ such that $\gamma \gg \gamma'$, and for all $p_x \in \proj$ feasible at $(\gamma'_q,\gamma_q;\psi_q)$, we have
	$$\sum_{\{w\;\mid \; p_x \trianglelefteq \beta^{\varphi_w}_{\gamma'},\; \beta^{\varphi_w}_{\gamma} \trianglelefteq p_{x+\kappa_q-1}\}}{\lambda_w} \geq \eta_q.$$
\end{theorem}
\begin{proof}
    \textbf{(Only if:)} We prove that for a random group min-max rule to satisfy \weakf, it is necessary that the given condition holds. Consider any $q, \gamma, \gamma',$ and $p_x$ as given in the theorem. Since $p_x$ is feasible at $(\gamma'_q,\gamma_q; \psi_q)$, we know that there exists a profile $\rprof_{N_q}$ such that exactly $\gamma'_q$ voters have their top-ranked projects before $p_x$, exactly $\gamma_q - \gamma'_q$ voters have their top-ranked projects in the interval $[p_x,p_{x+\kappa_q-1}]$, and $\psi_q(\rprof_{N_q}) = [p_x,p_{x+\kappa_q-1}]$. For every other group $h \in G \setminus \{q\}$, construct an arbitrary profile $\rprof_{N_h}$ such that exactly $\gamma'_h$ voters have their top-ranked projects before $p_x$ and exactly $\gamma_h - \gamma'_h$ voters have their top-ranked projects in the interval $[p_x,p_{x+\kappa_q-1}]$. Let the combined preference profile be $\rprof_N$.

Since the fairness requirement of group $q$ is met at $\rprof_N$, $\sum_{\{w: \varphi_w(\rprof_N) \in [p_x,p_{x+\kappa_q-1}]\}}{\lambda_w} \geq \eta_q$. Consider any $w$ such that $\varphi_w(\rprof_N) \in [p_x,p_{x+\kappa_q-1}]$. Since $\varphi_w$ is a group min-max rule, for any $\hat{\gamma} \in \Gamma$, $p_x \trianglelefteq \varphi_w(\rprof_N) \trianglelefteq \max\{\tau_{\hat{\gamma}_1}(\rprof_{N_1}),\ldots,\tau_{\hat{\gamma}_g}(\rprof_{N_g}), \beta^{\varphi_w}_{\hat{\gamma}}\}$. This implies $p_x \trianglelefteq \beta^{\varphi_w}_{\gamma'}$ since $\tau_{\gamma'_h}(\rprof_{N_h}) \triangleleft p_x$ for any $h \in G$. Similarly, since $\varphi_w$ is a group min-max rule, there exists $\hat{\gamma} \in \Gamma$ such that $\max\{\tau_{\hat{\gamma}_1}(\rprof_{N_1}),\ldots,\tau_{\hat{\gamma}_g}(\rprof_{N_g}), \beta^{\varphi_w}_{\hat{\gamma}}\} = \varphi_w(\rprof_N)$. This is not possible if $\hat{\gamma}_h > \gamma_h$ for any $h \in G$ since $\varphi_w(\rprof_N) \trianglelefteq p_{x+\kappa_q-1}$ and $p_{x+\kappa_q-1} \triangleleft \tau_{\gamma_h+1}(\rprof_{N_h})$ by construction. Therefore, there exists $\hat{\gamma}$ such that $\gamma \gg \hat{\gamma}$ and $\max\{\tau_{\hat{\gamma}_1}(\rprof_{N_1}),\ldots,\tau_{\hat{\gamma}_g}(\rprof_{N_g}), \beta^{\varphi_w}_{\hat{\gamma}}\}  = \varphi_w(\rprof_N)$. By the definition of group min-max rules, $\beta^{\varphi_w}_\gamma \trianglelefteq \beta^{\varphi_w}_{\hat{\gamma}}$. This implies, $\beta^{\varphi_w}_{\gamma} \trianglelefteq \beta^{\varphi_w}_{\hat{\gamma}} \trianglelefteq \varphi_w(\rprof_N) \trianglelefteq p_{x+\kappa_q-1}$. Hence, we have $p_x \trianglelefteq \beta^{\varphi_w}_{\gamma'}$ and $\beta^{\varphi_w}_{\gamma} \trianglelefteq p_{x+\kappa_q-1}$.

\textbf{(If:)} We now prove that any random group min-max rule that satisfies the given condition also satisfies \weakf. Consider any group $q \in G$ and an arbitrary preference profile $\rprof_N$. Let $p_x = \min{\{\psi_q(\rprof_{N_q})\}}$. Set $\gamma$ and $\gamma'$ such that for every $h \in G$, $\gamma'_h = |i \in N_h: \opat{i}{1} \triangleleft p_x|$ and $\gamma_h = |i \in N_h: \opat{i}{1} \trianglelefteq p_{x+\kappa_q-1}|$. By construction, $\gamma \gg \gamma'$. Also, since $\psi_q(\rprof_{N_q}) = [p_x,p_{x+\kappa_q-1}]$, $p_x$ is feasible at $(\gamma'_q,\gamma_q;\psi_q)$. Thus $q,\gamma,\gamma',$ and $p_x$ satisfy all the given conditions.

Consider any group min-max rule $\varphi_w$ such that $p_x \trianglelefteq \beta^{\varphi_w}_{\gamma'}$ and $\beta^{\varphi_w}_{\gamma} \trianglelefteq p_{x+\kappa_q-1}$. Since $\beta^{\varphi_w}_{\gamma} \trianglelefteq p_{x+\kappa_q-1}$ and also $\tau_{\gamma_h}(\rprof_{N_h}) \trianglelefteq p_{x+\kappa_q-1}$ for any $h \in G$, $\max\{\tau_{\gamma_1}(\rprof_{N_1}),\ldots,\tau_{\gamma_g}(\rprof_{N_g}), \beta^{\varphi_w}_\gamma\} \trianglelefteq p_{x+\kappa_q-1}$. This implies, $\varphi_w(\rprof_N) \trianglelefteq p_{x+\kappa_q-1}$. By construction, $p_x \trianglelefteq \tau_{\gamma'_h+1}$ for any $h \in G$. This implies, for any $\hat{\gamma}$ such that $\hat{\gamma}_h > \gamma_h$ for some $h \in G$, $p_x \trianglelefteq \max\{\tau_{\hat{\gamma}_1}(\rprof_{N_1}),\ldots,\tau_{\hat{\gamma}_g}(\rprof_{N_g}), \beta^{\varphi_w}_{\hat{\gamma}}\}$. For any $\hat{\gamma}$ such that $\gamma' \gg \hat{\gamma}$, since $p_x \trianglelefteq \beta^{\varphi_w}_{\gamma'}$, by definition of group min-max rule, $p_x \trianglelefteq \beta^{\varphi_w}_{\hat{\gamma}}$ and thus, $p_x \trianglelefteq \max\{\tau_{\hat{\gamma}_1}(\rprof_{N_1}),\ldots,\tau_{\hat{\gamma}_g}(\rprof_{N_g}), \beta^{\varphi_w}_{\hat{\gamma}}\}$. Therefore, for any $\hat{\gamma} \in \Gamma$, $p_x \trianglelefteq \max\{\tau_{\hat{\gamma}_1}(\rprof_{N_1}),\ldots,\tau_{\hat{\gamma}_g}(\rprof_{N_g}), \beta^{\varphi_w}_{\hat{\gamma}}\}$. This implies, $p_x \trianglelefteq \varphi_w(\rprof_N)$. Combining this with $\varphi_w(\rprof_N) \trianglelefteq p_{x+\kappa_q-1}$ gives $\varphi_w(\rprof_N) \in [p_x,p_{x+\kappa_q-1}]$. Fairness requirement of the group $q$ is met and this completes the proof.
\end{proof}

\subsection{EPC: Strong-GEG and Group-Wise Anonymity}\label{sec: ep_sfgwa}
Here, we characterize the random group min-max rules that are \strof (Def. \ref{def: strongfair}). Recall the notion of feasibility of a set of projects introduced in \Cref{sec: dc_sfgwa}. Basically, the following characterization ensures that there cannot be $\kappa_q$ consecutive projects which could be selected by $\psi_q$ under some conditions and are allotted a probability less than $\eta_q$ each.
\begin{theorem}\label{the: ep_strong}
  A RSCR \rscr is unanimous, strategy-proof, \community-wise anonymous, and \strof if and only if it is a random group min-max rule $\varphi=\sum_{w\in W} \lambda_w \varphi_w$ such that for all $q \in G$, for all $\gamma^0,\gamma^1,\ldots,\gamma^{\kappa_q} \in \Gamma$ such that $\gamma^{\kappa_q} \gg \ldots \gg \gamma^1 \gg \gamma^0$, and for all $p_x \in \proj$ such that $\{p_x,p_{x+1},\ldots,p_{x+\kappa_q-1}\}$ is feasible at $(\gamma^0,\gamma^1,\ldots,\gamma^{\kappa_q};\psi_q)$, there exists $t \in [0,{\kappa_q-1}]$ such that $$\sum_{\left\{w\;\mid \; p_{x+t} \trianglelefteq \beta^{\varphi_w}_{\gamma^t}\;,\; \beta^{\varphi_w}_{\gamma^{t+1}}\trianglelefteq p_{x+t}\right\}}\lambda_w\geq \eta_q.$$
\end{theorem}
\begin{proof}
    \textbf{(Only if:)} We first prove that a random group min-max rule satisfies \strof only when the given condition is satisfied. Consider any $q, \gamma^0,\gamma^1,\ldots,\gamma^{\kappa_q},$ and $p_x$ as given in the theorem. Construct a profile $\rprof_{N_q}$ such that: (i) For every group $N_h$, set the top-ranked projects of exactly $\gamma^0_h$ voters before $p_x$ and those of exactly $\gamma^{i+1}_h - \gamma^{i}_h$ voters at $p_{x+i}$ for every $i \in \{0,\ldots,\kappa_q-1\}$ (ii) $\psi_q(\rprof_{N_q}) = [p_x,p_{x+\kappa_q-1}]$. Note that such construction is possible since $\{p_x,p_{x+1},\ldots,p_{x+\kappa_q-1}\}$ is feasible at $(\gamma^0,\gamma^1,\ldots,\gamma^{\kappa_q};\psi_q)$.

Since the fairness requirement of group $q$ is met at $\rprof_N$, there exists $t \in [0,{\kappa_q-1}]$ such that $\sum_{\{w: \varphi_w(\rprof_N) = p_{x+t}\}}{\lambda_w} \geq \eta_q$. Consider any $\varphi_w$ such that $\varphi_w(\rprof_N) = p_{x+t}$.

Since $\varphi_w(\rprof_N) = p_{x+t}$ and $\varphi_w$ is a GMMR, $p_{x+t} \trianglelefteq \max\{\tau_{\gamma^t_1}(\rprof_{N_1}),\ldots,\tau_{\gamma^t_g}(\rprof_{N_g}), \beta^{\varphi_w}_{\gamma^t}\}$. But we know that, for any $h \in G$, $\tau_{\gamma^t_h}(\rprof_{N_h}) \trianglelefteq p_{x+t-1}$ by construction of $\rprof_N$. Therefore, $p_{x+t} \trianglelefteq \beta^{\varphi_w}_{\gamma^t}$. Since $\varphi_w(\rprof_N) = p_{x+t}$, there exists $\hat{\gamma}$ such that $\max\{\tau_{\hat{\gamma}_1}(\rprof_{N_1}),\ldots,\tau_{\hat{\gamma}_g}(\rprof_{N_g}), \beta^{\varphi_w}_{\hat{\gamma}}\} = p_{x+t}$. This is not possible if $\hat{\gamma}_h > \gamma^{t+1}_h$ for some $h \in G$ since $\tau_{\gamma^{t+1}_h+1}(\rprof_{N_h}) \succ p_{x+t}$ by the construction of $\rprof_N$. Thus, $\gamma^{t+1} \gg \hat{\gamma}$. Since $\beta^{\varphi_w}_{\hat{\gamma}} \trianglelefteq p_{x+t}$, by the definition of group min-max rules, $\beta^{\varphi_w}_{\gamma^{t+1}} \trianglelefteq p_{x+t}$. Thus it is proved that $p_{x+t} \trianglelefteq \beta^{\varphi_w}_{\gamma^t}$ and $\beta^{\varphi_w}_{\gamma^{t+1}} \trianglelefteq p_{x+t}$.

\textbf{(If:)} Now, we prove the converse direction by considering a random group min-max rule that satisfies the condition in the theorem. We prove that the rule also satisfies \strof. Consider any group $q \in G$ and an arbitrary preference profile $\rprof_N$. Let $p_x = \min{\{\psi_q(\rprof_{N_q})\}}$. For every $h \in G$, set $\gamma^0_h = |\{i \in N_h: \opat{i}{1} \triangleleft p_x\}|$. For every $h \in G$ and $j \in \{1,\ldots,\kappa_q\}$, set $\gamma^j_h = |\{i \in N_h: \opat{i}{1} \trianglelefteq p_{x+j-1}\}|$. By the construction, $\gamma^{\kappa_q} \gg \ldots \gg \gamma^1 \gg \gamma^0$. Also, $\{p_x,p_{x+1},\ldots,p_{x+\kappa_q-1}\}$ is feasible at $(\gamma^0,\gamma^1,\ldots,\gamma^{\kappa_q};\psi_q)$. Thus, they satisfy all the required conditions. Therefore, there exists $t \in [0,\kappa_q-1]$ such that the condition in the theorem holds.

Consider any group min-max rule $\varphi_w$ such that $p_{x+t} \trianglelefteq \beta^{\varphi_w}_{\gamma^t}$ and $\beta^{\varphi_w}_{\gamma^{t+1}}\trianglelefteq p_{x+t}$. It is enough to prove that $\varphi_w$ selects $p_{x+t}$.

For any $\gamma$ such that $\gamma_h > \gamma^{t}_h$ for some $h \in G$, $p_{x+t} \trianglelefteq \tau_{\gamma_h}(\rprof_{N_h})$ by the construction of $\gamma^{t}$. Hence for such a $\gamma$, $p_{x+t} \trianglelefteq \max\{\tau_{\gamma_1}(\rprof_{N_1}),\ldots,\tau_{\gamma_g}(\rprof_{N_g}), \beta^{\varphi_w}_\gamma\}$. Now consider any $\gamma$ such that $\gamma^{t} \gg \gamma$. By definition of group min-max rule, $\beta^{\varphi_w}_{\gamma^{t}} \trianglelefteq \beta^{\varphi_w}_{\gamma}$. Since $p_{x+t} \trianglelefteq \beta^{\varphi_w}_{\gamma^t}$, $p_{x+t} \trianglelefteq \beta^{\varphi_w}_{\gamma}$. Therefore, for any $\gamma \in \Gamma$, $p_{x+t} \trianglelefteq \max\{\tau_{\gamma_1}(\rprof_{N_1}),\ldots,\tau_{\gamma_g}(\rprof_{N_g}), \beta^{\varphi_w}_{\gamma}\}$. To prove that $\varphi_w$ selects $p_{x+t}$, it is now enough to prove that there exists a $\gamma \in \Gamma$ such that $\max\{\tau_{\gamma_1}(\rprof_{N_1}),\ldots,\tau_{\gamma_g}(\rprof_{N_g}), \beta^{\varphi_w}_\gamma\} = p_{x+t}$.

\noindent{\textbf{Case 1:} $p_{x+t}$ is the top-ranked project of some voter at $\rprof_N$.}

Therefore, there exists a group $N_d$ such that $\tau_{\gamma^{t+1}_d}(\rprof_{N_d}) = p_{x+t}$. By the construction of $\gamma^{t+1}$, for any $h \in G$, $\tau_{\gamma^{t+1}_h}(\rprof_{N_h}) \trianglelefteq p_{x+t}$. We know that $\beta^{\varphi_w}_{\gamma^{t+1}}\trianglelefteq p_{x+t}$. Combining all these observations, we have $\max\{\tau_{\gamma^{t+1}_1}(\rprof_{N_1}),\ldots,\tau_{\gamma^{t+1}_g}(\rprof_{N_g}), \beta^{\varphi_w}_{\gamma^{t+1}}\} = p_{x+t}$.

\noindent{\textbf{Case 2:} $p_{x+t}$ is not a top-ranked project for any voter.}

By construction of $\rprof_N$, exactly $\gamma^{t+1}_h - \gamma^{t}_h$ voters of every group $N_h$ have their top-ranked project at $p_{x+t}$. Therefore, $\gamma^{t+1} = \gamma^{t}$. Since $p_{x+t} \trianglelefteq \beta^{\varphi_w}_{\gamma^t}$ and $\beta^{\varphi_w}_{\gamma^{t+1}}\trianglelefteq p_{x+t}$, $\beta^{\varphi_w}_{\gamma^t} = \beta^{\varphi_w}_{\gamma^{t+1}} = p_{x+t}$.
Thus, either way, $\varphi_w(\rprof_N) = p_{x+t}$.
\end{proof}
\section{Individual-Fairness Notions}\label{sec: o-fl-indnotions}
When each group has exactly one voter (i.e., $G = N$), clearly, the concept of \community-wise anonymity is not relevant as any RSCR satisfies it trivially. In this case, the most natural and reasonable \compliant \repscene is that in which every representative function $\psi_i$ selects top $\kappa_i$ projects of voter $i$. We now define two fairness notions that are special cases of our \weakfness and \strofness notions where $G = N$ and the \repscene is as described above. $\boldsymbol{(\kappa_N,\eta_N)}$-\textbf{weak-IEG} (weak individual entitlement guarantee) ensures that the top $\kappa_i$ projects of every voter $i$ together receive a probability of at least $\eta_i$ while $\boldsymbol{(\kappa_N,\eta_N)}$-\textbf{strong-IEG} ensures that at least one project in the top $\kappa_i$ projects of each voter $i$ receives a probability of $\eta_i$. 
\begin{definition}\label{def: sp_weak}
	A RSCR $\varphi$ satisfies $\boldsymbol{(\kappa_N,\eta_N)}$-\textbf{weak-IEG} if for any voter $i \in N$ and any $\rprof_N \in \mathcal{D}^n$,
	$$\varphi_{U(\opat{i}{\kappa_i},\succ_i)}(\rprof_N)\geq \eta_i.$$
\end{definition}

\begin{definition}\label{def: sp_strong}
	A RSCR $\varphi$ satisfies $\boldsymbol{(\kappa_N,\eta_N)}$-\textbf{strong-IEG} if for any voter $i \in N$ and any $\rprof_N \in \mathcal{D}^n$, there exists $p \in U(\opat{i}{\kappa_i},\succ_i)$ such that
	$$\varphi_{p}(\rprof_N)\geq \eta_i.$$
\end{definition}
\section{Direct Characterizations (DC) of Individual-Fair Rules}\label{sec: o-fl-inddirect}
We may recall from \Cref{sec: un_sp_dc} that all the unanimous and strategy-proof RSCRs are characterized to be probabilistic fixed ballot rules whose definition is presented below. Let $S(t; \rprof_N)$ denote $\{i \in N: \opat{i}{1} \trianglelefteq p_t\}$.
\begin{definition}\label{def: pfbr-rep}
    A RSCR $\varphi$ on $\mathcal{D}^n$ is said to be a \textbf{probabilistic fixed ballot rule (PFBR)} if there is a collection $\{\delta_{S}\}_{S \subseteq N}$ of probability distributions satisfying the following two properties:
	\begin{enumerate}[(i)]
		\item  \textbf{Ballot Unanimity:} $\delta_{\emptyset}(p_m)=1$ and $\delta_{N}(p_1)=1$, and
		\item \textbf{Monotonicity}: for all $p_t \in \proj$, $S \subset T \subseteq N \implies \delta_S([p_1,p_t]) \leq \delta_T([p_1,p_t])$
	\end{enumerate}
	such that for all $\rprof_N\in \mathcal{D}^n$ and $p_t\in \proj$, we have
    $$\varphi_{p_t}(\rprof_N)= \delta_{S(t; \rprof_N)}([p_1, p_t])-\delta_{S(t-1; \rprof_N)}([p_1, p_{t-1}]);$$
    where $\delta_{S(0; \rprof_N)}([p_1, p_{0}])= 0$.
\end{definition}
As interpreted in \Cref{obs: pfgbr}, when the groups are singleton, PFBRs and PFGBRs become equivalent since each element in $\Gamma$ becomes a $n$-dimensional binary vector that can be directly interpreted as a subset of voters. In the subsequent sections, we present direct characterizations of the probabilistic fixed ballot rules that satisfy $\boldsymbol{(\kappa_N,\eta_N)}$-\textbf{weak-IEG} and $\boldsymbol{(\kappa_N,\eta_N)}$-\textbf{strong-IEG}. For this, we crucially use the characterization results for PFGBRs that satisfy \weakf (\Cref{the: dc_weak}) and \strof (\Cref{the: dc_strong}).

\subsection{DC: Weak-IEG}\label{sec: dc_wieg}
We characterize the probabilistic fixed ballot rules that satisfy our weak-IEG and strong-IEG notions. The notion of \spweakfness ensures that the top $\kappa_i$ projects of a voter $i$ collectively receive a probability of at least $\eta_i$. The characterization of rules satisfying this notion follows as a corollary from \Cref{the: dc_weak}, where every voter is considered to be a separate group.
\begin{corollary}\label{the: special_dc_weak}
    A RSCR \rscr is unanimous, strategy-proof, and satisfies \spweakf if and only if it is a probabilistic fixed ballot rule such that for all $S_2 \subseteq S_1 \subseteq N$ with $|S_1| \geq 1$, for all $i \in S_1 \setminus S_2$, and for all $x \in [1,{m-\kappa_i+1}]$, we have
    $$\delta_{S_1}([p_1, p_{x+\kappa_i-1}])-\delta_{S_2}([p_1, p_{x-1}]) \geq \eta_i.$$
\end{corollary}
The proof follows from \Cref{the: dc_weak} where $\gamma$ is a binary vector corresponding to the subset $S_1$, $\gamma'$ is a binary vector corresponding to the subset $S_2$, and $\psi_N$ is fixed to be the \repscene that chooses top $\kappa_i$ projects for every voter $i$.

\subsection{DC: Strong-IEG}\label{sec: dc_sieg}
The notion of \spstrofness ensures that some project in the top $\kappa_i$ projects receives a probability of at least $\eta_i$. Therefore, we ensure that there cannot be $\kappa_i$ consecutive projects which are allotted a probability less than $\eta_i$ each.
The characterization of rules satisfying this notion follows as a corollary from \Cref{the: dc_strong}, where every voter is considered to be a separate group.
\begin{corollary}\label{the: special_dc_strong}
	A RSCR \rscr is unanimous, strategy-proof, and satisfies \spstrof if and only if it is a probabilistic fixed ballot rule such that for all $S_2 \subseteq S_1 \subseteq N$ with $|S_1| \geq 1$, for all $i \in S_1 \setminus S_2$, and for all $p_t,p_u$ such that $p_t \triangleleft p_u$, it holds that
	$$\delta_{S_1}([p_1, p_j])-\delta_{S_2}([p_1, p_{j-1}])\! < \eta_i \; \forall p_j \in [p_t,p_u] \;\implies\; u-t \leq \kappa_i.$$
\end{corollary}
The proof follows from \Cref{the: dc_strong} where $\gamma$ is a binary vector corresponding to the subset $S_1$, $\gamma'$ is a binary vector corresponding to the subset $S_2$, and $\psi_N$ is fixed to be the \repscene that chooses top $\kappa_i$ projects for every voter $i$.

\section{Extreme Point Characterizations (EPC) of Individual-Fair Rules}\label{sec: o-fl-indextreme}
We may recall from \Cref{sec: un_sp_ep} that all the unanimous and strategy-proof RSCRs are characterized to be random min-max rules, which are the convex combinations of min-max rules defined below.
\begin{definition}\label{def: mmrr}
	A DSCR $f$ on $\mathcal{D}^n$ is a \textbf{min-max} rule if  for all $S \subseteq N$, there exists $\beta_S \in \proj$ satisfying $$\beta_{\emptyset}= p_m, \beta_N=p_1,  \mbox{ and  } \beta_T  \trianglelefteq \beta_{S} \mbox{ for all }S \subseteq T$$ such that $$f(\rprof_N)=\min_{S \subseteq N}\left [\max_{i \in S}\{\opat{i}{1}, \beta_S\}\right].$$ 	
\end{definition}
In the subsequent sections, we first present extreme point characterizations of the random min-max rules that satisfy $\boldsymbol{(\kappa_N,\eta_N)}$-\textbf{weak-IEG} and $\boldsymbol{(\kappa_N,\eta_N)}$-\textbf{strong-IEG}. Followed by it, we also present characterizations of the rules that additionally satisfy another property called \emph{anonymity}.

\subsection{EPC: Weak-IEG}\label{sec: epc_wieg}
As we saw in \Cref{sec: o-fl-indnotions}, the notion of \spweakfness ensures that the top $\kappa_i$ projects of a voter $i$ collectively receive a probability of at least $\eta_i$. Below, we characterize the random min-max rules that satisfy \spweakfness.
\begin{theorem}\label{the: special_ep_weak}
    A RSCR \rscr is unanimous, strategy-proof, and satisfies \spweakf if and only if it is a random min-max rule $\varphi=\sum_{w\in W} \lambda_w \varphi_w$ such that for all $S_1,S_2 \subseteq N$ with $S_1 \cap S_2 = \emptyset$ and $|S_2| \geq 1$, for all $i \in S_2$, and for all $x \in [1,{m-\kappa_i+1}]$,
    $$\sum_{\{w\;\mid \;p_x \trianglelefteq \beta^{\varphi_w}_{S_1},\; \beta^{\varphi_w}_{S_1 \cup S_2} \trianglelefteq p_{x+\kappa_i-1}\}}{\lambda_w} \geq \eta_i$$
\end{theorem}
\begin{proof}
    \textbf{(Only if:)} We first prove that for a random min-max rule to satisfy \spweakf, the above condition must be satisfied. Consider any $S_1, S_2, i,$ and $x$ as given in the theorem. Construct a profile $\rprof_N$ as follows: (i) $U(\opat{i}{\kappa_i},\succ_i) = [p_x,p_{x+\kappa_i-1}]$, (ii) $\opat{j}{1} \triangleleft p_x$ for all $j \in S_1$, (iii) $\opat{j}{1} \in [p_x,p_{x+\kappa_i-1}]$ for all $j \in S_2$, and (iv) $p_{x+\kappa_i-1} \triangleleft \opat{j}{1}$ for all $j \notin S_1\cup S_2$.

Since the fairness requirement of voter $i$ is met at $\rprof_N$, $\sum_{\{w: \varphi_w(\rprof_N) \in [p_x,p_{x+\kappa_i-1}]\}}{\lambda_w} \geq \eta_i$. Consider any $w$ such that $\varphi_w(\rprof_N) \in [p_x,p_{x+\kappa_i-1}]$. Since $\varphi_w$ is a min-max rule, for any $S \subseteq N$, $\max_{j \in S}\{\opat{j}{1}, \beta^{\varphi_w}_S\} \succeq \varphi_w(\rprof_N) \succeq p_x$. This implies $\beta^{\varphi_w}_{S_1} \succeq p_x$ since $\opat{j}{1} \triangleleft p_x$ for any $j \in S_1$. Similarly, since $\varphi_w$ is a min-max rule, there exists $S \subseteq N$ such that $\max_{j \in S}\{\opat{j}{1}, \beta^{\varphi_w}_S\} = \varphi_w(\rprof_N)$. For any $S \nsubseteq (S_1 \cup S_2)$, this is not possible since $\opat{j}{1} \succ p_{x+\kappa_i-1}$ for any $j \notin  S_1 \cup S_2$ and $\varphi_w(\rprof_N) \trianglelefteq p_{x+\kappa_i-1}$. Therefore, there exists $S \subseteq (S_1 \cup S_2)$ such that $\max_{j \in S}\{\opat{j}{1}, \beta^{\varphi_w}_S\} = \varphi_w(\rprof_N)$. By the definition of min-max rules, $\beta^{\varphi_w}_{S_1 \cup S_2} \trianglelefteq \beta^{\varphi_w}_S$. This implies, $\beta^{\varphi_w}_{S_1 \cup S_2} \trianglelefteq \beta^{\varphi_w}_S \trianglelefteq \varphi_w(\rprof_N) \trianglelefteq p_{x+\kappa_i-1}$. Hence, we have $p_x \trianglelefteq \beta^{\varphi_w}_{S_1}$ and $\beta^{\varphi_w}_{S_1 \cup S_2} \trianglelefteq p_{x+\kappa_i-1}$.

\textbf{(If:)}
Consider any arbitrary voter $i$. Set $S_1 = \{j \in N: \opat{j}{1} \triangleleft \min U(\opat{i}{\kappa_i},\succ_i)\}$ and $S_2 = \{j \in N: \opat{j}{1} \in U(\opat{i}{\kappa_i},\succ_i)\}$. Observe that $S_1 \cap S_2 = \emptyset$ by construction, and $|S_2| \geq 1$ since $i \in S_2$. Set $x$ such that $p_x = \min U(\opat{i}{\kappa_i},\succ_i)$. Since $|U(\opat{i}{\kappa_i},\succ_i)| = \kappa_i$, $x \in [1,m-\kappa_i+1]$. Hence, $S_1,S_2, i,$ and $x$ satisfy all the required conditions.

Consider any min-max rule $\varphi_w$ such that $p_x \trianglelefteq \beta^{\varphi_w}_{S_1}$ and $\beta^{\varphi_w}_{S_1 \cup S_2} \trianglelefteq p_{x+\kappa_i-1}$. Since $\beta^{\varphi_w}_{S_1 \cup S_2} \trianglelefteq p_{x+\kappa_i-1}$ and also $\opat{j}{1} \trianglelefteq p_{x+\kappa_i-1}$ for any $j \in (S_1 \cup S_2)$ by construction, $\max\limits_{j \in S_1 \cup S_2}\{\opat{j}{1}, \beta^{\varphi_w}_{S_1 \cup S_2}\} \trianglelefteq p_{x+\kappa_i-1}$. This implies, $\varphi_w(\rprof_N) \trianglelefteq p_{x+\kappa_i-1}$. By construction, $p_x \triangleleft \opat{j}{1}$ for any $j \notin S_1$. This implies, for any $S \nsubseteq S_1$, $p_x \triangleleft \max_{j \in S}\{\opat{j}{1}, \beta^{\varphi_w}_{S}\}$. Since $p_x \trianglelefteq \beta^{\varphi_w}_{S_1}$, by definition of min-max rule, $p_x \trianglelefteq \beta^{\varphi_w}_S$ for any $S \subseteq S_1$. This implies, $p_x \trianglelefteq \max_{j \in S}\{\opat{j}{1}, \beta^{\varphi_w}_{S}\}$ for any $S \subseteq S_1$. Therefore, for any $S \subseteq N$, $p_x \trianglelefteq \max_{j \in S}\{\opat{j}{1}, \beta^{\varphi_w}_{S}\}$. This implies, $p_x \trianglelefteq \varphi_w(\rprof_N)$. Combining this with $\varphi_w(\rprof_N) \trianglelefteq p_{x+\kappa_i-1}$ gives $\varphi_w(\rprof_N) \in U(\opat{i}{\kappa_i},\succ_i)$. Fairness requirement of $i$ is met and this completes the proof.
\end{proof}

\subsection{EPC: Strong-IEG}\label{sec: epc_sieg}
As we saw in \Cref{sec: o-fl-indnotions}, the notion of \spstrofness ensures that some project in top $\kappa_i$ projects of a voter $i$ receives a probability of at least $\eta_i$. Below, we characterize the random min-max rules that satisfy \spstrofness.
\begin{theorem}\label{the: special_ep_strong}
    A RSCR \rscr is unanimous, strategy-proof, and satisfies \spstrof if and only if it is a random min-max rule $\varphi=\sum_{w\in W} \lambda_w \varphi_w$ such that for all $S_1,S_2 \subseteq N$ with $S_1 \cap S_2 = \emptyset$ and $|S_2| \geq 1$, for all $i \in S_2$, for all $x \in [1,{m-\kappa_i+1}]$, and for all functions $f: S_2 \to [p_x,p_{x+\kappa_i-1}]$, at least one of the following conditions holds:
    \begin{enumerate}[label=(C\arabic*)]
		\item there exists $b_t \in range(f)$ such that $$\sum_{\left\{w\;\mid \; b_{t} \trianglelefteq \beta^{\varphi_w}_{S_1 \cup \{j \in S_2: f(j) \triangleleft b_t\}}\;,\; \beta^{\varphi_w}_{S_1 \cup \{j \in S_2: f(j) \trianglelefteq b_t\}}\trianglelefteq b_{t}\right\}}\lambda_w\geq \eta_i$$
		\item there exists $c \in [p_x,p_{x+\kappa_i-1}] \setminus range(f)$ such that $$\sum_{\left\{w \;\mid \; \beta^{\varphi_w}_{S_1 \cup \{j \in S_2: f(j) \trianglelefteq u\}}=c\right\}}\lambda_w\geq \eta_i,$$
		where $u=p_{x-1}$ if $c \triangleleft \min(range(f))$, $u = p_{x+\kappa_i-1}$ if $c \succ \max(range(f))$, and else $u = p_l$ such that $p_l \triangleleft c \triangleleft p_{l+1}$ and $[p_l,p_{l+1}] \cap range(f) = \{p_l,p_{l+1}\}$.
	\end{enumerate}
\end{theorem}
\begin{proof}
    \textbf{(Only if:)} Consider any $S_1,S_2,i,x,$ and $f$ as given in the theorem. Construct a profile $\rprof_N$ as follows: (i) $U(\opat{i}{\kappa_i},\succ_i) = [p_x,p_{x+\kappa_i-1}]$, (ii) $\opat{j}{1} \triangleleft p_x$ for all $j \in S_1$, (iii) $\opat{j}{1} = f(j)$ for all $j \in S_2$, and (iv) $p_{x+\kappa_i-1} \triangleleft \opat{j}{1}$ for all $j \notin S_1\cup S_2$. Since the fairness requirement of $i$ is met at $\rprof_N$, there exists $p \in [p_x,p_{x+\kappa_i-1}]$ such that $\sum_{\{w: \varphi_w(\rprof_N) = p\}}{\lambda_w} \geq \eta_i$. Consider any $\varphi_w$ such that $\varphi_w(\rprof_N) = p$.

\noindent{\textbf{Case 1:} $p \in range(f)$.}

Let $S = S_1 \cup \{j \in S_2: f(j) \triangleleft p\}$. Since by construction $\opat{j}{1} \triangleleft p_x$ for any $j \in S_1$, $\max_{j \in S}\{\opat{j}{1}\} \triangleleft p$. Since $\varphi_w(\rprof_N) = p$ and $\varphi_w$ is a min-max rule, $p \trianglelefteq \max_{j \in S}\{\opat{j}{1}, \beta^{\varphi_w}_S\}$. This implies $S \trianglelefteq \beta^{\varphi_w}_S$. Since $\varphi_w(\rprof_N) = p$, there exists some $S' \subseteq N$ such that $\max_{j \in S'}\{\opat{j}{1}, \beta^{\varphi_w}_{S'}\} = p$. This is not possible if $S' \nsubseteq S_1 \cup \{j \in S_2: f(j) \trianglelefteq p\}$ since $\opat{j}{1} \succ p$ for any $j \notin (S_1 \cup \{j \in S_2: f(j) \trianglelefteq p\})$. So, there exists $S' \subseteq S_1 \cup \{j \in S_2: f(j) \trianglelefteq p\}$ such that $\max_{j \in S'}\{\opat{j}{1}, \beta^{\varphi_w}_{S'}\} = p$. This implies, $\beta^{\varphi_w}_{S'} \trianglelefteq p$. By the definition of min-max rule, $\beta^{\varphi_w}_{S_1 \cup \{j \in S_2: f(j) \trianglelefteq p\}} \trianglelefteq \beta^{\varphi_w}_{S'} \trianglelefteq p$. Hence, (C1) holds for $b_t = p$.

\noindent{\textbf{Case 2:} $p \notin range(f)$.}

\textit{Case 2.1:} $p \triangleleft \min(range(f))$.

Since $\varphi_w(\rprof_N) = p$ and $\varphi_w$ is a min-max rule, $p \trianglelefteq \max_{j \in S_1}\{\opat{j}{1}, \beta^{\varphi_w}_{S_1}\}$. This implies $p \trianglelefteq \beta^{\varphi_w}_{S_1}$ since $\opat{j}{1} \triangleleft p_x$ for any $j \in S_1$. Since $\varphi_w(\rprof_N) = p$, there exists $S \subseteq N$ such that $\max_{j \in S}\{\opat{j}{1}, \beta^{\varphi_w}_S\} = p$. Since $p \triangleleft \min(range(f))$, $\opat{j}{1} \succ p$ for any $j \notin S_1$. Therefore, $S \subseteq S_1$. Since $\max_{j \in S}\{\opat{j}{1}, \beta^{\varphi_w}_S\} = p$ and $\opat{j}{1} \triangleleft p$ for any $j \in S$, $\beta^{\varphi_w}_S \trianglelefteq p$. By the definition of min-max rule, $\beta^{\varphi_w}_{S_1} \trianglelefteq \beta^{\varphi_w}_S \trianglelefteq p$. Combining this with $p \trianglelefteq \beta^{\varphi_w}_{S_1}$ implies (C2) is satisfied with $c = p$ and $u = p_{x-1}$.

\textit{Case 2.2:} $\max(range(f)) \triangleleft p$.

Since $\varphi_w(\rprof_N) = p$ and $\varphi_w$ is a min-max rule, $p \trianglelefteq \max_{j \in S_1 \cup S_2}\{\opat{j}{1}, \beta^{\varphi_w}_{S_1 \cup S_2}\}$. This implies $p \trianglelefteq \beta^{\varphi_w}_{S_1 \cup S_2}$ since $\opat{j}{1} \triangleleft p_x$ for any $j \in S_1$ and $\opat{j}{1} \triangleleft \max(range(f)) \triangleleft p$ for any $j \in S_2$. Since $\varphi_w$ is a min-max rule, there exists $S \subseteq N$ such that $\max_{j \in S}\{\opat{j}{1}, \beta^{\varphi_w}_S\} = p$. Since $p_{x+\kappa_i-1} \triangleleft \opat{j}{1}$ for any $j \notin (S_1 \cup S_2)$, there exists $S \subseteq (S_1 \cup S_2)$ such that $\max_{j \in S}\{\opat{j}{1}, \beta^{\varphi_w}_S\} = p$. Since $\max(range(f)) \triangleleft p$, $\opat{j}{1} \triangleleft p$ for any $j \in S$. This implies, $\beta^{\varphi_w}_S = p$. So, by the definition of min-max rule, $\beta^{\varphi_w}_{S_1 \cup S_2} \trianglelefteq p$. Combining this with $\beta^{\varphi_w}_{S_1 \cup S_2} \succeq p$ implies (C2) is satisfied with $c = p$ and $u = p_{x+\kappa_i-1}$.

\textit{Case 2.3:} $p$ lies in between two projects in $range(f)$.

Select $b_l,b_{l+1} \in range(f)$ such that $b_l \triangleleft p \triangleleft b_{l+1}$ and $[b_l,b_{l+1}] \cap range(f) = \{b_l,b_{l+1}\}$. Let $S = S_1 \cup \{j \in S_2: f(j) \trianglelefteq b_l\}$. Since $\varphi_w(\rprof_N) = p$ and $\varphi_w$ is a min-max rule, $p \trianglelefteq \max_{j \in S}\{\opat{j}{1}, \beta^{\varphi_w}_{S}\}$. This implies $p \trianglelefteq \beta^{\varphi_w}_{S}$ since $\opat{j}{1} \trianglelefteq b_l \triangleleft p$ for any $j \in S$. Since $\varphi_w(\rprof_N) = p$, there exists $S' \subseteq N$ such that $\max_{j \in S'}\{\opat{j}{1}, \beta^{\varphi_w}_{S'}\} = p$. This is not possible if $S' \nsubseteq S$ since $p \triangleleft b_{l+1} \trianglelefteq \opat{j}{1}$ for all $j \notin S$. This implies, there exists $S' \subseteq S$ such that $\max_{j \in S'}\{\opat{j}{1}, \beta^{\varphi_w}_{S'}\} = p$. Since $\opat{j}{1} \trianglelefteq b_l \triangleleft p$ for any $j \in S$, $\beta^{\varphi_w}_{S'} = p$. By the definition of median rule, $\beta^{\varphi_w}_{S} \trianglelefteq \beta^{\varphi_w}_{S'} = p$. Combining this with $p \trianglelefteq \beta^{\varphi_w}_{S}$ implies (C2) is satisfied with $c = p$ and $u = b_l$.

\textbf{(If:)} Consider any arbitrary voter $i$. Set $S_1 = \{j \in N: \opat{j}{1} \triangleleft \min U(\opat{i}{\kappa_i},\succ_i)\}$ and $S_2 = \{j \in N: \opat{j}{1} \in U(\opat{i}{\kappa_i},\succ_i)\}$. Observe that $S_1 \cap S_2 = \emptyset$ by construction, and $|S_2| \geq 1$ since $i \in S_2$. Set $x$ such that $p_x = \min U(\opat{i}{\kappa_i},\succ_i)$. Since $|U(\opat{i}{\kappa_i},\succ_i)| = \kappa_i$, $x \in [1,m-\kappa_i+1]$. Define a function $f$ such that $f(j) = \opat{j}{1}$ for every $j \in S_2$. Hence, $S_1,S_2, i,x,$ and $f$ satisfy all the required conditions.

suppose (C1) holds for some $b_t$ as given. Take a $\varphi_w$ such that $b_{t} \trianglelefteq \beta^{\varphi_w}_{S_1 \cup \{j \in S_2: f(j) \triangleleft b_t\}}$ and $\beta^{\varphi_w}_{S_1 \cup \{j \in S_2: f(j) \trianglelefteq b_t\}}\trianglelefteq b_{t}$. For any $j \notin (S_1 \cup S_2)$, $b_t \trianglelefteq p_{x+\kappa_i-1} \triangleleft \opat{j}{1}$. For any $j \in S_2$ such that $b_t \trianglelefteq f(j)$, $b_t \trianglelefteq \opat{j}{1}$ by construction. From $b_{t} \trianglelefteq \beta^{\varphi_w}_{S_1 \cup \{j \in S_2: f(j) \triangleleft b_t\}}$, any set $S \subseteq (S_1 \cup \{j \in S_2: f(j) \triangleleft b_t\})$ will have $b_{t} \trianglelefteq \max_{j \in S}\{\opat{j}{1}, \beta^{\varphi_w}_{S}\}$. Combining all the above together, we have, for all $S \subseteq N$, $b_{t} \trianglelefteq \max_{j \in S}\{\opat{j}{1}, \beta^{\varphi_w}_{S}\}$. Since $b_t \in range(f)$, there exists $j \in S_2$ such that $f(j) = b_t$. Since $\beta^{\varphi_w}_{S_1 \cup \{j \in S_2: f(j) \trianglelefteq b_t\}}\trianglelefteq b_{t}$, this implies, $\max_{j \in S}\{\opat{j}{1}, \beta^{\varphi_w}_{S}\} = b_{t}$ for $S = (S_1 \cup \{j \in S_2: f(j) \triangleleft b_t\})$. Hence, $\varphi_w(\rprof_N) = b_t$. Fairness requirement of $i$ is met.

suppose (C2) holds for some $c$ and $u$ as given in the theorem. If $u = p_{x-1}$, $S_1 \cup \{j \in S_2: f(j) \trianglelefteq u\} = S_1$. Consider any $\varphi_w$ such that $\beta^{\varphi_w}_{S_1} = c$. Then, $\max_{j \in S_1}\{\opat{j}{1}, \beta^{\varphi_w}_{S_1}\} = c$ since $\opat{j}{1} \triangleleft p_x$ for all $j \in S_1$ by construction. Since $c \triangleleft \min(range(f))$, $c \triangleleft \opat{j}{1}$ for any $j \notin S_1$. This implies, $\varphi_w(\rprof_N) = c$. Else if $u = p_{x+\kappa_i-1}$, $S_1 \cup \{j \in S_2: f(j) \trianglelefteq u\} = S_1 \cup S_2$. Consider any $\varphi_w$ such that $\beta^{\varphi_w}_{S_1 \cup S_2} = c$. Then, $\max_{j \in  S_1 \cup S_2}\{\opat{j}{1}, \beta^{\varphi_w}_{S_1 \cup S_2}\} = c$. For any $j \notin S_1 \cup S_2$, $c \triangleleft \opat{j}{1}$ by construction. This implies, $\varphi_w(\rprof_N) = c$. Finally, suppose $u = b_l$. Let $S = S_1 \cup \{j \in S_2: f(j) \trianglelefteq u\}$. For any $j \in S$, $\opat{j}{1} \trianglelefteq b_l \triangleleft c$. Since $\beta^{\varphi_w}_{S} = c$, $\max_{j \in S}\{\opat{j}{1}, \beta^{\varphi_w}_{S}\} = c$. For any $j \notin S$, $c \triangleleft b_{l+1} \trianglelefteq \opat{j}{1}$ by construction. This implies, $\varphi_w(\rprof_N) = c$. Hence, $c$ is allocated at least $\eta_i$. Fairness requirement of $i$ is met and this completes the proof.
\end{proof}
\subsection{Total Anonymity of the Voters}\label{sec: total_anonymity}
Up until now, we discussed the rules which are not necessarily anonymous across the groups. We now look at a more restricted special case where, in addition to groups having exactly one voter, anonymity across all the voters is required. This requirement is referred to as \emph{anonymity} of RSCR in the literature (i.e., permutation of the preferences of the voters do not change the outcome). This property is formally defined below. Let $\Sigma$ be the set of all permutations on $N$. For $\sigma \in \Sigma$ and $\rprof_N\in \mathcal{D}^n$, we define $\rprof^{\sigma}_N$ as $(\succ_{\sigma(1)},\ldots,\succ_{\sigma(n)})$. 
\begin{definition} \label{def: special_ta_anonymity}
	A RSCR $\varphi:\mathcal{D}^n \to \Delta \proj$ is said to be \textbf{anonymous} if for all $\rprof_N \in \mathcal{D}^n$ and all $\sigma \in \Sigma$, we have $$\varphi(\rprof_N)=\varphi(\rprof^{\sigma}_N).$$
\end{definition}
\subsubsection{EPC: Total Anonymity}\label{sec: median}
All the unanimous, strategy-proof, and anonymous RSCRs are characterized to be random median rules \cite{pycia2015decomposing}. In fact, as remarked in \Cref{obs: pfgbr}, when a single group has all the voters, our probabilistic fixed group ballot rules become equivalent to random median rules. Random median rules are convex combinations of DSCRs called median rules \cite{moulin1980strategy}. A median rule operates by selecting the median of all top-ranked projects of the voters and additional $n+1$ dummy projects fixed apriori as parameters by the PB organizer.

\begin{definition}\label{def: special_ta_median}
	A DSCR $f$ on $\mathcal{D}^n$ is a \textbf{median} rule if there are $n+1$ (dummy) projects  $p_1=\beta^{f}_0 \!\trianglelefteq\! \beta^{f}_1 \!\trianglelefteq\! \cdots \!\trianglelefteq\! \beta^{f}_{n-1} \!\trianglelefteq\! \beta^{f}_n=p_m$ such that for any profile $\rprof_N$, $f$ selects $$\mathsf{med}\!\left(\!\succ_1(1),\ldots,\succ_n(1),\beta^{f}_0,\beta^{f}_1,\ldots,\beta^{f}_n\!\right).$$
\end{definition}

\begin{example}\label{eg: median}
    Consider the scenario in \Cref{eg: dc_basic}. As there are four voters, a median rule $\varphi$ must have three parameters as $\beta^{\varphi}_{1}=p_1$, $\beta^{\varphi}_{2}=p_2$, and $\beta^{\varphi}_{3}=p_2$ ($\beta^{\varphi}_0$ and $\beta^{\varphi}_4$ are always fixed at $p_1$ and $p_3$ respectively). The median rule $\varphi$ selects the median of the set $\{p_1,p_3,p_2,p_3,p_1,p_1,p_2,p_2,p_3\}$, which is $p_2$.
\end{example}
A \textbf{random median rule} is a convex combination of median rules. Using this, we now present an extreme point characterization of the unanimous, strategy-proof, and anonymous RSCRs that satisfy our \spweakfness and \spweakfness notions.
\subsubsection{EPC: Weak-IEG and Total Anonymity}\label{sec: wiegta}
 The notion of \spweakfness ensures that the top $\kappa_i$ projects of a voter $i$ collectively receive a probability of at least $\eta_i$. Our characterization states that for any $\kappa_i$-sized interval $I$ of $\proj$ and any $r \in [0,n-1]$, the total probability given to the median rules having some element of $I$ in between their $r$-th and $(r+1)$-th parameters must be at least $\eta_i$. Let $\mathcal{I}_\kappa$ denote all the intervals of size $\kappa$. Before we present the characterization, we prove a necessary lemma.
 \begin{lemma}\label{lem: nrnr1}
For any median rule $\varphi$, a single-peaked profile $\rprof_N$, and a voter $i$ such that $|\{j\in N \mid \opat{j}{1}\trianglelefteq \opat{i}{1}\}|=r$,
$$\big[\beta^{\varphi}_{n-r},\beta^{\varphi}_{n-r+1}\big]\cap U(\opat{i}{\kappa_i},\succ_i) \neq \emptyset \implies \varphi(\rprof_N)\in U(\opat{i}{\kappa_i},\succ_i).$$
\end{lemma}
\begin{proof}
If $\opat{i}{1}\in \big[\beta^{\varphi}_{n-r},\beta^{\varphi}_{n-r+1}\big]$, then by the definition of a median rule, $\varphi(\rprof_N)=\opat{i}{1}$. Thus, the claim holds. Now suppose $\opat{i}{1}\triangleleft \beta^{\varphi}_{n-r}$. By the definition of a median rule this means, $\varphi(\rprof_N)\trianglelefteq \beta^{\varphi}_{n-r}$ and $\opat{i}{1}\triangleleft \varphi(\rprof_N)$. Combining these two observations, we have $\opat{i}{1}\triangleleft \varphi(\rprof_N)\trianglelefteq \beta^{\varphi}_{n-r}$. Additionally, as $\big[\beta^{\varphi}_{n-r},\beta^{\varphi}_{n-r+1}\big]\cap U(\opat{i}{\kappa_i},\succ_i)\neq \emptyset$ and $\succ_i$ is single-peaked with $\opat{i}{1}\triangleleft\beta^{\varphi}_{n-r}$, it follows that $[\opat{i}{1},\beta^{\varphi}_{n-r}]\subseteq U(\opat{i}{\kappa_i},\succ_i)$. This, together with $\opat{i}{1}\triangleleft \varphi(\rprof_N)\trianglelefteq \beta^{\varphi}_{n-r}$, implies $\varphi(\rprof_N)\in U(\opat{i}{\kappa_i},\succ_i)$. 
	
Finally, suppose $\beta^{\varphi}_{n-r+1}\triangleleft \opat{i}{1}$. By the definition of a median rule this means, $ \beta^{\varphi}_{n-r+1}\trianglelefteq \varphi(\rprof_N)\triangleleft \opat{i}{1}$. Moreover, as $\big[\beta^{\varphi}_{n-r},\beta^{\varphi}_{n-r+1}\big]\cap U(\opat{i}{\kappa_i},\succ_i)\neq \emptyset$ and $\succ_i$ is single-peaked with $\beta^{\varphi}_{n-r+1}\triangleleft \opat{i}{1}$, we have $[\beta^{\varphi}_{n-r+1},\opat{i}{1}]\subseteq U(\opat{i}{\kappa_i},\succ_i)$. All the above observations imply that $\varphi(\rprof_N)\in U(\opat{i}{\kappa_i},\succ_i)$.
\end{proof}

We are now ready to characterize all the fair rules.
 \begin{theorem}\label{the: special_ta_weak}
     A RSCR \rscr is unanimous, strategy-proof, anonymous, and satisfies \spweakf if and only if it is a random median rule $\varphi=\sum_{w\in W} \lambda_w \varphi_w$ such that for all $ r \in [0,n-1]$, all $i \in N$, and all $I  \in \mathcal{I}_{\kappa_i}$,  $$\sum_{\big\{w \; \mid \;   \big[ \beta^{\varphi_w}_{r}, \beta^{\varphi_w}_{r+1}\big]\cap I \neq \emptyset\big\}}\lambda_w\geq \eta_i.$$
\end{theorem}
\begin{proof}
    \textbf{(Only if:)} We first show that any random median rule $\varphi=\sum_{w\in W}\lambda_w\varphi_w$ satisfying \spweakf also satisfies the condition in the theorem. Take $r\in \{0,\ldots,n-1\}$, $i \in N$, and $x\in \{1,\ldots,m-\kappa_i+1\}$. Consider a profile $\hat{\rprof}_N$ such that (a) the top-ranked project of exactly $n-r-1$ voters is $p_1$, (b) $U(\opat{i}{\kappa_i},\succ_i)=[p_x,p_{x+\kappa_i-1}]$, and (c) $p_m$ is the top-ranked project of the remaining voters. By the definition of a median rule, for all $w\in W$,
	\begin{equation}\label{eq_1}
		\varphi_w(\hat{\rprof}_N)\in [\beta^{\varphi_w}_r,\beta^{\varphi_w}_{r+1}].
	\end{equation}
	Moreover, since $\varphi$ satisfies \spweakf, we have 
	\begin{equation}\label{eq_2}
		\sum_{\big\{w \in W \mid  \varphi_w(\hat{\rprof}_N)\in[p_x,p_{x+\kappa_i-1}]\big\}}\lambda_w\geq \eta_i.
	\end{equation}
	Combining \Cref{eq_1} and \Cref{eq_2}, we have 
	\begin{equation*}
		\sum_{\big\{w\in W \mid   \big[ \beta^{\varphi_w}_{r}, \beta^{\varphi_w}_{r+1}\big]\cap\big[p_x,p_{x+\kappa_i-1}\big]\neq \emptyset\big\}}\lambda_w\geq \eta_i.
	\end{equation*}

\textbf{(If:)} Suppose a random median $\varphi=\sum_{w\in W} \lambda_w \varphi_w$ satisfies the condition in the theorem. We show that $\varphi$ is \spweakf. Take $\rprof_N\in \mathcal{D}^n$ and a voter $i\in N$. Set $r = |\{j\in N \mid \opat{j}{1}\trianglelefteq \opat{i}{1}\}|$. Since $\succ_i$ is single-peaked, $U(\opat{i}{\kappa_i},\succ_i)\in \mathcal{I}_{\kappa_i}$. Let $U(\opat{i}{\kappa_i},\succ_i)=[p_t,p_{t+\kappa_i-1}]$ for some $p_t\in \proj$.  This implies $t\in [1,m-\kappa_i+1]$. Moreover, as $r\in  \{1,\ldots,n\}$, $n-r\in \{0,\ldots,n-1\}$.  Therefore, by the condition of the theorem, we have  $$\sum_{\big\{w \mid   \big[ \beta^{\varphi_w}_{n-r}, \beta^{\varphi_w}_{n-r+1}\big]\cap\big[p_t,p_{t+\kappa_i-1}\big]\neq \emptyset\big\}}\lambda_w\geq \eta_i.$$ By Lemma \ref{lem: nrnr1}, $\big[ \beta^{\varphi_w}_{n-r}, \beta^{\varphi_w}_{n-r+1}\big]\cap U(\opat{i}{\kappa_i},\succ_i)\neq \emptyset$ implies $\varphi^w(\rprof_N)\in U(\opat{i}{\kappa_i},\succ_i)$. Combining these observations together, we have $\varphi_{U(\opat{i}{\kappa_i},\succ_i)}(\rprof_N)\geq \eta_i$. Thus, the fairness requirement of voter $i$ is met.
\end{proof}
\subsubsection{EPC: Strong-IEG and Total Anonymity}\label{sec: siegta}
Our \spstrofness ensures that some project in the top $\kappa_i$ projects receives a probability of at least $\eta_i$. 
\begin{theorem}\label{the: special_ta_strong}
     A RSCR \rscr is unanimous, strategy-proof, anonymous, and \spstrof if and only if it is a random median rule $\varphi=\sum_{w\in W} \lambda_w \varphi_w$ such that for all $r \in [1,n]$, all $s \in [0,n-r]$, all $i \in N$, all $I \in \mathcal{I}_{\kappa_i}$, and all $\{b_{r},\ldots,b_{r+s}\} \subseteq I$ with $b_r\trianglelefteq \cdots \trianglelefteq b_{r+s}$,  at least one of the following holds:
	\begin{enumerate}[label=(C\arabic*)]
		\item there exists $c \in \{b_r,\ldots, b_{r+s}\}$ such that $$\sum_{\{w\;\mid \; \beta^{\varphi_w}_{n-t}\trianglelefteq b_{t} \trianglelefteq \beta^{\varphi_w}_{n-t+1} \mbox{ for some } b_t=c\}}\lambda_w\geq \eta_i,$$
		\item there exists $c \in I \setminus \{b_r,\ldots, b_{r+s}\}$ such that $$\sum_{\{w \;\mid \; \beta^{\varphi_w}_{n-u}=c\}}\lambda_w\geq \eta_i,$$
		where $u=r-1$ if $c \triangleleft b_{r}$, $u = r+s$ if $b_{r+s} \triangleleft c $, and else $u$ is such that $b_u \triangleleft c \triangleleft b_{u+1}$.
	\end{enumerate}
\end{theorem}
\begin{proof}
    \textbf{(Only if:)} Let $\varphi=\sum_{w\in W}\lambda_w \varphi_w$ be a \spstrof random median rule. We show that it satisfies (C1) and (C2). Take $r\in [1,n]$, $s\in [0,n-r]$, $i \in N$, $I\in \mathcal{I}_{\kappa_i}$, and $\{b_{r},\ldots,b_{r+s}\} \subseteq I$ with $b_r\trianglelefteq \cdots \trianglelefteq b_{r+s}$. Construct a preference profile $\rprof_N$ such that (i) $p_1$ is the top-ranked project of exactly $r-1$ voters, (ii) $U(\opat{i}{\kappa_i},\succ_i)=I$, (iii) $b_{r},\ldots,b_{r+s}$ are the top-ranked projects of any $s+1$ voters including $i$, and (iv) $p_m$ is the top-ranked project of the remaining voters. Since $\varphi$ is \spstrof, there exists a project $p\in I$ such that $\varphi_p(\rprof_N)\geq \eta_i$. This, together with $\varphi=\sum_{w\in W}\lambda_w\varphi_w$, implies, 
\begin{equation}\label{eq_-1}
\sum_{\{w\mid \varphi_w(\rprof_N)=p\}}\lambda_w\geq \eta_i.
\end{equation}

Suppose $p\in \{b_{r},\ldots,b_{r+s}\}$. This means there exist $s_1,s_2$ with $0\leq s_1\leq s_2\leq s$ such that $b_{r+s_1}=\cdots=b_{r+s_2}=p$. Let $\hat{\varphi}$ be a median rule such that $\hat{\varphi}(\rprof_N)=p$. We claim there exists $t\in [s_1,s_2]$ such that $\beta^{\hat{\varphi}}_{n-r-t}\trianglelefteq p\trianglelefteq \beta^{\hat{\varphi}}_{n-r-t+1}$. Suppose not. Then, either $\beta^{\hat{\varphi}}_{n-r-s_1+1}\triangleleft p$ or $p \triangleleft \beta^{\hat{\varphi}}_{n-r-s_2}$. Assume $\beta^{\hat{\varphi}}_{n-r-s_1+1}\triangleleft p$. This means $p_1\triangleleft p$. Moreover, at least $r+s_1-1$ voters have top-ranked projects before $p$. Combining these observations with the definition of a median rule, we have $\hat{\varphi}(\rprof_N)\triangleleft p$, which is a contradiction. Similarly, if $p \triangleleft \beta^{\hat{\varphi}}_{n-r-s_2}$, we can show that $p \triangleleft \hat{\varphi}(\rprof_N)$. Therefore,  there exists $t\in [s_1,s_2]$ such that $\beta^{\hat{\varphi}}_{n-r-t}\trianglelefteq a\trianglelefteq \beta^{\hat{\varphi}}_{n-r-t+1}$. This, together with \Cref{eq_-1}, implies  $$\sum_{\{w\;\mid \; \beta^{\varphi_w}_{n-t}\trianglelefteq b_{t} \trianglelefteq \beta^{\varphi_w}_{n-t+1} \mbox{ for some } b_t=a\}}\lambda_w\geq \eta_i.$$ Hence, (C1) holds with $c=p$.

Now suppose $p\in I\setminus \{b_{r},\ldots,b_{r+s}\}$. This implies either $p \in [\min I, b_r)$, or $p \in (b_{r+s}, \max I]$, or there exists  $b_{r+t} \in I$  such that $b_{r+t} \triangleleft p \triangleleft b_{r+t+1}$. Suppose  $p \in [\min I, b_r)$. Since $r>1$ implies there are exactly $r-1$ voters with top-ranked projects before $I$, for any median rule $\varphi$ with $\varphi(\rprof_N)=p$, we have $\beta^{\varphi}_{n-r+1} = p$. This, together with \Cref{eq_-1}, implies $$\sum_{\{w \;\mid \; \beta^{\varphi_w}_{n-r+1}=p\}}\lambda_w\geq \eta_i,$$  Thus, (C2) holds for $u = r-1$ and $c=p$. Similarly, if $p \in (b_{r+s}, \max I]$, then  (C2) holds for $u = r+s$ and $c=p$. Finally, if there exists $b_{r+t} \in I$ such that  $b_{r+t} \triangleleft p \triangleleft b_{r+t+1}$, then $\varphi_w(\rprof_N)=p$ implies $\beta^{\varphi_w}_{n-r-t} = p$. Therefore, (C2) holds for $u = r+t$ and $c=p$.

\textbf{(If:)}  Let $\varphi=\sum_{w\in W}\lambda_w\varphi_w$ be a random median rule such that for all $r \in [1,n]$, all $s \in [0,n-r]$, all $i \in N$, all $I \in \mathcal{I}_{\kappa_i}$, and all $\{b_{r},\ldots,b_{r+s}\} \subseteq I$ with $b_r\trianglelefteq \cdots \trianglelefteq b_{r+s}$,  at least one of (C1) and (C2) holds. We show that $\varphi$ is \spstrof. Take a preference profile $\rprof_N$ and an arbitrary voter $i$. We have to show that 
\begin{equation}\label{eq_-2}
	 \exists p\in U(\opat{i}{\kappa_i},\succ_i) \mbox{ such that } \varphi_p(\rprof_N)\geq \eta_i.
\end{equation}
Let $r = |j \in N: \opat{j}{1} \triangleleft \min{U(\opat{i}{\kappa_i},\succ_i)}|+1$, $T = \{j \in N: \opat{j}{1} \in U(\opat{i}{\kappa_i},\succ_i)\}$, and $s = |T|-1$. Let $I = U(\opat{i}{\kappa_i},\succ_i)$. Let $b_r,\ldots,b_{r+s}$ be the top-ranked projects of voters in $T$ (sorted in increasing order, with repetition). Now, for these values, at least one of (C1) and (C2) holds.
 Suppose (C1) holds. This means that there exists $c \in \{b_r,\ldots, b_{r+s}\}$ such that $$\sum_{\{w\;\mid \; \beta^{\varphi_w}_{n-t}\trianglelefteq b_{t} \trianglelefteq \beta^{\varphi_w}_{n-t+1} \mbox{ for some } b_t=c\}}\lambda_w\geq \eta_i.$$
 By the definition of a median rule, for all $\hat{\varphi}\in \{\varphi_w\mid w\in W \mbox{ and } \beta^{\varphi_w}_{n-t}\trianglelefteq b_t \trianglelefteq \beta^{\varphi_w}_{n-t+1} \mbox{ for some } b_t=c\}$, we have $\hat{\varphi}(\rprof_N)=c$. Therefore, by (C1), $\varphi_{c}(\rprof_N)\geq \eta_i$. Hence, \Cref{eq_-2} holds. Now suppose (C2) holds. This means there exists $c \in U(\opat{i}{\kappa_i},\succ_i) \setminus \{b_r,\ldots, b_{r+s}\}$ such that $$\sum_{\{w \;\mid \; \beta^{\varphi_w}_{n-u}=c\}}\lambda_w\geq \eta_i,$$
		where $u=r-1$ if $c \triangleleft b_{r}$, $u = r+s$ if $b_{r+s} \triangleleft c$, and else $u$ is such that $b_u \triangleleft c \triangleleft b_{u+1}$. First assume $u=r-1$. Since $\min U(\opat{i}{\kappa_i},\succ_i) \trianglelefteq c \triangleleft b_{r}$ and exactly $r-1$ voters have their top-ranked project before $\min U(\opat{i}{\kappa_i},\succ_i)$, for any $\varphi\in \{\varphi_w\mid w\in W \mbox{ and } \beta^{\varphi_w}_{n-p+1}=c\}$, $\varphi(\rprof_N)=c$. Hence, \Cref{eq_-2} holds with $p=c$. Now assume $u = r+s$ where $b_{r+s} \triangleleft c$. Since $b_{r+s}\triangleleft c \trianglelefteq \max U(\opat{i}{\kappa_i},\succ_i)$ and exactly $n-r-s$ voters have their top-ranked projects after $\max U(\opat{i}{\kappa_i},\succ_i)$, we have for all $\varphi\in \{\varphi_w\mid w\in W \mbox{ and } \beta^{\varphi_w}_{n-p-q}=c\}$, $\varphi(\rprof_N)=c$. Hence, \Cref{eq_-2} holds with $p=c$. Finally assume, $u$ is such that $b_u \triangleleft c \triangleleft b_{u+1}$. Since the number of voters with their top-ranked projects before $c$ is exactly $u$, for all $\varphi\in \{\varphi_w\mid w\in W \mbox{ and } \beta^{\varphi_w}_{n-q}=c\}$, $\varphi(\rprof_N)=c$. Hence, \Cref{eq_-2} holds with $p=c$.
\end{proof}
\section{Compliant Representation Scenarios Satisfying Additional Desiderata}\label{sec: o-fl-compliant}
Let us recall that a \repscene $\psi_G$ is said to be compliant with \repranges $\kappa_G$, if every representative function $\psi_q$ is \topi and the outcome of $\psi_q$ is always $\kappa_q$-sized interval (\Cref{def: compliance}). In this section, we discuss some \repscenes that are \compliant with $\kappa_G$ and also satisfy some additional desirable properties.

Given a parameter $k$ and the preferences of voters over the projects, the problem of multi-winner voting is to aggregate these preferences and choose exactly $k$ projects. This is very well studied in computer science and economics literature \cite{faliszewski2017multiwinner}. Each representative function $\psi_q$ can be viewed as a multi-winner voting rule with parameter $\kappa_q$. Motivated by this, we impose some well studied properties of multi-winner voting rules on the representative functions in the \repscene, in addition to ensuring that it is complaint with $\kappa_G$. The first property we look at is anonymity, which ensures that the permutation of preferences of voters in a \community does not change the representatives chosen for that \community.

\begin{definition}\label{def: rep_anonymity}
    A representative function $\psi_q: \mathcal{D}^{|N_q|} \to {2^\proj}$ is said to be {anonymous} if for any permutation $\pi$ of $N_q$ and $\rprof_{N_q}\in \mathcal{D}^{|N_q|}$, we have $\correspondence(\rprof_{N_q})=\correspondence(\rprof_{\pi(N_q)})$ where $\rprof_{\pi(N)}=(\succ_{\pi(1)},\ldots,\succ_{\pi(n)})$. A \repscene is said to be \textbf{anonymous} if $\psi_q$ is anonymous for every $q \in G$.
\end{definition}

The next property we define requires that for any group $N_q$ and a profile $\rprof_{N_q}$, there exists a voter whose top ranked project is selected by $\psi_q$.

\begin{definition}\label{def: rep_topcontaining}
    A representative function $\psi_q: \mathcal{D}^{|N_q|} \to {2^\proj}$ is said to be top-containing if for any $\rprof_{N_q} \in  \mathcal{D}^{|N_q|}$, there exists $i \in N_q$ such that $\opat{i}{1} \in \psi_q(\rprof_{N_q})$. A \repscene is said to be \textbf{top-containing} if $\psi_q$ is top-containing for every $q \in G$.
\end{definition}

The next property we look at is candidate monotonicity \cite{elkind2017properties}, which requires that for any group $N_q$, if a project selected by $\psi_q$ is shifted forward in the preference of some voter in $N_q$, that project should continue to be selected by $\psi_q$.

\begin{definition}\label{def: rep_candidatem}
    A representative function $\psi_q: \mathcal{D}^{|N_q|} \to {2^\proj}$ is said to satisfy candidate monotonicity if for any $\rprof \in  \mathcal{D}^{|N_q|}$, $p \in \psi_q(P)$, and $\rprof' \in  \mathcal{D}^{|N_q|}$ obtained by shifting $p$ forward in some preference in $\rprof$, we have $p \in \psi_q(\rprof')$. A \repscene is said to satisfy \textbf{candidate monotonicity} if $\psi_q$ satisfies candidate monotonicity for every $q \in G$.
\end{definition}

Finally, we introduce the concept of pareto-efficiency, a well-known property that ensures no voter can be improved without causing a detriment to another voter. To formally define pareto-efficiency, we need to consider various methods through which a voter can compare two intervals of projects. We use the term `x-comparison' to denote the specific way in which a voter compares intervals based on their preferences. The formal definition of pareto-efficiency is then established in relation to this x-comparison.

\begin{definition}\label{def: rep_pe}
    A representative function $\psi_q: \mathcal{D}^{|N_q|} \to {2^\proj}$ is said to be pareto-efficient w.r.t. x-comparison if for every $i \in N_q$ and every $\kappa_q$ sized interval $I$ of $\proj$, it holds that whenever $\correspondence(\rprof_{N_q})$ is preferred over $I$ by voter $i$ w.r.t. x-comparison, we have some voter $j \in N_q$ who prefers $I$ over $\correspondence(\rprof_{N_q})$ w.r.t. x-comparison. A \repscene is said to be \textbf{pareto-efficient} if $\psi_q$ is pareto-efficient for every $q \in G$.
\end{definition}

We examine two types of x-comparisons: b-comparison (also called best comparison) and lb-comparison (also called lexicographic best comparison). In the case of b-comparison, a voter favors an interval $I$ over another interval $I'$ if the highest-ranked project in $I$ is at least as desirable as the highest-ranked project in $I'$. If the highest-ranked projects in both intervals are the same, the voter considers both $I$ and $I'$ to be equally preferred. On the other hand, lb-comparison operates similarly to b-comparison, except that two intervals cannot be considered equally preferred by a voter. Therefore, in the event of a tie between two intervals, the next highest-ranked projects in the intervals are compared.


\begin{lemma}\label{the: rep_ape_b}
    A \repscene $\psi_G$ is compliant with $\kappa_G$ and pareto-efficient w.r.t. b-comparison or lb-comparison if and only if for every $q \in G$ and $\rprof_{N_q} \in \mathcal{D}^{|N_q|}$, $\psi_q$ selects a $\kappa_q$-sized interval $I$ such that $I \subseteq [\tau_1(\rprof_{N_q}),\tau_{|N_q|}(\rprof_{N_q})]$ (or $[\tau_1(\rprof_{N_q}),\tau_{|N_q|}(\rprof_{N_q})] \subseteq I$).
\end{lemma}
\begin{proof}
    \textbf{(Only if:)} First, we will prove that for pareto-efficiency to hold one of the two conditions must necessarily be satisfied. For the sake of contradiction, assume that $\psi_q$ selects an interval $I$ at $\rprof_{N_q}$ such that $I$ and $[\tau_1(\rprof_{N_q}),\tau_{|N_q|}(\rprof_{N_q})]$ are not contained in one another. This implies, exactly one of $\min\{I\} \triangleleft \tau_1(\rprof_{N_q})$ and $\tau_{|N_q|}(\rprof_{N_q}) \triangleleft \max\{I\}$ holds. We will prove that pareto-efficiency does not hold. First, suppose $\min\{I\} \triangleleft \tau_1(\rprof_{N_q})$ and $\max\{I\} \trianglelefteq \tau_{|N_q|}(\rprof_{N_q})$. If $\tau_{|N_q|}(\rprof_{N_q}) = \max\{I\}$, $[\tau_1(\rprof_{N_q}),\tau_{|N_q|}(\rprof_{N_q})] \subseteq I$ giving a contradiction. If $\max\{I\} \triangleleft \tau_{|N_q|}(\rprof_{N_q})$, consider the interval $I'$ obtained by shifting $I$ by one position towards $\tau_1(\rprof_{N_q})$. Clearly, $I$ is preferred over $I'$ by any voter $i$ w.r.t. both b-comparison and lb-comparison. Thus, pareto-efficiency does not hold. It can be argued similarly for the case with $\tau_{|N_q|}(\rprof_{N_q}) \triangleleft \max\{I\}$ (we obtain $I'$ by shifting $I$ by one position towards $\tau_{|N_q|}(\rprof_{N_q})$).
    
\textbf{(If:)} We now prove the sufficiency part. That is, we prove that if one of the two conditions is satisfied, then the \repscene is pareto-efficient and \compliant. Compliance with $\kappa_G$ follows directly since $I$ is $\kappa_q$ sized interval. We will prove that pareto-efficiency is satisfied w.r.t. b-comparison and lb-comparison. If $I \subseteq [\tau_1(\rprof_{N_q}),\tau_{|N_q|}(\rprof_{N_q})]$, shifting $I$ towards $\tau_1(\rprof_{N_q})$ will make it less preferable for voter having maximum top-ranked project. Shifting $I$ towards $\tau_{|N_q|}(\rprof_{N_q})$ will make it less preferable for voter having minimum top-ranked project. Thus, it satisfies pareto-efficiency. Now, suppose $[\tau_1(\rprof_{N_q}),\tau_{|N_q|}(\rprof_{N_q})] \subseteq I$. This implies that the top-ranked projects of all the voters in $N_q$ are selected as representatives, thereby making $I$ optimal. Thus, pareto-efficiency is satisfied.
\end{proof}

\begin{theorem}\label{the: rep_apecm}
    The following \repscenes are compliant with $\kappa_G$, anonymous, top-containing, candidate monotone, and pareto-efficient w.r.t. b-comparison and lb-comparison: for every $q \in G$,
    \begin{enumerate}[label=(R\arabic*)]
        \item \label{rule1} $\psi_q$ takes a parameter $r \in [1,|N_q|]$ and selects $\kappa_q$ consecutive projects starting from $\tau_r(\rprof_{N_q})$.
        \item \label{rule2} $\psi_q$ selects a $\kappa_q$ sized interval that has maximum number of projects in $\{\opat{i}{1}: i\in N_q\}$ (ties can be broken w.r.t. increasing order of the starting points of the intervals).
        \item \label{rule3} $\psi_q$ selects a $\kappa_q$ sized interval that maximizes the number of voters in $N_q$ whose top ranked project is in the interval.
        \item \label{rule4} $\psi_q$ selects $\kappa_q$ consecutive projects starting from a project $p$ that maximizes the number of voters in $N_q$ who rank it at top (i.e., $\argma{p \in \proj}{|\{i \in N_q: \opat{i}{1} = p\}|}$).
    \end{enumerate}
\end{theorem}
\begin{proof}
Anonymity and top-containingness follow from the definitions of the rules. Pareto-efficiency follows from \Cref{the: rep_ape_b}. We now prove candidate monotonicity.

First, consider rule (R1) with parameter $r$. Suppose the outcome interval is $I$ at $\rprof_{N_q}$. Now a project $p \in I$ is shifted forwards in the preference of a voter $i$ to obtain $\rprof'_{N_q}$. If $p$ is ${(r-1)}^{\text{st}}$ top-ranked project in $\rprof_{N_q}$ and $r^{\text{th}}$ top-ranked project in $\rprof'_{N_q}$, outcome interval at $\rprof'_{N_q}$ starts from $p$ and thus $p$ continues to be selected. Else, the outcome at $\rprof'_{N_q}$ continues to be $I$ and $p$ hence continues to be selected.

Consider rule (R2). Let the outcome interval at $\rprof_{N_q}$ be $I$. A project $p \in I$ is shifted forwards in the preference of a voter $i$ to obtain $\rprof'_{N_q}$. If at $\rprof'_{N_q}$, $p \notin \{\opat{}{1}: i\in N_q\}$, then the outcomes at $\rprof_{N_q}$ and $\rprof'_{N_q}$ remain the same. Similarly, if $p \in \{\opat{}{1}: i\in N_q\}$, $I$ continues to win.

Consider rule (R3). Let the outcome interval at $\rprof_{N_q}$ be $I$. Now a project $p \in I$ is shifted forwards in the preference of a voter $i$ to obtain $\rprof'_{N_q}$. The number of voters who rank $p$ at the top either increases by $1$ or remains the same. For the sake of contradiction, assume that at $\rprof'_{N_q}$, the number of voters having their top-ranked projects in $I$ decreases. This implies some project in $I$ has been shifted to second position from first, which in turn implies that the number of voters having their top-ranked project as $p$ increases by $1$. Thus, the number of voters having their top-ranked project in $I$ increases or remains the same.

The proof for (R4) is similar to that of (R1).
\end{proof}

Note that all the above four \repscenes give the output $\psi_1(\rprof_{N_1}) = \{p_1\}$ and $\psi_2(\rprof_{N_2}) = \{p_2,p_3\}$ in \Cref{eg: dc_fair}. Please note that these are only some of the \repscenes compliant with $\kappa_G$ and there can be many more such examples.

\begin{note}
    When the groups are singleton (i.e., each group has exactly one voter), the \repscene described in \Cref{sec: o-fl-indnotions} which chooses top $\kappa_i$ projects of each voter $i$ as her representatives is compliant with $\kappa_N$, anonymous, top-containing, pareto-efficient, and also candidate monotone (\repscenes in \ref{rule1} and \ref{rule4} become equivalent to this since all groups are singleton).
\end{note}
\section{Examples of Fair Rules}\label{sec: o-fl-examples}
In this section, we present a few simple looking and easily understandable RSCRs that satisfy \strofness under some specific conditions. All the RSCRs we mention in this section are expressed as convex combinations of DSCRs. Throughout this section, we denote $\min{\!\{\kappa_G\}}$ by $\boldsymbol{\kmin}$ and $\max{\!\{\eta_G\}}$ by $\boldsymbol{\emax}$.
\subsection*{Case I: $\boldsymbol{\suml{q \in G}{\eta_q} \leq 1}$ {and} $\boldsymbol{\kmin \geq \frac{m+1}{2}}$}
If the given condition is satisfied, for any $q \in G$, $\kappa_q \geq \frac{m+1}{2}$. Thus, $m - \kappa_q + 1 \leq \kappa_q$ and $p_{m - \kappa_q + 1} \trianglelefteq p_{\kappa_q}$. Now, we construct a RSCR that is unanimous, strategy-proof, \community-wise anonymous, and satisfies \strof. For each $q \in G$, we construct a group min-max rule $\phi_q$ whose parameters are as follows: $\beta^q_{\underline{\gamma}} = a_m$, $\beta^q_{\overline{\gamma}} = a_1$, and for every other $\gamma \in \Gamma$, $\beta^q_{\gamma} \in [p_{m-\kappa_q+1},p_{\kappa_q}]$. During the construction, it also needs to be ensured that whenever $\gamma \gg \gamma'$, $\beta^q_{\gamma} \trianglelefteq \beta^q_{\gamma'}$. One trivial way to do this is to set all the parameters except $\beta^q_{\underline{\gamma}}$ and $\beta^q_{\overline{\gamma}}$ to the same project which is in $[p_{m-\kappa_q+1},p_{\kappa_q}]$. Clearly, there are many other trivial ways to achieve it (start setting a few parameters to $p_{\kappa_q}$ and progressively move towards $p_{m-\kappa_q+1}$ as the $\gamma$ gets dominant).

Consider the RSCR $\sum_{q \in G}{\eta_q\phi_q}$. Take any $q, \gamma^0,\gamma^1,\ldots,\gamma^{\kappa_q},$ and $p_x$ as used in Theorem 5.6. Let $I = \{p_x,\ldots,p_{x+\kappa_q-1}\}$. Since $|I| = \kappa_q$, $\min\{I\} \trianglelefteq p_{m-\kappa_q+1}$ and $p_{\kappa_q} \trianglelefteq \max\{I\}$. Since $m - \kappa_q + 1 \leq \kappa_q$ and $p_{m - \kappa_q + 1} \trianglelefteq p_{\kappa_q}$. Therefore, $[p_{m-\kappa_q+1},p_{\kappa_q}] \subseteq I$. By the construction of $\phi_q$, for any $\gamma$ other than $\overline{\gamma}$ and $\underline{\gamma}$, $\beta^q_\gamma \in [p_{m-\kappa_q+1},p_{\kappa_q}]$. Number of projects in $[p_{m-\kappa_q+1},p_{\kappa_q}]$ is strictly lesser than $\kappa_q$. Therefore, among $\gamma^1,\ldots,\gamma^{\kappa_q}$, at least two consecutive values should have the same parameter from $[p_{m-\kappa_q+1},p_{\kappa_q}]$. Let the value of this parameter be $p_j$. Clearly $j \geq x$ since $x \leq m-\kappa_q+1$. Therefore, the condition in \Cref{the: ep_strong} is satisfied for $t=j-x$.

\subsection*{Case II: $\boldsymbol{\lfloor\frac{m}{\kmin}\rfloor\cdot\emax \leq 1}$}
Let $r$ be the reminder when $m$ is divided by \kmin. That is, $r = m - \kmin\lfloor\frac{m}{\kmin}\rfloor$. We select any $d \in [\kmin-r+1,\kmin]$. For each $i \in [0, \lfloor\frac{m}{\kmin}\rfloor-1]$, we construct a group min-max rule $\phi_i$ whose parameters are as follows: $\beta^i_{\underline{\gamma}} = p_m$, $\beta^i_{\overline{\gamma}} = p_1$, and for every other $\gamma \in \Gamma$, $\beta^i_{\gamma} = p_{d+i\kmin}$.

Consider the RSCR $\sum_{i}{\emax\phi_i}$. Take any $q, \gamma^0,\gamma^1,\ldots,\gamma^{\kappa_q},$ and $p_x$ as used in Theorem 5.6. Let $i = \lceil\frac{x}{\kmin}\rceil$.
Since $|\{p_x,\ldots,p_{x+\kappa_q-1}\}| = \kappa_q \geq \kmin$, by definition, $x \leq d+i\kmin \leq x+\kappa_q-1$. Let $d+i\kmin = x+t$ where $t \in [0,\kappa_q-1]$. Clearly, the condition in \Cref{the: ep_strong} is satisfied for this $t$.

\subsection*{Case III: $\boldsymbol{\emax \leq \frac{1}{n}}$, $\kmin < \frac{m+1}{2}$, \textbf{and }$\boldsymbol{\psi_G}$\textbf{ is top-containing}}
Recall \Cref{def: rep_topcontaining}. Let us define a function $s: \Gamma \to [0,n]$ as $s(\gamma) = \sum_{q \in G}{\gamma_q}$. Clearly, $s(\gamma) = n$ only when $\gamma = \overline{\gamma}$, and $s(\gamma) = 0$ only when $\gamma = \underline{\gamma}$. For each $i \in [1,n]$, we construct a group min-max rule $\phi_i$ whose parameters are as follows: $\beta^i_{\underline{\gamma}} = p_m$, $\beta^i_{\overline{\gamma}} = p_1$, and for every other $\gamma \in \Gamma$, $\beta^i_{\gamma} \leq p_{\kmin}$ if $s(\gamma) \geq i$ and $\beta^i_{\gamma} \geq p_{m-\kmin+1}$ otherwise. During this construction, it can be easily ensured that whenever $\gamma \gg \gamma'$, $\beta^q_{\gamma} \trianglelefteq \beta^q_{\gamma'}$ (similar to Case I).

Consider the RSCR $\sum_{i}{\frac{1}{n}\phi_i}$. Take any $q, \gamma^0,\gamma^1,\ldots,\gamma^{\kappa_q},$ and $p_x$ as used in Theorem 5.6. Since $\kmin < \frac{m+1}{2}$, $p_{\kmin} \triangleleft p_{m-\kmin+1}$. We know that $[p_x,p_{x+\kappa_q-1}] \cap [p_{\kmin},p_{m-\kmin+1}] \neq \emptyset$. Since $\psi_q$ is top-containing, $\gamma^{\kappa_q} \neq \gamma^0$. Let $T = s(\gamma^{\kappa_q})$. Consider the DSCF $\varphi_T$. By construction, $\beta^{\varphi_T}_{\gamma^{\kappa_q}} \trianglelefteq p_{\kmin}$ and $p_{m-\kmin+1} \trianglelefteq \beta^{\varphi_T}_{\gamma^{0}}$. Since $\gamma^{\kappa_q} \gg \ldots \gg \gamma^1 \gg \gamma^0$, we can conclude that there exits $t \in [0,\kappa_q-1]$ such that $\beta^{\varphi_T}_{\gamma^{t+1}} \trianglelefteq p_{x+t} \trianglelefteq \beta^{\varphi_T}_{\gamma^{t}}$. Therefore, the condition in \Cref{the: ep_strong} is satisfied for this $t$.
\section{Conclusion and Discussion}\label{sec: o-fl-conclusion}
This chapter studies the case where costs of the projects are totally flexible and each project can be allocated any amount. We normalize the budget to be $1$, thereby making our participatory budgeting problem equivalent to random social choice problem. A random social choice rule aggregates the preferences of voters and outputs a probability distribution over the projects as an outcome. Probability corresponding to each project may be interpreted as the fraction of budget allocated to the project. We propose fairness notions for both individuals and groups and characterize all random social choice rules that satisfy these properties in addition to other desiderata pursued in the literature.

We consider the model where there exists a natural partition of voters into groups based on attributes such as gender, race, economic status, and location. The goal is to be fair to each of these groups. We propose the notion of group-wise anonymity to ensure fairness within each group and the notions of weak-GEG and strong-GEG to ensure fairness across the groups. The proposed group-fairness notions are generalizations of existing individual-fairness notions and moreover provide non-trivial outcomes for strictly ordinal preferences, unlike the existing group-fairness notions. We characterize all the unanimous and strategy-proof social choice rules, both deterministic and random, that satisfy the three proposed notions. For the random rules, we give direct characterizations as well as extreme point characterizations (in which we express random rules as convex combinations of deterministic rules). We also give simpler characterizations for the special case where each group has only one voter and finally illustrate a few examples of families of fair rules under certain conditions.

It will be interesting to see how the rules satisfying weak-GEG and strong-GEG are structured in the domains beyond single-peakedness. Another promising question is to identify some tractable group-fair rules for cases that are not covered in \Cref{sec: o-fl-examples}.

\chapter{Summary and Future Directions} \label{chap: conclusion}
In this thesis, we studied participatory budgeting (PB), which is a voting paradigm to distribute a divisible resource called budget, among several projects, by aggregating the preferences of voters over the projects.

Preferences of voters could be elicited in two ways: (i) dichotomous preferences, where every voter approves/likes a subset of projects (ii) ordinal preferences, where every voter reports a strict/weak ranking of the projects. The amount that needs to be allocated to each project is referred to as its cost, and there are various approaches to imposing constraints on these costs: (i) the restricted cost model, which assigns a fixed cost value to each project, and (ii) the flexible cost model, which allows for multiple permissible cost values for each project.

The contributions of this thesis are presented in two parts. In \Cref{part: dichotomous}, we studied participatory budgeting under dichotomous preferences. \Cref{chap: d-re} in this part studies the restricted costs model, while the \Cref{chap: d-fl} studies the flexible costs model. In \Cref{part: ordinal} of this thesis, we studied participatory budgeting under ordinal preferences. \Cref{chap: o-re} in this part studies the restricted costs model, while the \Cref{chap: o-fl} studies the flexible costs model. In summary, this thesis introduced a multitude of fresh PB rules for the four considered PB models, along with a range of novel utility notions, axiomatic properties, and fairness notions.

\section{Summary of the Contributions}\label{summary}
\subsection{\Cref{part: dichotomous}: Dichotomous Preferences}\label{summary: d}
\subsubsection{\Cref{chap: d-re}: Restricted Costs: Egalitarian Participatory Budgeting}\label{summary: d-re}
We introduced the egalitarian rule, Maxmin Participatory Budgeting (\mmpb), that maximizes the egalitarian welfare and conducted a thorough computational and axiomatic analysis of the same.

On the computational front, we demonstrated that \mmpb is strongly NP-hard and identified special cases where the problem becomes tractable. As a part of this, we introduced two novel parameters to study fixed parameter tractability of \mmpb. We also proposed a LP-rounding based approximation algorithm and empirically proved that it gives optimal outcome in real-world. We proved that all these results carry over to minmax objective with minimal changes. We bounded the loss in approximation guarantee incurred due to strategy-proofness. On the axiomatic front, we analysed \mmpb with respect to existing axioms and also the novel fairness axiom, maximal participation, introduced by us.

\subsubsection{\Cref{chap: d-fl}: Flexible Costs: Welfare Maximization when Projects have Multiple Degrees of Sophistication}\label{summary: d-fl}
We studied PB under dichotomous preferences when the project costs are partially flexible. We introduced a generalization of dichotomous preferences, called ranged dichotomous preferences, where each voter approves a range of costs for each project. We generalized two utility notions defined for PB under restricted costs to our model. We also proposed two other utility notions unique to our model. We analyzed all the corresponding utilitarian welfare maximizing rules computationally and axiomatically.

Our computational part strengthened all the existing positive results, and also introduced several new parameterized tractability results (FPT, parameterized FPTAS) by introducing and studying novel parameters. On the axiomatic front, we introduced several axioms for our model with ranged approval votes and investigated their satisfiability by our utilitarian PB rules. We concluded that two of our rules, \nrule and \drule, exhibit remarkable axiomatic performance.

\subsection{\Cref{part: ordinal}: Ordinal Preferences}\label{summary: o}
\subsubsection{\Cref{chap: o-re}:Restricted Costs: Welfare Maximization and Fairness under Incomplete Weakly Ordinal Preferences}\label{summary: o-re}
We studied PB under incomplete weakly ordinal preferences and investigated utilitarian welfare maximization and fairness. For the welfare maximization, we extended the existing rules in the literature on dichotomous preferences and strictly ordinal preferences to propose a new family of rules called dichotomous translation rules and a rule called the PB-CC rule. We proved that our extensions perform as good as, and even better than, their parent rules. For the fairness, we identified some major drawbacks suffered by the existing fairness notions and proposed two families of rules, average rank guarantee rules and share guarantee rules, that fill this gap. We studied the computational complexity of these families of rules.
\subsubsection{\Cref{chap: o-fl}: Flexible Costs: Characterization of Group-Fair and Individual-Fair Rules under Single-Peaked Preferences}\label{summary: o-fl}
We studied PB under strictly ordinal single-peaked preferences for the case where costs of the projects are totally flexible and each project can be allocated any amount. We normalize the budget to be $1$, thereby making our participatory budgeting problem equivalent to random social choice problem. A random social choice rule aggregates the preferences of voters and outputs a probability distribution over the projects as an outcome. Probability corresponding to each project may be interpreted as the fraction of budget allocated to the project. We assume that the voters are partitioned into groups and defined three novel group-fairness notions: (i) group-wise anonymity (ii) weak group entitlement guarantee (weak-GEG) and (iii) strong group entitlement guarantee (strong-GEG). We characterized all random social choice rules that satisfy these properties in addition to other desiderata pursued in the literature. We also studied a special case where every voter is considered to be a group in its own right, and characterized all the individually-fair rules.

\section{Open Research Directions}\label{future}
\subsection{Dichotomous Preferences and Restricted Costs}
Utilitarian welfare and fairness objectives are well studied in the literature. This thesis studies the maxmin and minmax objectives for egalitarian welfare maximization. Going further, studying leximin objective is an interesting direction. Our focus also extends to addressing intractability through the use of FPT and the design of approximation algorithms. A promising avenue lies in identifying special structures within instances that render our objectives polynomial time tractable.
\subsection{Dichotomous Preferences and Flexible Costs}
This thesis initiated the investigation into utilitarian welfare maximization in the partially flexible costs model, marking the initial stride in this direction. Moving forward, exploring rules that optimize egalitarian welfare within this model holds significant potential. Investigating the impact of strategic agents is also a promising direction. Additionally, formulating novel fairness axioms that are relevant to this model and identifying participatory budgeting rules that satisfy these axioms present captivating avenues for further exploration.
\subsection{Ordinal Preferences and Restricted Costs}
This thesis introduced various families of PB rules under incomplete weakly ordinal preferences with the goals of maximizing utilitarian welfare and achieving fairness. Regarding utilitarian welfare, exploring new utility notions or approaching utility as a multi-objective problem hold promise as future directions. In terms of fairness, investigating tractable special cases with structured preferences would be interesting since the proposed ARSG rules were shown to be computationally hard. Additionally, studying manipulability of the proposed rules and egalitarian welfare optimization in the context of PB under incomplete weakly ordinal preferences are other avenues for future research.
\subsection{Ordinal Preferences and Flexible Costs}
This thesis provided characterizations of group-fair and individual-fair rules in the single-peaked domain. Exploring the behavior of these rules beyond the single-peaked domain poses an interesting avenue for future research. Another intriguing question arises when considering the possibility of relaxing the constraints of unanimity or strategy-proofness in characterizing fair rules. Finding tractable fair rules for scenarios not covered in this thesis would be highly valuable. Additionally, the pursuit of characterizing rules that maximize utilitarian welfare or egalitarian welfare presents an open and promising direction for further investigation.

\bibliographystyle{plainnat}
\bibliography{references}
\blankpage
\addcontentsline{toc}{chapter}{Appendix}
\chapter*{Appendix}
\addtocounter{chapter}{1}
\appendix



\section{Restricted Costs: Egalitarian Participatory Budgeting}\label{chap: first_app}
\subsection{Minmax Objective}\label{app: minimax}
We define the disutility of a voter to be $\bud - \uof{}{S}$. This notion effectively handles the scenarios where the voter, being a taxpayer and contributor to the budget, pays a fixed amount irrespective of the cost of its approval vote. The voter expects superior projects (which it believes are worth the tax paid) to be proposed and also to be selected. Unsatisfactory proposals or non-funded approved projects make a voter unhappy. In such a scenario, it is reasonable to assume that every dollar of the budget not used for approved projects causes disutility to the voter since it amounts to a waste.

The egalitarian goal now is to minimize the disutility of the voter with maximum disutility. From the objective function, it is clear that minimizing maximum disutility is equivalent to maximizing minimum utility. However, the approximation results do not transfer. To derive the approximation guarantee of our algorithm \lpalgo, we first formulate the ILP.

\begin{align}
    \nn
    \min\; &q\\
    \label{eq: ipl1mm}
    \text{subject to }&q \geq \bud - \suml{p \in \aof{}}{\cxof{p}} \quad \forall i \in \voters\\\nn
    &\suml{p \in \proj}{\cxof{p}} \leq \bud\\
    \label{eq: ilp2mm}
    &x_p \in \curly{0,1} \quad \forall p \in \proj\\\nn
    &q \geq 0
\end{align}

We consider the LP-relaxation for minimax objective by relaxing \Cref{eq: ilp2mm} to $0 \leq x_p \leq 1$. Now, again consider the same \lpalgo algorithm: Solve the relaxed LP to get $(q^*,x^*)$. Let $S = \emptyset$ be the initial outcome. Add the project with the highest value of $\cxsof{p}$ to $S$, followed by the one with the second highest value and so on till the next project does not fit.

\begin{theorem}
The algorithm \lpalgo achieves an approximation guarantee of $\left(\;2-\frac{1}{\ho}\right)$ for \hcbp instances for the objective of minimizing the maximum disutility.
\end{theorem}
\begin{proof}
We obtain the below step following the steps similarly as in the proof of \Cref{lem: algobound}:
\begin{align}
    \nn
    \suml{p \in \aof{}}{\cxsof{p}} &\leq \frac{\bud+\firstfrac{}\;\uof{}{S}}{\firstfrac{}+1}\\
    \label{app: eq1mm}
    \lb \frac{\bud - \uof{}{S}}{1 + \lb \frack{} \rb} \rb &\leq \bud - \suml{p \in \aof{}}{\cxsof{p}}
\end{align}
Let $j = \arg\max_{i}{(\fdisuof{}{S})}$ and the maximum disutility achieved by \lpalgo be $ALG = \fdisuof{j}{S}$. Let $\opt\;$ be the value of the maximum disutility in the optimal solution of \mmpb objective. Since the optimal solution also belongs to the feasible region of the relaxed LP, we know that $q^* \leq \opt$. From \Cref{eq: ipl1mm}, we know that $q^* \geq \bud - \suml{p \in \aof{j}}{\cxsof{p}}$. Combining these, we have,
\begin{align}
    \nn
    \opt &\geq \lb \frac{\bud - \uof{}{S}}{1 + \lb \frack{} \rb} \rb\\
    \nn
    ALG &\leq \lb 1 + \lb \frack{} \rb \rb\opt\\
    \nn
    &\leq \lb 2 - \revfrac{j} \rb\opt
\end{align}
Since $Y_j \subseteq S$, $|S| - |Y_j| \leq |S| \leq \ho$. Since the instance satisfies \hcbp, $|S| \geq \lo > \ha \geq |\aof{j}|$. Therefore $|S| > |\aof{j}|$, which implies $|S| - |\aof{j}| \geq 1$. Therefore, $\revfrac{j} \geq \frac{1}{\ho}$ and $ALG \leq \lb 2 - \frac{1}{\ho} \rb \cdot \opt$
\end{proof}
\begin{note}
    It is worth highlighting that all the other results in \Cref{chap: d-re}, except \Cref{the: exsp}, hold also for minmax objective, without changes. Axiomatic results and hardness results follow by their definition. \Cref{sec: d-re-fpt} also follows from the same proofs since the number of constraints or variables or the coefficients in the ILP remain same for both maxmin and minmax objectives.
\end{note}
\clearpage

\addtocounter{app}{1}
\section{Restricted Costs: Welfare Maximization and Fairness under
Incomplete Weakly Ordinal Preferences}\label{chap: third_app}
\subsection{Axiomatic Analysis of Dichotomous Translation Rules with Cost-Worthy Translation Scheme}\label{app: cwaxioms}
Next, we look at the axiomatic analysis. Proofs for the axioms in \Cref{axioms: generic} follow without any changes.
\begin{theorem}\label{the: ctaxioms1}
	When $f(S) = |S|$, the rule $\langle CT,f \rangle$ satisfies all the properties mentioned in \Cref{tab: o-re-results} except limit monotonicity. It satisfies limit monotonicity if and only if $\cwparaof{1} < \cwparaof{m}+2$.
\end{theorem}
\begin{proof}
    \begin{enumerate}[(a)]
        \item \textbf{Splitting Monotonicity:} The proof is exactly the same as (a)-(i) in \Cref{the: splittingmp}, except for the argument that for any $x_1 \in X$, if $x \in A_i$ in \instance, then $x_1 \in A_i$ in $\instance'$. This follows for $CT$ scheme since all the projects in $X$ cost lesser than $x$. Thus, splitting monotonicity is satisfied.
        \item \textbf{Discount Monotonicity:} This follows since the set of approved projects of each voter in both \instance and $\instance'$ remain the same. Thus, discount monotonicity is satisfied.
        \item \textbf{Limit Monotonicity:} First, let us consider the case where $\cwparaof{1} < \cwparaof{m}+2$. Without loss of generality, assume that the first few entries in the worth vector are $\cwparaof{m}+1$ and the remaining are \cwparaof{m}. Consider the rule $\langle \CC_\cwpara,f \rangle$. Suppose $f(S) = |S|$. Take an instance \instance with budget \bud, and a project $x \in \winners{\cwr}{}$. Construct $\instance'$ by increasing the budget to $\bud+1$. Since $\rprof$ and costs are unchanged, the score of any $S$ is the same in both \instance and $\instance'$. Consider any set $S' \in \ruleof{\cwr}{'}$ and any set $S_x \in \ruleof{\cwr}{}$ such that $x \in S_x$. For the sake of contradiction, assume $S_x \notin \ruleof{\cwr}{'}$. So, there exists $a \in S'$ and $b \in S_x$ such that $\scoreof{a}>\scoreof{b}$. However, since $S_x\in\ruleof{\cwr}{}$, $(S_x \setminus \curly{b}) \cup \curly{a}$ must not be feasible in \instance (else, it would have strictly more utility than $S_x$). Hence, $\cof{S_x}+\cof{a}-\cof{b} > \bud$. Since $\cof{S_x} \leq \bud$, this implies, $\cof{b} < \cof{a}$. The score of any project whose cost is at most \cwparaof{m} is $n$. The score of any project whose cost is greater than $\cwparaof{m}+1$ is $0$. The score of any project whose cost is exactly $\cwparaof{m}+1$ will belong to $[0,n]$. Since $\scoreof{a} > \scoreof{b}$, one of these holds: (i) $\cof{a} \leq \cwparaof{m}$ and $\cof{b} \geq \cwparaof{m}+1$ (ii) $\cof{a} = \cwparaof{m}+1$ and $\cof{b} \geq \cwparaof{m}+1$. That is, $\cof{a} \leq \cwparaof{m}+1$ and $\cof{b} \geq  \cwparaof{m}+1$. This contradicts $\cof{b} < \cof{a}$. Thus, limit monotonicity is satisfied.

        Now, let us look at the case where $\cwparaof{1} \geq \cwparaof{m}+2$. Construct an instance \instance as follows: budget $\bud = \cwparaof{1}+\cwparaof{2}-1$, set of projects $\proj = \curly{p_1,\ldots,p_m}$, two voters have rankings $p_1 \succ \ldots \succ p_m$ and $p_1 \succ p_2 \succ p_m \succ p_4 \ldots \succ p_{m-1} \succ p_3$ respectively. Set $\cof{p_i}$ as \cwparaof{i} if $i \in \{1,2\}$ and as $\cwparaof{i}+1$ otherwise. Note that no project in $\proj$ costs exactly $\bud+1$. This is because, for any project $p_i$ such that $i \notin \curly{1,2}$, $\cof{p_i} = \cwparaof{i}+1$. Assume that this is equal to $\bud+1$, i.e., $\cwparaof{1}+\cwparaof{2}-\cwparaof{i}=1$. Since $\cwparaof{2} \geq \cwparaof{i} \geq 0$, this implies $\cwparaof{1} = 1$. This contradicts $\cwparaof{1} \geq \cwparaof{m}+2$. Also, note that $p_1$ and $p_2$ have a score of $2$ each while $p_m$ might have a score of $1$. All other projects have a score of $0$.

         Consider the set $\curly{p_2,p_m}$. We know that $\cof{\curly{p_2,p_m}} = \cwparaof{2}+\cwparaof{m}+1$. Since $\cwparaof{m} \leq \cwparaof{1}-2$, $\cof{\curly{p_2,p_m}} \leq \bud$. Since \curly{p_1,p_2} is infeasible, score of \curly{p_2,p_m} is optimal. Therefore, $\curly{p_2,p_m} \in \ruleof{\cwr}{}$ and $p_m \in \winners{\cwr}{}$. Now, construct an instance $\instance'$ by increasing \bud to $\cwparaof{1}+\cwparaof{2}$. Now, \curly{p_1,p_2} is feasible and has the optimal maximum score of $4$. Hence, $p_m \notin \winners{\cwr}{'}$ and limit monotonicity is not satisfied.
         \item \textbf{Inclusion Maximality:} This follows from the fact that $|S|$ is subset monotone. In other words, utility from a set $S$ is at least as much as the utility from a subset of $S$. Thus, inclusion maximality is satisfied.
         \item \textbf{Candidate Monotonicity:} Depending on \cwparaof{j}, \cwparaof{j-1}, \cof{x}, and \cof{x'}, $x$ will be approved by the same voters and possibly by one more voter. Likewise, $x$ will be approved by the same number of voters or possibly by exactly one voter less than that number. The scores of all other projects will remain the same in $\instance'$. Since $x \in \winners{\cwr}{}$, there exists $S_x \in \ruleof{\cwr}{}$ such that $x \in S_x$. By definition, the utility of any set is the sum of scores of all projects in that set. Hence, some set containing $x$ will continue to win. Therefore, $x \in \winners{\cwr}{'}$. Thus, candidate monotonicity is satisfied.
         \item \textbf{Non-Crossing Monotoncity:} The proof is exactly same as that of \Cref{the: noncrossingmp}.
         \item \textbf{Pro-Affordability:}  Proof is exactly same as that of (a)-(ii) in the proof of \Cref{the: proaffordp}.
    \end{enumerate}
This completes the proof.
\end{proof}
\begin{theorem}\label{the: ctaxioms2}
	When $f(S) = c(S)$, the rule $\langle CT,f \rangle$ satisfies all the properties mentioned in \Cref{tab: o-re-results} except discount monotonicity, limit monotonicity, and pro-affordability. It satisfies limit monotonicity and pro-affordability if and only if $\cwparaof{1} \leq 1$.
\end{theorem}
\begin{proof}
    \begin{enumerate}[(a)]
        \item \textbf{Splitting Monotonicity:} The proof is exactly the same as (a)-(ii) in \Cref{the: splittingmp}, except for the argument that for any $x_1 \in X$, if $x \in A_i$ in \instance, then $x_1 \in A_i$ in $\instance'$. This follows for $CT$ scheme since all the projects in $X$ cost lesser than $x$. Thus, splitting monotonicity is satisfied.
        \item \textbf{Discount Monotonicity:} Consider an instance \instance with budget $\bud = \cwparaof{1}$, projects $\proj = \curly{p_1,p_2,d_1,\ldots,d_{m-2}}$ such that $\cof{p_1} = \cof{p_2} = \cwparaof{1}$ and for any $i \in \curly{1,\ldots,m-2}$ $\cof{d_i} = \cwparaof{i+2}+1$. Let there be two voters with preferences $p_1 \succ_1 p_2 \succ_1 d_1 \succ_1 \ldots \succ_1 d_{m-2}$ and $p_2 \succ_2 p_1 \succ_1 d_1 \succ_1 \ldots \succ_1 d_{m-2}$ respectively. Then, $\winners{\cwr}{} = \curly{p_1,p_2}$. However, if $\cof{p_1}$ is reduced by $1$, $\ruleof{\cwr}{'} = \curly{\curly{p_2}}$ and $p_1 \notin \winners{\cwr}{'}$. Thus, discount monotonicity is not satisfied.
        \item \textbf{Limit Monotonicity:} First, we look at the case where $\cwparaof{1} \leq 1$. The translation scheme selects only unit cost projects into each $A_i$ and increasing \bud won't change it.

        Finally, we look at the case where $\cwparaof{1} > 1$. Consider an instance \instance with $\bud = 2\cwparaof{1}-2$. Say there are three projects $\curly{p_1,p_2,p_3}$ respectively costing $\curly{\cwparaof{1},\cwparaof{1}-1,1}$. Say there are $m-3$ projects $d_1,\ldots,d_{m-3}$ such that $\cof{d_i} = \cwparaof{i+1}+1$. Say there is only one voter and her ranking is $\curly{p_1,p_2,p_3} \succ d_1 \succ \ldots \succ d_{m-3}$. Then, $\ruleof{\cwr}{} = \curly{\curly{p_1,p_3}}$. If the budget is increased by $1$ to get $\instance'$, $\ruleof{\cwr}{'} = \curly{\curly{p_1,p_2}}$. Hence, $p_3 \notin \winners{\cwr}{'}$. Thus, limit monotonicity is not satisfied.
        \item \textbf{Inclusion Maximality:} This follows from the fact that $\cof{S}$ is subset monotone. In other words, utility from a set $S$ is at least as much as the utility from any subset of $S$. Thus, inclusion maximality is satisfied.
        \item \textbf{Candidate Monotonicity:} Depending on \cwparaof{j}, \cwparaof{j-1}, \cof{x}, and \cof{x'}, $x$ will be approved by the same voters and possibly by one more voter. Likewise, $x$ will be approved by the same number of voters or possibly by exactly one voter less than that number. The scores of all other projects will remain the same in $\instance'$. Since $x \in \winners{\cwr}{}$, there exists $S_x \in \ruleof{\cwr}{}$ such that $x \in S_x$. By definition, the utility of any set is the sum of scores of all projects in that set. Hence, some set containing $x$ will continue to win. Therefore, $x \in \winners{\cwr}{'}$. Thus, candidate monotonicity is satisfied.
        \item \textbf{Non-Crossing Monotoncity:} The proof is exactly same as that of \Cref{the: noncrossingmp}.
        \item \textbf{Pro-Affordability:} The case of $\cwparaof{1} \leq 1$ is trivial since all the approved projects for each voter cost at most $1$ and there cannot be any project whose cost is lesser than $1$.

        Finally, consider the case where $\cwparaof{1} > 1$. Take any arbitrary voter $i$. If $\rof{x'} < \trunkkcc{\ebpara}{S_x} \leq \rof{x}$, then $\trunkkcc{\ebpara}{S'} < \trunkkcc{\ebpara}{S_x}$. Else, $\trunkkcc{\ebpara}{S'} = \trunkkcc{\ebpara}{S_x}$. Hence, utility of $S'$ is at least that of $S_x$. Since $S_x \in \ruleof{\cwr}{}$, $S' \in \ruleof{\cwr}{}$.
    \end{enumerate}
This completes the proof.
\end{proof}
\begin{theorem}\label{the: ctaxioms3}
	When $f(S) = \bool(|S|>0)$, the rule $\langle CT,f \rangle$ satisfies all the properties mentioned in \Cref{tab: o-re-results} except limit monotonicity and non-crossing monotonicity. It satisfies limit monotonicity if and only if $\cwparaof{1} < \cwparaof{m}+2$ and non-crossing monotonicity if and only if $\cwparaof{1} = \cwparaof{m}$.
\end{theorem}
\begin{proof}
    \begin{enumerate}[(a)]
        \item \textbf{Splitting Monotonicity:} The proof is exactly the same as (a)-(iii) in \Cref{the: splittingmp}, except for the argument that for any $x_1 \in X$, if $x \in A_i$ in \instance, then $x_1 \in A_i$ in $\instance'$. This follows for $CT$ scheme since all the projects in $X$ cost lesser than $x$. Thus, splitting monotonicity is satisfied.
        \item \textbf{Discount Monotonicity:} This follows since the set of approved projects of each voter in both \instance and $\instance'$ remain the same. Thus, discount monotonicity is satisfied.
        \item \textbf{Limit Monotonicity:} First, let us consider the case where $\cwparaof{1} < \cwparaof{m}+2$. Without loss of generality, assume that the first few entries in the worth vector are $\cwparaof{m}+1$ and the remaining are \cwparaof{m}. Consider the rule $\langle \CC_\cwpara,f \rangle$. Suppose $f(S) = \bool(|S|>0)$. If there is any project costing at most $\cwparaof{m}$, it covers all the voters and the claim follows. If $\CC_\cwpara$ selects only projects costing $\cwparaof{m}+1$ into each $A_i$, increasing \bud won't change the outcome. Thus, limit monotonicity is satisfied.

        Now, let us look at the case where $\cwparaof{1} \geq \cwparaof{m}+2$. Suppose we have an instance \instance with $\bud = 2\cwparaof{1}-1$. Let $\curly{p_1,p_2,p_3}$ be projects respectively costing $\curly{\cwparaof{1}-1,\cwparaof{1},\cwparaof{1}}$. Let $\curly{d_1,\ldots,d_{m-3}}$ be projects such that $\cof{d_i} = \cwparaof{i+1}+1$. Let there be six voters such that the ranking of one voter is $p_1 \succ d_1 \succ \ldots \succ d_{m-3} \succ p_2 \succ p_3$, ranking of two voters is $p_2 \succ d_1 \succ \ldots \succ d_{m-3} \succ p_3 \succ p_1$, and ranking of remaining three voters is $p_3 \succ d_1 \succ \ldots \succ d_{m-3} \succ p_2 \succ p_1$. Clearly, $\ruleof{\cwr}{} = \curly{\curly{p_1,p_3}}$. If we increased budget by $1$ to get $\instance'$, then $\ruleof{\cwr}{'} = \curly{\curly{p_2,p_3}}$. Thus, $p_1 \notin \winners{\cwr}{'}$ and limit monotonicity is not satisfied.
        \item \textbf{Inclusion Maximality:} This follows from the fact that $\bool(|S|>0)$ is subset monotone. In other words, utility from a set $S$ is at least as much as the utility from any subset of $S$. Thus, inclusion maximality is satisfied.
        \item \textbf{Candidate Monotoncity:} Utility of a set $S$ is the number of voters who have some project of $S \in A_i$. Utility of any $S_x \in \ruleof{\cwr}{}$ increases by $1$ in $\instance'$ (when $x$ is the only project in $S_x$ in $A_i$) or remains the same, otherwise. Similarly, the utility of any set without $x$ stays the same or decreases by $1$ in $\instance'$. Therefore, $x \in \winners{\cwr}{'}$. Thus, candidate monotonicity is satisfied.
        \item \textbf{Non-Crossing Monotonicity:} First, we look at the case where $\cwparaof{1}=\cwparaof{m}$. Let $\cwparaof{1} = Q$. All the voters approve all and only the projects whose cost is at most $Q$. Hence, the outcome is independent of the ranks of the projects in the preferences and thus, non-crossing monotonicity is satisfied.
        
        Now, consider the case where $\cwparaof{1}>\cwparaof{m}$. This implies that there exists some $t \in \{1,\ldots,m\}$ such that $\cwparaof{t} > \cwparaof{t+1}$. Consider an instance \instance with budget $\bud = \cwparaof{1}+\cwparaof{t}$ and a set of projects $\proj = \curly{p_1,\ldots,p_5,d_1,\ldots,d_{m-5}}$ whose costs are as follows: $\cof{p_1} = \cof{p_4} = \cwparaof{1}$; $\cof{p_2} = \cof{p_3} = \cwparaof{t}$; $\cof{p_5} = \cwparaof{1}+\cwparaof{t}$; for every $i \in \curly{1,\ldots,t-2}$ $\cof{d_i} = \cwparaof{i+1}+1$; for every $i  \in \curly{t-1,\ldots,m-5}$, $\cof{d_i} = \cwparaof{i+3}+1$. Suppose we have three voters whose preferences are (1) $p_1 \succ d_1 \succ \ldots \succ d_{t-2} \succ p_2 \succ p_3 \succ d_{t-1} \succ \ldots \succ d_{m-5} \succ p_5 \succ p_4$ (2) $p_4 \succ d_1 \succ \ldots \succ d_{t-2} \succ p_2 \succ p_3 \succ d_{t-1} \succ \ldots \succ d_{m-5} \succ p_5 \succ p_1$ and (3) $p_5 \succ d_1 \succ \ldots \succ d_{t-2} \succ p_3 \succ p_2 \succ d_{t-1} \succ \ldots \succ d_{m-5} \succ p_4 \succ p_1$. Since $\cwparaof{t} > \cwparaof{t+1}$, Algorithm \ref{algo: ctscheme} gives $A_1 = \curly{p_1,p_2}$, $A_2 = \curly{p_4,p_2}$, $A_3$ is either $\curly{p_5,p_3}$ or $\curly{p_3}$. Hence, the outcome of \cwr with $f(S) = \bool(|S|>0)$ is $\curly{\curly{p_1,p_2},\curly{p_1,p_3},\curly{p_4,p_2},\curly{p_4,p_3}}$ since each of these sets has an utility of exactly $2$. Now, let $S = \curly{p_1,p_3}$. See that $p_2 \notin S$. Exchange $p_2$ and $p_3$ in the first preference to obtain new instance $\instance'$. In $\instance'$, $\curly{p_3,p_4}$ has a strictly better utility than $S$ and hence $S \notin \ruleof{\cwr}{'}$. Thus, non-crossing monotonicity is not satisfied.
        \item \textbf{Pro-Affordability:} Proof is exactly same as that of (a)-(ii) in the proof of \Cref{the: proaffordp}.
    \end{enumerate}
This completes the proof.
\end{proof}

\addtocounter{app}{1}
\section{Flexible Costs: Characterization of Group-Fair and Individual-Fair Rules under Single-Peaked Preferences}\label{chap: fourth_app}
\subsection{Proof of \Cref{le_1}} \label{app4: le_1}
Take a cut $C$. Note that if $C$ contains either $x$ or $y$ then $c(C)\geq 1$. So, assume $C$ does not contain $x$ and $y$. We show that there must exist $\gamma_{i_1},\ldots,\gamma_{i_m}$ such that $\gamma_{i_j}\gg\gamma_{i_{j+1}}$ for all $j\in 1,\ldots,m-1$, and $(\gamma_{i_j},j)\in C$. First note that for each $j\leq m$, there exists some $(\gamma_{i_j},j)\in C$ as otherwise we can construct a deterministic rule $f$ where $f_{rj}(\rprof^{\gamma_i})=1$ for all $\gamma_i\in\Gamma$. Hence, $C$ can be written in the form $$C=\{(\gamma_{1,1},1),\ldots,(\gamma_{1,k_1},1),(\gamma_{2,1},2),\ldots,(\gamma_{2,k_2},2),\ldots,(\gamma_{m,1},m),\ldots,(\gamma_{m,k_m},m)\}.$$

Let us define the deterministic rule $f$ as follows. Let $\gamma \in \Gamma$. If there are $(\gamma_{i_{m-1}},m-1),\ldots,(\gamma_{i_{1}},1)\in C$ such that $\gamma_{i_{1}}\gg\cdots\gg\gamma_{i_{m-1}}\gg\gamma$ then $f(\rprof^\gamma)=p_m$. Otherwise, if there are $(\gamma_{i_{m-2}},m-1),\ldots,(\gamma_{i_{1}},1)\in C$ such that $\gamma_{i_{1}}\gg\cdots\gg\gamma_{i_{m-2}}\gg\gamma$ then $f(\rprof^\gamma)=p_{m-1}$. Continuing in this manner, if there are $(\gamma_{i_{1}},1)\in C$ such that $\gamma_{i_{1}}\gg\gamma$ then $f(\rprof^\gamma)=p_{2}$. Finally, for all the remaining cases, let us define $f(\rprof^\gamma)=p_1$.

It follows from the definition that  $f$ is strategy-proof. Therefore, the path in the network induced by $f$ must intersect $C$. Let the vertex where $f$ intersects the cut be $(\gamma,k)$, that is, $f(\rprof^\gamma)=p_k$. Suppose $k<m$. Then by definition, there are $(\gamma_{i_{k-1}},k-1),\ldots,(\gamma_{i_{1}},1)\in C$ such that $\gamma_{i_{1}}\gg\cdots\gg\gamma_{i_{k-1}}\gg\gamma$. Since $(\gamma,k)\in C$, this implies  $\gamma_{i_{1}}\gg\cdots\gg\gamma_{i_{k-1}}\gg\gamma_{i_{k}}=\gamma$, which in turn means $f(\rprof^\gamma)\neq a_k$. Therefore, it follows that $k=m$, and hence  $\gamma_{i_{1}}\gg\cdots\gg\gamma_{i_{m-1}}\gg\gamma_{i_{m}}=\gamma$. 

Consider  $\gamma_{i_{1}},\cdots,\gamma_{i_{m-1}},\gamma_{i_{m}}$ as defined  in the preceding paragraph. By strategy-proofness,   $\varphi_{p_1}(\rprof^\gamma_{i_{1}})\geq \varphi_{p_1}(\rprof^\gamma_{i_{2}})$. Therefore, $$c(\gamma_{i_{1}},1)+c(\gamma_{i_2},2)\geq \varphi_{p_1}(\rprof^\gamma_{i_{2}})+\varphi_{p_2}(\rprof^\gamma_{i_{2}}).$$ By induction, it follows that  
$$c(\gamma_{i_{1}},1)+\ldots+c(\gamma_{i_{k+1}},k+1)\geq \varphi_{p_1}(\rprof^\gamma_{i_{k+1}})+\ldots+\varphi_{p_{k+1}}(\rprof^\gamma_{i_{k+1}}).$$
Therefore, we have 
\begin{align*}
	c(C)&\geq c(\gamma_{i_{1}},1)+\ldots+c(\gamma_{i_{m}},m)\\
	&=\varphi_{p_1}(\rprof^\gamma_{i_{m}})+\ldots+\varphi_{p_{m}}(\rprof^\gamma_{i_{m}})=1.\\
\end{align*}
This completes the proof.


%

\end{document}